\documentclass[11pt]{article}
\usepackage[authoryear]{natbib}
\usepackage{amssymb}
\usepackage{amsfonts}
\usepackage{amsmath}
\usepackage{dsfont}
\usepackage{soul}
\usepackage[nohead]{geometry}
\usepackage{graphicx}
\usepackage{amsthm}
\usepackage{color}
\usepackage{comment}
\usepackage{setspace}
\usepackage{framed}
\usepackage{enumitem}
\usepackage{todonotes}
\usepackage{tikz}
\usepackage{hyperref}
\usepackage{rotating}
\usepackage{subfigure}
\usepackage[flushleft]{threeparttable}
\usepackage{multirow}
\newcommand{\cref}[2][1]{{\textup{(\hyperref[#2]{\ref*{#2}$_{#1}$})}}}
\usepackage{xr}
\usepackage{bm}
\usepackage[toc,page,header]{appendix}
\usepackage{minitoc}
\usepackage[linesnumbered,ruled]{algorithm2e}

\SetKwInput{KwInput}{Input}                
\SetKwInput{KwOutput}{Output}              

7

\newcommand{\eps}[0]{\ensuremath{\varepsilon}}
\newcommand{\dt}{\delta}

\newcommand{\ap}{\alpha}
\newcommand{\bt}{\beta}
\newcommand{\ld}{\lambda}
\newcommand{\Ld}{\Lambda}
\newcommand{\gm}{\gamma}
\newcommand{\sgn}{\mathrm{sgn}}

\newcommand{\diag}{\mathrm{diag}}

\newcommand{\var}{\mathrm{var}}
\newcommand{\beq}{\begin{eqnarray*}}
\newcommand{\eeq}{\end{eqnarray*}}
\newtheorem{thm}{Theorem}[section]

\newtheorem{lem}{Lemma}[section]

\newtheorem{assum}{Assumption}
\newtheorem{pro}{Proposition}[section]
\numberwithin{equation}{section}
\theoremstyle{definition}

\newtheorem{remark}{Remark}[section]
\makeatletter
\def\@biblabel#1{\hspace*{-\labelsep}}
\makeatother
\DeclareMathOperator*{\argmin}{argmin}
\DeclareMathOperator*{\plim}{plim}

\input{tcilatex}

\begin{document}
\doparttoc 
\faketableofcontents 


\title{Factor-Driven Two-Regime Regression
\thanks{We would like to thank Don Andrews, Mehmet Caner, Greg   Cox, Bruce Hansen, Zhongjun Qu  and the seminar participants at BU, Emory, Michigan State, NYU, Wisconsin-Madison, Northwestern, Yale, and 2018 ASSA Winter Meeting for helpful comments. We would like to thank 
the Ministry of Education of the Republic of Korea and the National Research Foundation of Korea (NRF-0405-20180026), 
the Social Sciences and Humanities Research Council of Canada (SSHRC-435-2018-0275), 
the European Research Council for financial support (ERC-2014-CoG-646917-ROMIA) and 
the UK Economic and Social Research Council for research grant (ES/P008909/1) to the CeMMAP.}
}
\date{August 13, 2020}

\author{
Sokbae  Lee\thanks{%
Address: 420 West 118th Street,  New York, NY 10027, USA. E-mail: \texttt{sl3841@columbia.edu}.} \\ \footnotesize Columbia University  \and
Yuan Liao\thanks{Address: 75 Hamilton St., New Brunswick, NJ 08901, USA. Email:
\texttt{yuan.liao@rutgers.edu}.}\\  \footnotesize   Rutgers University 
\and Myung Hwan Seo\thanks{%
Address: 1 Gwanak-ro, Gwanak-gu, Seoul 08826, Korea. E-mail: 
\texttt{myunghseo@snu.ac.kr}.} \\  \footnotesize  Seoul National University\and
Youngki  Shin\thanks{Address: 1280 Main St.\ W.,\ Hamiloton, ON L8S 4L8, Canada. Email:
\texttt{shiny11@mcmaster.ca}.}\\  \footnotesize   McMaster University
}

\maketitle

\begin{abstract}
 We propose a novel two-regime regression model where regime switching  is driven by a vector of possibly unobservable factors. When the factors are latent, we estimate them by the principal component analysis of a   panel data set.
 We show that the optimization problem can be reformulated as  mixed integer optimization, and we present two alternative computational algorithms. We derive the asymptotic distribution of the resulting estimator under the scheme that the threshold effect shrinks to zero.  In particular,  we establish a phase transition that describes the effect of first-stage factor estimation as the cross-sectional dimension of panel data increases relative to the time-series dimension.
 Moreover, we develop bootstrap inference and illustrate our methods via numerical studies.
\end{abstract}

\thispagestyle{empty}



\onehalfspacing

\newpage
\setcounter{page}{1}
\pagenumbering{arabic}

\section{Introduction}

Suppose that  $y_t$ is generated from
 \begin{align}
 	y_{t}                                             & = x_{t}^{\prime }\beta _{0}+x_{t}^{\prime }\delta_{0}1\{f_t'\gamma_0>0\}+\varepsilon _{t},  \label{model1} \\
 	\mathbb{E}\left( \varepsilon _{t}|\mathcal{F}_{t-1}\right) & = 0,  \; t =1, \ldots, T, \label{model2}
 \end{align}
where $x_{t}\ $and $f_{t}$ are adapted to the filtration $\mathcal{F}_{t-1}$,
$(\beta_0, \delta_0, \gamma_0)$ is a vector of unknown parameters,
and the unobserved random variable $\varepsilon _{t}$ satisfies the conditional mean restriction
in \eqref{model2}. We interpret $f_t$ to be a vector of  {factors} determining regime switching.   When  $f_t'\gamma_0 > 0$, the regression function becomes
$x_{t}^{\prime }(\beta _{0}+\delta_0)$;  if $f_t'\gamma_0 \leq 0$, it reduces to $x_{t}^{\prime }\beta _{0}$.
 We allow for either observable or unobservable factors.
For the latter, we assume that they can be recovered from a   panel data set.
In light of this feature, we call the  model in \eqref{model1} and \eqref{model2}
a \emph{factor-driven two-regime regression model}.

Our paper is closely related to the  literature on threshold models with unknown change points
(see, e.g., \citep{chan1993consistency}, \citep{hansen2000sample}, \citep{ling_1999},
\citep{Seijo:Sen:11a}, \citep{Seo-Linton}, and \citep{tong1990non},  among many others).
In the conventional threshold regression model,
an intercept term and a scalar observed random variable constitute  $f_t$.
For instance,  Chan \citep{chan1993consistency} and Hansen \citep{hansen2000sample} studied the model in which $1\{f_t'\gamma_0>0\}$ in (\ref{model1})  is replaced by $1\{q_t> \widetilde{\gamma }_0 \}$ for some observable scalar variable $q_t$ with a scalar unknown parameter $\widetilde{\gamma}_0$.
In practice,
it might be controversial to choose which observed variable plays the role of $q_t$. For example, if the two different regimes represent
the status of two   environments of the population,  arguably it is difficult to assume that the change of the environment is governed by just a single variable.
On the contrary, our   proposed  model introduces a regime change due to a single index of factors that can be ``learned" from a potentially much larger dataset. Specifically, we consider the   framework of latent approximate factor models in order  to model a regime switch based on a potentially  large number of covariates.

In view of  the conditional mean restriction
in \eqref{model2}, a natural strategy to estimate $(\beta_0, \delta_0, \gamma_0)$ is to rely on least squares.
A least-squares estimator for our model brings new challenges  in terms of both computation and asymptotic theory.
First of all, when the dimension of $f_t$ is larger than 2, it is computationally demanding to estimate $(\beta_0, \delta_0, \gamma_0)$. We overcome this difficulty by developing new computational algorithms based on the method of mixed integer optimization (MIO). 
See, for example,  section~2.1 in Bertsimas et al.\ \citep{bertsimas2016} for a discussion on computational advances in solving the MIO problems.


Second, we establish asymptotic properties of our proposed estimator  by adopting  a diminishing thresholding effect.
That is, we assume that  $\delta_0=T^{-\varphi}d_0$ for some unknown $\varphi \in (0, 1/2)$ and unknown non-diminishing vector $d_0$.
The  diminishing threshold has been one of the standard frameworks in the change point literature (e.g.,  \citep{bai1994least,hawkins1986simple,horvath1997effect}).
The unknown parameter $ \varphi $ reflects the difficulty of estimating $ \gamma_0 $ and affects the identification and estimation of the change-point $ \gamma_0 $.
Both the rate of convergence and the asymptotic distribution depend  on $\varphi$.
This is a widely employed tool to allow for flexible signal strengths of the parameters in the nonlinear model.
For instance, McKeague and Sen \citep{mckeague2010fractals} studied a  ``{point impact}" linear model, where the  identification and estimation of $\gamma_0$ are affected by an unknown slope  $\delta_0$.
While specifically assuming $\delta_0\neq0$, they encountered a similar parameter $\varphi$, reflecting the difficulty of estimating $\gamma_0$.
The asymptotic theory for the estimated $\delta_0$ under the diminishing jump setting is fundamentally different from the fixed jump setting:  the former is determined by a Gaussian process (e.g., \citep{hansen2000sample}), and the latter  by a compound poison process (e.g., \citep{chan1993consistency}). While both settings lead to important asymptotic implications, we focus on the diminishing setting  because when the factors are estimated, there is a new and interesting \emph{phase transition} phenomenon  that smoothly  appears in  the ``bias'' term of the Gaussian process. The phase transition characterizes the continuous change of the asymptotic distribution as the precision of the estimated factors increases relative to the size of the jump, which we shall detail below.

When the factor $f_t$ is  latent, we   estimate it using  principal component analysis (PCA) from a potentially much larger dataset, whose dimension is $N$.    It turns out that the asymptotic distribution for the estimator of  $\alpha_0 \equiv (\beta_0',\delta_0')'$ is identical to that when $\gamma_0$ were known, regardless of  whether factors are  directly observable or not; therefore, the estimator of $\alpha_0$ enjoys  an oracle property.

 The issue is more sophisticated for the distribution of the estimator of $\gamma_0$. When factors are directly observable, we prove that
  \begin{align*}
 &  T^{1-2\varphi }
 \left( \widehat{\gamma }-\gamma _{0}\right)\overset{d}{\longrightarrow }%
 \argmin_{g\in \mathcal{G}}B(g) +2W\left( g\right) ,
 \end{align*}%
 where $B(g)$  represents a ``drift function" of the criterion function, which is linear with  a kink at zero,
  $W(g)$ is a mean-zero Gaussian process
  and $\mathcal{G}$ is a  rescaled parameter space.  However, when factors are not directly observable,  the estimation error from the PCA plays an essential role and may slow down the rates of convergence, depending on the relation between $N$ and $T$. Specifically,  we show that
  \begin{align*}
 & \left( \left( NT^{1-2\varphi }\right) ^{1/3}\wedge T^{1-2\varphi }\right)
 \left( \widehat{\gamma }-\gamma _{0}\right)\overset{d}{\longrightarrow }%
 \argmin_{g\in \mathcal{G}}A\left( \omega, g\right) +2W\left( g\right),
 \end{align*}%
 with a new drift function $A\left( \omega, g\right)$ that depends on  $\omega=\lim \sqrt{N} T^{-(1-2\varphi)}\in[0,\infty]$. On one hand, when $\omega=\infty$,  we find that $A(\omega, g)= B(g)$, so the limiting distribution becomes the same as if the factors were observable. This case
corresponds  to
 the \emph{super-consistency rate} (e.g., \citep{hansen2000sample}).
 On the other hand,  when $\omega=0$, it turns out that $A\left( \omega, g\right)$ is  quadratic in $g$,  corresponding to  a  \emph{cube root rate} similar to the maximum score estimator (e.g., \citep{kim1990, seo2018local}).
Furthermore, both the drift function and the resulting rates of convergence have continuous transitions  as $\omega$ changes between
$0$ and $\infty$.
Therefore, one of our key  findings for  the estimator of $\gamma_0$  is the occurrence of a {phase transition} from a \emph{weak-oracle} limiting distribution to a \emph{semi-strong} oracle one, and then to a \emph{strong} oracle one as $\omega$ increases.

As the asymptotic distribution of $\widehat\gamma$ is non-pivotal, we propose a wild bootstrap  for inference of $\gamma_0$.  Importantly,  we construct  bootstrap confidence intervals for $\gamma_0$ that  do not require knowledge of $\varphi$.
  This facilitates applications in which the  jump  diminishing speed  is not known in advance.



The remainder of the paper is organized as follows.
In Section \ref{sec:opt}, we propose the least-squares estimator  and algorithms to compute the proposed estimator.
In Section \ref{sec:known factors}, we establish asymptotic theory when $f_t$ is directly observed.
In Section \ref{model:est:factors}, we consider estimation when $f_t$ is a vector of latent factors, we propose a two-step estimator via PCA, and we analyze asymptotic properties of our proposed estimator. In Section \ref{sec:inference}, we develop  bootstrap inference, and in Section \ref{sec:MC} we give the results of Monte Carlo experiments. In Section \ref{sec:real-data-app}, we illustrate our methods by applying them to threshold autoregressive models of unemployment. We conclude in Section \ref{sec:conclusions}. The online appendices provides details that are omitted from the main text.

%
The notation used in the paper is as follows. The sample size is denoted by $T$ and the transpose of a matrix is denoted by  a prime.
The true parameter is denoted by the subscript $0$, whereas a generic element has no subscript. 
The Euclidean norm is denoted by $| \cdot |_2$,
the Frobenius norm of  a matrix  is $| \cdot |_F$, the spectral norm of a matrix is  $|\cdot|_2$,
and the $\ell_0$-norm is $|\cdot |_0$.
For a generic random variable or vector $z_{t}$, let its density function be
denoted by $p_{z_{t}}$. Similarly, let
$p_{y_t|x_t}(y)$ denote the conditional density of $y_t$  given $x_t$ for the random vectors $y_t  $ and $x_t$.
The abbreviation \emph{a.s.} means almost surely.



\section{Least-Squares Estimator via Mixed Integer Optimization}\label{sec:opt}

\subsection{Identifiability}
 We use the convention that the constant $ 1 $ is the first element of $ x_t $ and $ -1 $ is the last element of $ f_t $. Define $\alpha :=(\beta ^{\prime },\delta ^{\prime
})^{\prime }$ and $Z_{t}(\gamma ):=(x_{t}^{\prime },x_{t}^{\prime }1\{f_{t}^{\prime
}\gamma >0\})^{\prime }$.   Then, we can rewrite the model as
\begin{equation*}
y_{t}=Z_{t}\left( \gamma _{0}\right) ^{\prime }\alpha _{0}+\varepsilon _{t}.
\end{equation*}
Because only the sign of the index $f_t'\gamma_0$ determines the regime switching,
the scale of $\gamma_0$ is not identifiable.
We  assume that the first element of $\gamma_0$ equals 1.
Let $d_x$ and $d_f$ denote  the dimensions of $x_t$ and  $f_t$, respectively.

\begin{assum}
\label{scale-normalization}
 $\alpha_0 \in \mathbb{R}^{2d_x}$  and
 $\gamma_0
\in \Gamma :=
\{ (1, \gamma_2')':  \gamma_2 \in \Gamma_2  \}$, where
$\Gamma_2 \subset \mathbb{R}^{d_f-1}$ is a compact set.
\end{assum}

We decompose $f_{t}$ into a scalar random variable $f_{1t}$ and other variables $f_{2t}$, so that $f_{t}^{\prime}\gamma  \equiv f_{1t} + f_{2t}^{\prime} \gamma_2$.
In view of the conditional mean zero restriction in \eqref{model2}, it is natural to impose
conditions under which  both $\alpha _{0}$ and $\gamma _{0}$ are
identified by the $L_{2}$-loss. Introduce the excess loss%
\begin{align}\label{eq:excess}
R(\alpha ,\gamma ) := \mathbb E(y_{t}-x_{t}^{\prime }\beta -x_{t}^{\prime }\delta
1\{f_{t}^{\prime }\gamma >0\})^{2}- \mathbb E( \varepsilon _{t}^{2} ).
\end{align}%
In order to establish that
$R\left( \alpha,\gamma \right) >
R\left( \alpha _{0},\gamma _{0}\right) =0$ whenever
$(\alpha,\gamma) \neq (\alpha_0,\gamma_0)$,
we make the following regularity conditions.

\begin{assum}
\label{iden-assump}
For any $\varepsilon >0$,
$\left( \alpha_{0},\gamma _{0}\right) $ satisfies
\begin{equation*}
\inf_{\{(\alpha', \gamma')' \in \mathbb{R}^{2d_x} \times \Gamma: |(\alpha', \gamma') - (\alpha_0', \gamma_0')|_2 > \varepsilon \} }
R\left( \alpha,\gamma \right) >0.
\end{equation*}%
\end{assum}

Online Appendix \ref{sec:identification} provides sufficient conditions for Assumption \ref{iden-assump}.

\subsection{Estimator}
We now propose the least-squares estimator and  two alternative algorithms to compute the proposed estimator.
For  computational purposes, we assume that $\alpha \in \mathcal{A} \subset \mathbb{R}^{2d_x}$ for some known compact set $\mathcal{A}$.
In practice, we can take a large $2 d_x$-dimensional hyper-rectangle so that the resulting estimator is not on the boundary of  $\mathcal{A}$.
The unknown parameters can be estimated by  least squares:
 $\left( \widehat{\alpha},\widehat{\gamma}\right)$ solves
 \begin{align}
\min_{(\alpha', \gamma')' \in \mathcal{A} \times \Gamma }%
&\mathbb{S}_{T}\left( \alpha ,\gamma \right) \equiv \frac{1}{T}%
\sum_{t=1}^{T}(y_{t}-x_{t}^{\prime }\beta -x_{t}^{\prime }\delta
1\{f_{t}^{\prime }\gamma >0\})^{2} \label{original-prob}
\\
&\text{subject to:}\quad \tau_1 \leq \frac{1}{T} \sum_{t=1}^T 1\{f_{t}^{\prime }\gamma >0\} \leq \tau_2.
\label{restrction:para:original-form}
\end{align}
We assume that  the   restriction (\ref{restrction:para:original-form}) is satisfied when $\gamma=\gamma_0$ \emph{a.s.}
 Here, $0 < \tau_1 < \tau_2 < 1$ for some predetermined $\tau_1$ and $\tau_2$ (e.g., $\tau_1 = 0.05$ and $\tau_2 = 0.95$).
 In the special case that  $1\{f_{t}^{\prime }\gamma _{0}>0\} = 1\{q_t > \widetilde{\gamma}_0\}$ with a scalar variable $q_t$ and a parameter $\widetilde{\gamma}_0$, it is standard to assume that the parameter space for $\widetilde{\gamma}_0$ is between the $\tau$ and $(1-\tau)$ quantiles of $q_t$ for some known $0 < \tau < 1$.  We can interpret \eqref{restrction:para:original-form} as a natural generalization of this  restriction so that the proportion of one regime is never too close to 0 or 1.


When
$\gamma $ is of high dimension, the naive  grid search  would not work well.
Dynamic programming (e.g., \citep{Bai:Perron:2003})
or
smooth global optimization  (e.g., \citep{Qu:Tkachenko:17})
might be considered but are
not readily available. 
We overcome this computational difficulty by replacing the naive grid search with MIO.
We present two alternative algorithms based on MIO below.

\subsection{Mixed Integer Quadratic Programming}\label{sec:joint:estimation}

Our first algorithm is based on mixed integer quadratic programming (MIQP), which jointly estimates $(\alpha,\gamma )$. It is guaranteed to obtain a global solution once it is found. To write the original least-squares problem in MIQP,  we introduce
 $d_t:=1\{f_t'\gamma>0\} $ and $\ell_{t} := \delta d_t$ for $ t = 1,\ldots,T$.
 Then, rewrite the objective function as
\begin{equation} \label{prob2}
\frac{1}{T}\sum_{t=1}^{T}
(y_{t}-x_{t}'{\beta}-  x_t'\ell_{t})^{2},
\end{equation}
 which is a quadratic function of $\beta$ and $\ell_t$.
 The goal is to introduce only linear constraints with respect to variables of optimization, and to construct an MIQP that is equivalent to the original least-squares problem. Then, we can apply  modern MIO packages (e.g., Gurobi) to solve MIQP.
The assumption
 $\alpha \in \mathcal{A}$ implies that
there exist  known upper and lower bounds for $\delta_{j}$:  $L_j \leq \delta_{j} \leq  U_j$,
where $\delta_{j}$ denotes the $j$th element of $\delta$ for $j =1,\ldots,d_x$.
In addition, to make sure that $\ell_{j,t} = \delta_{j} d_t$ for each $j$ and $t$, we impose two additional restrictions:
\begin{align}\label{eq:bilinear:d:L}
d_t L_j \leq \ell_{j,t} \leq d_t U_j
\ \ \text{ and } \ \
L_{j} (1-d_t) \leq \delta_{j} - \ell_{j,t} \leq U_{j} (1-d_t).
\end{align}
  It is then straightforward to check that these constraints imply $\ell_{j, t} =\delta_j d_t$.  To introduce another key constraint,  we define
$
M_t \equiv \max_{\gamma \in \Gamma} | f_t' \gamma |
$
for each $t=1,\ldots,T$, where $\Gamma$ is the parameter space for $\gamma_0$.
We can compute $M_t$ easily for each $t$ using linear programming.
 We store them as inputs to our algorithm.  The following new constraints along with
 \eqref{restrction:para:original-form} and
 \eqref{eq:bilinear:d:L} ensure that
 the reformulated problem  \eqref{prob2} is the same as  the original problem:
$$
 (d_t - 1) (M_t + \epsilon) < f_t' \gamma \leq d_t M_t,
 $$
where $\epsilon > 0$ is a small predetermined constant (e.g., $\epsilon = 10^{-6}$).
The following defines an algorithm for the MIQP algorithm.

\begin{algorithm}[h]
 \KwInput{$\{(y_t, x_t, f_t, M_t): t=1,\ldots,T \}$}
\KwOutput{$\left( \widehat{\alpha},\widehat{\gamma}\right)$}
Let $\bm{d} =  (d_1, \ldots, d_T)'$ and $\bm{\ell} = \{\ell_{j,t}: j =1,\ldots,d_x, t = 1,\ldots,T \}$, where
 $\ell_{j,t}$ is a real-valued variable. Solve the following problem:
\begin{align}\label{prob3}
\min_{\beta, \delta, \gamma, \bm{d}, \bm{\ell} }
\mathbb{Q}_{T}\left( \beta ,\boldsymbol{\ell }\right) \equiv
\frac{1}{T}\sum_{t=1}^{T}
(y_{t}-x_{t}'{\beta}- \sum_{j=1}^{d_x} x_{j,t} \ell_{j,t} )^{2}
\end{align}
subject to
\begin{align}\label{main-constraints}
\begin{split}
& (\beta, \delta) \in \mathcal{A}, \; \gamma \in \Gamma,  \; d_t \in \{0, 1\}, \; L_j \leq \delta_{j} \leq  U_j, \\
& (d_t - 1) (M_t + \epsilon) < f_t' \gamma \leq d_t M_t, \\
& d_t L_j \leq \ell_{j,t} \leq d_t U_j, \\
& L_{j} (1-d_t) \leq \delta_{j} - \ell_{j,t} \leq U_{j} (1-d_t), \\
& \tau_1 \leq T^{-1} \sum_{t=1}^T d_t \leq \tau_2
\end{split}
\end{align}
for each $t=1,\ldots,T$ and each $j=1,\ldots,d_x$, where  $0 < \tau_1 < \tau_2 < 1$.
\caption{Mixed Integer Quadratic Programming (MIQP)}
\end{algorithm}

Our proposed algorithm  is mathematically equivalent to the original least-squares problem  \eqref{original-prob} subject to \eqref{restrction:para:original-form} in terms of values of objective functions. Formally, we state it as the following theorem.

\begin{thm}\label{thm:computation:joint}
Let $(\bar{\alpha},\bar{\gamma})$ denote a solution using MIQP as described above.
Then, $\mathbb{S}_{T}\left( \widehat{\alpha},\widehat{\gamma}\right) =\mathbb{S}_{T}\left(
\bar{\alpha},\bar{\gamma}\right) $, where
$(\widehat{\alpha},\widehat{\gamma})$ is defined in \eqref{original-prob}.
\end{thm}

The proposed algorithm in Section \ref{sec:joint:estimation} may run slowly when the dimension $d_x$ of $x_t$ is large. To mitigate this problem, we
 reformulate MIQP in Appendix  \ref{sec:joint:estimation:appendix} and use  the alternative  formulation in our numerical work; however, we present a simpler form here to help readers follow our basic ideas more easily.

\subsection{Block Coordinate Descent}

\begin{algorithm}[h!tb]
 \KwInput{$\{(y_t, x_t, f_t, M_t): t=1,\ldots,T \}$, \texttt{MaxTime\_1}, \texttt{MaxTime\_2}}
\KwOutput{$\left( \widehat{\alpha},\widehat{\gamma}\right)$}
Set $k = 1$\;
Step 1. Obtain an initial estimate
$\left( \widehat{\alpha}^{0},\widehat{\gamma}^{0}\right)$
using MIQP with the pre-specified  time limit \texttt{MaxTime\_1}\;

\If{a solution is found before reaching  \texttt{MaxTime\_1},}{
      set the initial estimate as the final estimate and terminate\;
   }

\While{elapsed time is no greater than  \texttt{MaxTime\_2}}{

	 Step 2. For the given $\widehat{\alpha}^{k-1},$ obtain an estimate $\widehat{\gamma}%
	^{k}$ via MILP:%
\begin{align}\label{prob3-local}
\min_{\gamma \in \Gamma, d_1,\ldots,d_T}\frac{1}{T}\sum_{t=1}^{T}
\left\{  (x_{t}'\widehat{\delta}^{k-1} )^2
- 2(y_{t}- x_{t}'\widehat{\beta}^{k-1}) x_t'\widehat{\delta}^{k-1} \right\} d_t
\end{align}
		subject to
	\begin{align}\label{main-constraints-iterative}
	\begin{split}
	& (d_{t}-1)(M_{t}+\epsilon )<f_{t}^{\prime }\gamma \leq d_{t}M_{t}, \\
	& d_{t}\in \{0,1\} \ \ \text{for each $t =1,\ldots,T$}, \\
	& \tau_1 \leq \frac{1}{T} \sum_{t=1}^T d_t \leq \tau_2\text{\;}
	\end{split}
	\end{align}

\If{$\mathbb{S}_{T}\left( \widehat{\alpha}^{k-1} ,\widehat{\gamma}^{k} \right)
 \geq \mathbb{S}_{T}\left( \widehat{\alpha}^{k-1} ,\widehat{\gamma}^{k-1} \right)$,
}
{terminate\;}

Step 3. For the given $\widehat{\gamma}^{k},$ obtain
		\begin{equation*}
	\widehat{\alpha}^{k}= \left[\frac{1}{T}\sum_{t=1}^T Z_{t}\left( \widehat{\gamma}^{k} \right)Z_{t}\left( \widehat{\gamma}^{k} \right)' \right]^{-1} \frac{1}{T}\sum_{t=1}^T Z_{t}\left( \widehat{\gamma}^{k} \right) y_t\text{\;}
	\end{equation*}

 Let $k=k+1$\;

}

 \caption{Block Coordinate Descent (BCD)}\label{algo:BCD}
\end{algorithm}


While the MIQP jointly estimates $(\alpha,\gamma )$ and aims at obtaining a  global solution,  it might not compute  as fast as necessary in large-scale problems.   To mitigate the issue of scalability, we introduce a faster alternative approach based on mixed integer linear programming (MILP), whose objective function is linear in $d_t$. The algorithm  solves for $\alpha$ and $\gamma$ iteratively, which we call a block coordinate descent (BCD) algorithm, starting with an initial value that can be obtained through  MIQP with an early stopping rule. At step $k$,  given  $\widehat\alpha^{k-1}$, which is obtained in the previous step, we estimate $\gamma$ by solving
  \begin{equation}\label{eq3.6}
	\min_{\gamma \in \Gamma, d_{1},\ldots ,d_{T}}\frac{1}{T}\sum_{t=1}^{T}\left(
	y_{t}-x_{t}^{\prime }\widehat{\beta}^{k-1}-x_{t}^{\prime }\widehat{\delta}%
	^{k-1}d_{t}\right) ^{2}
	\end{equation}
subject to  similar  constraints as in MIQP. Note that   the least-squares problem (\ref{eq3.6})  is linear in $d_t$ as  $d_t^2=d_t$. The BCD  algorithm is defined in Algorithm \ref{algo:BCD}. Intuitively speaking, it runs  the MIQP algorithm for the amount of time \texttt{MaxTime\_1}, then switches to the MILP for the amount of time \texttt{MaxTime\_2}.
 The BCD approach is a descent algorithm in the sense that
the least-squares objective function is a non-increasing function of $k$.
In other words, BCD in Steps 2 and 3 can provide a higher-quality solution than MIQP with an early stopping rule \texttt{MaxTime\_1}.
The time limit \texttt{MaxTime\_2} in Step 2 can be smaller than \texttt{MaxTime\_1} as
it is easier to solve an MILP problem than to solve an MIQP problem.
Furthermore, the alternative minimization approach efficiently solves for $\widehat\alpha^k$ because it has
 an explicit solution.

\begin{figure}[h]
\centering
\caption{Computation Example of MIQP and BCD}\label{fig-comp_ex}
\includegraphics[width=7cm, height=7cm]{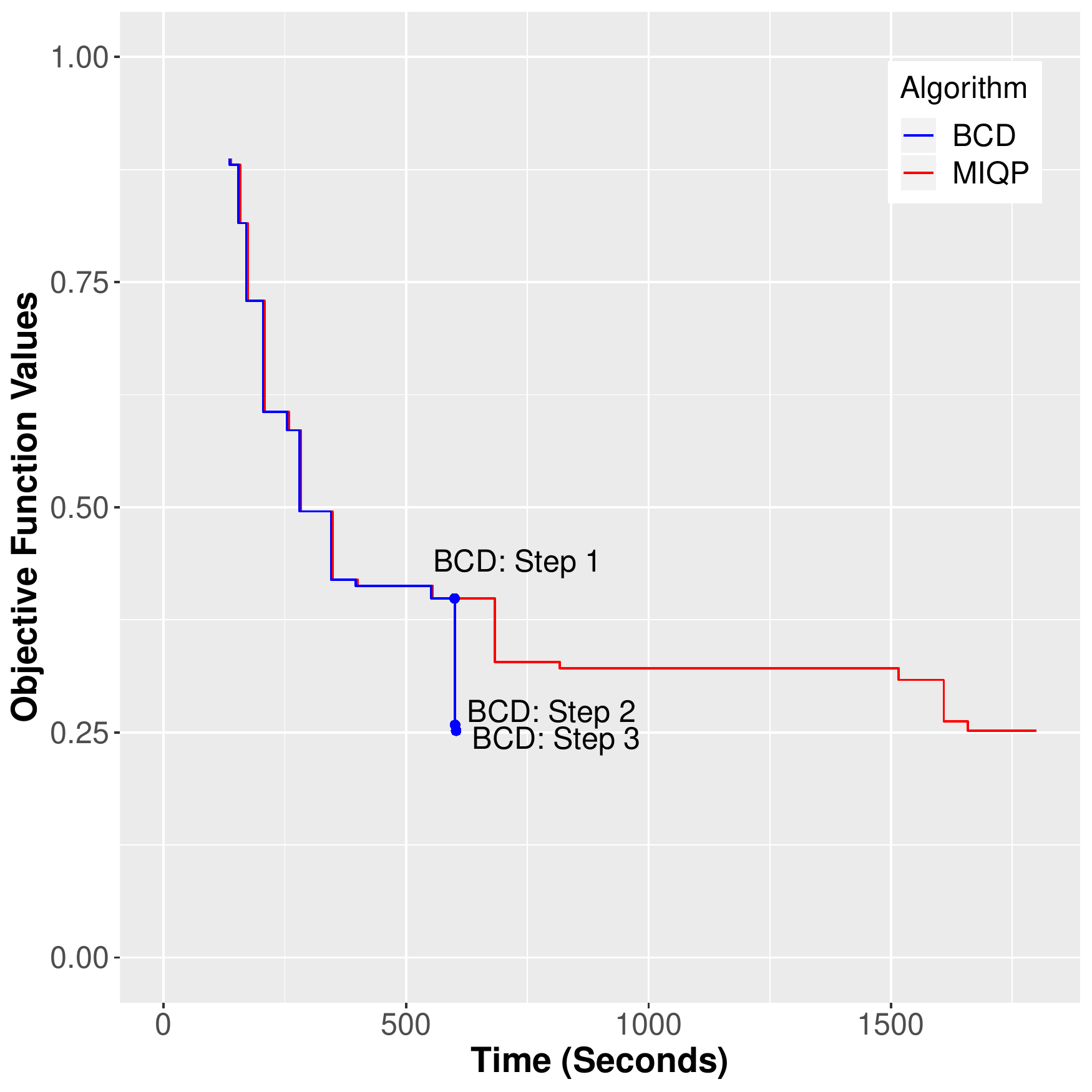}
\end{figure}

Figure \ref{fig-comp_ex} illustrates the performance of MIQP and BCD in one simulation draw. After spending \texttt{MaxTime\_1} (600 seconds) in Step 1, BCD switches into Step 2 and it converges to the solution quickly just in one iteration. Meanwhile, MIQP achieves a similar objective function value after spending the whole time budget of 1800 seconds.
In Monte Carlo experiments, we compare MIQP with BCD more thoroughly, subject to the same total computing time restrictions, and we demonstrate the efficiency of BCD.

\section{Asymptotic Properties with Known Factors}\label{sec:known factors}


We split the asymptotic properties of the estimator into two cases:  known and unknown factors.
In this section, we consider the former.

\begin{assum}
	\label{A-mixing}
	\begin{enumerate}[label=(\roman*)]
		\item\label{A-mixing:itm1}
		$\left\{ x_{t},f_{t},\varepsilon
		_{t}\right\} $ is a sequence of strictly stationary, ergodic, and $\rho $%
		-mixing random vectors with $\sum_{m=1}^{\infty }\rho _{m}^{1/2}<\infty $, $ \mathbb E\left\vert
		x_{t}\right\vert_2 ^{4}<\infty $, and  there exists a constant $C < \infty$ such that  $ \mathbb  E (
		\left\vert x_{t}\right\vert_2 ^{8} \big| f_{t}'\gamma=0 ) < C$ and  $ \mathbb E (
		\varepsilon_{t}^{8} \big| f_{t}'\gamma=0 ) <C$  for all $\gamma \in \Gamma$.


		\item\label{A-mixing:itm2}
		$\left\{ \varepsilon _{t}\right\}$ is a martingale
		difference sequence, that is, $\mathbb{E}\left( \varepsilon _{t}|\mathcal{F}_{t-1}\right)  = 0$,
where $x_{t}\ $and $f_{t}$ are adapted to the filtration $\mathcal{F}_{t-1}$.

		\item\label{A-mixing:itm3}
		The smallest eigenvalue of $ \mathbb E [  Z_{t}\left( \gamma \right)  Z_{t}\left( \gamma \right)^{\prime } ]$ is bounded away from zero for all $\gamma \in \Gamma$.

	\end{enumerate}
\end{assum}

We decompose $f_{t}$ into a scalar random variable $f_{1t}$ and the other variables $f_{2t}$, so that $f_{t}^{\prime}\gamma  \equiv f_{1t} + f_{2t}^{\prime} \gamma_2$.
Define $u_t := f_t'\gamma_{0}$.

\begin{assum}
	\label{A:diminishing dt}
	\begin{enumerate}[label=(\roman*)]
		\item \label{A:dim:itm1} For some $0<\varphi <1/2\ $and $d_{0}\neq 0,\ $%
		$\delta _{0}=d_{0}T^{-\varphi }$.

		\item \label{A:dim:itm2}	$p_{u_t|f_{2t}}(u) $,
		$\mathbb  E [ ( x_{t}^{\prime }d_{0})^{2}|f_{2t},u_t=u ]$ and
		$ \mathbb E [ ( \varepsilon_{t}x_{t}^{\prime }d_{0} ) ^{2}|f_{2t},u_t=u ]$ are continuous and bounded away from zero
		at $ u=0 $ a.s.

		\item \label{A:dim:itm3} For some $M<\infty $,
		$
		\inf_{\left\vert r\right\vert_2 =1}\mathbb{E}\left( \left\vert f_{2t}^{\prime
		}r\right\vert 1\left\{ \left\vert f_{2t}\right\vert_2 \leq M\right\}
		\right) >0. $
	\end{enumerate}

\end{assum}

Most of the conditions in Assumptions \ref{A-mixing} and \ref{A:diminishing dt} are a natural extension of the scalar case   in the literature, when  $ f_t = (q_t,-1)' $ for a  scalar random variable (e.g., \citep{hansen2000sample}).
 Assumption \ref{A:diminishing dt}\ref{A:dim:itm3} is a rank condition on $ f_{2t} $ due to the vector of threshold parameter to be estimated and it is in terms of the first moment because of the asymptotic linear approximation of criterion function near $ \gamma_0 $. It also allows for discrete variables in $ f_{2t} $.
Assumption \ref{A:diminishing dt}\ref{A:dim:itm2}
 ensures the presence of a jump, not just a kink at the change point.

\begin{thm}\label{asdist-alpha-gamma}
Let
$\mathcal{G} := \{ g \in \mathbb{R}^{d_f}: g_1 = 0 \}$.
Let Assumptions  \ref{scale-normalization}, \ref{iden-assump}, \ref{A-mixing}, and \ref{A:diminishing dt} hold.
Assume further that
$\alpha_0$ is in the interior of  $\mathcal{A}$ and    that
$\gamma_0$ is  in the interior of $\Gamma$.
In addition, let  $W$ denote a mean-zero Gaussian process whose covariance kernel is given by \begin{equation}\label{H-Gaussian-process-covariance-kernel}
H\left( s,g\right) :=\frac{1}{2}  \mathbb E \left[ \left(
\varepsilon _{t}x_{t}^{\prime }d_{0}\right) ^{2}\left( \left\vert
f_{t}^{\prime }g\right\vert +\left\vert f_{t}^{\prime }s\right\vert
-\left\vert f_{t}^{\prime }\left( g-s\right) \right\vert \right) p_{u_t|f_{2t}}(0) \right].
\end{equation}
Then, as $T \rightarrow \infty$, we have
\begin{align*}
\sqrt{T}(\widehat\alpha-\alpha_0) &\overset{d}{\longrightarrow } \mathcal N(0, (  \mathbb EZ_t(\gamma_0)Z_t(\gamma_0)')^{-1}\var(Z_t(\gamma_0)\varepsilon_t) (   \mathbb EZ_t(\gamma_0)Z_t(\gamma_0)')^{-1}  ), \\
T^{1-2\varphi }\left( \widehat{\gamma}-\gamma _{0}\right) &\overset{d}{%
\longrightarrow }\argmin_{g \in \mathcal{G}}
\left\{
 \mathbb E\left[ \left( x_{t}^{\prime }d_{0}\right) ^{2}\left\vert
f_{t}^{\prime }g\right\vert p_{u_t|f_{2t}}(0) \right] +2W\left( g\right)
\right\},
\end{align*}
where $\sqrt{T}(\widehat\alpha-\alpha_0)$ and $T^{1-2\varphi }\left( \widehat{\gamma}-\gamma _{0}\right)$ are asymptotically independent.
\end{thm}

The normalization scheme is embedded in the asymptotic distribution.
Because $\gamma _{1}=1$, the minimum in the limit is taken after
fixing the first element of $g$ at zero (recall that $\mathcal{G} = \{ g \in \mathbb{R}^{d_f}: g_1 = 0 \}$).  Also note that, in the scalar  threshold case,  $f_{t}=\left( q_{t},-1\right) ^{\prime }$ and  $\gamma_0 = (1, \widetilde{\gamma}_0)'$,
$$
H(s,g)= \frac{1}{2}  \mathbb E \left[ \left(
\varepsilon _{t}x_{t}^{\prime }d_{0}\right) ^{2}\left(  2\min\left(\left|g_{2}\right|,\left|s_{2}\right|\right)1\left\{ \sgn\left(g_{2}\right)=\sgn\left(s_{2}\right)\right\} \right) p_{u_t|f_{2t}}(0) \right],
$$
which  becomes the two-sided
Brownian motion,  as in   Hansen \citep{hansen2000sample}.



\section{Estimation with Unobserved Factors}\label{model:est:factors}

In this section, we consider the case in which the factors are estimated.

\subsection{The Model}


Consider the following factor model,
\begin{align}\label{factor-model-reg}
\mathcal Y_t =\Lambda g_{1t}+  e_t, \; t=1,\ldots, T,
\end{align}
where $\mathcal Y_t$ is an $N \times 1$ vector of time series,
$\Lambda$ is an $N\times K$ matrix of factor loadings, $g_{1t}$ is a $K \times 1$ vector of common factors, and $e_t$ is an $N \times 1$ vector of idiosyncratic components.
Throughout this section,
we make it explicit that there is a constant term in  the factors, and  we
replace the regression model in \eqref{model1} with
\begin{align}\label{model1-est-factor}
y_t=x_t'\beta_0+x_t'\delta_01\{g_{t}'\phi_{0}  > 0 \} + \varepsilon_t,
\end{align}
where $g_t = (g_{1t}', -1)'$ is a vector of unknown factors in \eqref{factor-model-reg} plus a constant term ($-1$), and $\phi_0$ is a vector of unknown parameters.  In addition, we allow $g_{1t}$ to contain lagged (dynamic) factors, but we treat them as static factors and estimate them using the PCA without losing the validity of the estimated factors. 

It is well known that $g_t$ is identifiable and estimable by the PCA up to an invertible
 matrix transformation (i.e., $H_T'g_t$),   whose exact form will be given in Section \ref{asymp:est:factors}.   Therefore,
 it is customary in the literature (see, e.g., \citep{bai03, BN06}) to treat $H_T' g_t$ as a centering object in the limiting distribution of estimated factors. Following  this convention,
in this section, let
\begin{align}\label{rotation-convention}
f_t:=H_T'g_t\quad \text{and} \quad \gamma_0:=H_T^{-1}\phi_0.
\end{align}
Using  the fact that $ g_t'\phi_0=f_t'\gamma_0$, we can  rewrite (\ref{model1-est-factor}) as  the original formulation in \eqref{model1}:
$$
y_t=x_t'\beta_0+x_t'\delta_01\{f_{t}'\gamma_{0}  > 0 \} + \varepsilon_t.
$$
Hence, $\gamma_0$ depends on the sample in this section but we suppress dependence on $T$ for the sake of notational simplicity. 



Our estimation procedure now consists of two steps. In the first step, a $(K+1) \times 1$ vector of estimated factors and the constant term (i.e.,  $\widetilde f_t := (\widetilde f_{1t}', 1)'$) are obtained  by the method of principal components.
To describe estimated factors, let $\mathcal Y  $ be the $T \times N$ matrix whose $t$-th row is  $\mathcal Y_t' $.  Let $ (\widetilde f_{11}, \ldots, \widetilde f_{1T})$ be the    $  K\times T$ matrix, whose rows are $K$ eigenvectors (multiplied by $\sqrt{T}$) associated with the  largest $K$ eigenvalues of $ \mathcal Y  \mathcal Y  '/{NT}$
 in decreasing order.
In the second step, unknown parameters $(\alpha_0, \gamma_0)$ are estimated by  the same algorithm in Section  \ref{sec:opt} with  $\widetilde f_t$ as inputs.

\subsection{Regularity Conditions}\label{subsection:asymp:est:factors}

We introduce assumptions needed for asymptotic results with estimated factors.
We first replace  Assumptions \ref{scale-normalization}--\ref{A:diminishing dt} with the following assumption.
Define
\begin{align}\label{def-Phi-T}
\Phi_T := \{ \phi: \phi = H_T \gamma \; \text{ for some $\gamma \in \Gamma_\epsilon$} \},
\end{align}
where $\Gamma_\epsilon$ is an $\epsilon$-enlargement of $\Gamma$. Note that $ \phi $ cannot be a vector whose first $ K $ elements are zeros due to the normalization on $ \gamma $ and the block diagonal structure of $ H_T$ that will be defined in \eqref{Bai-type-expansion}.
The space $\Phi_T$ for $\phi$ is defined through $H_T$ and excludes the case that $g_t'\phi$ is degenerate.
The $\epsilon$-enlargement of $\Gamma$ is needed because the factors are latent.

\begin{assum} \label{as9}
\begin{enumerate}[label=(\roman*)]
\item\label{as9:itm1}
Assumptions  \ref{scale-normalization},  \ref{iden-assump}, and
 \ref{A:diminishing dt}\ref{A:dim:itm1} hold after replacing $f_t$ and $\gamma_0$ with $g_t$ and $\phi_0$, respectively.

\item\label{as9:itm11}
		$\left\{ x_{t},g_{t}, e_{t}, \varepsilon
		_{t}\right\} $ is a sequence of strictly stationary, ergodic, and $\rho $	-mixing random vectors with $\sum_{m=1}^{\infty }\rho _{m}^{1/2}<\infty $,
		and there exists a constant $C < \infty$ such that
		  $\mathbb E(\left| x_t\right|_2^8 |g_t, e_t)<C$, $\mathbb E(\varepsilon_{t}^8 |g_t, e_t)<C$ a.s.,  and   $ g_t'\phi$ has a density that is continuous and bounded by  $C$ for all $\phi \in \Phi_T$.

\end{enumerate}
\end{assum}


Recall that by  the normalization in Assumption \ref{scale-normalization}, the first element of $ \gamma $ is fixed at 1. One caveat of this normalization scheme is that the sign of the first element of $f_t$ might not be the same as that of the first element of $g_t$ due to random rotation $H_T$; however, if we assume that $\delta_0 \neq 0$ and we also know the sign of one of the non-zero coefficients of $\delta_0$, then  we can determine the sign of the first element of $f_t$ after estimating the model. This is a ``labeling'' problem that is common in models with hidden regimes. For simplicity, we assume that the first element of $\gamma_0$ is 1.

The following assumption is  standard in the literature.    In particular, we allow  weak serial correlation among  $e_t$.

 \begin{assum}\label{assmp:factor}
 \begin{enumerate}[label=(\roman*)]
\item\label{assmp:factor:itm1}
$\lim_{N\to\infty}\frac{1}{N}\Lambda'\Lambda =\Sigma_{\Lambda}$ for some $K\times K$ matrix $\Sigma_{\Lambda}$, whose eigenvalues are bounded away from both zero and infinity.
\item\label{assmp:factor:itm2}
 The eigenvalues of  $\Sigma_{\Lambda}^{1/2} \mathbb E (g_{1t}g_{1t}') \Sigma_{\Lambda}^{1/2}$ are distinct.
\item\label{assmp:factor:itm3}
All the eigenvalues of the $N\times N$ covariance $\var(e_t)$ are bounded away from both zero and infinity.
\item\label{assmp:factor:itm4}  For any $ t $, $\frac{1}{N}\sum_{s=1}^{T}\sum_{i=1}^{N}|\mathbb Ee_{it}e_{is}|<C$   for some $C>0.$
\end{enumerate}
\end{assum}

Define  $\lambda_i'$ to be the $i$th row of $\Lambda$, so that $\Lambda=(\lambda_1, \ldots,\lambda_N)'$.
Further, let
\begin{align*}
\xi_{s,t}&:= N^{-1/2}\sum_{i=1}^N(e_{is}e_{it}- \mathbb Ee_{is}e_{it}),\quad \psi := (TN)^{-1/2}\sum_{t=1}^T\sum_{i=1}^N     g_t e_{it}   \lambda_i' ,\cr
\eta_t&:= (TN)^{-1/2}\sum_{s=1}^T \sum_{i=1}^Ng_{1s}  (e_{is}e_{it}- \mathbb Ee_{is}e_{it}),\quad \zeta_t:=  N^{-1/2}\sum_{i=1}^N \lambda_{it} e_{it}.
\end{align*}

 We require the following additional exponential-tail conditions.

\begin{assum}
   There exist finite, positive constants  $C, C_1$ and $c_1$ such that   for any $x>0$ and
    for any $\varpi\in\Xi:=\{e_{it}, g_{1t}, \xi_{s,t}, \zeta_t, vec(\psi),\eta_{t}\}$,
   \[\mathbb P(|\varpi|_2>x) \leq C \exp(-C_1x^{c_1}).\]

\end{assum}


These conditions impose exponential tail conditions on various terms. First, it requires weak cross-sectional correlations among $e_{it}.$  This assumption can be verified under some low-level conditions such as the $ \alpha$-mixing condition of the type of Merlev{\`e}de et al. \citep{MPR-2011} across both $(i,t)$ and individual  exponential-tailed distributions on $\{e_{it}, g_t\}$.
 While the  quantities in $\Xi$ are often assumed to have finite moments in the high-dimensional factor model literature, these moment bounds would no longer be sufficient in the current context. Instead, exponential-type probability bounds are more useful for us to characterize the  effect of the estimated factors. To see the point, note that we have the following asymptotic expansion:
\begin{align}\label{Bai-expansion}
\widetilde f_t = \widehat f_t + r_t,\quad \widehat f_t:=H_T'(g_t+N^{-1/2}h_t).
\end{align}
Here, $r_t$ is a remainder term,
\begin{align}\label{e:6.7:h}
H_T':=  \begin{pmatrix}
 \widetilde H_T'&0\\
 0&1
 \end{pmatrix},\quad
 h_t:= \begin{pmatrix}
h_{1t} \\
 0
\end{pmatrix}, \quad
 h_{1t}:=
(\frac{1}{N}\Lambda'\Lambda)^{-1}\frac{1}{\sqrt{N}}\Lambda' e_t,
\end{align}
and the exact form of $\widetilde H_T$ is given  in \eqref{Bai-type-expansion}.
 The diagonality in $H_T$ and the zero element in $h_t$ reflect the inclusion of the constant in $g_t$.
	  We  establish the following uniform approximation result:
uniformly for $\gamma$ over a compact set,
	$$
	\max_{t\leq T}\left| \mathbb P(\widetilde f_t'\gamma>0)-  \mathbb P(\widehat f_t'\gamma>0)\right|\leq O \left(\frac{(\log T)^c}{T} \right) + \max_{t\leq T}\mathbb P\left(   |r_t|>C\frac{(\log T)^c}{T} \right)
	$$
	for some constants $C, c>0$.
	The above exponential-tail assumption then enables us to derive a sharp bound so that
	$
	\max_{t\leq T}\mathbb P(   | r_t |>C(\log T)^c T^{-1})
	$
	is asymptotically negligible.


Next, we state important technical conditions to facilitate the local asymptotic expansion of the least-squares criterion function.
A technical challenge in the analysis is that even the expected criterion function  is non-smooth with respect to the factors.  As such, we  introduce some conditional density conditions to  study the effect of estimating factors $ H_T'h_t= \sqrt{N} (\widehat f_t- f_t) $.

 \begin{assum}\label{as5}
 \begin{enumerate}[label=(\roman*)]
\item\label{as5:itm1}  	
$
\sup_{x_t, g_t} \left| \mathbb P(h_t'\phi_0<0|x_t, g_t)- ({1}/{2}) \right|=O(N^{-1/2}).
 $
 \item\label{as5:itm2}
Let $\sigma^2_{h, x_t, g_t}:=\plim_{N\to\infty} \mathbb E[(h_t'\phi_0)^2|x_t, g_t]$ and let $\mathcal Z_t$ be a sequence of Gaussian random variables  whose conditional distribution,  given $x_t$ and $g_t$, is $\mathcal N(0,\sigma^2_{h, x_t, g_t})$.
 Then, there are positive constants  $c$, $c_0$, and $C$ such that $
\sigma^2_{h,x_t, g_t}  >c_0  $ a.s., $ \sup_{x_t, g_t} \sup_{|z|<c} p_{h_t^{\prime }\phi_0|		g_t,x_t}(z) < C$, and
\begin{align*}
\sup_{x_t, g_t} \sup_{|z|<c} |p_{h_t'\phi_0| g_t,x_t}( z) -p_{\mathcal Z_t|g_t,x_t}(z)  |  =o(1).
\end{align*}
\end{enumerate}
\end{assum}

Assumption \ref{as5} is concerned with the asymptotic behavior of  the distribution of $h_t$ as $N\to\infty.$
  The rate $ N^{-1/2} $ in Assumption \ref{as5}\ref{as5:itm1} is a reminiscent of the Berry--Essen theorem.
 The Edgeworth expansion of the sample means at zero implies that the approximation
error is $C N^{-1/2}$, where the universal constant $C$ depends on the moments of
the summand up to the third order \citep{hall1992bootstrap}. Thus, condition \ref{as5:itm1}   holds
for a broad range of setups including heteroskedastic errors $e_{it}$. For
instance, if the idiosyncratic error has the form $e_{it}=\sigma \left( g_{t}\right) \xi
_{it}$, where $g_{t}$ and $\xi _{it}$ are two independent sequences and $%
\left\{ \xi _{it}\right\} $ is an independent and identically distributed (i.i.d.) sequence across $i$, then the
condition is satisfied as long as both $\sigma \left( g_{t}\right) ^{3}$ and
$\mathbb{E} \left\vert \xi _{it}\right\vert ^{3}$ are bounded. Furthermore, it holds trivially if the conditional distribution of $h_t'\phi_0$ given $x_t$ and $g_t$ is symmetric around zero or more generally if its median is zero.
Assumption \ref{as5} ensures, among other things, that
for some function $\Psi(\cdot)$ such that $\mathbb E|\Psi(x_t, g_t)|<\infty$,
	$$
	\mathbb E\left[\Psi(x_t, g_t)\left(1\{h_t'\phi_0\leq 0\}-1\{\mathcal Z_t\leq 0\}\right)\bigg{|}x_t,g_t\right]=O(N^{-1/2}).
	$$
Above all, because $h_t$ is a cross-sectional average multiplied by $\sqrt{N}$, this assumption can be verified by a   cross-sectional central limit theorem (CLT), if $\{e_{it}: i\leq N\}$ satisfies some cross-sectional mixing condition.

In the next assumption, recall that, by the identification condition, we can write $\gamma=(1, \gamma_2)$, where $1$ is the first element of $\gamma$. Correspondingly, let $f_{2t}$ and $\widehat f_{2t}$ be the subvectors of $f_t$ and $\widehat f_t$, excluding their first elements.  Also, let
$ u_t:=g_t'\phi_0 = f_t'\gamma_0 $ and $ \breve{g}_t:=g_t + h_t/\sqrt{N} $.


 \begin{assum}\label{as8}
 	There exist positive constants $c$, $c_0$, $M_0$, and $M$ such that  the following hold  \emph{a.s.}.
 	\begin{enumerate}[label=(\roman*)]
 		\item \label{as8:itm1}  $\inf_{|u|<c}p_{\widehat f_t'\gamma_0| \widehat f_{2t}, x_t } (u) \geq c_0$  and $ \sup_{|f|_2<M_0}p_{f_{2t}|h_t}(f)<M$.

 		\item \label{as8:itm2} $\inf_{|u|<c}p_{  u_t|  f_{2t}, h_t, x_t } (u) \geq c_0 $.
 		For all $|u_1| < c, |u_2| < c$,
 		$$ |p_{u_t|h_t'\phi_0,  f_{2t},x_{t} }(u_1)- p_{u_t| h_t'\phi_0,   f_{2t},x_{t} }(u_2)|\leq M|u_1-u_2|.$$

 		\item\label{as8:itm3}
 		$\inf_{|r|_2=1} \mathbb E \left[ |f_{2t}'r|^k  1\{|f_{2t}|_2< M_0\} \right] \geq c_0$ for $ k=1,2. $

 		\item\label{as8:itm4}
 		$\sup_{|r|_2=1}\sup_{|u|<c}p_{g_t'r|h_t}(u) \leq M$.

 		\item\label{as8:itm6}
 		Each of $ \inf_{\phi \in \Phi_T} |
 		g_t'\phi|$, $\inf_{\phi \in \Phi_T}|\breve{g}_t'\phi|$,      $ \sup_{\phi \in \Phi_T} |
 		h_t'\phi|$, and $ \breve{g}_t'\phi_0$ has a density function bounded and continuous at zero, with
 		$\Phi_T$ given in \eqref{def-Phi-T}.

 		\item \label{as8:itm7}
 		$   \mathbb E [ (  x_{t}^{\prime }d_{0} ) ^{2}|g_{t}, h_t ] $ is bounded above by $ M_0 $ and below by $ c_0 $.

 		\item \label{as8:itm_AD} For any $ s$ and $ w $ that are linearly independent of $ \phi_{0} $, $ p_{\breve{g}_{t}'\phi_0|\breve{g}_{t}'s, \breve{g}_t'w}(u) $
		and $\mathbb E( (\varepsilon_t x_t'd_0)^{2}|\breve{g}_{t}'\phi_0=u,\breve{g}_{t}'s, \breve{g}_t'w ) $
		are continuously differentiable at $u=0$ with bounded
 		derivatives. Furthermore, $\mathbb E( (\varepsilon_t x_t'd_0)^{4}\left\vert
 		\breve{g}_{t}\right\vert _{2}^{2}|\breve{g}_{t}'\phi_0 ) \leq M $.

 	\end{enumerate}
 \end{assum}


These conditions control the local characteristics of the centered least-squares criterion function near the true parameter value. As the model is perturbed by the error in the estimated factors,  the centered criterion is a drifting sequence $ \widehat{f}_t $.
	Its leading term changes depending on whether $N=O(T^{2-4\varphi})$ or not. The lower bounds in the above assumption are part of rank conditions that ensure that the leading terms are well defined. As a result, it entails a phase transition on the distribution of $\widehat\gamma$. Because they are rather technical, we provide a more detailed discussion on Assumption \ref{as8} in Online Appendix \ref{sec:e:discuss}.

\subsection{Rates of Convergence}\label{rates:est:factors}

The following theorem presents the rates of convergence for the estimators.

\begin{thm} \label{thm:rate}
Let Assumptions  \ref{as9}--\ref{as8} hold.  Suppose $T=O(N)$. Then
\begin{align*}
|\widehat\alpha-\alpha_0|_2=O_P\left(\frac{1}{\sqrt{T}}\right)
\; \text{ and } \;
  |\widehat\gamma-\gamma_0|_2 =O_P\left(\frac{1}{T^{1-2\varphi}}+\frac{1}{\left( NT^{1-2\varphi }\right) ^{1/3}}\right).
\end{align*}
 \end{thm}

While the convergence rate for $\widehat\alpha$ is standard, 
the convergence rate of $\widehat{\gamma}$ merits further explanation.
First of all, when $N$ is relatively large so that  $T^{2-4\varphi }=o\left( N\right)$, $\widehat{\gamma}-\gamma_0$ converges  at a super-consistent rate of ${T^{-(1-2\varphi)}}$.
Contrary to this case, when  $N=o(T^{2-4\varphi})$,  the estimated threshold parameter  has a cube root rate, which is similar to that of the maximum score type estimators \citep{kim1990}.
Therefore,   as  $\sqrt{N}/ T^{1-2\varphi}$ varies in $[0,\infty]$,  the rate of convergence  varies between the super-consistency rate of the usual threshold models to the cube root rate of the maximum score type estimators.

The convergence rates exhibit a continuous transition from one to the other.
To explain this  transition phenomenon, we can show that uniformly in $(\alpha,\gamma)$, the objective function has the following expansion:    there  are  functions   $ R_1(\cdot)$ and $ R_2(\cdot, \cdot)$   such that
$$
{\mathbb{S}}_{T}\left( \alpha , \gamma\right)- {\mathbb{S}}_{T}\left( \alpha_0 , \gamma_0 \right)
=  R_1(\gamma)  +   R_2(\alpha, \gamma),
$$
where $\gamma \mapsto R_1(\gamma)$ is a non-stochastic function, representing the ``mean" of the loss function, but is also highly non-smooth with respect to $\gamma$,
and $R_2(\alpha, \gamma)$ is the remaining stochastic part.
A key step   is to derive a sharp lower bound for $R_1(\gamma)$.
When $N$ is relatively large, the effect of estimating latent factors is negligible, and $R_1(\gamma)$ has a high degree of non-smoothness.  Similar to the usual threshold model, we have
$$
R_1(\gamma)\geq CT^{-2\varphi} |\gamma-\gamma_0|_2 -  O_P(T^{-1}).
$$
This lower bound leads to a super-consistency rate.  On the other hand, when $N$ is relatively small, there are  extra noises arising from the cross-sectional idiosyncratic errors when estimating the latent factors, which we call ``cross-sectional noises."  A remarkable feature of our model is that the cross-sectional noises help  {smooth} the  objective function in this case. As a result,
the behavior of  $R_1(\gamma) $ is  similar to that of the maximum score type estimators, where a quadratic lower bound can be derived:
$$
R_1(\gamma)\geq CT^{-2\varphi} \sqrt{N} |\gamma-\gamma_0|_2^2
-O_P(T^{-2\varphi }N^{-5/6}) .
$$
The quadratic lower bound, together with a larger error rate, then leads to a cube root rate type of convergence.
See  Online Appendix \ref{sec:roadmap}  for a detailed description of the roadmap of the proof.

\subsection{Consistency of Regime-Classification}

We introduce an  error rate in (in-sample)
regime-classification,
\[
\widehat{R}_{T}=\frac{1}{T}\sum_{t=1}^{T}\left\vert 1\left\{ \widetilde{f}%
_{t}^{\prime }\widehat{\gamma}>0\right\} -1\left\{ f_{t}^{\prime }\gamma
_{0}>0\right\} \right\vert.
\]
The uncertainty about the regime classification comes from either   $\widetilde f_t$ or  $\widehat{\gamma}$ or both.  We establish
its convergence rate in the following theorem.

\begin{thm} \label{thm:classification}
Let Assumptions  \ref{as9}--\ref{as8} hold.  Suppose $T=O(N)$. Then
\[
\widehat{R}_{T}=O_P\left( \left( NT^{1-2\varphi }\right)
^{-1/3}+T^{-1+2\varphi }+N^{-1/2}\right) .
\]
\end{thm}

This is a useful corollary of the derivation of the rates of convergence for the
threshold estimator. We expect a good performance of our regime classification rule even with a moderate size of $T$. 

 \subsection{Asymptotic Distribution}\label{asymp:est:factors}

To describe the asymptotic distribution, we introduce additional notation.
Let $V_T$ denote the $K \times K$ diagonal matrix whose elements are the $K$ largest eigenvalues of $ \mathcal Y \mathcal Y  '/{NT}$.
Define
\begin{align}
 \widetilde H_T':=V_T^{-1}   \frac{1}{T}\sum_{t=1}^T\widetilde f_{1t} g_{1t}' \frac{1}{N}\Lambda'\Lambda,\quad
 H_T:=\diag(\widetilde H_T, 1),
  \label{Bai-type-expansion}
\end{align}
and
$ H:=\plim_{T, N\rightarrow \infty }H_{T} $, which is well defined, following Bai \citep{bai03}.
Let
\begin{equation*}
\omega:=\lim_{N,T\rightarrow \infty }\frac{\sqrt{N}}{T^{1-2\varphi }}\in \lbrack
0,\infty ],\quad
\zeta_\omega :=\max\{\omega, \omega^{1/3}\},\quad \text{and} \quad
M_\omega:=\max\{1, \omega^{-1/3}\} .
\end{equation*}
Define, for $ u_t=f_t'\gamma_0$,
\begin{align*}
A(\omega,g)&:= M_\omega   \mathbb E \left[(x_td_0)^2\left(\left|    f_t'g + \zeta_\omega^{-1}  \mathcal Z_t \right |   -
  \left | \zeta_\omega^{-1} \mathcal Z_t   \right | \right) \bigg{|} u_t=0\right]   p_{u_t}(0)  \cr
\end{align*}%
for $\omega\in \left( 0,\infty\right] $, with the convention that $1/\omega=0$ for $
\omega=\infty $, and%
\begin{align*}
A(0, g)&:=  \mathbb{E}\left[(x_t^{\prime }d_0)^2 (  f_t'g)^2\bigg{|}u_t=0, \mathcal Z_t=0\right]p_{u_t, \mathcal Z_t}(0,0)
\end{align*}%
for $\omega=0 $.
Recall $  Z_{t}(\gamma) :=(x_{t}^{\prime },x_{t}^{\prime }1\{f_{t}^{\prime
}\gamma>0\})^{\prime }$.




\begin{thm}
\label{thm:AD with Estimated f}
Let Assumptions  \ref{as9}--\ref{as8} hold.  Suppose $T=O(N)$.
Let $\mathcal{G}:= \{0\} \times \mathbb{R}^{K} $.
In addition, let $W$ denote the same Gaussian process as in Theorem \ref{asdist-alpha-gamma}. 
Then,  as $N,T\rightarrow \infty $, we have
\begin{align*}
& \sqrt{T}(\widehat{\alpha }-\alpha _{0})\overset{d}{\longrightarrow }%
\mathcal{N}\left( 0,\left( \mathbb{E}Z_t(\gamma_0)Z_t(\gamma_0)^{\prime }\right) ^{-1}\mathbb{E}\left( Z_t(\gamma_0)Z_t(\gamma_0)^{\prime }\varepsilon _{t}^{2}\right) \left( \mathbb{E}Z_t(\gamma_0)Z_t(\gamma_0)^{\prime }\right) ^{-1}\right) , \\
& \left( \left( NT^{1-2\varphi }\right) ^{1/3}\wedge T^{1-2\varphi }\right)
\left( \widehat{\gamma }-\gamma _{0}\right)\overset{d}{\longrightarrow }%
\argmin_{g\in \mathcal{G}}A\left( \omega, g\right) +2W\left( g\right) ,
\end{align*}%
and $\sqrt{T}(\widehat{\alpha }-\alpha _{0})$ and $( (
NT^{1-2\varphi }) ^{1/3}\wedge T^{1-2\varphi }) ( \widehat{\gamma }-\gamma _{0}) $ are asymptotically independent. Moreover,
$A(0, g)=\lim_{w\to0}A( w, g).$
\end{thm}

It is worth noting that 
 $A\left( \omega, g\right) $ is continuous
everywhere, 
which implies
that the distribution of the argmin  of the limit processes $A\left(
\omega,g\right) +2W\left( g\right) $ is  also continuous in $\omega$ in virtue of
the argmax continuous mapping theorem [see e.g.,\citep{VW}].
Furthermore,
 the asymptotic distribution of $\widehat{\gamma}$ is well defined for any $%
\omega $ due to Lemma 2.6 of Kim and Pollard \citep{kim1990}.  Specifically, the argmin of the limit Gaussian process is $O_P\left(
1\right) $ since $A\left( \omega, g\right) $ is a deterministic function of order
at least $\left\vert g\right\vert $ for any $\omega$ while the variance of $%
W\left( g\right) $ grows at the rate of $\left\vert g\right\vert $ as $%
g\rightarrow \infty $. It also possesses a unique minimizer almost
surely.  

    In the literature, 	Bai and Ng \citep{BN06, BN08}  have shown that the oracle property (with regard to the estimation of the factors) holds for the linear regression if $T^{1/2}=o\left( N\right) $ and for the extremum estimation   if $T^{5/8}=o\left( N\right) $,  in the presence of  estimated factors. Thus, it appears that the oracle property demands a larger $N$ as the nonlinearity of the estimating equation rises. In view of this, we regard our condition, $T=O(N)$, as not too stringent because we need to deal with estimated factors inside the indicator functions.

\subsection{Phase Transition}


To demonstrate that our asymptotic results are sharp, we consider a special case that
$N = T^\kappa$ for $\kappa \geq 1$.
In this case, the asymptotic results can be depicted on  the $(\kappa, \varphi)$-space.

We categorize the results  of Theorem \ref{thm:AD with Estimated f} into three groups.
In all three cases, the estimators enjoy certain oracle properties.

\begin{itemize}
\item Strong oracle: $T^{2-4\varphi }=o\left( N\right)
$ or $ \omega = \infty $.
This is equivalent to $\kappa> 2-4\varphi$.
 The drift function $A(\infty, g)$  has a kink at $g=0$.  Intuitively, a bigger $N$ makes the
estimated factors  more precise. This  yields the oracle result for both $\widehat{\gamma }$ and $ \widehat{\alpha} $, and the same asymptotic distribution as in the known factor case.

\item Weak oracle: $N=o\left(
T^{2-4\varphi }\right) $ or $ \omega = 0 $.  This is equivalent to $\kappa< 2-4\varphi$. The drift function $A(0, g)$
 is approximately quadratic in $g$ near the origin.
Because it is harder to  identify the minimum when the function is
smooth than when it has a
kink at the minimum, this results in a non-oracle asymptotic distribution as well as a slower rate of convergence for $\widehat\gamma$ to $\left( NT^{1-2\varphi }\right) ^{-1/3}$. However, the asymptotic distribution for $\widehat \alpha$ are still the same as those when the unknown factors are observed. So  the oracle property for $ \widehat{\alpha} $ is preserved.

\item Semi-strong oracle: $N\asymp T^{2-4\varphi }$ or $\omega\in (0,\infty)$.  This is equivalent to $\kappa=2-4\varphi$. In this case, $A(\omega, g)$ has a continuous  transition between the two polar cases discussed above.  The effect of estimating factors is non-negligible for $ \widehat{\gamma} $ and yet the estimator enjoys the same rate of convergence. The estimator $ \widehat{\alpha} $ continues to achieve the oracle efficiency.
\end{itemize}

The phase transition occurs when   $\kappa=2-4\varphi$, which is the \emph{semi-strong oracle case} and the \emph{critical boundary} of the phase transition. Changes in the convergence rates and asymptotic distributions are continuous along the critical boundary.

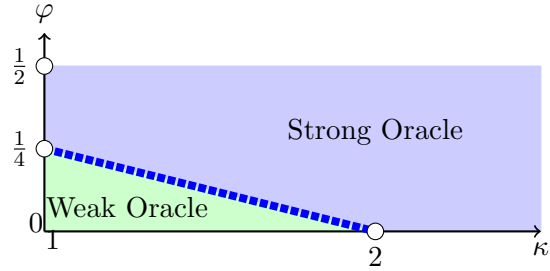
\begin{figure}[htb]
\caption{Phase Diagram}
\label{fig:phase}
\begin{center}
\begin{tikzpicture}[scale = 1.1]
\draw [green!20!white, fill=green!20!white] (0,0) -- (4,0) -- (0,1) -- cycle;
\draw [blue!20!white, fill=blue!20!white] (0,1) -- (0,2) -- (6,2) -- (6,0) -- (4,0);

\draw[thick, ->] (0,0)  -- (6,0) node[below] {$\kappa$};
\draw[thick, ->] (0,0)  -- (0,2.4) node[above] {$\varphi$};

\draw[blue,line width=3pt,densely dotted] (0,1) -- (4,0);

\draw[color=black, fill=white] (0,2) circle (.1);
\draw[color=black, fill=white] (0,1) circle (.1);
\draw[color=black, fill=white] (4,0) circle (.1);

\node[black] at (-0.3,2) {$\frac{1}{2}$};
\node[black] at (-0.3,1) {$\frac{1}{4}$};
\node[black] at (0.1,-0.1) {$1$};
\node[black] at (-0.1,0.1) {$0$};
\node[black] at (4,-0.3) {$2$};

\node[black] at (4,1.2) {Strong Oracle};
\node[black] at (1,0.3) {Weak Oracle};
\end{tikzpicture}
\end{center}
\end{figure}

Figure \ref{fig:phase} depicts a phase transition from the strong oracle phase to the weak oracle phase.
The critical boundary   $\kappa=2-4\varphi$  is shown by closely dotted points in the figure.
On  one hand, as $\varphi$ moves from 0 to $1/2$, the strong oracle region for $\kappa$ increases. That is, as the convergence rate for $\widehat \gamma$ becomes slower, the requirement for the minimal sample size $N$ for factor estimation becomes less stringent.
On the other hand,  as $\kappa$ becomes larger, the strong oracle region for $\varphi$ increases. In other words, as $N$ becomes larger,  the range of attainable oracle rates of convergence for $\widehat \gamma$ becomes wider. In this way, we provide a thorough characterization of the effect of estimated factors.

\subsection{Graphical Representation of $A\left( \omega, g\right)$}

\begin{figure}[htbp]
	\caption{An Example of $A\left( \omega, g\right)$}
	\label{fig-Akg}
	\begin{center}
		\graphicspath{ {plot/} }
		\includegraphics[scale=0.25]{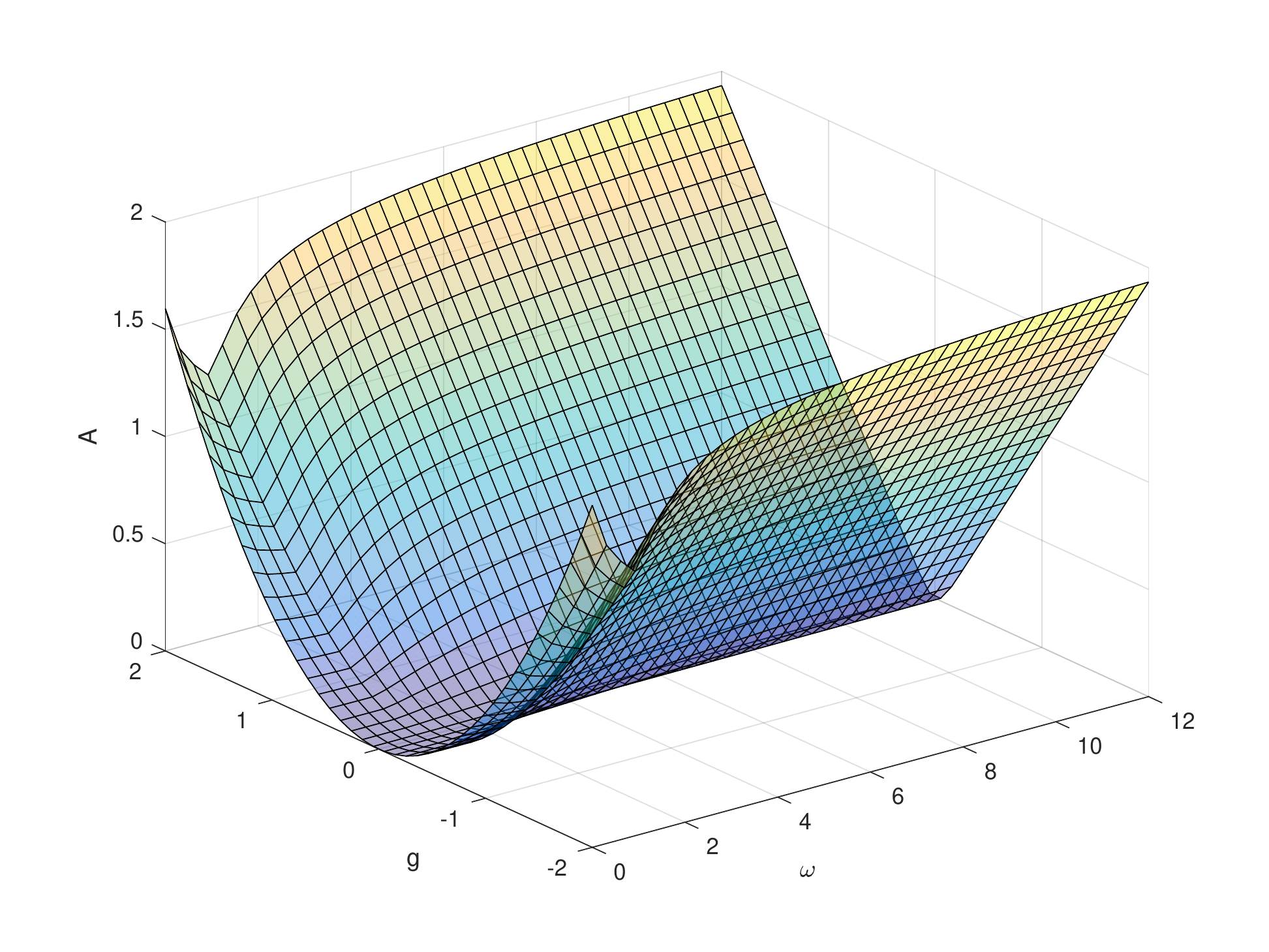}
		\includegraphics[scale=0.25]{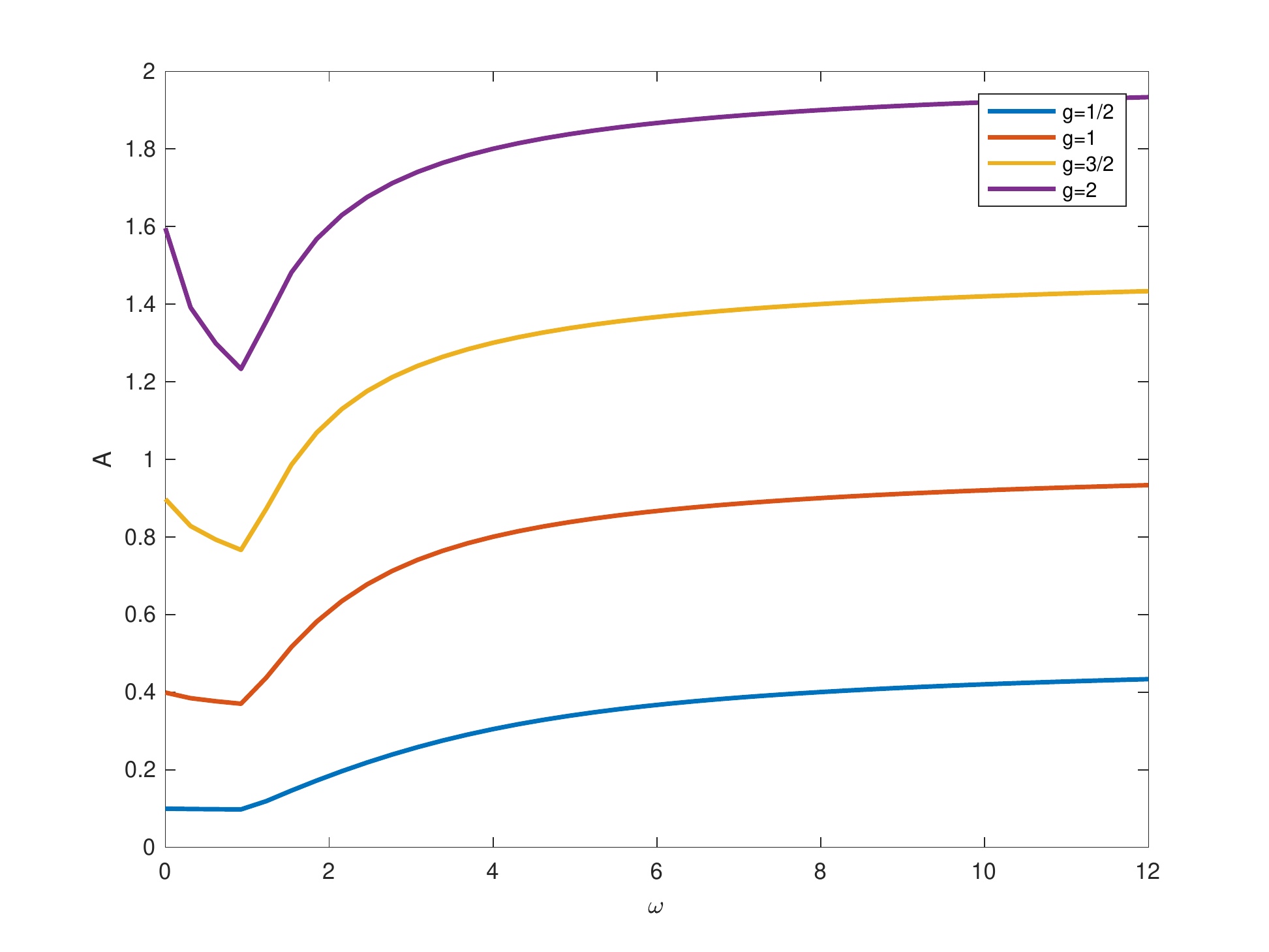}
		\includegraphics[scale=0.25]{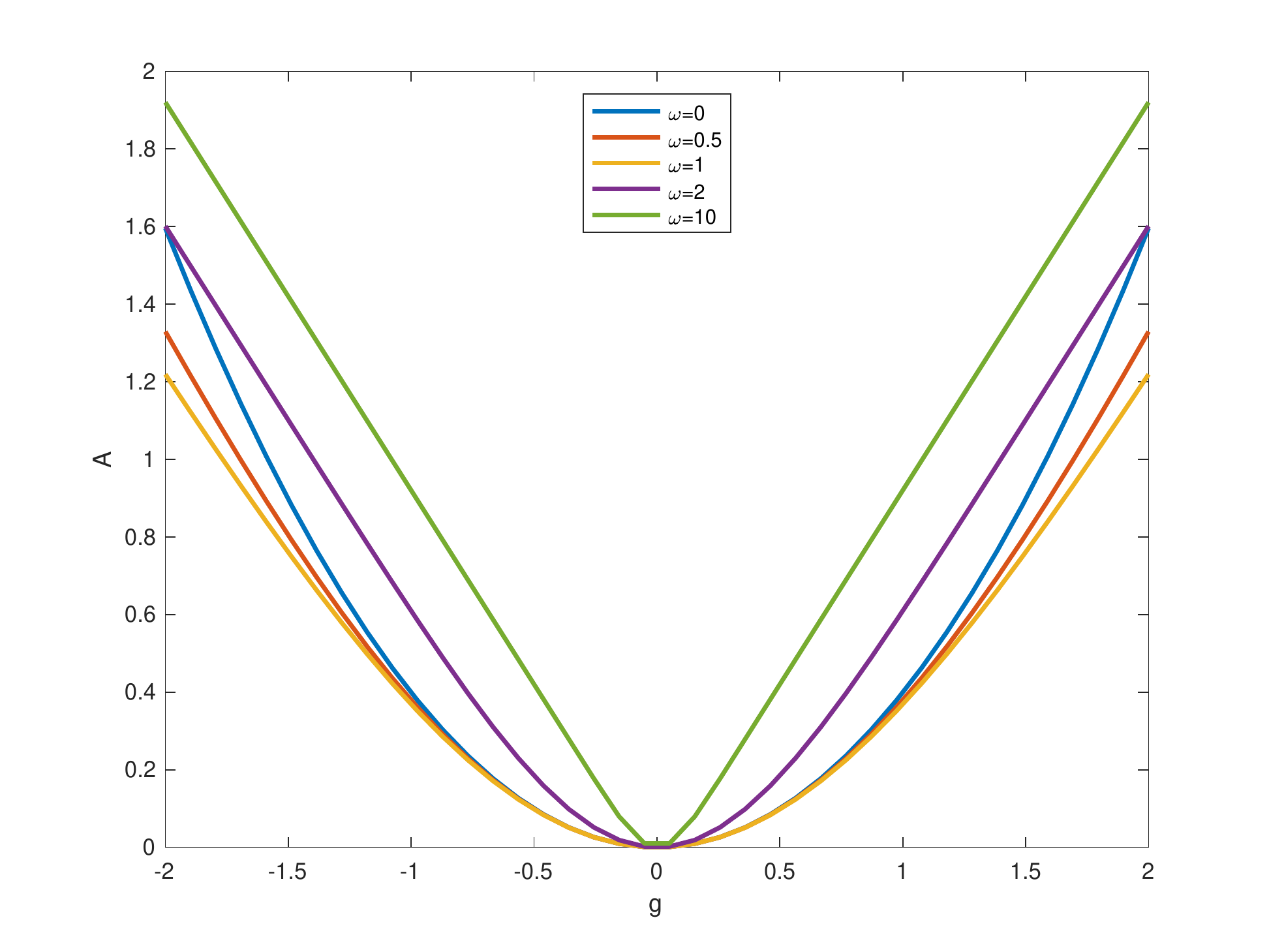}
	\end{center}
\end{figure}

To plot $A\left( \omega, g\right)$, we consider the simple case that
$g_{t}=\left( q_{t},-1\right) ^{\prime }$,   $g=\left( 0,g_{2}\right) ^{\prime },$ $x_{t}=1$,   $d_{0}=1$,
and $h_{t}$ and $q_{t}$ are independent of each other. We  write $g_2=g$ for simplicity.
The left panel of Figure \ref{fig-Akg} shows the three-dimensional graph of $A\left( \omega, g\right)$,
the middle panel depicts the profile of $A\left( \omega, g\right)$ as a function of $\omega$ for several values of $g$,
and the right panel exhibits that of $A\left( \omega, g\right)$ as a function of $g$ for given values of $\omega$.
First of all, it can be  seen that $A\left( \omega, g\right)$ is continuous everywhere but has a kink at $\omega=1$.
As $\omega$ approaches zero, the shape of  $A\left( \omega, g\right)$ is clearly quadratic in $g$; whereas, as $\omega$ becomes larger,
it becomes almost linear in $g$. Also, $A\left( \omega, g\right)$ is quite flat around its minimum at $g=0$ when $\omega$ is close to zero;  however, $A\left( \omega, g\right)$ has a sharp minimum at zero for a larger value of $\omega$. This reflects the fact that the rate of convergence increases as $\omega$ becomes larger.

\section{Inference}\label{sec:inference}

In this section, we consider inference.
Regarding $\alpha_0$,
Theorems \ref{asdist-alpha-gamma} and \ref{thm:AD with Estimated f} imply that inference for $\alpha_0$ can be carried out as if $\gamma_0$ were known. Therefore, the standard inference method based on the asymptotic normality can be carried out for $\alpha_0$ for both observed and estimated $f_t$.

We now focus on the inference issue regarding $\gamma_0$.
 Let $ \theta_0 = h(\gamma_0)$ denote the parameter of interest for some known linear transformation  $h(\cdot)$. For instance, this can be a particular element of $\gamma_0$ or a linear combination of the elements of $\gamma_0$.   We use a quasi-likelihood ratio statistic:
 \begin{align*}
LR(\theta) &:=\frac{\mathbb{%
S}_{T}\left( \widehat\alpha_h ,\widehat\gamma_h \right) - {\mathbb{S}}_{T}(\widehat\alpha,\widehat\gamma)}{ {\mathbb{S}}_{T}(\widehat\alpha,\widehat\gamma)},\cr
( \widehat\alpha_h ,\widehat\gamma_h)&:=\arg\min_{\alpha ,h\left( \gamma \right) = \theta}\mathbb{%
	S}_{T}\left( \alpha ,\gamma \right),\quad
( \widehat\alpha ,\widehat\gamma):=\arg\min_{\alpha ,\gamma}\mathbb{%
	S}_{T}\left( \alpha ,\gamma \right),
\end{align*}
where $\mathbb S_T$ denotes the least-squares loss function, using $f_t$ when factors are observable, and $\widetilde f_t$ when factors are estimated.
Then, the $ 100(1-a) \% $-level confidence set for $ \theta_0 $ is $ \{\theta : LR(\theta) \leq \texttt{cv}_a \} $, where $ \texttt{cv}_a $ denotes a critical value.
As Theorem \ref{t6.1} shows, the asymptotic distribution is non-pivotal, so the  critical value  is   computed based on the  bootstrap.

\subsection{The Bootstrap with Estimated Factors}

We focus on  the case of estimated factors, where we use
$\widetilde f_t$ as the ``true" factors, and denote by  $f_t^*$ as the 	\emph{estimated factors}
in the bootstrap world.
To preserve the phase transition brought by the effect of PCA factor estimators,
$f_t^*$ should be a ``perturbed" version of $\widetilde f_t$.  Specifically,
let $f_t^*$ be re-estimated factors in the bootstrap sample via PCA. This is given by Gon{\c{c}}alves and Perron \citep{gonccalves2018bootstrapping}. To maintain the cross-sectional dependence among the idiosyncratic components in the bootstrap factor models, we  generate bootstrap data by
$$
\mathcal Y_t^*:= \widehat\Lambda \widetilde f_{t} + \widehat \var(e_t)^{1/2}\mathcal W_t^*,
$$
where $\{\mathcal W_t^*:t\leq T\}$ is a sequence of independent $N\times 1$    multivariate standard normal random vectors and $ \widehat \var(e_t)$ is the estimated covariance matrix of $e_t$. If the covariance is a sparse matrix,
we apply the thresholding covariance estimator of Fan, Liao, and Mincheva \citep{POET}.
Then, we  apply PCA to estimate factors to obtain $\widetilde F_{t}^*$.
However,
$\widetilde F_{t}^*$ estimates $\widetilde f_t$, the ``true factors" in the bootstrap sample, up to a new rotation matrix $H_T^*$. Fortunately, such a rotation indeterminacy can be removed  because $H_T^*$ is known in the bootstrap world. Following Gon{\c{c}}alves and Perron \citep{gonccalves2014bootstrapping, gonccalves2018bootstrapping}, we  define
\begin{align}\label{f-star-bootstrap}
f_t^*:= H_T^{*'-1} \widetilde F_{t}^*
\end{align}
as the final ``estimated factors" in the bootstrap sample.
The bootstrap distribution of $f_t^*- \widetilde f_t$    mimics well the asymptotic sampling distribution of $\widetilde f_t -H_T'g_t$, that is $\mathcal N(0, \Sigma_h)$.
We  give more details of this method,  the definition of $H_T^*$,
and an alternative method based on Gaussian perturbation in Online Appendix  \ref{sec:appendix:est:factors}.

\subsection{The k-Step Bootstrap Algorithm}

We now  describe the bootstrap algorithm in detail.
Define
\begin{align}\label{Z-star:bootstrap}
	\widetilde Z_t(\gamma):= (x_t', x_t' 1\{\widetilde f_t'\gamma>0\})'
	\; \text{ and } \;
		Z^*_t(\gamma):=  (x_t', x_t'1\{f_t^{*'}\gamma>0\})'.
\end{align}
For each $t=1,\ldots, T$,  construct $\left\{ y_{t}^{\ast }\right\}_{t\leq T} $  by
	\begin{align}\label{y-star:bootstrap}
	y_{t}^{\ast }:=
	\widetilde Z_t \left( \widehat \gamma \right)^{\prime }\widehat{\alpha}+\eta
	_{t}\widehat{\varepsilon}_{t}
	\; \text{ with } \;
	\widehat\varepsilon_t:=y_t-\widetilde Z_{t}\left( \widehat{\gamma}\right) ^{\prime }\widehat{\alpha},
	\end{align}
where $\eta_t$ is an i.i.d.\ sequence  whose mean is
	zero and whose variance is one.
For example, $\eta_t \sim \mathcal{N} (0,1)$ or it can be simulated from a discrete distribution
(e.g., the Rademacher distribution).
The bootstrap least-squares loss is given by
\begin{equation}\label{eq7.1}
\mathbb S_T^*(\alpha,\gamma):=\frac{1}{T}\sum_{t=1}^T[y_t^*-Z^*_{t}\left(  {\gamma}\right)'\alpha]^2.
\end{equation}
In principle,  the bootstrap analog of the original constraint
is $h(\gamma)=h(\widehat\gamma)$ and the bootstrap analogous $ LR $ is defined as
 $$
 \widetilde{LR}^*:=\frac{\min_{\alpha ,h\left(\gamma \right) =h(\widehat\gamma)}\mathbb{%
S}_{T}^*\left( \alpha ,\gamma \right) -\min_{\alpha, \gamma} {\mathbb{S}}^*_{T}( \alpha, \gamma)}{  \min_{\alpha, \gamma} {\mathbb{S}}^*_{T}( \alpha, \gamma)}.
 $$

 A potential computational problem   for  $ \widetilde{LR}^*$     is that it is necessary to fully solve two joint  MIO problems:  $ \min_{\alpha, \gamma} {\mathbb{S}}^*_{T}( \alpha, \gamma) $ and $  \min_{\alpha, h(\gamma)=h(\widehat\gamma)} {\mathbb{S}}^*_{T}( \alpha, \gamma)$   in each of the bootstrap repetitions.   To circumvent this problem,
we adopt the approach of Andrews \citep{andrews2002higher}. Because a solution based on the original data  should be close to a solution   based on the bootstrapped data, within each bootstrap replication, we can employ the MILP algorithm, with
$(\widehat\alpha,\widehat\gamma)$ as the initial value,  and iteratively update the  algorithm for $k$ steps rather than computing the full bootstrap solutions.
A computationally convenient $k$-step LR statistic (${LR}_k^*$) and
its computational details are given in Algorithm \ref{algo:bootstrap}.

\begin{algorithm}[h!tb]
	\KwInput{$\{(y_t, x_t, \widetilde f_{t}, M_t, \widehat{\varepsilon}_t): t=1,\ldots,T \}$, $\widehat \var(e_t)$, $\widehat\Lambda$,  $ \widehat{\alpha} $, $\widehat\gamma$, $\widehat\gamma_h$, $ B $}

	\KwOutput{bootstrap critical value $\texttt{cv}_a^*$}

	Set $ b=1 $\;

	\While{$ b \leq B $}{

	Generate an i.i.d. sequence $\left\{ \eta _{t}\right\}_{t\leq T} $ whose mean is
	zero and variance is one and an i.i.d. sequence of multivariate vectors $\left\{ \mathcal W^* _{t}\right\}_{t\leq T} $  from $\mathcal N(0, I)$\;

	Generate $
	\mathcal Y_t^*= \widehat\Lambda \widetilde f_{t} + \widehat \var(e_t)^{1/2}\mathcal W_t^*, t=1,...,T
	$\;

Apply PCA to $ \{\mathcal Y_t^* \}$ and obtain $\widetilde F_{t}^*$ as the PCA factor estimates\;

Compute $ H_T^* $	and $ f_t^* = H_T^{*'-1}  \widetilde{F}_t $, $ t=1,...,T$\;

 Construct $
y_{t}^{\ast }=
\widetilde Z_t \left( \widehat \gamma \right)^{\prime }\widehat{\alpha}+\eta
_{t}\widehat{\varepsilon}_{t}$, $ t=1,...,T $, where $ 		\widetilde Z_t(\gamma)= (x_t', x_t' 1\{\widetilde f_t'\gamma>0\})'
$\;

Initialize at $\widehat\gamma^{* ,0}=\widehat\gamma,  $ $\widehat\gamma_h^{* , 0}=\widehat\gamma_h  $\;

    Set $ l=1$;

	\While{$ l \leq k $}{

	Compute $	\widehat{\alpha}^{*, l}= \alpha^*(\widehat{\gamma}^{*, l-1} )$ and $\widehat{\alpha}_h^{*, l}=\alpha^*(\widehat{\gamma}_h^{*, l-1})$, where
\begin{eqnarray*}
			\alpha^*(\gamma)&=& \left[\frac{1}{T}\sum_{t=1}^T Z^*_{t}\left(\gamma\right)Z^*_{t}\left( \gamma \right)' \right]^{-1} \frac{1}{T}\sum_{t=1}^T Z^*_{t}\left( \gamma \right) y_t^*\text{\;} 		\end{eqnarray*}

	For the given $(\widehat{\alpha}^{*l},\widehat{\alpha}^{*l}_h)$,  compute the following by MILP:
		\begin{eqnarray*}
			\widehat\gamma^{*,l} &=&\arg\min_{\gamma}\mathbb S_T^*(\widehat{\alpha}^{*, l},  \gamma ),\cr
		\widehat\gamma^{*,l}_h &=&\arg\min_{h(\gamma)=h(\widehat\gamma)}\mathbb S_T^*(\widehat{\alpha}_h^{*, l},  \gamma )\text{\;}
		\end{eqnarray*}

		Let $l=l+1$\;
	}

	Compute
	\[  {LR}_k^*:=\frac{ \mathbb{%
			S}_{T}^*\left( \widehat \alpha_h^{*, k} ,\widehat \gamma_h^{*, k} \right)-  \mathbb{%
			S}_{T}^*\left( \widehat \alpha^* ,\widehat \gamma^* \right) }{   \mathbb{%
			S}_{T}^*\left( \widehat \alpha^* ,\widehat \gamma^* \right) }\text{\;}
	 \]


	Let $b=b+1$\;

}

Obtain $\texttt{cv}_a^*$ by the $(1-a)$ th quantile of the empirical distribution of $LR^*_k$.

\caption{Bootstrap for Estimated Factors }\label{algo:bootstrap}
\end{algorithm}

\subsection{Asymptotic Distribution}

 To describe the asymptotic distribution of  the quasi-likelihood ratio statistic, let
 $\sigma_\varepsilon^{2}$ be the variance of $\varepsilon_t$.  In addition, recall the asymptotic distributions of $\widehat\gamma$, the minimizer of
 $$
 \mathbb Q(\omega, g) :=   A(\omega, g) +2W\left( g\right),
 $$
and,
 as we discussed for Theorem \ref{thm:AD with Estimated f}, $\omega=\infty$   also corresponds to the  case of known factors. 

 Note that $A(\omega, g)$ depends on the true value $\phi_0$, the rotation matrix $H$, and the covariance matrix $\Sigma_h$.  For the bootstrap sampling distribution, we  consider  drifting sequences around these values. For this, define
\begin{align*}
&\mathbb A(\omega, g, \Sigma,  \bar H,\phi)\cr
&:=M_\omega   \mathbb E \left[(x_td_0)^2\left(\left|    g_t'Hg + \zeta_\omega^{-1}  \mathcal W_t^{*'}\Sigma^{1/2}\bar H^{-1} \phi\right |   -
  \left | \zeta_\omega^{-1}   \mathcal W_t^{*'}\Sigma^{1/2}\bar H^{-1} \phi  \right | \right) \bigg{|} g_t'\phi=0\right]   p_{g_t'\phi}(0)
\end{align*}
 for $\omega\in(0,\infty]$, and
\begin{align*}
&\mathbb A( 0, g, \Sigma, \bar H,\phi)\cr
&:=  \mathbb{E}\left[(x_t^{\prime }d_0)^2 (  g_t'Hg)^2\bigg{|}g_t'\phi=0,    \mathcal W_t^{*'}\Sigma^{1/2}\bar H^{-1} \phi =0\right]p_{g_t'\phi,   \mathcal W_t^{*'}\Sigma^{1/2}\bar H^{-1} \phi }(0,0) .
\end{align*}
Note that   $A(\omega, g)= \mathbb A(\omega, g, H'\Sigma_h H,   H,\phi_0)$.

\begin{assum} \label{a9.1}
(i) Uniformly  for $\phi$ inside a  neighborhood of $\phi_0$,

  $
 \sup_{x_t, f_{2t}}| p_{\breve g_t'\phi|x_t,  f_{2t}}(0)-p_{g_t'\phi_1|x_t, f_{2t}}(0)|=o(1).
 $

 (ii) For each fixed $\omega\in[0,\infty]$ and $g$, $\mathbb A(\omega, g,  S)$ is continuous with respect to $S=(\Sigma, \bar H, \phi)$.

 (iii) The factor idiosyncratic component $e_t$ is independent of $(x_t, g_t)$, and
 $|\widehat \var(e_t)- \var(e_t)|_2=o_P(1)$ under the matrix spectral norm.

 (iv)  $\inf_{\gamma}|\widehat f_t^{*'}\gamma|$ has a density (jointly with respect to $(e_t, g_t, \mathcal W_t^*)$) bounded and continuous  at zero, where $\widehat f_t^{*}=\widehat f_t+ N^{-1/2} \widehat\Sigma_h^{1/2}\mathcal W_t^*$.

\end{assum}

Fan, Liao, and Mincheva \citep{POET} showed that under mild sparsity assumptions, for the matrix spectral norm, $|\widehat \var(e_t)- \var(e_t)|_2=o_P(1)$, given that $\log N$ does not grow too fast relative to $T$. The following theorem presents the asymptotic distribution of $LR$, and the validity of the $k$-step bootstrap procedure.

\begin{thm}\label{t6.1}
Suppose that   assumptions of Theorem \ref{asdist-alpha-gamma} (for the known factor case) or assumptions of Theorem \ref{thm:AD with Estimated f} (for the estimated factor case)
and Assumption \ref{a9.1} hold.
  Let $h(\cdot)$ be a $\mathbb R^m$-valued linear function with a fixed $m$ and  let $
r_{NT}:=\left( NT^{1-2\varphi }\right) ^{1/3}\wedge T^{1-2\varphi } $, where we set  $N = T^2 $ in case of the known factor.
 Then,  under $\mathcal H_0: h(\gamma_0)=\theta$, we have
 $$
  \sqrt{r_{NT}T^{1+2\varphi }} \cdot LR\to^d \sigma_{\varepsilon}^{-2}\min_{ g_h'\nabla h=0} \mathbb Q( \omega,g_h)   - \sigma_{\varepsilon}^{-2}\min_{ g} \mathbb Q( \omega,g),
 $$
and  for any $k\geq 1$ as the number of iterations in the $k$-step bootstrap,
 $$
\sqrt{r_{NT}T^{1+2\varphi }} \cdot  LR_k^*\to^{d^*} \sigma_{\varepsilon}^{-2}\min_{ g_h'\nabla h=0} \mathbb Q( \omega,g_h)   - \sigma_{\varepsilon}^{-2}\min_{ g} \mathbb Q( \omega,g).
 $$
In the above, $\to^{d^*}$ represents the convergence in distribution with respect to the conditional distribution of $\left\{ \eta _{t},\mathcal W_t^*\right\}_{t\leq T} $  given the original data. Also, $\nabla h$ denotes the gradient of $h(\cdot)$, which is independent of $\gamma_0$ as $h$ is linear.
\end{thm}


\section{Monte Carlo Experiments}\label{sec:MC}

In this section, we study the finite sample properties of the proposed method via Monte Carlo experiments. The data are generated from the following design:
\begin{align*}
y_{t} & = x_{t}^{\prime }\beta _{0}+x_{t}^{\prime }\delta
_{0}1\{g_{t}^{\prime }\phi _{0}>0\}+\varepsilon _{t} ~~\text{for}~~t=1,\ldots,T,
\end{align*}
where $\eps_t \sim N(0,0.5^2)$, $x_t \equiv (1,x_{2,t}')'$, and $g_t \equiv (g_{1,t}',-1)'$. Both $x_{2,t}$ and $g_{1,t}$ follow the vector autoregressive model of order 1:
$
x_{2,t}  = \rho_x x_{2, t-1} + \nu_t,
g_{1,t}  = \rho_g g_{1, t-1} + u_t,
$
where $\nu_t \sim N(0, I_{d_x-1})$ and $u_t \sim N(0,I_{K})$. When the factor $g_t$ is not observable, we instead observe $\mathcal{Y}_t$ that is generated from
$\mathcal{Y}_t = \Lambda g_{1,t} + \sqrt{K}e_t, e_t = \rho_e e_{t-1} + \omega_t$,
where $\mathcal{Y}_t$ is an $N\times 1$ vector and $\omega_t$ is an i.i.d.\ innovation generated from $N(0,I_N)$.
The terms $\eps_t$, $\nu_t$, $u_t$, and $\omega_t$ are mutually independent.

In the baseline model, we set $T=N=200$, $d_x=2$, and $K=3$, and apply the MIQP algorithm.
The additional parameter values are set as follows:
$\bt_0=\dt_0=(1,1)$;
$\phi_0=(1,2/3,0,2/3)$;
$\rho_x = \diag(0.5,\ldots,0.5)$;
$\rho_g = \diag(\rho_{g,1},\ldots,\rho_{g,K})$, where $\rho_{g,k} \sim U(0.2,0.8)$ for $k=1,\ldots,K$,
the $i$th row of $\Ld$, $\ld_i' \sim N(0', K \cdot I_K)$; and
$\rho_e = \diag(\rho_{e,1},\ldots,\rho_{e,N})$, where $\rho_{e,i} \sim U(0.3,0.5)$ for $i=1,\ldots,N$.
The values of $\rho_g$ and $\rho_{e}$ are drawn only once and kept for the whole replications. The factor model design is similar to Bai and Ng \citep{bai2009boosting} and Cheng and Hansen \citep{cheng2015forecasting}.
All simulation results are based on 1,000 replications unless otherwise mentioned. We use a desktop computer equipped with an AMD RYZEN Threadripper 1950X CPU (16 cores with 3.4 GHz) and 64 GB RAM. The replication R codes for both the Monte Carlo experiments and empirical applications are available at \url{https://github.com/yshin12/fadtwo}. Also, the full simulation results can be found in Tables \ref{tb1:base}--\ref{tb:com-large-1000}  in Online Appendix \ref{sec:add:sim:appendix}.

\begin{figure}[thb]
	\centering
	\caption{Simulation Results: Baseline Model}\label{fig-base-model}
	\begin{tabular}[t]{c c}
		\includegraphics[width=6cm, height=4cm]{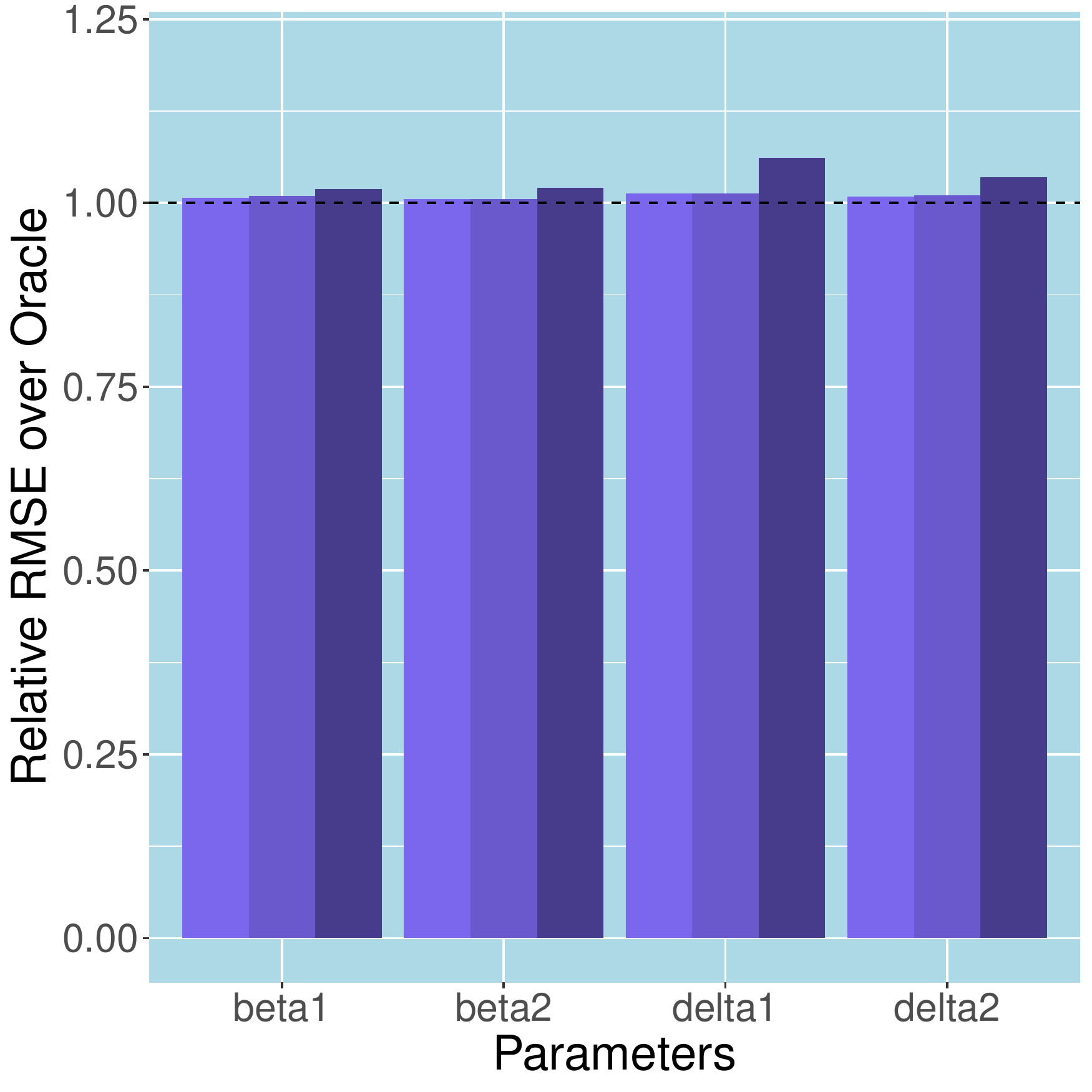} & \includegraphics[width=6cm, height=4cm]{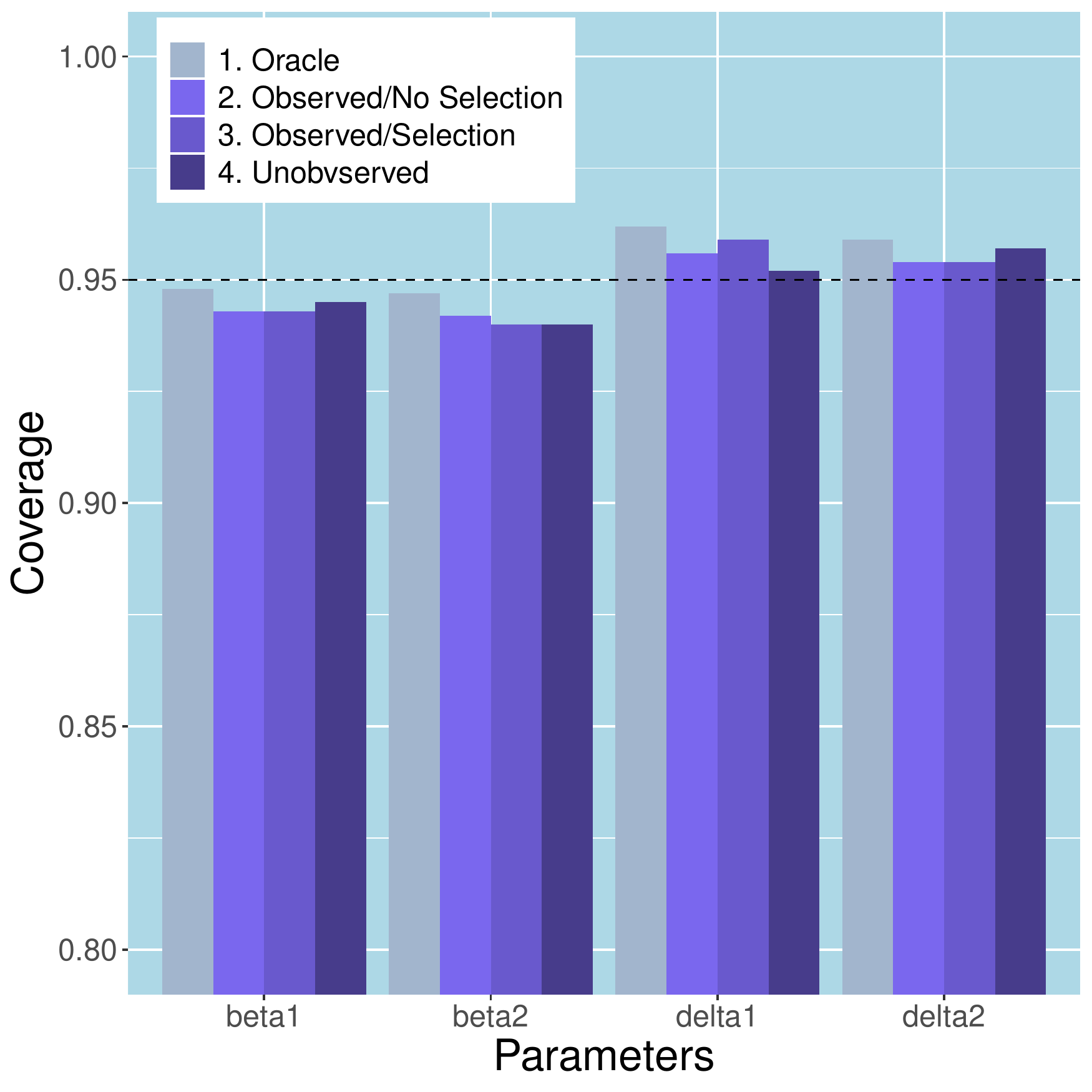}
	\end{tabular}
\end{figure}

First, we study the baseline model under four scenarios: (i) when we know the correct regime, i.e.\ $\phi_0$, (Oracle); (ii) when we observe $g_t$ and know that the third factor is irrelevant (Observed Factors/No Selection); (iii) when we observe $g_t$ and have to select the relevant factors (Observed Factors/Selection); and (iv) when we do not observe $g_t$ but estimate factors from $\mathcal{Y}_t$ by PCA. We set the dimension of $\gm$ to be 4 in (iv). Figure \ref{fig-base-model} reports the relative size of the root-mean-square errors (RMSEs) for $\bt$, $\dt$ as well as the coverage rate for the 95\% confidence intervals. As predicted by the asymptotic theory in the previous sections, the relative RMSEs over Oracle are close to 1 in all scenarios. The coverage rates for the 95\% confidence intervals are also close to the nominal value. Not surprisingly, these results on $\alpha$ are based on the good estimation performance of $\phi$ (or $\gm$).

\begin{figure}[h]
\centering
\caption{Unobserved Factors with Different $N$}\label{fig-sim2}
\begin{tabular}[t]{c c c}
\includegraphics[scale=0.25]{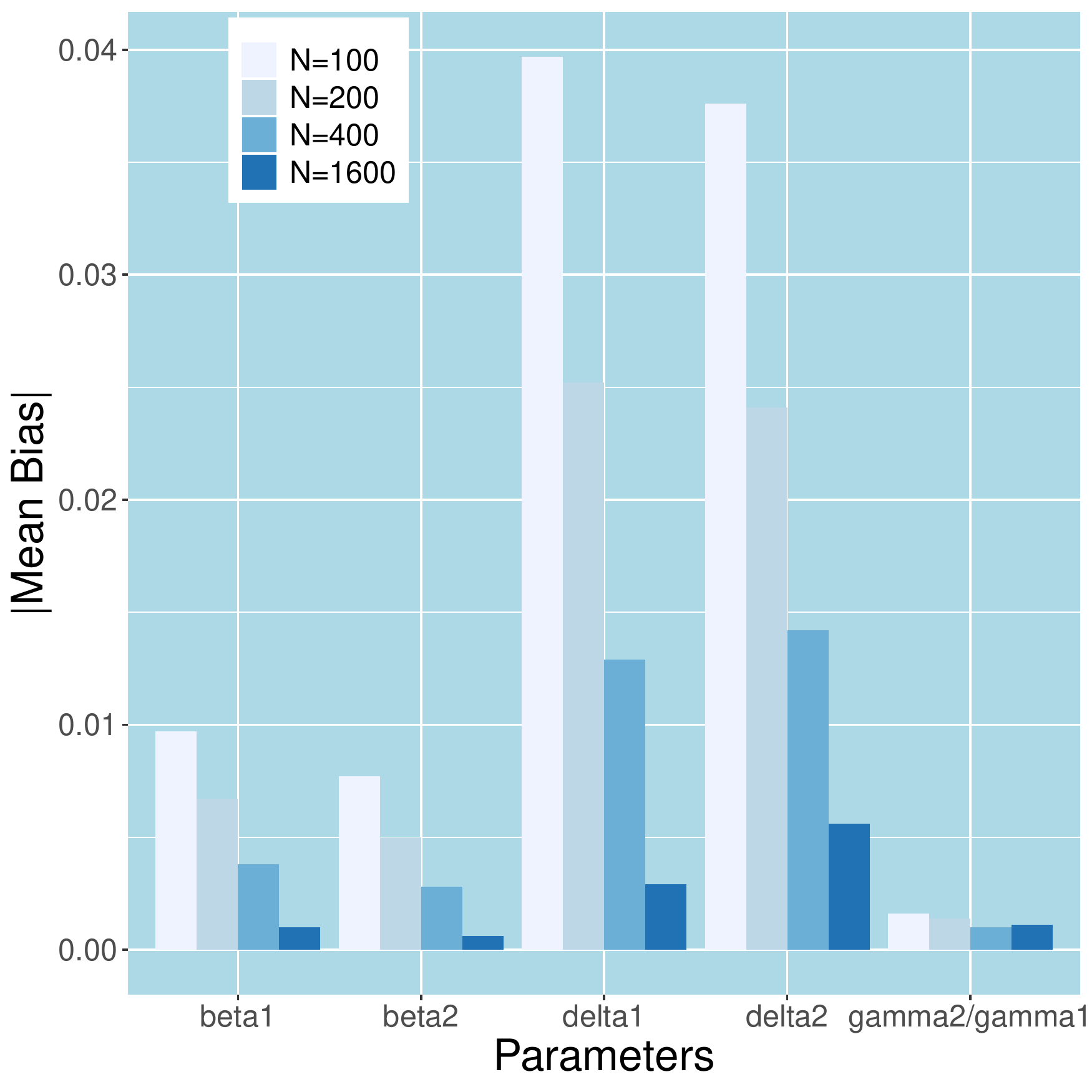} & \includegraphics[scale=0.25]{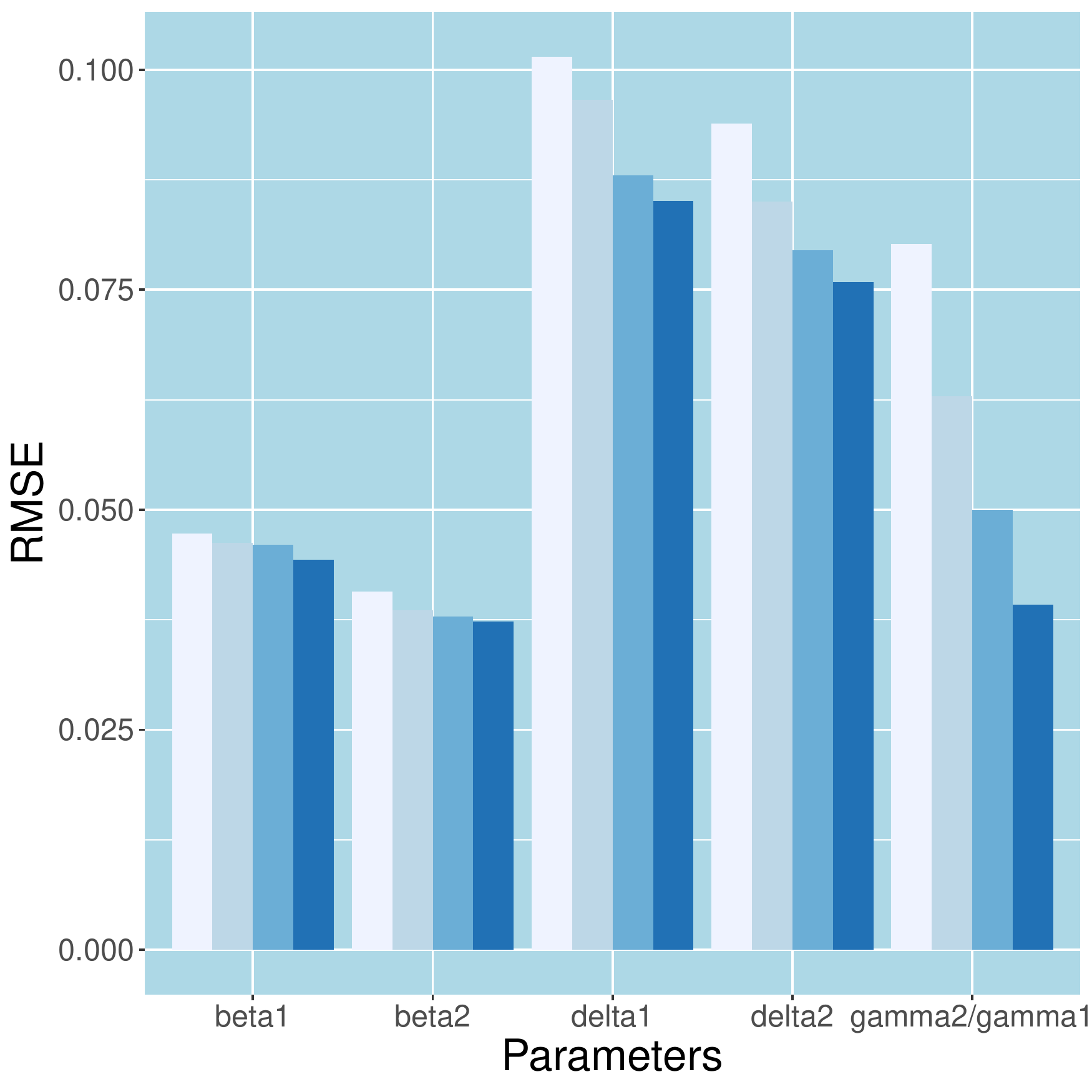}  & \includegraphics[scale=0.25]{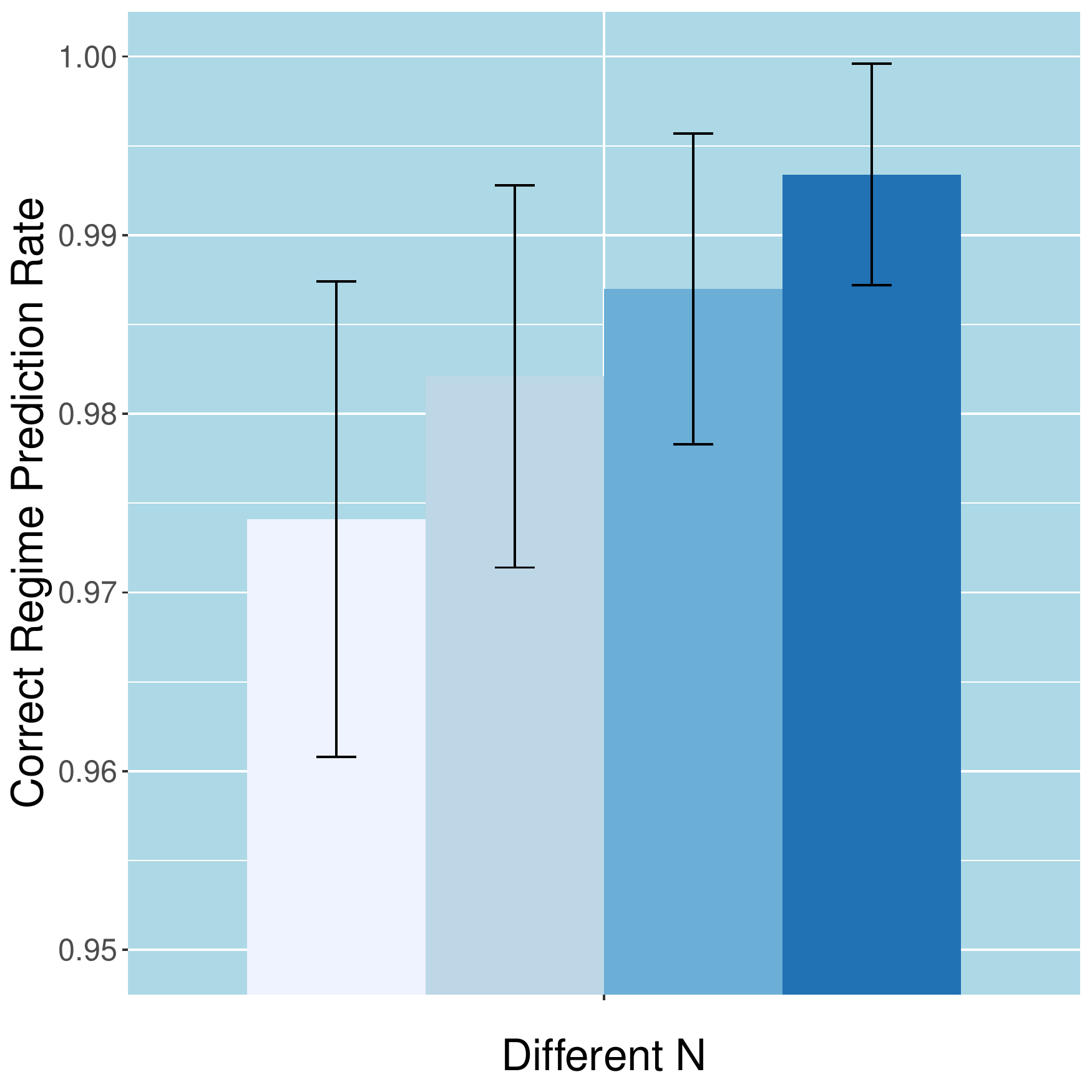}
\end{tabular}
\par
\parbox{5in}{\footnotesize Note. The whisker plot in the panel on the right denotes one standard deviation computed over replication draws.}
\end{figure}

\begin{table}[h]
\caption{Size of Bootstrap Test}\label{tb:bootstrap}
\centering
\begin{tabular}{llcc}
\hline\hline
\multirow{2}{*}{Null hypothesis} & \multirow{2}{*}{Scenarios}                   & \multicolumn{2}{c}{Significance level}                                    \\
                   \cline{3-4}
&                   & \multicolumn{1}{c}{5\%} & \multicolumn{1}{c}{1\%} \\
\hline
$H_0: \gm_{02}=0$ & Estimated factor   & 3.8\%                   & 0.7\%                   \\
$H_0: \phi_{02}=0$ & Known factor       & 4.3\%                   & 0.5\%                   \\
$H_0: \phi_{02}=0 \mbox{ and } \phi_{03}=0$& Many known factors & 7.5\%                   & 1.1\%         \\
\hline
\end{tabular}
\end{table}

\begin{figure}[h]
\centering
\caption{Computation Time over $T$, $d_x$, and $d_g$}\label{fig-time1}
\begin{tabular}[t]{c c c}
\includegraphics[width=5cm, height=4cm]{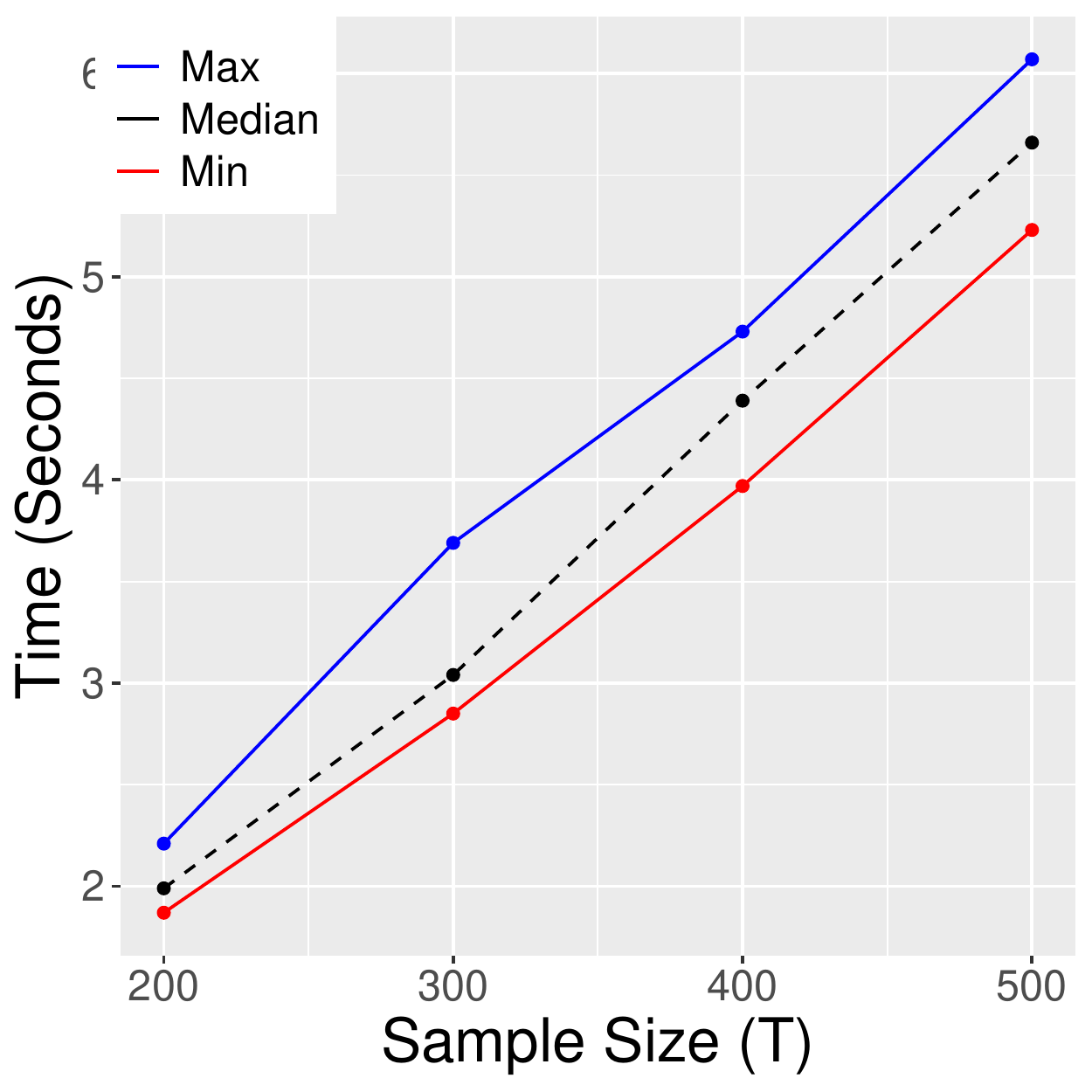} &
\includegraphics[width=5cm, height=4cm]{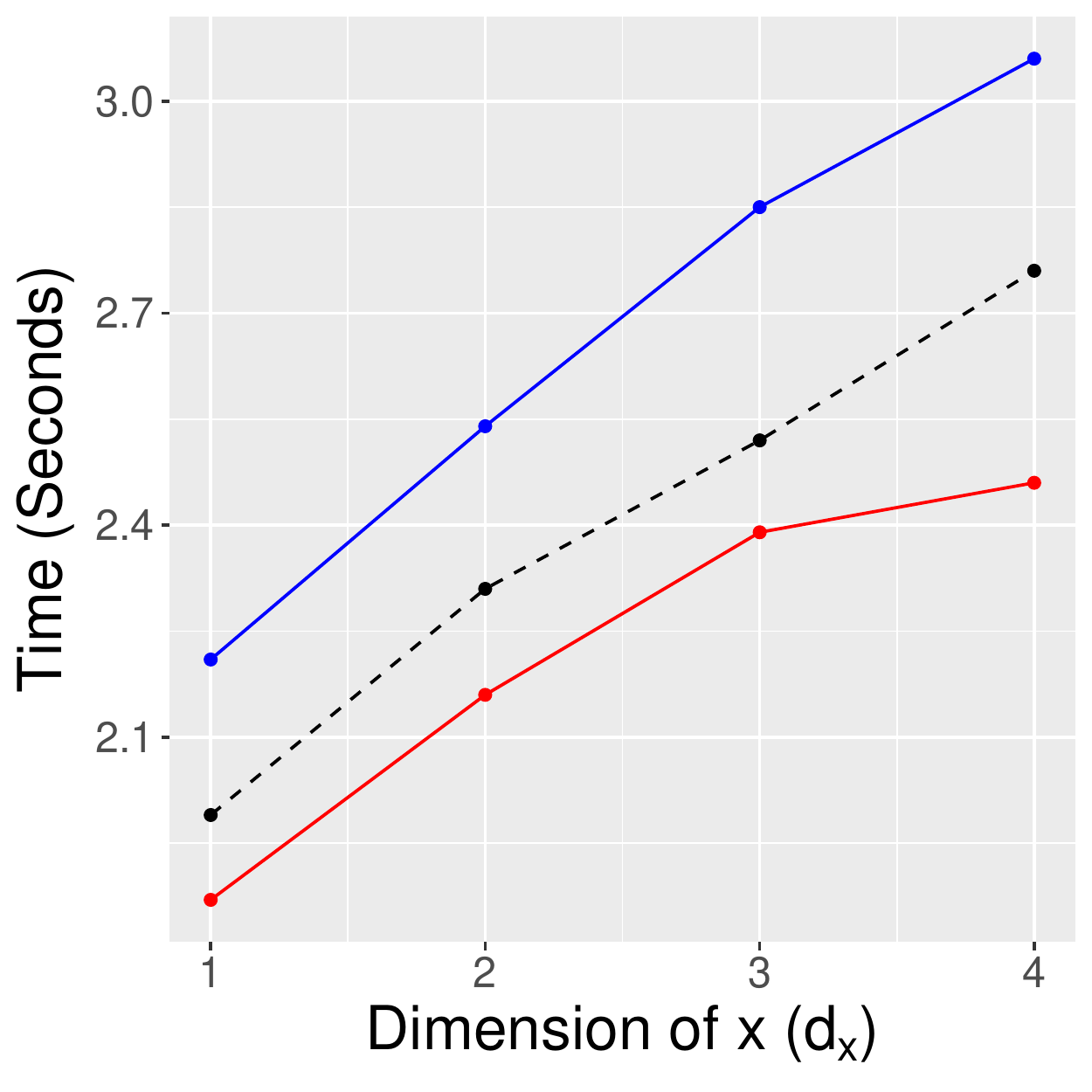}  &
\includegraphics[width=5cm, height=4cm]{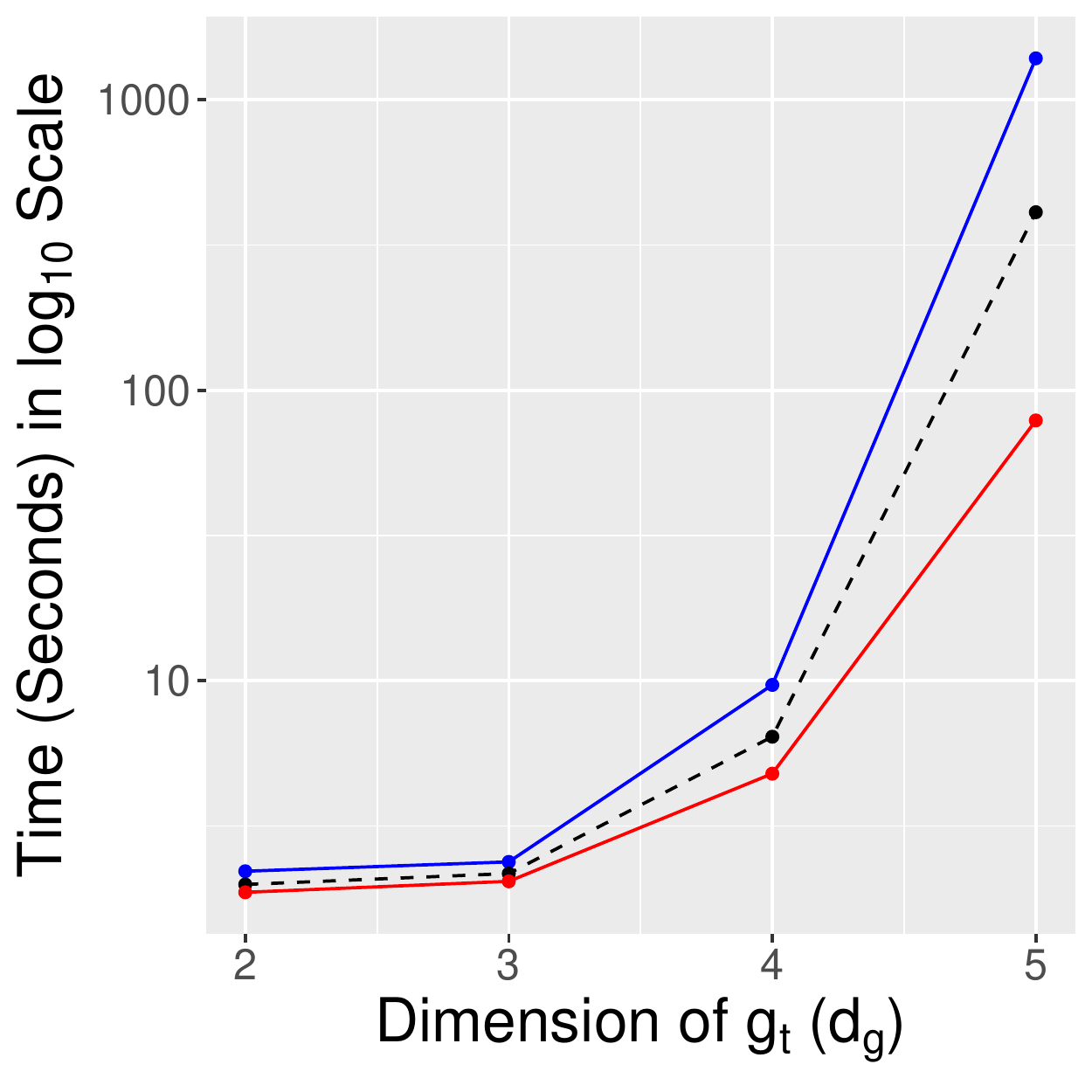}
\end{tabular}
\end{figure}

\begin{figure}[h]
\centering
\caption{Large Dimensional Models ($T=500$)}\label{fig-large}
\begin{tabular}[t]{c c}
$d_x = 6$ & $d_x=10$  \\
 \includegraphics[width=5cm, height=5cm]{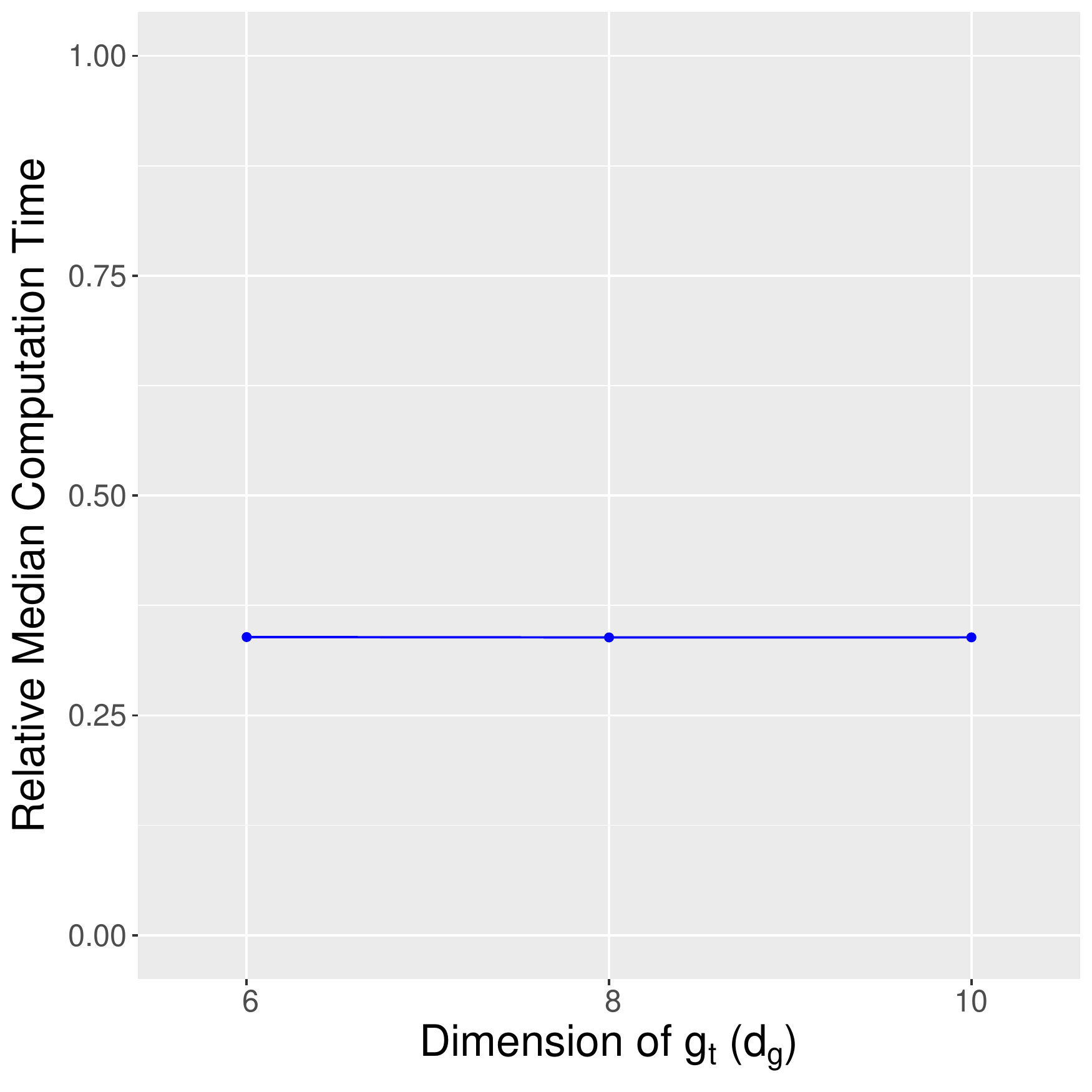} &
 \includegraphics[width=5cm, height=5cm]{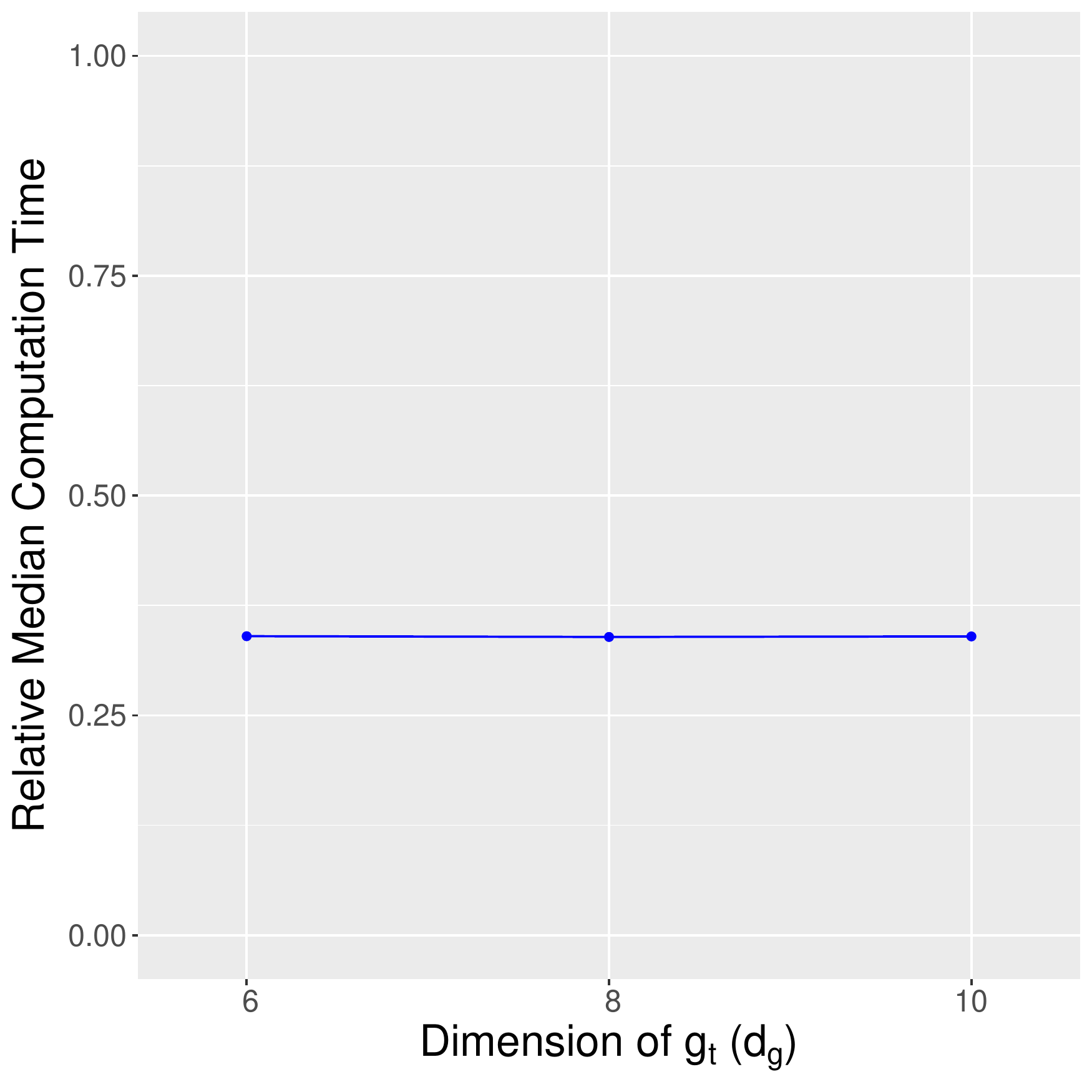}
  \\
 \includegraphics[width=5cm, height=5cm]{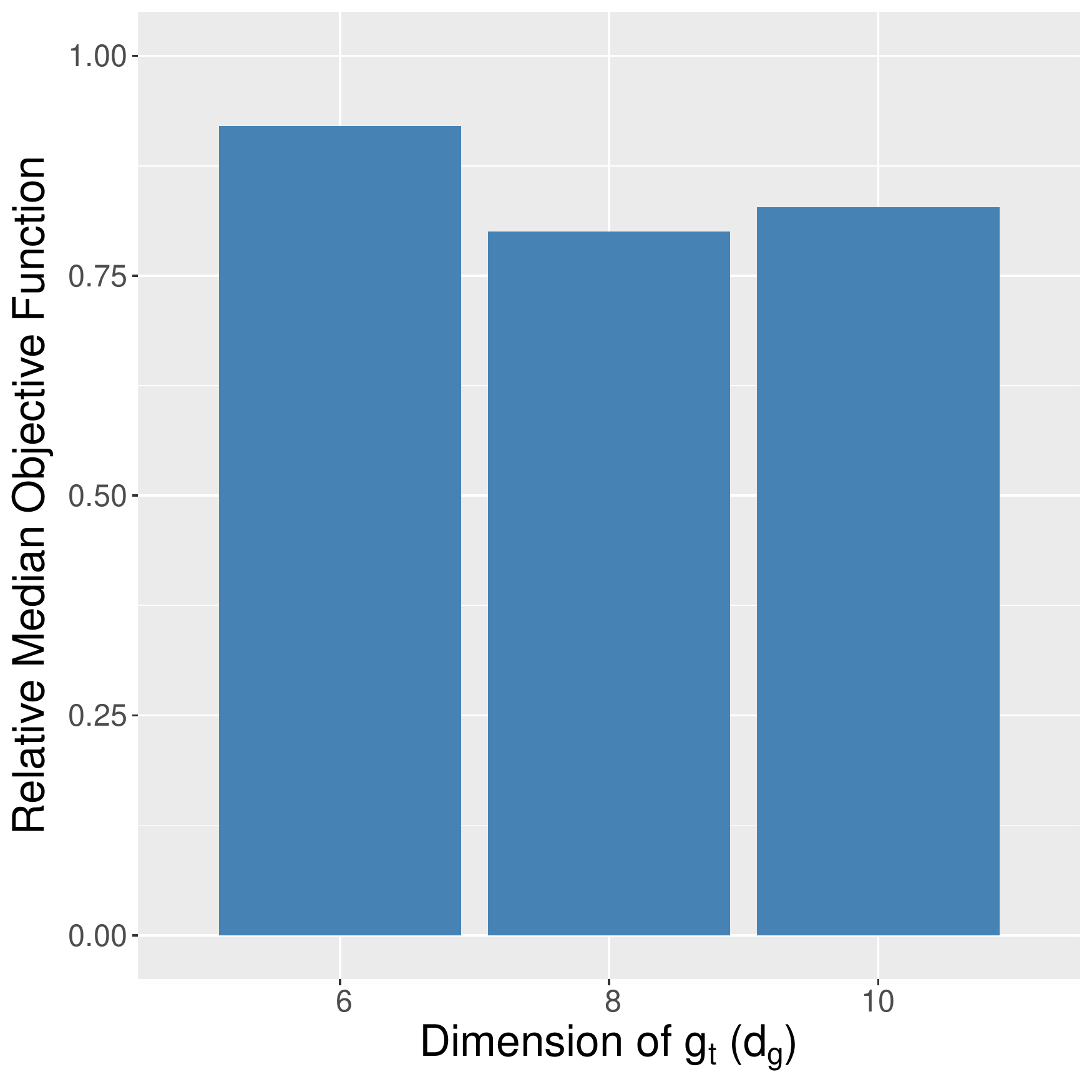} &
 \includegraphics[width=5cm, height=5cm]{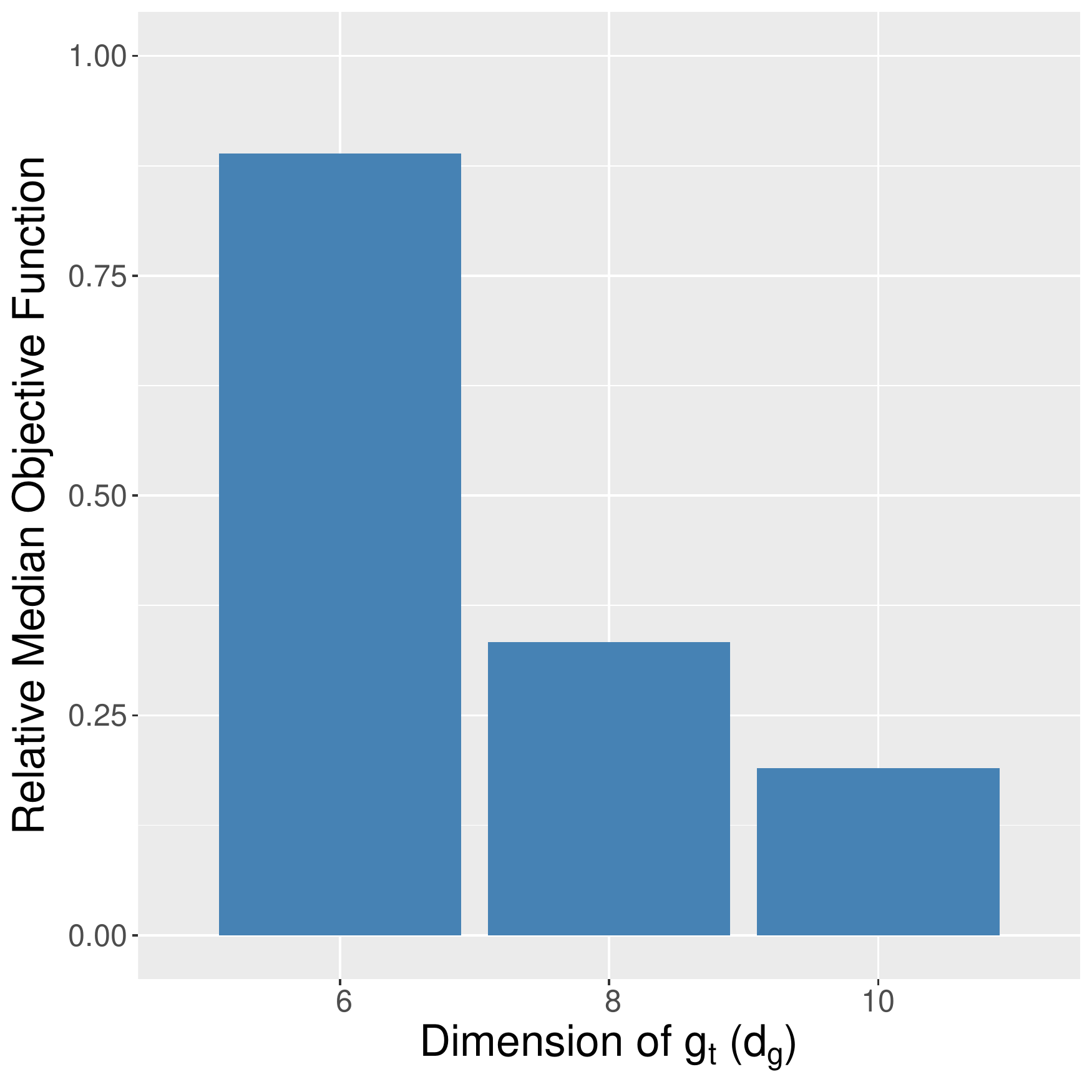}   \\
\end{tabular}
\par
\parbox{5.5in}{\footnotesize Note. The relative measures  are calculated by dividing the outcome of BCD by that of MIQP.}
\end{figure}

Second, we focus on the unobserved factor model and investigate the performance as $N$ increases. For each simulated sample of $\{y_t, x_t, g_t\}$, we generate $\mathcal{Y}_t$ with $N=100, 200, 400, 1600$. We use the same baseline design with $T=200$, $d_x=2$, but $K=1$ to speed up computations. Figure \ref{fig-sim2} summarizes the results. The regimes are predicted more precisely as $N$ increases and the performance of the estimator improves. We observe relatively more improvements in $\gm$ rather than $\alpha$. This is because $\widehat \alpha$ already enjoys the oracle property, provided that $T= O(N)$.

Third, we investigate the performance of the bootstrap test under three scenarios: (i) an estimated factor; (ii) a known factor;  (iii) many known factors. The parameters are set as follows: $T=200$, $N=400$, $B=499$, $\eps_t\sim N(0,1)$, $\eta_t \sim N(0,1)$, $\bt_0=(1,1)$, $\dt_0=(0.5,0.5)$, $\gm_0 (\mbox{or }\phi_0) = (1,0)$ in (i) and (ii), and $\phi_0 = (1,0,0,0)$ in (iii). We test a simple null hypothesis of $H_0: \gm_{02}(\mbox{or }\phi_{02})=0$ in (i) and (ii) and a joint hypothesis of $H_0: \phi_{02}=\phi_{03}=0$ in (iii). There is no serial correlation in the model ($\rho_x=\rho_g=\rho_e=0$). Table \ref{tb:bootstrap} reports the size of the bootstrap test in each scenario and it is satisfactory but we observe over-rejection in the joint hypothesis case.

We next investigate the computation time. We start from a set of simple models and extend to large dimensional models.
We simplify the baseline model by considering scenario (ii) (i.e., Observed/No Selection), and by setting $\rho_x=\rho_g=0$. The results are based on 100 replications. We set $T=200$, $d_x=1$, and $d_g=2$, initially and increase each dimension as follows: $T=\{200, 300, 400, 500\}$, $d_x=\{1, 2, 3, 4\}$ while keeping $T=200$ and $d_g=2$; $d_g=\{2,3,4,5\}$ while keeping $T=200$ and $d_x=1$.
Figure \ref{fig-time1} reports the computation time of MIQP.
The results indicate that the computation time stays in a reasonable bound and increases linearly as $T$ and $d_x$ increase. However, it increases exponentially as $d_g$ increases.

\begin{figure}[h]
\centering
\caption{Larger Dimensional Models using BCD ($T=1000$ and $d_x = 6$)}\label{fig-very-large}
\begin{tabular}[t]{c c}
 \includegraphics[width=5cm, height=5cm]{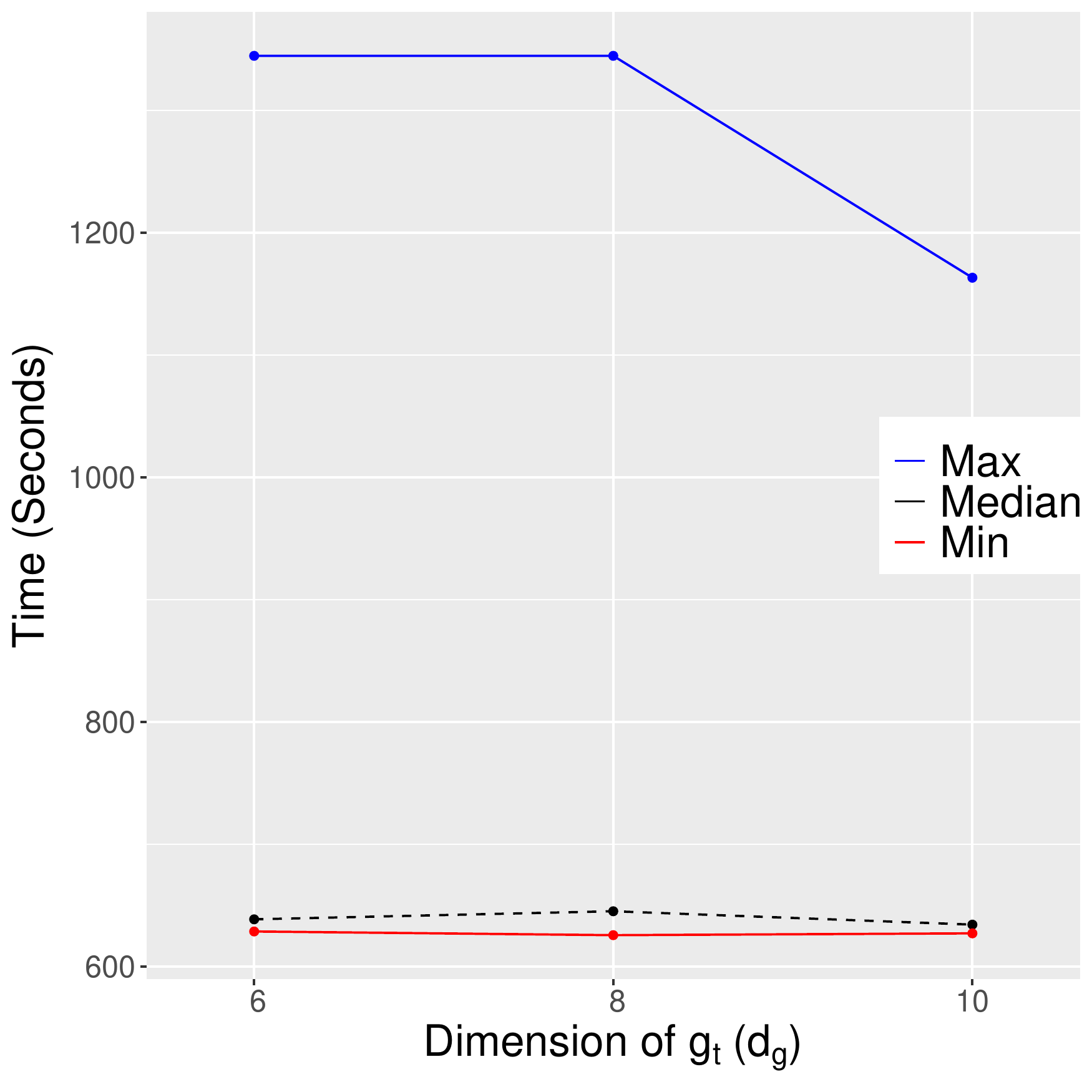}  &
 \includegraphics[width=5cm, height=5cm]{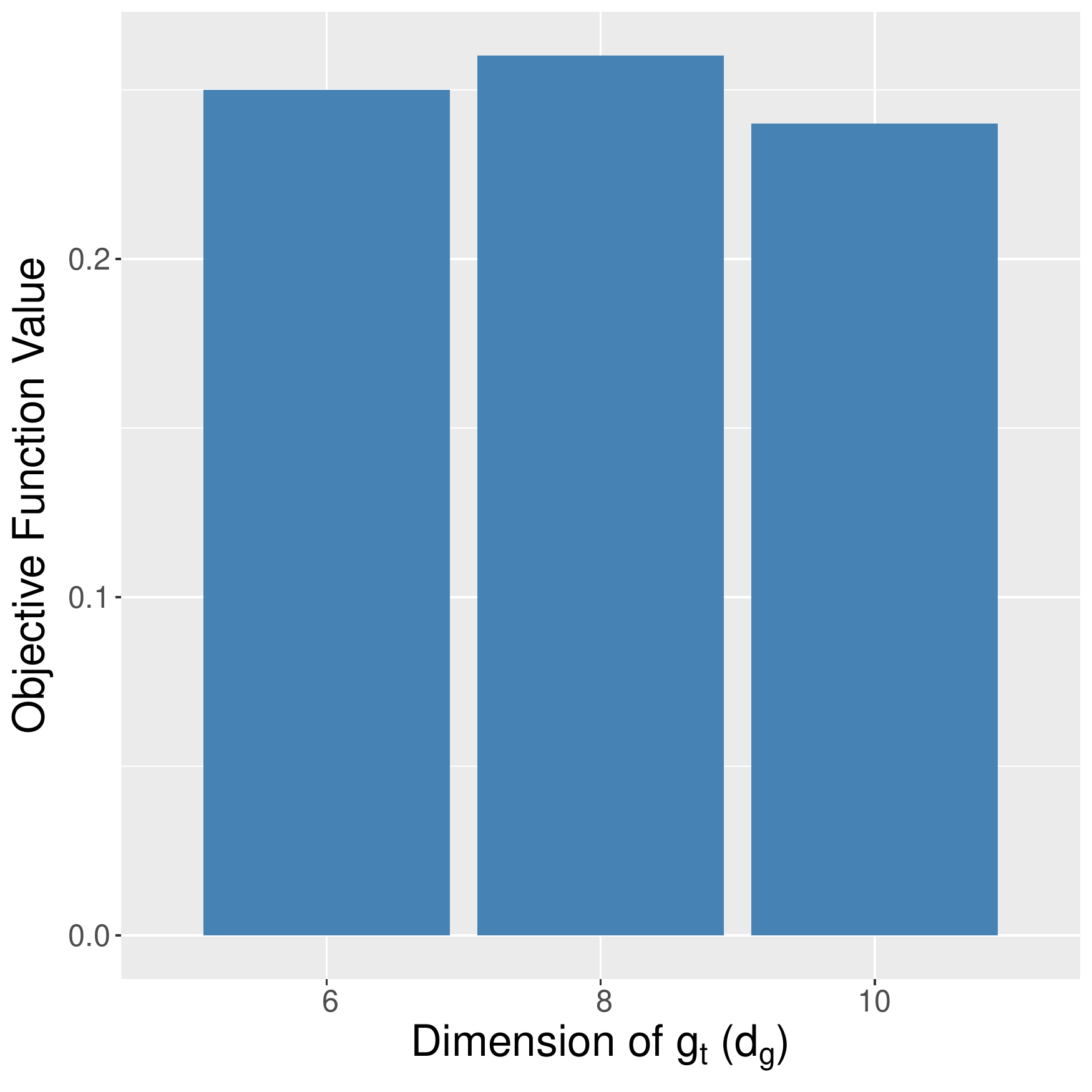}  \\
\end{tabular}
\end{figure}

We now consider large dimensional models and handle the computational challenge by implementing the BCD algorithm in addition to MIQP. We extend the dimension of the models as $T=\{500, 1000\}$, $d_x=\{6, 8, 10\}$, and $d_g = \{6, 8, 10\}$. Note that $d_g=10$ would be quite challenging and the standard grid search method would be infeasible in practice with $T=$ 1,000. The results are based on 10 iterations of each model. We set the total time budget as 1,800 seconds for both MIQP and BCD so that each estimation terminates after that even if it does not converge. In BCD, we set \texttt{MaxTime\_1}=600 (seconds) and \texttt{MaxTime\_2}=60 (seconds).
Figure \ref{fig-large} reports the ratio of the median computation time and median objective function values between BCD and MIQP when $T=500$. BCD spends a third of the computation time, whereas MIQP spends the total time budget. BCD achieves better objective function values in all cases and the performance of MIQP deteriorates quickly as $d_g$ increases when $d_x = 10$. Figure \ref{fig-very-large} reports the summary statistics of computation time and the median objective function values of BCD when $T=$ 1,000. As the computation is more challenging, we observe that the maximum computation time is higher for all $d_g$. However, the median computation time is still around 600 seconds and the achieved objective function values are quite stable.

Based on our simulation studies, we propose to use the BCD algorithm by assigning $1/3$ of the total time budget into the maximum time (\texttt{MaxTime\_1}) for Step 1 (MIQP). When the global solution is not attainable within \texttt{MaxTime\_1}, the BCD algorithm would switch into Steps 2--3 (MILP) automatically. We recommend assigning 1/30 of the total time budget into the maximum time (\texttt{MaxTime\_2}) for each cycle of Step 2. 

In summary, the simulation studies reveal that the proposed method achieves the properties predicted by the asymptotic theory, especially the oracle property of $\ap$ and the inference based on the bootstrap method. The BCD algorithm also shows quite satisfactory results in a large dimensional change-point model whose computation is infeasible with grid search.

\section{Classifying the Regimes of US Unemployment}\label{sec:real-data-app}

We revisit the empirical application of Hansen \citep{Hansen:97},  who
considered  threshold autoregressive models for the US unemployment rate.
Specifically, Hansen \citep{Hansen:97} used monthly unemployment rates (i.e., $u_t$) for males age 20 and over, and  set $y_t = \Delta u_t$ in \eqref{model1}. The lag length in the autoregressive model was $p=12$ and the preferred threshold variable was
$q_{t-1} = u_{t-1} - u_{t-12}$.  In this section, we investigate the usefulness of using unknown but estimated factors. We use  the first principal component (i.e., $F_t$) of Ludvigson and Ng \citep{Ludvigson:Ng:09}  that is estimated from 132 macroeconomic variables. 
This  factor not only explains the largest fraction of the total variation in their panel data set but also loads heavily on employment, production, and so on. Ludvigson and Ng call it a \emph{real factor} and thus it is a legitimate candidate for explaining the unemployment rate.
We consider three different specifications for $f_t$:  (1) $f_{1t} = (q_{t-1}, -1)$,  (2) $f_{2t} = (F_{t-1}, -1)$, and (3) $f_{3t} = (q_{t-1}, F_{t-1}, -1)$.
We combined the updated estimates of the real factor, which are available on Ludvigson's web page at \url{https://www.sydneyludvigson.com},
with Hansen's data, yielding a  monthly sample from March 1960 to July 1996.

\begin{table}[h]
\caption{Estimation Results}
\begin{center}
\begin{tabular}{lrrrrrr}
\hline\hline
 Specification   & \multicolumn{2}{c}{(1) $f_{1t} = (q_{t-1}, -1)$} & \multicolumn{2}{c}{(2) $f_{2t} = (F_{t-1}, -1)$} & \multicolumn{2}{c}{(3) $f_{3t} = (q_{t-1}, F_{t-1}, -1)$} \vspace*{1ex} \\
\hline
Regime 1    (``Expansion'')        & \multicolumn{2}{c}{$q_{t-1} \leq 0.302$} & \multicolumn{2}{c}{$F_{t-1} \leq -0.28$} & \multicolumn{2}{c}{$q_{t-1} + 3.55 F_{t-1} \leq -1.60$} \vspace*{1ex} \\
 Prediction error      & \multicolumn{2}{c}{0.0264}  & \multicolumn{2}{c}{0.0272} & \multicolumn{2}{c}{0.0252} \vspace*{1ex} \\
Classification error         & \multicolumn{2}{c}{0.193}  & \multicolumn{2}{c}{0.106} & \multicolumn{2}{c}{0.104} \vspace*{1ex} \\
\hline
\end{tabular}
\end{center}
\label{table-hansen97-short}
\parbox{6.5in}{\footnotesize Note. See Table \ref{table-hansen97} in the Online Appendix for estimated coefficients and their heteroskedasticity-robust standard errors.
Regime 2 (``Contraction'') is the complement of regime 1.
``Prediction Error'' refers to the average of squared residuals $(T^{-1} \sum_{i=1}^T \widehat{\varepsilon}_t^2)$.
``Classification Error'' corresponds to the proportion of misclassification defined in \eqref{def:correct_matches}.
}
\end{table}

\begin{figure}[h]
	\caption{Regime Classification}
	\label{fig-hansen97}
	\begin{center}
		\graphicspath{ {plot/} }
		\includegraphics[scale=0.3]{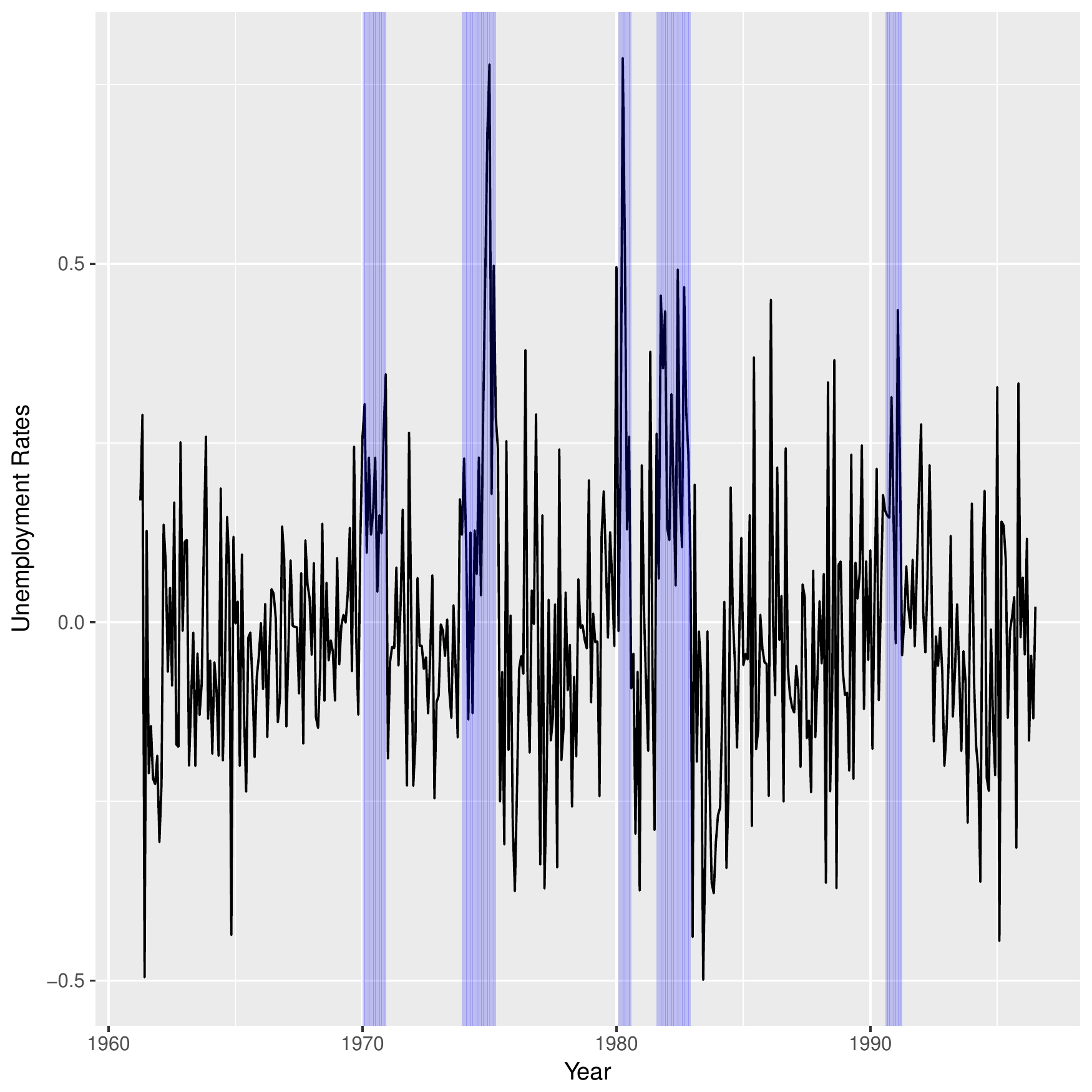}
		\includegraphics[scale=0.3]{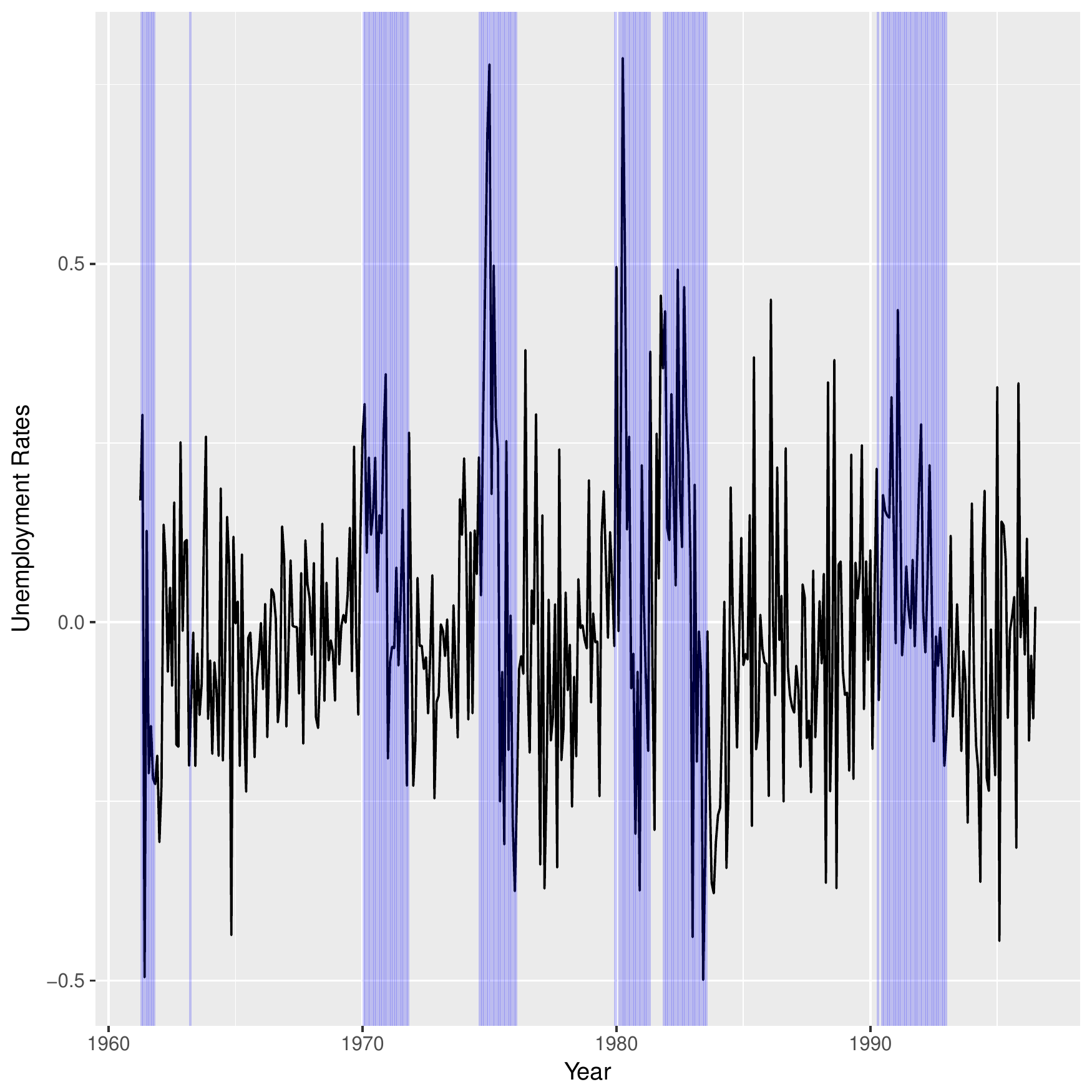}
		\includegraphics[scale=0.3]{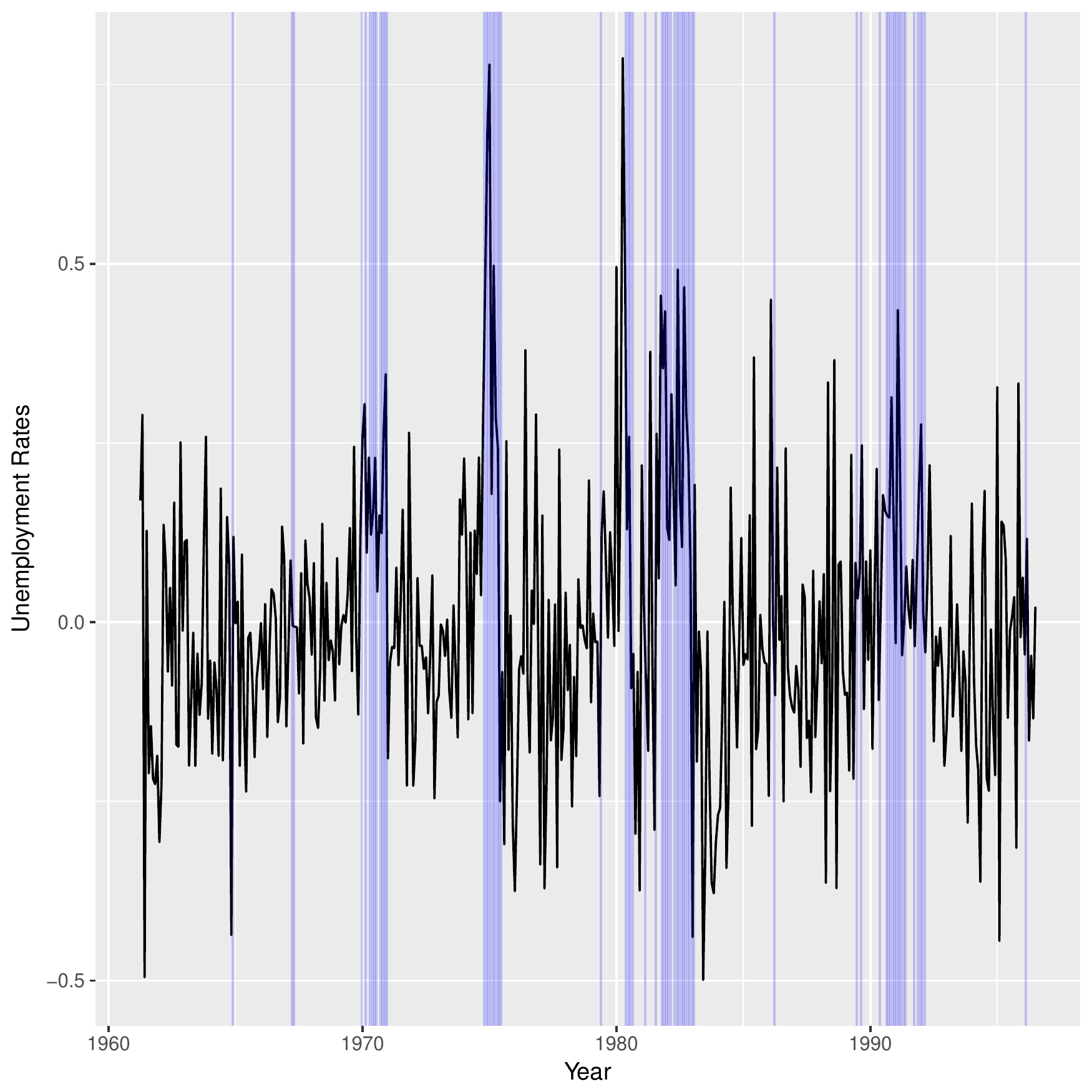}
		\includegraphics[scale=0.3]{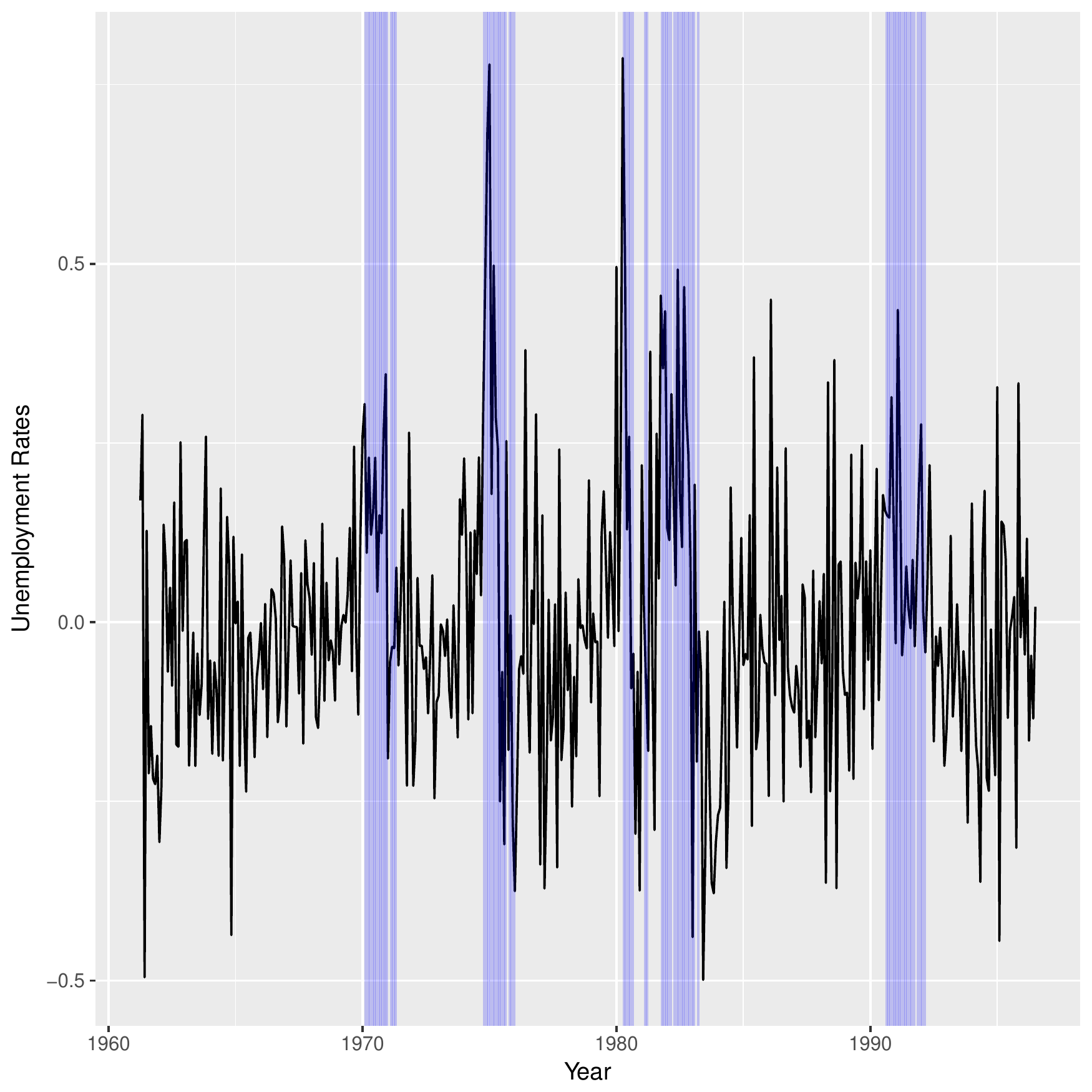}
	\end{center}

\parbox{5in}{Note. The leftmost panel shows NBER recession dates in the shaded area, and
the other three panels display those with specifications (1), (2) and (3), respectively.}
\end{figure}

Table \ref{table-hansen97-short} reports estimation results that are obtained by the MIQP algorithm. We
show the goodness of fit by reporting the average  squared residuals and also the results of regime misclassification relative to the NBER business cycle dates.
The latter is obtained by
\begin{align}\label{def:correct_matches}
 \frac{1}{T}\sum_{t=1}^{T}\Big\vert 1\left\{ {f}_{jt}^{\prime }\widehat{\gamma}_j>0\right\} - 1_{\text{NBER}, t}  \Big\vert   \; \text{ for each $j=1,2,3$},
\end{align}
where $1_{\text{NBER}, t}$ is the indicator function that has value 1 if and only if the economy is in contraction according to the NBER dates.
Accordingly, we label regime 1 ``expansion'' and regime 2  ``contraction'', respectively.
Figure \ref{fig-hansen97} gives the graphical representation of regime classification.
Specification (1) suffers from the highest level of misclassification and tends to classify recessions more often than NBER.
Specification (2) mitigates the misclassification risk but at the expense of a worse goodness of fit. On one hand,
the threshold autoregressive model  solely by $q_{t-1}$  fittingly explains the unemployment rate but  is short of classifying the overall economic conditions satisfactorily. On the other hand, the model based only on $F_{t-1}$ is adequate at describing the underlying overall economy but does not  explain the unemployment rate well.
It turns out that specification (3) has the lowest misclassification error and best explains unemployment. Thus, we have shown the real benefits of  using a vector of possibly unobserved factors to explain the unemployment dynamics.

As an additional check,  we tested the null hypothesis of no threshold effect.   The resulting  $p$-value is 0.002 based on 500 bootstrap replications, thus providing strong evidence for the existence of two regimes. See Table \ref{table-hansen97} in the Online Appendix for details and additional results.

\section{Conclusions}\label{sec:conclusions}



We have proposed a new method for estimating a two-regime regression model where
regime switching  is driven by a vector of possibly unobservable factors. We show that our optimization problem can be reformulated as  MIO and have presented two alternative computational algorithms.
We have also derived the asymptotic distribution of the resulting estimator under the scheme that the threshold effect shrinks to zero as the sample size tends to infinity.
 As a possible interesting extension,  we can consider nonparametric regime switching, where the switching indicator is replaced by $1\{F(w_t)>0\}$ with  a vector of observables $w_t$ and a nonparametric function $F(\cdot)$.  We intend to study this in the future.



%
%
\newpage

\appendix

\addcontentsline{toc}{section}{Appendix} 
\part{Appendix} 
\parttoc 

\newpage

\section{Identification}
\label{sec:identification}

In this section, we  establish sufficient conditions under which $(\beta_0', \delta_0', \gamma_0')'$ is identified.
Recall that the covariates $x_t$ and $f_t$ may not be directly observable  in our general setup; however, since we assume that they can be consistently estimable, it suffices to consider the identification of the unknown parameters under the simple setup that $x_t$ and $f_t$ are observed directly from the data.

If  there is no random variable in $f_t$ with a non-zero coefficient, $\gamma_0$ is unidentifiable.
 Assumption  \ref{scale-normalization} in the main text
avoids this directly by assuming that the first coefficient of $\gamma_0$ is 1.\footnote{Alternatively, it could be $-1$; however, the choice between $+1$ and $-1$ is just a labelling issue since two regimes are equivalent up to reparametrization of $\alpha_0$ under either scale normalization.}
We partition $f_{t}=\left( f_{1t},f_{2t}^{\prime }\right) ^{\prime }$ and $%
\gamma =\left( 1,\gamma _{2}^{\prime }\right) ^{\prime }$,
and write, occasionally,
$
1\left\{ f_{1t}>f_{2t}^{\prime }\gamma \right\}
$
instead of
$
1\left\{ f_{t}^{\prime }\gamma >0\right\}.
$

\begin{remark}[Alternative Scale Normalization]
We may consider an alternative parameter space for $\gamma_0$:
$
\gamma_0 \in \Gamma  \equiv
    \{  \gamma:  |\gamma|_2 = 1, \gamma \neq (0,...,0,1)' $, and $ \gamma \neq (0,...,0,-1)' \}.
 $
This parameter space excludes the case of no real threshold variable by assuming that both $|\gamma|_2 = 1$ and $\gamma \neq (0,...,0,\pm 1)'$ (recall that the last element of $f_t$ is $-1$).
Assumption \ref{scale-normalization} is more convenient for computation since it reduces the number of unknown parameters but it requires to know which factor has a non-zero coefficient.
On the other hand, the alternative parameter space might be more attractive when it is difficult to know which factor has a non-zero coefficient  \emph{a priori}. We focus on the former throughout the paper; however, the main results of the paper could be obtained under the latter.
\end{remark}

We make the following regularity conditions.

\begin{assum}[Identification]
\label{A-f1 conti}
\begin{enumerate}[label=(\roman*)]
\item\label{A1:itm1}
There exists an element $f_{jt}$ in $f_t$ such that   $\gamma _{j0}\neq 0$ and the conditional
distribution of $f_{jt}$ given $f_{-j,t}$ is continuous with
probability one, where $f_{-j,t}$ is the subvector of $f_{t}$ excluding $%
f_{jt}$.
\item\label{A1:itm2}
Let $B_{\gamma t} \equiv \left\{ f_t'\gamma_0\leq
0<f_{t}^{\prime }\gamma \right\} \cup \left\{ f_{t}^{\prime }\gamma \leq
0<f_t'\gamma_0\right\} .$ Then, for any $\gamma \in \Gamma$ such that $\gamma \neq \gamma_{0}$,
\begin{equation}
 \mathbb    E\left[ \left( x_{t}^{\prime }\delta _{0}\right) ^{2}1\left\{ B_{\gamma t
}\right\} \right] >0.  \label{eq:rank1}
\end{equation}%
\item\label{A1:itm3}
Let  $A_{1 \gamma t} \equiv \left\{ f_t'\gamma_0>0\right\} \cap \left\{ f_{t}^{\prime
}\gamma >0\right\} $ and
$A_{2 \gamma t} \equiv \left\{ f_t'\gamma_0\leq 0\right\} \cap \left\{
f_{t}^{\prime }\gamma \leq 0\right\}$. Then,
\begin{equation}
 \inf_{\gamma \in \Gamma }\mathbb E \left[ x_{t}x_{t}^{\prime }1\left\{ A_{1 \gamma t}\right\} \right] >0\; \; \text{and } \; \;
 \inf_{\gamma \in \Gamma }\mathbb E \left[ x_{t}x_{t}^{\prime }1\left\{ A_{2 \gamma t}\right\} \right] >0.  \label{eq:rank2}
\end{equation}
\end{enumerate}
\end{assum}

Recall that
\begin{align}\label{eq:excess:supp}
R(\alpha ,\gamma ) \equiv \mathbb E(y_{t}-x_{t}^{\prime }\beta -x_{t}^{\prime }\delta
1\{f_{t}^{\prime }\gamma >0\})^{2}- \mathbb E(y_{t}-x_{t}^{\prime }\beta
_{0}-x_{t}^{\prime }\delta _{0}1\{f_t'\gamma_0>0\})^{2}.
\end{align}%
Note that  under Assumption \ref{A-f1 conti}\ref{A1:itm1}, $R\left( \cdot ,\cdot \right) $
is continuous.
The condition (\ref{eq:rank1}) ensures the presence of a change in
the regression function. If $\delta_0 = 0$, then \eqref{eq:rank1} is not satisfied.
A sufficient condition for \eqref{eq:rank1} is to assume that  that there exists
some $\eta >0$ such that any open subset  of $F_{\eta } \equiv \left\{
f_{t}:\left\vert f_t'\gamma_0\right\vert \leq \eta \right\} $
possesses a positive probability (dense support) and that%
\begin{equation*}
 \mathbb E\left[ \left( x_{t}^{\prime }\delta _{0}\right) ^{2} \big| f_{t}=z\right] >0
\end{equation*}%
for all but finitely many $z\in \left\{ z:\left\vert z^{\prime }\gamma
_{0}\right\vert \leq \eta \right\} $ (rank condition).

The condition (\ref{eq:rank2}) is satisfied, for example, if
\begin{equation}\label{rank2-suff}
 \mathbb E \left[ x_{t}x_{t}^{\prime }1\left\{ \inf_{\gamma \in \Gamma} f_{t}^{\prime }\gamma >0\right\} \right] >0\ \text{and }%
 \mathbb E \left[ x_{t}x_{t}^{\prime }1\left\{ \sup_{\gamma \in \Gamma} f_{t}^{\prime }\gamma \leq 0\right\}
\right] >0.
\end{equation}
Note that \eqref{rank2-suff} requires that (i) the parameter space $\Gamma$ satisfies \begin{equation}
\mathbb{P}\left( \bigcap_{\gamma \in \Gamma }\left\{ f_{t}^{\prime }\gamma >0\right\}
\right) >0\ \text{and }\mathbb{P}\left( \bigcap_{\gamma  \in \Gamma}\left\{
f_{t}^{\prime }\gamma \leq 0\right\} \right) >0  \label{eq:Gamma}
\end{equation}%
and (ii)
$ \mathbb E\left( x_{t}x_{t}^{\prime }|f_{t}=z\right) $ has full rank for
some $z$ belonging to $\left\{z: \inf_{\gamma \in \Gamma}z^{\prime }\gamma >0\right\} $
and also for some $z$ such that $\left\{z: \sup_{\gamma \in \Gamma}z^{\prime }\gamma \leq 0\right\}$.
In other words, there should be
some non-negligible fraction of observations in each regime for any $\gamma \in \Gamma$.
However, we cannot simply assume that $ \mathbb E\left( x_{t}x_{t}^{\prime }|f_{t}=z%
\right) >0$ for all $z$ since $x_{t}$ may contain $f_{t}$ and thus the
positive-definiteness may not hold for all $z$.

\begin{remark}
It is possible to provide sufficient conditions for Assumption \ref{A-f1 conti}  in a more compact form if $x_{t}$ does not contain $f_{t} = (f_{1t}, f_{2t}')'$ other than the constant $1$.
 For instance, in that case,
 it suffices to assume that $\delta _{0}\neq 0$, the conditional
distribution of $f_{1t}$ given $f_{2t}$ has everywhere positive density with
respect to Lebesque measure for almost every $f_{2t}$,  and both $ \mathbb E \left( f_{2t}f_{2t}^{\prime } \right) >0$ and $ \mathbb E\left( x_{t}x_{t}^{\prime }|f_{t}\right)
>0 $ a.s.
\end{remark}

The following theorem gives the identification and well-separability of $(\alpha_0', \gamma_0')'$.

\begin{thm}[Identification]\label{iden-thm}
If Assumptions \ref{scale-normalization} and \ref{A-f1 conti} hold, then $\left( \alpha_{0}',\gamma _{0}'\right) $ is the unique solution to%
\begin{equation*}
\min_{(\alpha', \gamma')' \in \mathbb{R}^{2d_x} \times \Gamma }
\mathbb E(y_{t}-x_{t}^{\prime }\beta -x_{t}^{\prime }\delta
1\{f_{t}^{\prime }\gamma >0\})^{2}
\end{equation*}%
and
\begin{equation*}
\inf_{\{(\alpha', \gamma')' \in \mathbb{R}^{2d_x} \times \Gamma: |(\alpha', \gamma') - (\alpha_0', \gamma_0')|_2 > \varepsilon \} }
R\left( \alpha,\gamma \right) >0
\end{equation*}%
for any $\varepsilon >0$.
\end{thm}
Theorem \ref{iden-thm} gives the basis for our estimator given in the main text.


\begin{proof}[Proof of Theorem \protect\ref{iden-thm}]
Note that
\begin{equation*}
R\left( \alpha ,\gamma \right) = \mathbb E\left( Z_{t}\left( \gamma \right) ^{\prime
}\alpha -Z_{t}\left( \gamma _{0}\right) ^{\prime }\alpha _{0}\right) ^{2}
\end{equation*}%
due to $\left( \ref{model1}\right) $ and $\left( \ref{model2}\right) $. We
consider two cases separately: (1) $\alpha =\alpha _{0}$ and $\gamma \neq
\gamma _{0}$ and (2) $\alpha \neq \alpha _{0}$.

First, when $\alpha =\alpha _{0}$ and $\gamma \neq \gamma _{0},$
\begin{equation*}
\left( Z_{t}\left( \gamma \right) ^{\prime }\alpha -Z_{t}\left( \gamma
_{0}\right) ^{\prime }\alpha _{0}\right) ^{2}=\left( x_{t}^{\prime }\delta
_{0}\right) ^{2}
\end{equation*}%
on $B_{\gamma }=\left\{ f_{t}^{\prime }\gamma _{0}\leq 0<f_{t}^{\prime
}\gamma \right\} \cup \left\{ f_{t}^{\prime }\gamma \leq 0<f_{t}^{\prime
}\gamma _{0}\right\} .$ Thus,
\begin{equation*}
R\left( \alpha _{0},\gamma \right) \geq \mathbb{E}\left[ \left( x_{t}^{\prime }\delta
_{0}\right) ^{2}1\left\{ B_{\gamma }\right\} \right] >0
\end{equation*}%
by \eqref{eq:rank1} and $R\left( \alpha _{0},\gamma \right) $ is continuous
at $\gamma =\gamma _{0}$ due to Assumption \ref{A-f1 conti} (i).

Second, if $\alpha \neq \alpha _{0},$%
\begin{equation*}
\left( Z_{t}\left( \gamma \right) ^{\prime }\alpha -Z_{t}\left( \gamma
_{0}\right) ^{\prime }\alpha _{0}\right) ^{2}=\left( x_{t}^{\prime }\left(
\beta -\beta _{0}+\delta -\delta _{0}\right) \right) ^{2}
\end{equation*}%
on $\left\{ f_{t}^{\prime }\gamma _{0}>0\right\} \cap \left\{ f_{t}^{\prime
}\gamma >0\right\} $ and
\begin{equation*}
\left( Z_{t}\left( \gamma \right) ^{\prime }\alpha -Z_{t}\left( \gamma
_{0}\right) ^{\prime }\alpha _{0}\right) ^{2}=\left( x_{t}^{\prime }\left(
\beta -\beta _{0}\right) \right) ^{2}
\end{equation*}%
on $\left\{ f_{t}^{\prime }\gamma _{0}\leq 0\right\} \cap \left\{
f_{t}^{\prime }\gamma \leq 0\right\} $. Thus,
\begin{align}\label{iden-eq-lower-bound}
\begin{split}
R\left( \alpha ,\gamma \right) &\geq \mathbb{E}\left( x_{t}^{\prime }\left( \beta
-\beta _{0}+\delta -\delta _{0}\right) \right) ^{2}1\left\{ A_{1\gamma
t}\right\} \\
&+\mathbb{E}\left( x_{t}^{\prime }\left( \beta -\beta _{0}\right) \right)
^{2}1\left\{ A_{2\gamma t}\right\} \\
&>c\left\vert \alpha -\alpha _{0}\right\vert _{2}^{2},
\end{split}
\end{align}%
for some $c>0\ $due to the rank condition in (\ref{eq:rank2}).

Together, they imply that the minimizer of $R$ is unique and well-separated.
\end{proof}

\section{Additional Details on Computation}\label{sec:add:computation:appendix}

In this section, we provide additional details on computation.
We give the proof of  Theorem \ref{thm:computation:joint} and  present an alternative form of the proposed algorithm in Section \ref{sec:joint:estimation}.

\subsection{Proof for Section \ref{sec:opt}}

\begin{proof}[Proof of Theorem \ref{thm:computation:joint}]
For convenience, we number constraints in the following way:
$\forall t,j,$%
\begin{itemize}
\item[1.]  $(\beta ,\delta )\in \mathcal{A},\;\gamma \in \Gamma$,
\item[2.]  $L_{j}\leq \delta _{j}\leq U_{j}$,
\item[3.]  $(d_{t}-1)(M_{t}+\epsilon )<f_{t}^{\prime }\gamma \leq d_{t}M_{t}$,
\item[4.]  $d_{t}\in \{0,1\}$,
\item[5.]  $d_{t}L_{j}\leq \ell _{j,t}\leq d_{t}U_{j}$,
\item[6.]  $L_{j}(1-d_{t})\leq \delta _{j}-\ell _{j,t}\leq U_{j}(1-d_{t})$,
\item[7.]  $\tau _{1}\leq \frac{1}{T}\sum_{t=1}^{T}d_{t}\leq \tau _{2}$.
\end{itemize}
Recall that
\begin{equation*}
\mathbb{Q}_{T}\left( \beta ,\boldsymbol{\ell }\right) \equiv \frac{1}{T}%
\sum_{t=1}^{T}\left( y_{t}-x_{t}^{\prime }{\beta }-\sum_{j=1}^{d_x}x_{j,t}\ell
_{j,t}\right) ^{2},
\end{equation*}%
where $\boldsymbol{\ell }=\left( \ell _{1,1},\ell _{1,2},...,\ell
_{d_{x},T}\right) ^{\prime }$,
\begin{equation*}
\left( \bar{\beta},\bar{\delta},\bar{\gamma},\boldsymbol{\bar{d},\bar{\ell}}%
\right) =\limfunc{argmin}_{\beta ,\delta ,\gamma ,\boldsymbol{d,\ell }}%
\mathbb{Q}_{T}\left( \beta ,\boldsymbol{\ell }\right) \ \text{under
conditions 1-7},
\end{equation*}%
and  $\mathbb{S}_{T}\left( \alpha ,\gamma \right) \equiv \frac{1}{T}%
\sum_{t=1}^{T}(y_{t}-x_{t}^{\prime }\beta -x_{t}^{\prime }\delta
1\{f_{t}^{\prime }\gamma >0\})^{2}$ and $\widehat{\alpha}$ and $\widehat{\gamma}$
denote the argmin of $\mathbb{S}_{T}$.

To prove the theorem, we show that (i)  $\mathbb{S}_{T}\left( \bar\alpha, \bar\gamma
\right)  =\mathbb{Q}_{T}\left( \bar\beta, \boldsymbol{ \bar{\ell}}
\right)$; (ii)
 $\mathbb{Q}_{T}\left( \bar{\beta},\boldsymbol{\bar{\ell}}%
\right) \geq \mathbb{S}_{T}\left(\widehat \alpha ,  \widehat\gamma \right) $; (iii) $\mathbb{S}_{T}\left( \widehat{%
\alpha},\widehat{\gamma}\right) \geq \mathbb{Q}_{T}\left( \bar{\beta},\boldsymbol{\bar{\ell}}\right) $.

Proof of (i):  By definition, $\mathbb{S}_{T}\left( \bar{\alpha},\bar{\gamma}\right)  = \frac{1}{T}%
\sum_{t=1}^{T}(y_{t}-x_{t}^{\prime }\bar{\beta}-x_{t}^{\prime }\bar{\delta}%
1\{f_{t}^{\prime }\bar{\gamma}>0\})^{2}$. Hence we need to show
\begin{eqnarray*}
 \frac{1}{T}%
\sum_{t=1}^{T}(y_{t}-x_{t}^{\prime }\bar{\beta}-x_{t}^{\prime }\bar{\delta}%
1\{f_{t}^{\prime }\bar{\gamma}>0\})^{2}
= \frac{1}{T}\sum_{t=1}^{T}\left( y_{t}-x_{t}^{\prime }{\bar{\beta}}%
-\sum_{j=1}^{d_x}x_{j,t}\bar{\ell}_{j,t}\right) ^{2}.
\end{eqnarray*}%
We show $\bar\ell_{j,t}=\bar\delta_j 1\{f_t'\bar \gamma>0\}$ for all $(t,j)$.
If $f_{t}^{\prime }\bar{\gamma}>0,$ $\bar{d}_{t}=1$ by condition 3
and 4, and $\bar{\ell}_{j,t}=\bar{\delta}_{j}$ by condition 6. If $f_t'\bar\gamma\leq 0$,  $%
\bar{d}_{t}=0$ by condtion 3 and 4 and $\bar{\ell}_{j,t}=0$ by condtion 5.

Proof of (ii): By part (i), we have
\begin{equation*}
\mathbb{Q}_{T}\left( \bar{\beta},\boldsymbol{\bar{\ell}}\right) =\mathbb{S}_{T}\left( \bar\alpha, \bar\gamma
\right) \geq
\min_{\alpha \in \mathcal{A},\gamma \in \Gamma }\mathbb{S}_{T}\left( \alpha
,\gamma \right) =\mathbb{S}_{T}\left( \widehat{\alpha},\widehat{\gamma}\right) .
\end{equation*}

Proof of (iii): Define $\widehat{\ell}_{j,t}:=\widehat{\delta}_{j}\widehat d_t $, where $\widehat d_t=1\left\{
f_{t}^{\prime }\widehat{\gamma}>0\right\}$. Then $\mathbb{S}_{T}\left( \widehat{\alpha},\widehat{\gamma}\right) =\mathbb{Q}_{T}\left(
\widehat{\beta},\boldsymbol{\widehat{\ell}}\right)$, where $\boldsymbol{\widehat\ell }=\left( \widehat\ell _{1,1} ,...,\widehat\ell
_{d_{x},T}\right) ^{\prime }$. Now
  it is straightforward to check that  $\left( \widehat{\beta},\widehat{\delta},\widehat{\gamma},\boldsymbol{\widehat{d},\widehat{\ell}}%
\right) $ satisfy conditions 1-7 for all $j$ and $t$.  For simplicity, we just give the details of checking condition 3. When $f_t'\widehat\gamma>0$, then $\widehat d_t=1$. Condition 3 becomes $0<f_t'\widehat\gamma\leq M_t=\sup_{\gamma\in\Gamma}|f_t'\gamma|$, which is satisfied. When $f_t'\widehat\gamma\leq 0$, $\widehat d_t=0$.  Condition 3 becomes $-M_t-\epsilon<f_t'\widehat\gamma\leq 0$, which holds for  any $\epsilon>0$.
So it is a feasible to the optimization problem $\min \mathbb Q_T$ with conditions 1-7. Consequently,
\begin{equation*}
\mathbb{S}_{T}\left( \widehat{\alpha},\widehat{\gamma}\right) =\mathbb{Q}_{T}\left(
\widehat{\beta},\boldsymbol{\widehat{\ell}}\right) \geq  \mathbb{Q}_{T}\left( \bar{\beta},\boldsymbol{\bar{\ell}}%
\right)
\end{equation*}%
by the definition of $(\bar{\beta},\boldsymbol{\bar{\ell}}). $ Combining parts (i),(ii) and (iii),   $\mathbb{S}_{T}\left( \bar\alpha, \bar\gamma
\right)  =\mathbb{Q}_{T}\left( \bar\beta, \boldsymbol{ \bar{\ell}}
\right)=\mathbb{S}_{T}\left( \widehat\alpha, \widehat\gamma
\right)  . $
\end{proof}

\subsection{Alternative MIQP Algorithm} \label{sec:joint:estimation:appendix}

The proposed algorithm in Section \ref{sec:joint:estimation} may run slowly when the dimension of $x_t$ is large. To mitigate this problem, we reformulate MIQP in the following way.

 {\singlespacing
\noindent\rule{15.5cm}{0.3pt} \\
\noindent
\texttt{[MIQP (Alternative Form)]} Let $\bm{d} =  (d_1, \ldots, d_T)'$
and $\bm{\widetilde{\ell}} = \{\widetilde{\ell}_{j,t}: j =1,\ldots,d_x, t = 1,\ldots,T \}$, where
 $\widetilde{\ell}_{j,t}$ is a real-valued variable. Solve the following problem:
\begin{align}\label{prob3-alt}
\min_{\beta, \widetilde{\delta}, \gamma, \bm{d}, \bm{\widetilde{\ell}} }\frac{1}{T}\sum_{t=1}^{T}
\left(y_{t}-x_{t}'{\beta} - \sum_{j=1} x_{j,t} \widetilde{\ell}_{j,t} - \left[ \sum_{j=1}^{d_x} x_{j,t} L_j  \right] d_t\right)^{2}
\end{align}
subject to
\begin{align}\label{main-constraints-alt}
\begin{split}
& (\beta, \delta) \in \mathcal{A}, \; \gamma \in \Gamma,  \\
& 0 \leq \widetilde{\delta}_{j} \leq  (U_j - L_j), \\
& 0 \leq  \widetilde{\ell}_{j,t} \leq \widetilde{\delta}_{j}, \\
& (d_t - 1) (M_t + \epsilon) < f_t' \gamma \leq d_t M_t, \\
& d_t \in \{0, 1\}, \\
& 0 \leq \sum_{j=1}^{d_x} \widetilde{\ell}_{j,t} \leq d_t \sum_{j=1}^{d_x} (U_j - L_j), \\
& 0 \leq \sum_{j=1}^{d_x} \left[ \widetilde{\delta}_{j} - \widetilde{\ell}_{j,t} \right] \leq  (1-d_t)  \sum_{j=1}^{d_x} (U_{j} - L_j), \\
& \tau_1 \leq \frac{1}{T} \sum_{t=1}^T d_t \leq \tau_2
\end{split}
\end{align}
for each $t=1,\ldots,T$ and each $j=1,\ldots,d_x$, where  $0 < \tau_1 < \tau_2 < 1$. \\
\noindent\rule{15.5cm}{0.3pt} \\
}

Note that $\widetilde{\delta}_{j}$ and $\widetilde{\ell}_{j,t}$ are transformed to be positive.
Using the positivity of these variables, one can sum up restrictions across $j$'s, where $j=1,\ldots,d_x$, while ensuring that optimization problem \eqref{prob3-alt} under \eqref{main-constraints-alt}
is mathematically equivalent to optimization problem \eqref{prob3} under \eqref{main-constraints}
in Section \ref{sec:joint:estimation}.
We use the alternative form of formulation in our numerical work; however, we present a simpler form in Section \ref{sec:joint:estimation} to help readers follow our basic ideas more easily.

\section{Selecting Relevant Factors}\label{sec:factor-selection}

In this section,
we consider factor selections.
In applications, it is often difficult to have \emph{a priori} knowledge regarding which variables constitute  $f_t$ in  \eqref{model1}.
Suppose that there are a mildly large number of factors; however,  we are willing to assume that only a small number of factors are active (i.e. their $\gamma$ coefficients are non-zero), although we do not know their identities. This is an unordered combinatorial selection problem, but can be easily adopted in the $\ell_0$-penalization framework with the help of MIO, so long as the number of candidate factors is fixed (\citet{bertsimas2016}).

To be specific, decompose $f_t = (f_{1t}', f_{2t}', -1)'$, where $f_t$ can be either the observed factors or estimated factors \footnote{For this section only, we use $ f_{2t} $ excluding $ -1 $. This is to reflect our setup where  the constant term $-1$ is always included among active factors.}  and $\gamma = (\gamma_{1}', \gamma_{2}', \gamma_{3})'$.
Assume that $f_{1t}$ is known to be active for certainty, but $f_{2t}$ may or may not be active.
Let  $p = |f_{2t}|_0$.
Suppose that each element of $\gamma_{2}$ is bounded between known values of  $\underline{\gamma_2}$ and $\overline{\gamma_2}$.
Let $\gamma_{2j}$ denote the $j$-th element of $\gamma_2$, where $j = 1,\ldots, p$.
Assume further that we know the lower and upper bounds, say $\underline{p}$ and $\overline{p}$, of the number of active elements of $\gamma_2$.    A  default choice of $(\underline{p},\overline{p})$ is
$\underline{p} = 0$ and $\overline{p} = p$; however, a strictly smaller choice of $\overline{p}$ might help estimation in practice when $p$ is relatively large and it is plausible to assume that the maximal number of factors is much less than $p$.

For a given  penalty parameter $\lambda > 0$ (here $f_t$ is either observed or estimated factors),
define
\begin{align}\label{l0-variable-selection}
\begin{split}
\widetilde\gamma= &\arg\min_{\gamma \in \Gamma} \min_{\beta,\delta} \frac{1}{T}\sum_{t=1}^{T}\left(
y_{t}-x_{t}^{\prime}{{\beta}}  -x_{t}^{\prime}%
{\delta}  1\left\{  f_{t}^{\prime}\gamma>0\right\}
\right)  ^{2} +\lambda|\gamma|_0 \\
&
\; \text{subject to  \eqref{restrction:para:original-form}}.
\end{split}
\end{align}
Computation of $\widetilde\gamma$ can be formulated using the following optimization.

{\singlespacing
\noindent\rule{13cm}{0.3pt} \\
\texttt{[MIQP with Factor Selection]} In addition to $\bm{d}$ and $\bm{\ell}$,  let $\bm{e} =  (e_1, \ldots, e_{p})'$.
Choose a penalty parameter $\lambda > 0$. Then
solve the following problem:
\begin{align}\label{prob3-factor}
\min_{\beta, \delta, \gamma, \bm{d}, \bm{\ell}, \bm{e} }
\widetilde{\mathbb{Q}}_{T}\left( \beta ,\boldsymbol{\ell }\right) \equiv
\frac{1}{T}\sum_{t=1}^{T}
\left(y_{t}-x_{t}'{\beta}- \sum_{j=1} x_{j,t} \ell_{j,t} \right)^{2} + \lambda \sum_{m=1}^{p} e_m
\end{align}
subject to  \eqref{main-constraints}
and
\begin{align}\label{selection-constraints}
\begin{split}
& e_m \underline{\gamma_2} \leq \gamma_{2m} \leq e_m  \overline{\gamma_2}, \\
& \underline{p} \leq \sum_{m=1}^p e_m \leq \overline{p}, \\
& e_m \in \{0, 1\} \ \ \text{for each $m=1,\ldots,p$}.
\end{split}
\end{align}
Finally, re-estimate the model using only selected factors via the method given in Section \ref{sec:joint:estimation}. \\
\noindent\rule{13cm}{0.3pt} \\
}

The new indicator variable $e_m$ turns on and off the $m$-th factor in estimation. The complexity of the regression model  is penalized  by the $\ell_0$ norm $(\sum_{m=1}^{p} e_m)$.

When $f_t$ contains only observed factors, we provide selection consistency below.

\begin{thm}\label{thm:factor:selection:consistency}
 Consider the known factor case. Let $S\left(  \gamma\right)  =\left\{  j:\gamma_{j}\neq0\right\}  $ and
$S_{0}=S\left(  \gamma_{0}\right)  $. Let Assumptions \ref{scale-normalization}, \ref{iden-assump},
\ref{A-mixing}, and \ref{A:diminishing dt} hold. Suppose that $\lambda \rightarrow 0$, $\lambda T \rightarrow \infty$, and $p$ is fixed.
 Then,%
\[
\mathbb{P}\left\{  S\left(  \widetilde{\gamma}\right)  =S_{0}\right\}  \rightarrow1.
\]
\end{thm}




 When  factors are unobservable but estimated via the PCA,
 the optimization (\ref{l0-variable-selection}) can be still used to select the estimated factors. Indeed, \citet{bai2008forecasting}  first obtained a set of PCA factors, then applied BIC to select among them for diffusion index forecasts.
 In theory, however,
  interpretation of each estimated factor is more involved since factors are identified only up to some random rotation. The rotation   indeterminacy creates difficulties to define ``the true factors" statistically, and therefore the ``factor selection consistency".  Thus we do not pursue the selection consistency in the estimated factor case.

\begin{proof}[Proof of Theorem \ref{thm:factor:selection:consistency}]
For a given $\gamma,$ let%
\[
{\mathbb{Q}}_{T}\left(  \gamma\right)  \equiv\frac{1}{T}\sum_{t=1}^{T}\left(
y_{t}-x_{t}^{\prime}{\widehat{\beta}}\left(  \gamma\right)  -x_{t}^{\prime}%
\widehat{\delta}\left(  \gamma\right)  1\left\{  f_{t}^{\prime}\gamma>0\right\}
\right)  ^{2}%
\]
and
\[
\widetilde{\mathbb{Q}}_{T}(\gamma)={\mathbb{Q}}_{T}\left(  \gamma\right)
+\lambda\left\vert \gamma\right\vert _{0},
\]
where $\widehat{\alpha}\left(  \gamma\right)  =\left(  \widehat{\beta}\left(
\gamma\right)  ^{\prime},\widehat{\delta}\left(  \gamma\right)  ^{\prime}\right)
^{\prime}$ is the OLS estimate of $\alpha$ for the given $\gamma$. The former
is a profiled criterion function of the original criterion. Define
\[
\widetilde{\gamma}=\arg\min_{\gamma}\widetilde{\mathbb{Q}}_{T}(\gamma).
\]
Our proof is
divided into the following steps.

\begin{enumerate}
\item[] \textbf{Step 1}.
Show that $S_{0}\subset S(\widetilde{\gamma})$ with probability
approaching one.

\item[] \textbf{Step 2}.
Show that $\min_{\gamma: S(\gamma)=S_{0}}\mathbb{Q}_{T}(\gamma)
\leq\min_{\gamma}\mathbb{Q}_{T}(\gamma)+O_{P}(T^{-1}). $

\item[] \textbf{Step 3}.
Show that for $\Gamma_{b}:=\{\gamma:S_{0}\subset S(\gamma),S_{0}\neq
S(\gamma)\}$,
\[
\min_{\gamma\in\Gamma_{b}}\widetilde{\mathbb{Q}}_{T}(\gamma)-\min
_{\gamma:S(\gamma)=S_{0}}\widetilde{\mathbb{Q}}_{T}(\gamma)>\lambda/2
\]
with probability approaching one. 
\end{enumerate}
Now suppose $S_{0}\neq S(\widetilde{\gamma})$. Then by step 1, $\widetilde
{\gamma}\in\Gamma_{b}$, then by step 3,
\[
\widetilde{\mathbb{Q}}_{T}(\widetilde{\gamma})\geq\min_{\gamma\in\Gamma_{b}%
}\widetilde{\mathbb{Q}}_{T}(\gamma)>\min_{\gamma:S(\gamma)=S_{0}}%
\widetilde{\mathbb{Q}}_{T}(\gamma)+\lambda/2,
\]
which contradicts with the definition of $\widetilde{\gamma}$. Consequently,
we must have $S_{0}=S(\widetilde{\gamma})$ with probability approaching one.
\end{proof}

\begin{proof}[Proof of Step 1]
Let $\alpha^{\ast}\left(  \gamma\right)  =\left(  \mathbb{E}Z_{t}\left(
\gamma\right)  Z_{t}\left(  \gamma\right)  ^{\prime}\right)  ^{-1}%
\mathbb{E}Z_{t}\left(  \gamma\right)  Z_{t}\left(  \gamma_{0}\right)
^{\prime}\alpha_{0}. $ Also let
\[
\mathbb{Q}\left(  \gamma\right)  \equiv\mathbb{E}\left(  y_{t}-Z_{t}\left(
\gamma\right)  ^{\prime}{\alpha}^{\ast}\left(  \gamma\right)  \right)
^{2}=\sigma^{2}+\mathbb{E}\left(  \alpha^{\ast}\left(  \gamma\right)
^{\prime}Z_{t}\left(  \gamma\right)  -\alpha_{0}^{\prime}Z_{t}\left(
\gamma_{0}\right)  \right)  ^{2}.
\]

Then, by the ULLN and the CMT and the fact that $\lambda\rightarrow0$,
uniformly in $\gamma$,
\[
\widehat{\alpha}\left(  \gamma\right)  -\alpha^{\ast}\left(  \gamma\right)
=o_P\left(  1\right)  ,\qquad\widetilde{\mathbb{Q}}_{T}\left(
\gamma\right)  -\mathbb{Q}\left(  \gamma\right)  =o_P\left(  1\right)  .
\]
Also, $\alpha^{\ast}(\gamma_{0})=\alpha_{0}$ implies $\mathbb{Q}(\gamma
_{0})=\sigma^{2}$ and
\[
\mathbb{Q}(\widetilde{\gamma})=\widetilde{\mathbb{Q}}_{T}(\widetilde{\gamma
})+o_{P}(1)\leq\widetilde{\mathbb{Q}}_{T}(\gamma_{0})+o_{P}(1)=\mathbb{Q}%
(\gamma_{0})+o_{P}(1)=\sigma^{2}+o_P\left(  1\right)  .
\]
On the other hand, for $\Gamma_{a}=\left\{  \gamma:S_{0}\nsubseteqq S\left(
\gamma\right)  \right\}  $, due to Theorem \ref{iden-thm},
\[
\min_{\gamma\in\Gamma_{a}}\mathbb{E}\left(  \alpha^{\ast}\left(
\gamma\right)  ^{\prime}Z_{t}\left(  \gamma\right)  -\alpha_{0}^{\prime}%
Z_{t}\left(  \gamma_{0}\right)  \right)  ^{2}>0.
\]
So $\min_{\gamma\in\Gamma_{a}}\mathbb{Q}\left(  \gamma\right)  >\sigma^{2}.$
This implies $\widetilde{\gamma}\notin\Gamma_{a}$, thus $S_{0}\subset
S(\widetilde{\gamma})$ with probability approaching one.
\end{proof}

\begin{proof}[Proof of Step 2]
Uniformly over pairs $(\gamma_{1},\gamma_{2})$ in a shrinking
neighborhood of $\gamma_{0}$, ($B_{C}(\gamma_{0})=\{|\gamma-\gamma_{0}%
|_{2}\leq CT^{-(1-2\varphi)}\}$ for any $C>0$),
\[
\mathbb{Q}_{T}(\gamma_{1})-\mathbb{Q}_{T}(\gamma_{2})=R_{T}(\gamma_{1}%
)-R_{T}(\gamma_{2})+\mathbb{G}_{T}(\gamma_{2})-\mathbb{G}_{T}(\gamma_{1}),
\]
where $R_{T}(\gamma)=\frac{1}{T}\sum_{t}[Z_{t}(\gamma)^{\prime}\widehat{\alpha
}(\gamma)-Z_{t}(\gamma_{0})^{\prime}\alpha_{0}]^{2}$ and $\mathbb{G}%
_{T}(\gamma)=\frac{2}{T}\sum_{t}\varepsilon_{t}Z_{t}(\gamma)\widehat{\alpha
}(\gamma)$. Note that $\sup_{\gamma\in B_{C}(\gamma_{0})}|\widehat{\alpha}%
(\gamma)-\alpha_{0}|_{2}=O_{P}(T^{-1/2})$, $\sup_{\gamma\in B_{C}(\gamma_{0}%
)}|R_{T}(\gamma)|=O_{P}(T^{-1})$, and

 $\sup_{\gamma_{1},\gamma_{2}\in
B_{C}(\gamma_{0})}|\mathbb{G}_{T}(\gamma_{1})-\mathbb{G}_{T}(\gamma
_{2})|=O_{P}(T^{-1})$. Therefore,
\[
\sup_{\gamma_{1},\gamma_{2}\in B_{C}(\gamma_{0})}|\mathbb{Q}_{T}(\gamma
_{1})-\mathbb{Q}_{T}(\gamma_{2})|=O_{P}(T^{-1}).
\]

Let $\widehat\gamma_{1}$ and $\widehat\gamma_{2}$ respectively denote the argument of
$\min_{S(\gamma)=S_{0}}\mathbb{Q}_{T}(\gamma)$ and $\min_{\gamma}%
\mathbb{Q}_{T}(\gamma)$. Then for both $j=1,2$, $\mathbb{Q}_{T}(\widehat\gamma
_{j}) \leq\mathbb{Q}_{T}(\gamma_{0}). $ Then it follows from the proof of
Theorem \ref{asdist-alpha-gamma} that $\widehat\gamma_{j}-\gamma_{0}%
=O_{P}(T^{-(1-2\varphi)}),\quad j=1,2. $ As a result,
\[
0\leq\min_{\gamma: S(\gamma)=S_{0}}\mathbb{Q}_{T}(\gamma) - \min_{\gamma
}\mathbb{Q}_{T}(\gamma) =\mathbb{Q}_{T}(\widehat\gamma_{1}) -\mathbb{Q}_{T}%
(\widehat\gamma_{2}) = O_{P}(T ^{-1}).
\]
\end{proof}

\begin{proof}[Proof of Step 3]
Let $\Gamma_{b}:=\{\gamma: S_{0}\subset S(\gamma),
S_{0}\neq S(\gamma)\}$. Then we have
\begin{align*}
\min_{\gamma\in\Gamma_{b}} \widetilde{\mathbb{Q}}_{T}(\gamma) -\min_{\gamma:
S(\gamma)=S_{0}} \widetilde{\mathbb{Q}}_{T}(\gamma)  & \geq^{(1)} \min
_{\gamma}\mathbb{Q}_{T}(\gamma)+\lambda\min_{\gamma\in\Gamma_{b}}|\gamma
|_{0}-\min_{\gamma: S(\gamma)=S_{0}} \widetilde{\mathbb{Q}}_{T}(\gamma
)\cr &=^{(2)} \min_{\gamma}\mathbb{Q}_{T}(\gamma) -\min_{S(\gamma)=S_{0}%
}{\mathbb{Q}}_{T}(\gamma)+\lambda\min_{\gamma\in\Gamma_{b}}|\gamma
|_{0}-\lambda|\gamma_{0}|_{0}\cr &\geq^{(3)} O_{P}(T^{-1}) +\lambda
\cr& >^{(4)}\lambda/2\quad(\text{with probability approaching one})
\end{align*}
where (1) is due to $\min_{\gamma\in\Gamma_{b}} \widetilde{\mathbb{Q}}%
_{T}(\gamma) \geq\min_{\gamma}\mathbb{Q}_{T}(\gamma) +\lambda\min_{\gamma
\in\Gamma_{b}}|\gamma|_{0}$; (2) is due to the fact that $\arg\min_{\gamma:
S(\gamma)=S_{0}} \widetilde{\mathbb{Q}}_{T}(\gamma)= \arg\min_{\gamma:
S(\gamma)=S_{0}} {\mathbb{Q}}_{T}(\gamma)$, and $|\gamma|_{0}=|\gamma_{0}%
|_{0}$ for all $\gamma\in\{\gamma: S(\gamma)=S_{0}\}$; (3) is due to step 2
and $\min_{\gamma\in\Gamma_{b}}|\gamma|_{0}-|\gamma_{0}|_{0} \geq1$. Finally,
(4) is due to $T\lambda\to\infty$.
\end{proof}

\section{Test of Linearity}\label{sec:test:linearity}

In some applications, we are interested in testing the linearity of the regression model in \eqref{model1}. That is,  we may want to test the following null hypothesis:
\begin{equation*}
\mathcal{H}_{0}:\delta_0 =0\ \ \ \text{for all }\gamma_0 \in \Gamma .
\end{equation*}%
Under the null hypothesis the model becomes the linear regression model and thus
$\gamma_0$ is not identified. This testing problem has been studied
intensively in the literature when $f_t$ is directly observed and the dimension of an unidentifiable component of $\gamma_0$ is  1 (see, e.g., \citet{Hansen:1996} and \citet{lee2012testing} among many others).

We propose to use the  following statistic:
\begin{align}\label{def:supQ}
\begin{split}
\text{supQ} &
=\sup_{\gamma \in \Gamma}T%
\frac{\min_{\alpha :\delta =0}\mathbb{S}_{T}\left( \alpha ,\gamma \right)
-\min_{\alpha }\mathbb{S}_{T}\left( \alpha ,\gamma \right) }{\min_{\alpha }%
\mathbb{S}_{T}\left( \alpha ,\gamma \right) } \\
&=T\frac{\min_{\alpha :\delta =0}\mathbb{S}_{T}\left( \alpha ,\gamma
\right) -\min_{\alpha ,\gamma }\mathbb{S}_{T}\left( \alpha ,\gamma \right) }{%
\min_{\alpha ,\gamma }\mathbb{S}_{T}\left( \alpha ,\gamma \right) },
\end{split}
\end{align}%
where $\mathbb{S}_{T}\left( \alpha ,\gamma \right)$  is  the least squares criterion function, using either the observed or estimated factor.

For both observed and latent factor cases, we establish the following result.

\begin{thm}\label{Thm:Linearity Test}
 Suppose that either assumptions of Theorems \ref{asdist-alpha-gamma} (for the known factor case) or  assumptions of  Theorems \ref{thm:AD with Estimated f} (for the estimated factor case) hold.
 Then, under $\mathcal{H}_{0}$,
\begin{equation*}
\ \text{supQ}\overset{d}{\longrightarrow }\sup_{\gamma \in \Gamma} W\left(
\gamma \right) ^{\prime }\left( R\left( \mathbb{E}Z_{t}\left( \gamma \right)
Z_{t}\left( \gamma \right) ^{\prime }\right) ^{-1}\mathbb{E}\varepsilon
_{t}^{2}R^{\prime }\right) ^{-1}W\left( \gamma \right) ,
\end{equation*}%
where $W\left( \gamma \right) $ is a vector of centered Gaussian processes
with covariance kernel
\begin{equation*}
K\left( \gamma _{1},\gamma _{2}\right) =R\left( \mathbb{E}Z_{t}\left( \gamma
_{1}\right) Z_{t}\left( \gamma _{1}\right) ^{\prime }\right) ^{-1}\mathbb{E}%
 \left[ Z_{t}\left( \gamma _{1}\right) Z_{t}\left( \gamma _{2}\right) ^{\prime
}\varepsilon _{t}^{2} \right] \left( \mathbb{E} Z_{t}\left( \gamma _{2}\right)
Z_{t}\left( \gamma _{2}\right) ^{\prime }\right) ^{-1}R^{\prime }
\end{equation*}%
and $R=\left( 0_{{d_x}}, I_{d_{x}}\right)$ is the $(d_x \times 2 d_x)$-dimensional selection matrix.\footnote{Here, $0_{{d_x}}$ and $I_{d_{x}}$, respectively,  denote the $d_x$-dimensional square matrix with all elements being zeros and the $d_x$-dimensional identity matrix.}
\end{thm}



Below we present a bootstrap algorithm for the p-value.

{\singlespacing
\noindent\rule{13cm}{0.3pt} \\
\texttt{[Computation of Bootstrap $p$-Values]}

\begin{enumerate}
\item Generate an iid sequence $\left\{ \eta _{t}\right\} $ whose mean is
zero and variance is one.

\item Construct $\left\{ y_{t}^{\ast }\right\} $ by
\begin{equation*}
y_{t}^{\ast }=x_{t}^{\prime }\widehat{\beta }+\eta _{t}\widehat{\varepsilon }
_{t},
\end{equation*}
where $\widehat{\beta }$ is the unconstrained estimator of $\beta_0$ and $\widehat{\varepsilon }
_{t}$ is the estimated residual from unconstrained estimation.

\item Construct the bootstrap statistic $\text{supQ}^{\ast }$ by \eqref{def:supQ} with the bootstrap sample $\{ y_{t}^{\ast }, x_t, f_t: t=1,\ldots,T \}$ if $f_t$ is known and
$\{ y_{t}^{\ast }, x_t, \widetilde{f}_t: t=1,\ldots,T \}$ if $f_t$ is estimated, respectively.

\item Repeat 1-3 many times and compute the empirical distribution of supQ$%
^{\ast }$.

\item
Then, with the obtained empirical distribution, say $F_{T}^{\ast }\left(
\cdot \right) ,$ one can compute the bootstrap $p$-value by
\begin{equation*}
p^{\ast }=1-F_{T}^{\ast }\left( \text{supQ}\right) ,
\end{equation*}%
or $a$-level critical value%
\begin{equation*}
c_{a}^{\ast }=F_{T}^{\ast ^{-1}}\left( 1-a\right) .
\end{equation*}
\end{enumerate}
\noindent\rule{13cm}{0.3pt} \\
}

The proposed bootstrap is standard and thus its asymptotic validity follows from the
standard manner in view of Lemma \ref{l:relf} and the conditional martingale difference sequence   central limit theorem (e.g.  Theorem 3.2 of \citet{Hall-Heyde}). The details are omitted for the sake of brevity.
Furthermore, it is straightforward to establish conditions for the consistency of our proposed test.

\section{Additional Empirical Results}\label{sec:add:emp:appendix}

In this part of the appendix, we provide additional empirical results that are omitted from the main text.

\subsection{Testing the Linearity of US GNP and Selecting Factors}\label{sec:add:example}

In this section, we revisit the empirical application in \citet{Hansen:1996}, who tested \citet{Potter:95}'s model of US GNP.
\citet{Hansen:1996} used annualized quarterly growth rates, say $y_t$, for the period 1947-1990.
His estimates were as follows:
\begin{align}\label{hansen-estimates}
\begin{split}
y_t &=  - 3.21 + 0.51 y_{t-1} - 0.93 y_{t-2} - 0.38 y_{t-5} + \widehat{\varepsilon}_t    \;\;\;\; \text{if $y_{t-2} \leq 0.01$} \\
      &   \;\;\;\;\;\; (2.12) \;\;\; (0.25)   \;\;\;\;\;\;\;\; (0.31)   \;\;\;\;\;\;\;\; (0.25)     \\
y_t &=  2.14 + 0.30 y_{t-1} + 0.18 y_{t-2} - 0.16 y_{t-5} + \widehat{\varepsilon}_t    \;\;\;\; \text{if $y_{t-2} > 0.01$,}  \\
      &   \;\;\; (0.77) \;\;\; (0.10)   \;\;\;\;\;\;\;\; (0.10)   \;\;\;\;\;\;\;\; (0.07)
\end{split}
\end{align}
where heteroskedasticity-robust standard errors are given in parenthesis. His heteroskedasticity-robust LM-based tests for the hypothesis of no threshold effect were all far from usual rejection regions (the smallest p-value was 0.17).
Using the same dataset, we carry out the following two exercises: (1) selecting relevant factors and (2) testing the linearity of the model.
For the former, we
keep $y_{t-2}$ as $f_{1t}$ and add  $(y_{t-1}, y_{t-5})$ as $f_{2t}$. That is, we allow for the possibility that the regimes can be determined by a linear combination of $(y_{t-1}, y_{t-2}, y_{t-5})$.
The choice of penalization parameter $\lambda$ is important. Recall that we require $\lambda \rightarrow 0$ and $\lambda T \rightarrow \infty$. In this application,
we set
\begin{align*}
\lambda = \widehat{\sigma}_{\text{Hansen}}^2 \frac{\log T}{T},
\end{align*}
where $\widehat{\sigma}_{\text{Hansen}}^2 = T^{-1} \sum_{t=1}^T \widehat{\varepsilon}_t^2$ and the estimated residual $\widehat{\varepsilon}_t$
is obtained from \citet{Hansen:1996}'s estimates in \eqref{hansen-estimates}.
By implementing joint optimization with this choice of $\lambda$, we select only $y_{t-5}$ but drop $y_{t-1}$ in $f_{2t}$. Our estimated index
is
\[
f_t'\widehat{\gamma} = y_{t-2} - 0.91 y_{t-5}  + 0.50.
\]
If we compare this with Hansen's estimate $f_t'\widehat{\gamma} =  y_{t-2} - 0.01$, we can see that in Hansen's model, the regime is determined by the level of GNP growth in $t-2$; on the contrary, in our model, it is determined by $y_{t-2} - 0.91 y_{t-5}$, roughly speaking the changes in growth rates from $t-5$ to $t-2$. Specifically, the regime is determined whether $y_{t-2} - 0.91 y_{t-5}$ is  above or below $- 0.50$.
Our estimates suggest that a recession might be captured better by a decrease in growth rates from $t-5$ to $t-2$, compared to a low level of growth rates in $t-2$. Our estimated coefficients and their standard errors are as follows:
\begin{align}\label{new-estimates}
\begin{split}
y_t &=  - 2.07 + 0.28 y_{t-1} - 0.33 y_{t-2} + 0.62 y_{t-5} + \widehat{\varepsilon}_t    \;\;\;\; \text{if $y_{t-2} - 0.91 y_{t-5} \leq -0.50$} \\
      &   \;\;\;\;\;\; (1.33) \;\;\; (0.13)   \;\;\;\;\;\;\;\; (0.16)   \;\;\;\;\;\;\;\; (0.19)     \\
y_t &=  2.76 + 0.35 y_{t-1} + 0.07 y_{t-2} - 0.21 y_{t-5} + \widehat{\varepsilon}_t    \;\;\;\; \text{if $y_{t-2} - 0.91 y_{t-5} > -0.50$.}  \\
      &   \;\;\; (0.96) \;\;\; (0.12)   \;\;\;\;\;\;\;\; (0.12)   \;\;\;\;\;\;\;\; (0.10)
\end{split}
\end{align}
We now report the result of testing the null hypothesis of no threshold effect.  We take our estimates in \eqref{new-estimates} as unconstrained estimates. The resulting LR test statistic is 28.19 and the p-value is 0.056 based on 500 bootstrap replications. This implies that the null hypothesis is rejected at the 10\% level but not at the 5\% level. There are two main differences between our test result and \citet{Hansen:1996}'s. We use the LR statistic, whereas \citet{Hansen:1996} considered the LM statistic. Furthermore, his alternative only allows for the scalar threshold variable $y_{t-2}$ but we consider a single index using $y_{t-2}$ and $y_{t-5}$.

\subsection{Details on Estimation Results for Table \ref{table-hansen97-short}}

Table \ref{table-hansen97} shows full estimation results for Table \ref{table-hansen97-short}.

\begin{table}[ h!tbp]
\caption{Estimation Results}
\begin{center}
\begin{tabular}{lrrrrrr}
\hline\hline
Specification     & \multicolumn{2}{c}{(1)} & \multicolumn{2}{c}{(2)} & \multicolumn{2}{c}{(3)} \vspace*{1ex} \\
    & \multicolumn{2}{c}{$f_{1t} = (q_{t-1}, -1)$} & \multicolumn{2}{c}{$f_{2t} = (F_{t-1}, -1)$} & \multicolumn{2}{c}{$f_{3t} = (q_{t-1}, F_{t-1}, -1)$} \vspace*{1ex} \\
                  &   Estimate & Std. Err. &   Estimate & Std. Err.  &   Estimate & Std. Err. \\
\hline
\\
Regime 1          & \multicolumn{2}{c}{$q_{t-1} \leq 0.302$} & \multicolumn{2}{c}{$F_{t-1} \leq -0.28$} & \multicolumn{2}{c}{$q_{t-1} + 3.55 F_{t-1}$} \\
 (``Expansion'')      &  & & & & \multicolumn{2}{c}{$\leq -1.60$} \vspace*{1ex} \\
Intercept         &  -0.0214   &  0.0126 & -0.0255  &  0.0101 &   -0.0294   &   0.0101   \\
$\Delta u_{t-1}$  &  -0.1696   &  0.0640 & -0.1182  &  0.0629 &   -0.1628   &   0.0601   \\
$\Delta u_{t-2}$  &   0.0382   &  0.0650 &  0.0774  &  0.0558 &    0.0264   &   0.0600   \\
$\Delta u_{t-3}$  &   0.1896   &  0.0587 &  0.2097  &  0.0645 &    0.1933   &   0.0520   \\
$\Delta u_{t-4}$  &   0.1399   &  0.0630 &  0.1039  &  0.0523 &    0.1445   &   0.0552   \\
$\Delta u_{t-5}$  &   0.0858   &  0.0749 &  0.0622  &  0.0600 &    0.0699   &   0.0656   \\
$\Delta u_{t-6}$  &   0.0214   &  0.0653 &  0.0193  &  0.0558 &    0.0177   &   0.0613   \\
$\Delta u_{t-7}$  &   0.0318   &  0.0678 & -0.0268  &  0.0596 &    0.0174   &   0.0613   \\
$\Delta u_{t-8}$  &   0.0402   &  0.0599 & -0.0006  &  0.0617 &    0.0103   &   0.0626   \\
$\Delta u_{t-9}$  &  -0.0667   &  0.0663 & -0.0766  &  0.0660 &   -0.0637   &   0.0656   \\
$\Delta u_{t-10}$ &  -0.0540   &  0.0640 & -0.0120  &  0.0559 &   -0.0467   &   0.0575   \\
$\Delta u_{t-11}$ &   0.0782   &  0.0568 &  0.0162  &  0.0529 &    0.0196   &   0.0528   \\
$\Delta u_{t-12}$ &  -0.0899   &  0.0641 & -0.1216  &  0.0576 &   -0.1224   &   0.0572   \\
\hline
\\
Regime 2          & \multicolumn{2}{c}{$q_{t-1} > 0.302$} & \multicolumn{2}{c}{$F_{t-1} > -0.28$} & \multicolumn{2}{c}{$q_{t-1} + 3.55 F_{t-1}$} \\
(``Contraction'')         &  & & & & \multicolumn{2}{c}{$> -1.60$} \vspace*{1ex} \\
%
Intercept         &   0.0876   &  0.0375  &  0.0509  &  0.0560  &    0.1893   &    0.0576   \\
$\Delta u_{t-1}$  &   0.2406   &  0.1179  &  0.3671  &  0.2011  &    0.2937   &    0.1665   \\
$\Delta u_{t-2}$  &   0.2455   &  0.0932  &  0.2198  &  0.1634  &    0.1420   &    0.1279   \\
$\Delta u_{t-3}$  &   0.1283   &  0.1038  &  0.0936  &  0.1563  &    0.1042   &    0.1549   \\
$\Delta u_{t-4}$  &  -0.0222   &  0.1033  & -0.0053  &  0.1883  &   -0.1035   &    0.1690   \\
$\Delta u_{t-5}$  &  -0.0272   &  0.1104  & -0.1804  &  0.2188  &   -0.0723   &    0.1868   \\
$\Delta u_{t-6}$  &  -0.0851   &  0.1083  & -0.0500  &  0.2125  &   -0.0821   &    0.1400   \\
$\Delta u_{t-7}$  &  -0.1562   &  0.1057  & -0.0297  &  0.2027  &   -0.1853   &    0.1443   \\
$\Delta u_{t-8}$  &  -0.0372   &  0.1357  &  0.0021  &  0.2923  &   -0.1214   &    0.2038   \\
$\Delta u_{t-9}$  &   0.0991   &  0.1358  &  0.0754  &  0.1754  &   -0.0861   &    0.1475   \\
$\Delta u_{t-10}$ &   0.1149   &  0.1125  &  0.0445  &  0.1574  &    0.0392   &    0.1426   \\
$\Delta u_{t-11}$ &  -0.1012   &  0.1256  &  0.1872  &  0.1995  &   -0.0307   &    0.1840   \\
$\Delta u_{t-12}$ &  -0.4440   &  0.1144  & -0.2269  &  0.1668  &   -0.3807   &    0.1542   \\
\hline
\multicolumn{7}{l}{Avg. of squared residuals} \\
$(T^{-1} \sum_{i=1}^T \widehat{\varepsilon}_t^2)$          & \multicolumn{2}{c}{0.0264}  & \multicolumn{2}{c}{0.0272} & \multicolumn{2}{c}{0.0252} \vspace*{1ex} \\
\hline
\multicolumn{7}{l}{Proportion of matches between NBER recession dates and threshold estimates} \\
         & \multicolumn{2}{c}{0.807}  & \multicolumn{2}{c}{0.894} & \multicolumn{2}{c}{0.896} \vspace*{1ex} \\
\hline
\end{tabular}
\end{center}
\label{table-hansen97}
\end{table}

\section{Proofs of the Asymptotic Distribution in Section \ref{sec:known factors}: Known $ f $}\label{proof:asymp:known:factors}

Recall that we have proposed two computing algorithms for the estimators of $(\alpha_0, \gamma_0)$.
Throughout this part of the appendix, we assume that
$(\widehat{\alpha }, \widehat{\gamma })$ is the global solution to the optimization problem.
The proof is divided into the following subsections.

\subsection{Consistency}

\begin{lem}[Consistency]\label{cons-thm}
	Let Assumptions \ref{scale-normalization}, \ref{A-f1 conti} and \ref{A-mixing} \ref{A-mixing:itm1} and \ref{A-mixing:itm2}
	hold. Then  as $T \rightarrow \infty$,
	\begin{equation*}
	\left\vert \widehat{\alpha }-\alpha _{0}\right\vert _{2}=o_{P}\left( 1 \right) \ \text{and\ }\left\vert \widehat{\gamma }-\gamma
	_{0}\right\vert _{2}=o_{P}\left( 1 \right) .
	\end{equation*}
\end{lem}

\begin{proof}[Proof of Lemma \ref{cons-thm}]

We begin with stating the following standard ULLN for $ \rho $-mixing sequences, see e.g. \cite{davidson1994}, for which Assumption \ref{A-mixing} \ref{A-mixing:itm1} and \ref{A-mixing:itm2} suffice.
\begin{enumerate}[label=(\roman*)]
	\item\label{A-LLN:itm1}
	$
	\sup_{\gamma \in \Gamma}
	|  \frac{1}{T}\sum_{t=1}^{T}
	Z_{ti}\left( \gamma \right) Z_{tj}\left( \gamma \right)
	- \mathbb E \left[ Z_{ti}\left( \gamma \right) Z_{tj}\left( \gamma \right) \right]
	| = o_{P}\left( 1 \right).
	$
	\item\label{A-LLN:itm2}
	$
	\sup_{\gamma \in \Gamma}
	| \frac{1}{T}\sum_{t=1}^{T}\varepsilon _{t}Z_{t}\left( \gamma \right)
	| = o_{P}\left(1 \right).
	$
\end{enumerate}
These will be cited as ULLN hereafter.

We begin with the consistency of $ \widehat{\gamma} $.
Recall that the least squares estimate of $\alpha $ for a given $\gamma $ is
the OLS\ estimate and construct the profiled least squares criterion $%
\mathbb{S}_{T}\left( \gamma \right) $, that is,
\begin{eqnarray*}
	\mathbb{S}_{T}\left( \gamma \right) &=&\mathbb{S}_{T}\left( \widehat{\alpha}%
	\left( \gamma \right) ,\gamma \right) =\frac{1}{T}Y^{\prime }\left(
	I-P\left( \gamma \right) \right) Y \\
	&=&\frac{1}{T}\left( e^{\prime }\left( I-P\left( \gamma \right) \right)
	e+2\delta _{0}^{\prime }X_{0}\left( I-P\left( \gamma \right) \right)
	e+\delta _{0}^{\prime }X_{0}^{\prime }\left( I-P\left( \gamma \right)
	\right) X_{0}\delta _{0}\right) ,
\end{eqnarray*}%
where $e,Y,$ and $X_{0}$ are the matrices stacking $\varepsilon _{t}$'s, $%
y_{t}^{{}}$'s and $x_{t}^{\prime }1_{t}$'s, respectively$,$ and $P\left(
\gamma \right) $ is the orthogonal projection matrix onto $Z_{t}\left(
\gamma \right) $'s.

Let $\widetilde{\gamma}$ be an estimator such that
\begin{equation}
	\mathbb{S}_{T}\left( \widetilde{\gamma}\right) \leq \mathbb{S}_{T}\left( \gamma
	_{0}\right) +o_P\left( T^{-2\varphi }\right) .  \label{eq:approx gm}
\end{equation}%
Then, by Lemma \ref{A-rates:lemma-used}, the ULLN for $T^{-1}%
\sum_{t=1}^{T}Z_{t}\left( \gamma \right) Z_{t}\left( \gamma \right) ^{\prime
},$ the rank condition for $\mathbb EZ_{t}\left( \gamma \right) Z_{t}\left( \gamma
\right) ^{\prime }$ in Assumption \ref{A-mixing} (iii), the fact that $%
P\left( \gamma _{0}\right) X_{0}=X_{0}$,%
\begin{eqnarray*}
	0 &\geq &T^{2\varphi }\left( \mathbb{S}_{T}\left( \widetilde{\gamma}\right) -%
	\mathbb{S}_{T}\left( \gamma _{0}\right) \right) -o_P\left( 1\right) \\
	&=&\frac{T^{2\varphi }}{T}\left( e^{\prime }\left( P\left( \gamma
	_{0}\right) -P\left( \widetilde{\gamma}\right) \right) e+2\delta _{0}^{\prime
	}X_{0}\left( P\left( \gamma _{0}\right) -P\left( \widetilde{\gamma}\right)
	\right) e+\delta _{0}^{\prime }X_{0}^{\prime }\left( P\left( \gamma
	_{0}\right) -P\left( \widetilde{\gamma}\right) \right) X_{0}\delta _{0}\right) \\
	&=&o_P\left( 1\right) +\frac{1}{T}d_{0}^{\prime }X_{0}^{\prime }\left(
	I-P\left( \widetilde{\gamma}\right) \right) X_{0}d_{0}, \\
	&=&o_P\left( 1\right) +\mathbb Ed_{0}^{\prime }x_{t}x_{t}^{\prime }d_{0}1_{t}-%
	\underset{A\left( \widetilde{\gamma}\right) }{\underbrace{\left( \mathbb Ed_{0}^{\prime
			}x_{t}1_{t}Z_{t}\left( \widetilde{\gamma}\right) ^{\prime }\right) \left(
			\mathbb EZ_{t}\left( \widetilde{\gamma}\right) Z_{t}\left( \widetilde{\gamma}\right)
			^{\prime }\right) ^{-1}\mathbb EZ_{t}\left( \widetilde{\gamma}\right) 1_{t}x_{t}^{\prime
			}d_{0}}}.
\end{eqnarray*}%
However, the term $A\left( \widetilde{\gamma}\right) $ is continuous by Assumption \ref{A-f1 conti} and
has maximum at $\widetilde{\gamma}=\gamma _{0}$ by the property of the orthogonal projection, and
$ \mathbb Ed_{0}^{\prime }x_{t}x_{t}^{\prime }d_{0}1_{t} - A(\gamma) > 0 $ for any $ \gamma \neq \gamma_0 $ due to Assumptions \ref{A-f1 conti} \ref{A1:itm2} and \ref{A-mixing} \ref{A-mixing:itm3}.
Finally, the compact parameter space yields the consistency of $ \widehat{\gamma} $ by the argmax continuous mapping theorem (see, e.g.,   \citet[p.286]{VW}).

Turning to $ \widehat{\alpha} $, note that
	\begin{eqnarray}
	0 &\geq &\mathbb{S}_{T}\left( \widehat{\alpha },\widehat{\gamma }\right) -%
	\mathbb{S}_{T}\left( \alpha _{0},\gamma _{0}\right)  \notag \\
	&=&\mathbb{R}_{T}\left( \widehat{\alpha },\widehat{\gamma }\right) -\mathbb{G%
	}_{T}\left( \widehat{\alpha },\widehat{\gamma }\right) +\mathbb{G}_{T}\left(
	\alpha _{0},\gamma _{0}\right) ,  \label{eq:S-S0}
	\end{eqnarray}%
	where
	\begin{align*}
	\mathbb{R}_{T}\left( \alpha ,\gamma \right) &\equiv \frac{1}{T}%
	\sum_{t=1}^{T}\left( Z_{t}\left( \gamma \right) ^{\prime }\alpha
	-Z_{t}\left( \gamma _{0}\right) ^{\prime }\alpha _{0}\right) ^{2} \\
	\mathbb{G}_{T}\left( \alpha ,\gamma \right) &\equiv \frac{2}{T}%
	\sum_{t=1}^{T}\varepsilon _{t}Z_{t}\left( \gamma \right) ^{\prime }\alpha .
	\end{align*}%
	First, note that
	\begin{align}\label{R-term-LLN}
	\begin{split}
	& \mathbb{R}_{T}\left( \alpha ,\gamma \right) -R\left( \alpha
	,\gamma \right)  \\
	&=\left( \alpha -\alpha _{0}\right) ^{\prime }\frac{1}{T}%
	\sum_{t=1}^{T}\left( Z_{t}\left( \gamma \right) Z_{t}\left( \gamma \right)
	^{\prime }- \mathbb EZ_{t}\left( \gamma \right) Z_{t}\left( \gamma \right) ^{\prime
	}\right) \left( \alpha -\alpha _{0}\right)  \\
	&+\frac{1}{T}\sum_{t=1}^{T}\left( x_{t}^{\prime }\delta _{0}\right)
	^{2}\left\vert 1_{t}\left( \gamma \right) -1_{t}\left( \gamma _{0}\right)
	\right\vert - \mathbb E\left( x_{t}^{\prime }\delta _{0}\right) ^{2}\left\vert
	1_{t}\left( \gamma \right) -1_{t}\left( \gamma _{0}\right) \right\vert
	\\
	&+\frac{2\delta _{0}^{\prime }}{T}\sum_{t=1}^{T}\Big[ x_{t}\left(
	1_{t}\left( \gamma \right) -1_{t}\left( \gamma _{0}\right) \right) Z%
	_{t}\left( \gamma \right) - \mathbb E\left[ x_{t}\left( 1_{t}\left( \gamma \right)
	-1_{t}\left( \gamma _{0}\right) \right) Z_{t}\left( \gamma \right) \right] %
	\Big] ^{\prime }\left( \alpha -\alpha _{0}\right) \\
	&= o_P(1)(|\alpha -\alpha_0 |_2^2 + |\alpha -\alpha_0 |_2 )  \; \text{ uniformly in $ \gamma \in  \Gamma$},
	\end{split}
	\end{align}
	by ULLN.
	Similarly,
	\begin{align}\label{G-term-LLN}
	\begin{split}
	&  \mathbb{G}_{T}\left( \alpha ,\gamma \right) -\mathbb{G}_{T}\left( \alpha _{0},\gamma _{0}\right) \\
	&=\frac{2}{T}\sum_{t=1}^{T}\varepsilon _{t}Z_{t}\left( \gamma \right)
	^{\prime }\left( \alpha -\alpha _{0}\right) +\frac{2}{T}\sum_{t=1}^{T}%
	\varepsilon _{t}x_{t}^{\prime }\delta _{0}\left( 1_{t}\left( \gamma \right)
	-1_{t}\left( \gamma _{0}\right) \right), \\
	&= o_P(1)( |\alpha -\alpha_0 |_2 )  \;
	\text{ uniformly in $ \gamma \in  \Gamma$}
	\end{split}
	\end{align}
	Combining these results together  implies that
	\begin{equation*}
	R\left( \widehat{\alpha},\widehat{\gamma} \right) \leq
	o_P (1)(|\widehat{\alpha} -\alpha_0 |_2 + |\widehat{\alpha} -\alpha_0 |_2^2 ).
	\end{equation*}%
	Then, combining this result with  the proof of Theorem \ref{iden-thm} implies that
	$ \widehat{\alpha} - \alpha_{0} = o_P (1)$ as \eqref{iden-eq-lower-bound} shows that $ R $ is bounded below by some positive constant times $ |\alpha - \alpha_0 |_2^2 $.
\end{proof}

\subsection{Rates of Convergence}


To begin with, we assume $ \gamma $ belongs to a small neighborhood of $ \gamma_0 $ due to the preceding consistency proof.
It is useful to introduce additional notation.
Let $1_{t}\left( \gamma \right) \equiv 1\left\{ f_{t}^{\prime }\gamma >0\right\} $
while $1_{t} \equiv 1_{t}\left( \gamma _{0}\right) $. Similarly, let $1_{t}\left(
\gamma ,\bar{\gamma}\right) \equiv 1\left\{ f_{t}^{\prime }\gamma \leq
0<f_{t}^{\prime }\bar{\gamma}\right\} $. Clearly, $1_{t}\left( \gamma
\right) =1_{t}\left( 0,\gamma \right) $.

Define
\begin{align*}
H_{1,t}(\gamma)
&:= \varepsilon _{t}x_{t}^{\prime }\delta _{0}\left( 1_{t}\left( \gamma \right)-1_{t}\left( \gamma _{0}\right) \right), \\
H_{2,t}(\gamma)
&:=
\left( x_{t}^{\prime }\delta _{0}\right)
^{2}\left\vert 1_{t}\left( \gamma \right) -1_{t}\left( \gamma _{0}\right)
\right\vert, \\
H_{3,t}(\gamma)
&:= (x_{t}'\delta _{0}) \left(
1_{t}\left( \gamma \right) -1_{t}\left( \gamma _{0}\right) \right) Z_{tj} \left( \gamma \right),
\end{align*}
where $Z_{tj}\left( \gamma \right)$ is the $j$-th element of  $Z_{t}\left( \gamma \right)$.
For the simplicity of notation, we suppress the dependence of $H_{3,t}(\gamma)$ on $j$.
We first state a lemma that is a direct consequence of Lemmas \ref{Lem:modul1} and \ref{Lem:rate_gm} for an easy reference.

\begin{lem}
	\label{A-rates:lemma-used}
	There exists a constant $C_2 > 0$ such that for any $\eta > 0$,
	\begin{align*}
	& \sup_{\left\vert \gamma - \gamma_0 \right\vert_2
		\leq T^{-1+2\varphi } }\left\vert \frac{1}{T}\sum_{t=1}^{T}
	\left\{ H_{k,t}(\gamma)
	-\mathbb{E} H_{k,t}(\gamma) \right\} \right\vert
	= O_{P}\left( \frac{1}{T}\right), \\
		& \sup_{\left\vert \gamma - \gamma_0 \right\vert_2
		\leq T^{-1+2\varphi } }\left\vert \frac{1}{T}\sum_{t=1}^{T}
	\left\{ H_{2,t}(\gamma)
	-\mathbb{E} H_{2,t}(\gamma) \right\} \right\vert
	= O_{P}\left( \frac{1}{T^{1+\varphi}}\right), \\
	&\sup_{T^{-1+2\varphi }<\left\vert \gamma -\gamma_{0}\right\vert _{2}<C_2}
	\Bigg|
	\left| \frac{1}{T}\sum_{t=1}^{T}%
	\left\{ H_{k,t}(\gamma)
	-\mathbb{E} H_{k,t}(\gamma) \right\} \right|
	- \eta T^{-2\varphi }\left\vert \gamma -\gamma
	_{0}\right\vert _{2} \Bigg|
	=
	O_{P}\left( \frac{1}{T}\right),
	\end{align*}%
	where $k=1,2,3$.
\end{lem}

\begin{lem}[Rates of Convergence]\label{rates-thm}
	Let Assumptions \ref{scale-normalization}, \ref{A-f1 conti}, \ref{A-mixing}, and
	\ref{A:diminishing dt}	hold. Then  as $T \rightarrow \infty$,
	\begin{equation*}
	\left\vert \widehat{\alpha }-\alpha _{0}\right\vert _{2}=O_{P}\left( \frac{1%
	}{\sqrt{T}}\right) \ \text{and\ }\left\vert \widehat{\gamma }-\gamma
	_{0}\right\vert _{2}=O_{P}\left( \frac{1}{T^{1-2\varphi }}\right) .
	\end{equation*}
\end{lem}

\begin{proof}[Proof of Lemma \ref{rates-thm}]
The proof is based on the following two steps, which will be shown later.

\noindent
\textit{Step 1}. As $T \rightarrow \infty$, there exist positive constants $c$ and $ e $, with probability approaching one,
\begin{align*}
R\left( \alpha ,\gamma \right) &\geq c \left\vert \alpha -\alpha _{0}\right\vert _{2}^{2}+ c T^{-2\varphi
}\left\vert \gamma -\gamma _{0}\right\vert _{2},
\end{align*}%
for any $ \alpha $ and $ \gamma $ such that $ |\alpha-\alpha_0|<e $ and $ |\gamma - \gamma_0|<e $. Recall $R(\alpha,\gamma)$ is defined in (\ref{eq:excess:supp}).

\noindent
\textit{Step 2}. There exists a positive constant $\eta < c/2$ such that
\begin{align}
\left\vert \mathbb{G}_{T}\left( \alpha ,\gamma \right) -\mathbb{G}_{T}\left(
\alpha _{0},\gamma _{0}\right) \right\vert &\leq O_{P}\left( \frac{1}{\sqrt{%
T}}\right) \left\vert \alpha -\alpha _{0}\right\vert _{2}+\eta T^{-2\varphi
}\left\vert \gamma -\gamma _{0}\right\vert _{2}+O_{P}\left( \frac{1}{T}%
\right)  \label{eq:boundG} \\
\left\vert \mathbb{R}_{T}\left( \alpha ,\gamma \right) -R\left( \alpha
,\gamma \right) \right\vert &\leq \eta \left\vert \alpha -\alpha
_{0}\right\vert _{2}^{2}+\eta T^{-2\varphi }\left\vert \gamma -\gamma
_{0}\right\vert _{2}+O_{P}\left( \frac{1}{T}\right), \label{eq:boundR}
\end{align}%
where the inequalities above are uniform in $\alpha $ and $\gamma $ such that $ |\alpha-\alpha_0|<e $ and $ |\gamma - \gamma_0|<e $,
in the sense that
the sequences $O_{P}\left( \cdot \right) $ and $o_P\left( \cdot \right) $
do not depend on $\alpha $ and $\gamma $.

Given Steps 1 and 2, since
\begin{equation*}
R\left( \widehat{\alpha },\widehat{\gamma }\right) \leq \left\vert \mathbb{G}%
_{T}\left( \widehat{\alpha },\widehat{\gamma }\right) -\mathbb{G}_{T}\left(
\alpha _{0},\gamma _{0}\right) \right\vert +\left\vert \mathbb{R}_{T}\left(
\widehat{\alpha },\widehat{\gamma }\right) -R\left( \widehat{\alpha },%
\widehat{\gamma }\right) \right\vert ,
\end{equation*}%
we conclude that
\begin{equation}
\left( c-2\eta \right) \left( \left\vert \widehat{\alpha }-\alpha
_{0}\right\vert _{2}^{2}+T^{-2\varphi }\left\vert \widehat{\gamma }-\gamma
_{0}\right\vert _{2}\right) \leq O_{P}\left( \frac{1}{\sqrt{T}}\right)
\left\vert \widehat{\alpha }-\alpha _{0}\right\vert _{2}+O_{P}\left( \frac{1%
}{T}\right) .  \label{eq:ratebound}
\end{equation}%
That is,
\begin{equation*}
\left\vert \widehat{\alpha }-\alpha _{0}\right\vert _{2}^{2}\leq O_{P}\left(
\frac{1}{\sqrt{T}}\right) \left\vert \widehat{\alpha }-\alpha
_{0}\right\vert _{2}+O_{P}\left( \frac{1}{T}\right) ,
\end{equation*}%
implying
\begin{equation*}
\left\vert \widehat{\alpha }-\alpha _{0}\right\vert _{2}=O_{P}\left( \frac{1%
}{\sqrt{T}}\right) \ \text{and thus\ }\left\vert \widehat{\gamma }-\gamma
_{0}\right\vert _{2}=O_{P}\left( \frac{1}{T^{1-2\varphi }}\right) .
\end{equation*}
\end{proof}

\proof[Proof of Step 1]

Due to Assumption \ref{A:diminishing dt} and then Assumption \ref{A-f1 conti} we can find positive constants $ c, c_0 $ such that
\begin{align*}
 \mathbb E\left( x_{t}^{\prime }\delta _{0}\left( 1_{t}\left( \gamma \right)
-1_{t}\left( \gamma _{0}\right) \right) \right) ^{2}
&\geq
T^{-2\varphi
}c \mathbb E\left\vert 1_{t}\left( \gamma \right) -1_{t}\left( \gamma _{0}\right)
\right\vert
\\
&\geq
c_0 T^{-2\varphi }\left\vert \gamma -\gamma _{0}\right\vert
_{2}.
\end{align*}
More specifically, we need to show that there exists a constant $c>0$ and a neighborhood of $%
\gamma _{0}$ such that for all $\gamma $ in the neighborhood
\[
G\left( \gamma \right) =\mathbb{E}\left\vert 1_{t}\left( \gamma \right)
-1_{t}\left( \gamma _{0}\right) \right\vert \geq c\left\vert \gamma -\gamma
_{0}\right\vert _{2}.
\]%
Note that $f_{t}^{\prime }\gamma _{0}=u_{t}$ and the first element of $%
\left( \gamma -\gamma _{0}\right) $ is zero due to the normalization. Then,
\[
G\left( \gamma \right) =\mathbb{P}\left\{ -f_{2t}^{\prime }\left( \gamma
_{2}-\gamma _{20}\right) \leq u_{t}<0\right\} +\mathbb{P}\left\{ 0<u_{t}\leq
-f_{2t}^{\prime }\left( \gamma _{2}-\gamma _{20}\right) \right\} .
\]%
Since the conditional density of $u_{t}$ is bounded away from zero and
continuous, we can find a strictly positive lower bound, say $c_{1},$ of the
conditional density of $u_{t}$ if we choose a sufficiently small open
neighborhood $\epsilon $ of zero. Then,
\[
\mathbb{P}\left\{ -f_{2t}^{\prime }\left( \gamma _{2}-\gamma _{20}\right)
\leq u_{t}<0\right\} \geq c_{1}\mathbb{E}\left( f_{2t}^{\prime }\left(
\gamma _{2}-\gamma _{20}\right) 1\left\{ f_{2t}^{\prime }\left( \gamma
_{2}-\gamma _{20}\right) >0\right\} 1\left\{ \left\vert f_{2t}^{\prime
}\right\vert \leq M\right\} \right) ,
\]%
where $M$ satisfies that $\max \left\vert \gamma -\gamma _{0}\right\vert
_{2}M$ belongs to $\epsilon $. This is always feasible because we can make $%
\max \left\vert \gamma -\gamma _{0}\right\vert _{2}$ as small as necessary
due to the consistency of $\hat{\gamma}$. Similarly,
\[
\mathbb{P}\left\{ 0<u_{t}\leq -f_{2t}^{\prime }\left( \gamma _{2}-\gamma
_{20}\right) \right\} \geq c_{1}\mathbb{E}\left( -f_{2t}^{\prime }\left(
\gamma _{2}-\gamma _{20}\right) 1\left\{ f_{2t}^{\prime }\left( \gamma
_{2}-\gamma _{20}\right) <0\right\} 1\left\{ \left\vert f_{2t}^{\prime
}\right\vert \leq M\right\} \right) .
\]%
Thus,
\[
G\left( \gamma \right) \geq c_{1}\mathbb{E}\left( \left\vert f_{2t}^{\prime
}\left( \gamma _{2}-\gamma _{20}\right) \right\vert 1\left\{ \left\vert
f_{2t}^{\prime }\right\vert \leq M\right\} \right) \geq c_{2}\left\vert
\gamma -\gamma _{0}\right\vert _{2}
\]%
for some $c_{2}>0$ because
\[
\inf_{\left\vert r\right\vert =1}\mathbb{E}\left( \left\vert f_{2t}^{\prime
}r\right\vert 1\left\{ \left\vert f_{2t}^{\prime }\right\vert \leq M\right\}
\right) >0
\]%
for some $M<\infty $ due to Assumption \ref{A:diminishing dt}.

Next,
\begin{align*}
 \mathbb E\left( Z_{t}\left( \gamma \right)
^{\prime }\left( \alpha -\alpha _{0}\right) \right) ^{2} &
\geq c_1 \left\vert \alpha -\alpha _{0}\right\vert _{2}^{2},
\end{align*}%
due to Assumption \ref{A-mixing} \ref{A-mixing:itm3}.

Also, note that
\begin{align*}
&\left\vert  \mathbb E\left( x_{t}^{\prime }\delta _{0}\left( 1_{t}\left( \gamma
\right) -1_{t}\left( \gamma _{0}\right) \right) \right) Z_{t}\left( \gamma
\right) ^{\prime }\left( \alpha -\alpha _{0}\right) \right\vert \\
&\leq T^{-\varphi }   \mathbb E \Big[ \left| x_{t}^{\prime }d _{0}\right| \left| 1_{t}\left( \gamma \right)
-1_{t}\left( \gamma _{0}\right) \right|
 \left| Z_{t}\left( \gamma
\right) \right|_2 \left| \alpha -\alpha _{0} \right|_2 \Big]  \\
&\leq 2 T^{-\varphi }  | d_0 |_2 C_0 C_1 \left\vert \gamma -\gamma _{0}\right\vert
_{2} \left\vert \alpha -\alpha _{0}\right\vert _{2},
\end{align*}%
where  the second inequality comes from Assumption \ref{A-mixing} \ref{A-mixing:itm1} and Assumption \ref{A-f1 conti} \ref{A1:itm1}.
Combining the inequalities above together yields that
\begin{align}\label{expansion-in-R}
\begin{split}
R\left( \alpha ,\gamma \right) &= \mathbb E\left( Z_{t}\left( \gamma \right)
^{\prime }\left( \alpha -\alpha _{0}\right) \right) ^{2}+ \mathbb E\left(
x_{t}^{\prime }\delta _{0}\left( 1_{t}\left( \gamma \right) -1_{t}\left(
\gamma _{0}\right) \right) \right) ^{2} \\
&+2 \mathbb E\left( x_{t}^{\prime }\delta _{0}\left( 1_{t}\left( \gamma \right)
-1_{t}\left( \gamma _{0}\right) \right) \right) Z_{t}\left( \gamma \right)
^{\prime }\left( \alpha -\alpha _{0}\right) \\
&\geq c_1 \left\vert \alpha -\alpha _{0}\right\vert _{2}^{2}+c_0 T^{-2\varphi
}\left\vert \gamma -\gamma _{0}\right\vert _{2}- C_2 T^{-\varphi }  \left\vert \gamma -\gamma _{0}\right\vert
_{2} \left\vert \alpha -\alpha _{0}\right\vert _{2},
\end{split}
\end{align}%
where $C_2 = 2   | d_0 |_2 C_0 C_1$.

We consider two cases: (i)
$
 c_1 \left\vert \alpha -\alpha _{0}\right\vert _{2}
\geq
2C_2 T^{-\varphi }  \left\vert \gamma -\gamma _{0}\right\vert_{2}
$
and (ii)
$
 c_1 \left\vert \alpha -\alpha _{0}\right\vert _{2}
<
2C_2 T^{-\varphi }  \left\vert \gamma -\gamma _{0}\right\vert_{2}.
$
When (i) holds,
\begin{align*}
R\left( \alpha ,\gamma \right) &\geq \frac{c_1}{2} \left\vert \alpha -\alpha _{0}\right\vert _{2}^{2}+c_0 T^{-2\varphi
}\left\vert \gamma -\gamma _{0}\right\vert _{2}.
\end{align*}%
When (ii) holds, we have that
$$
 C_2 T^{-\varphi }  \left\vert \gamma -\gamma _{0}\right\vert
_{2}  \left\vert \alpha -\alpha _{0}\right\vert _{2}
<
2 c_1^{-1} C_2^2 T^{-2\varphi }  \left\vert \gamma -\gamma _{0}\right\vert
_{2}^2.
$$
Then under (ii),
\begin{align*}
&c_0 T^{-2\varphi
}\left\vert \gamma -\gamma _{0}\right\vert _{2} - C_2
 T^{-\varphi }  \left\vert \gamma -\gamma _{0}\right\vert
_{2}  \left\vert \alpha -\alpha _{0}\right\vert _{2} \\
&> T^{-2\varphi
}\left\vert \gamma -\gamma _{0}\right\vert _{2}
\left[c_0 - 2 c_1^{-1} C_2^2   \left\vert \gamma -\gamma _{0}\right\vert
_{2} \right].
\end{align*}
Thus, as long as $\left\vert \gamma -\gamma _{0}\right\vert
_{2} \leq c_0 c_1/(4 C_2^2)$,  we obtain the desired result.
This completes the proof of Step 1 by taking $c = \min \{ c_0, c_1 \}/2$ since
$\left\vert \widehat{\gamma }-\gamma
_{0}\right\vert _{2}=o_{P}\left( 1 \right)$
by Lemma  \ref{cons-thm}.
$\qed$

\proof[Proof of Step 2]

To prove (\ref{eq:boundG}), note that as in \eqref{G-term-LLN},
\begin{align}
&\frac{1}{2}\left| \mathbb{G}_{T}\left( \alpha ,\gamma \right) -\mathbb{G}%
_{T}\left( \alpha _{0},\gamma _{0}\right) \right|  \label{eq:G-G0} \\
&\leq \left| \frac{1}{T}\sum_{t=1}^{T}\varepsilon _{t}Z_{t}\left( \gamma \right)
^{\prime }\left( \alpha -\alpha _{0}\right) \right| + \left| \frac{1}{T}\sum_{t=1}^{T}%
\varepsilon _{t}x_{t}^{\prime }\delta _{0}\left( 1_{t}\left( \gamma \right)
-1_{t}\left( \gamma _{0}\right) \right) \right|  \notag \\
&= O_{P}\left( \frac{1}{\sqrt{T}}\right) \left\vert \alpha -\alpha
_{2}\right\vert _{2}+\eta T^{-2\varphi }\left\vert \gamma -\gamma
_{0}\right\vert _{2}+O_{P}\left( \frac{1}{T}\right)  \notag
\end{align}%
for any $0<\eta <c/2$,
by the MDS CLT and Lemma \ref{Lem:modul1}  for the first term
$T^{-1/2}\sum_{t=1}^{T}\varepsilon _{t}Z_{t}\left(
\gamma \right) $ and by Assumption \ref{A-rates:lemma-used} for the second term
 $T^{-1}\sum_{t=1}^{T}\varepsilon _{t}x_{t}^{\prime }\delta
_{0}\left( 1_{t}\left( \gamma \right) -1_{t}\left( \gamma _{0}\right)
\right) $.

We now prove \eqref{eq:boundR}.
Note that for any $0<\eta <c/2$,  as in \eqref{R-term-LLN},
\begin{align}
&\left\vert \mathbb{R}_{T}\left( \alpha ,\gamma \right) -R\left( \alpha
,\gamma \right) \right\vert  \label{eq:R_T - R} \\
&\leq \left| \left( \alpha -\alpha _{0}\right) ^{\prime }\frac{1}{T}%
\sum_{t=1}^{T}\left( Z_{t}\left( \gamma \right) Z_{t}\left( \gamma \right)
^{\prime }-\mathbb EZ_{t}\left( \gamma \right) Z_{t}\left( \gamma \right) ^{\prime
}\right) \left( \alpha -\alpha _{0}\right)  \right| \notag \\
&+ \left| \frac{1}{T}\sum_{t=1}^{T}\left( x_{t}^{\prime }\delta _{0}\right)
^{2}\left\vert 1_{t}\left( \gamma \right) -1_{t}\left( \gamma _{0}\right)
\right\vert -\mathbb{E}\left( x_{t}^{\prime }\delta _{0}\right) ^{2}\left\vert
1_{t}\left( \gamma \right) -1_{t}\left( \gamma _{0}\right) \right\vert \right|
\notag \\
&+\left| \frac{2}{T}\sum_{t=1}^{T}\delta _{0}^{\prime } \left[ x_{t}\left(
1_{t}\left( \gamma \right) -1_{t}\left( \gamma _{0}\right) \right) Z%
_{t}\left( \gamma \right) -\mathbb{E}\left[ x_{t}\left( 1_{t}\left( \gamma \right)
-1_{t}\left( \gamma _{0}\right) \right) Z_{t}\left( \gamma \right) \right] %
\right] ^{\prime }\left( \alpha -\alpha _{0}\right) \right|  \notag \\
&\leq o_P\left( \left\vert \alpha -\alpha _{0}\right\vert _{2}^{2}\right)
+O_{P}\left( \frac{1}{T}\right) +\eta T^{-2\varphi }\left\vert \gamma
-\gamma _{0}\right\vert _{2}  \notag
\end{align}%
by ULLN for the first term and by Lemma \ref{A-rates:lemma-used} for the second and third terms. This completes the proof.
$\qed$

\subsection{Asymptotic Distribution}\label{sec:asymptotic-distn:known-f}

\begin{proof}[Proof of Theorem \ref{asdist-alpha-gamma}]
	Let $r_T \equiv  T^{1-2\varphi }$, $a \equiv \sqrt{T}\left( \alpha -\alpha
	_{0}\right) $ and $g \equiv r_T \left( \gamma -\gamma _{0}\right) $.
	To prove the theorem, we first
	derive the weak convergence of the process
	\begin{equation*}
	\mathbb{K}_{T}\left( a,g\right) \equiv T\left( \mathbb{S}_{T}\left( \alpha
	_{0}+a\cdot T^{-1/2},\gamma _{0}+g\cdot r_{T}^{-1}\right) -\mathbb{S}%
	_{T}\left( \alpha _{0},\gamma _{0}\right) \right) ,
	\end{equation*}%
	over an arbitrary compact set, say $\mathcal{AG}$, and then apply the argmax continuous mapping
	theorem to obtain the limit distribution of $\widehat{\alpha}$ and $\widehat{\gamma}$.
	\medskip

	\noindent
	\textit{Step 1}.
	The following decomposition holds uniformly in $(a,g) \in \mathcal{AG}$:
	\begin{equation*}
	\mathbb{K}_{T}\left( a,g\right) =\mathbb{K}_{1T}\left( a\right) +\mathbb{K}%
	_{2T}\left( g\right) - 2\mathbb{K}_{3T}\left( g\right) +o_P\left( 1\right) ,
	\end{equation*}%
	where
	\begin{align*}
	\mathbb{K}_{1T}\left( a\right) &:=a^{\prime }\mathbb EZ_{t}\left( \gamma _{0}\right)
	Z_{t}\left( \gamma _{0}\right) ^{\prime }a-\frac{2}{\sqrt{T}}%
	\sum_{t=1}^{T}\varepsilon _{t}Z_{t}\left( \gamma _{0}\right) ^{\prime }a, \\
	\mathbb{K}_{2T}\left( g\right) &:= T\cdot  \mathbb E\left[\left( x_{t}^{\prime }\delta
	_{0}\right) ^{2}\left\vert 1_{t}\left( \gamma _{0}+g\cdot r_{T}^{-1}\right)
	-1_{t}\right\vert\right], \\
	\mathbb{K}_{3T}\left( g\right) &:= \sum_{t=1}^{T}\varepsilon
	_{t}x_{t}^{\prime }\delta _{0}\left( 1_{t}\left( \gamma _{0}+g\cdot
	r_{T}^{-1}\right) -1_{t}\right).
	\end{align*}

	\proof[Proof of Step 1]
	To begin with, note that (\ref{eq:R_T - R}) and Lemma \ref{A-rates:lemma-used} together imply that
	\begin{align}\label{tmp-1}
	\begin{split}
	&T\cdot \left[ \mathbb{R}_{T}\left( \alpha _{0}+a\cdot T^{-1/2},\gamma
	_{0}+g\cdot r_{T}^{-1}\right) -R\left( \alpha _{0}+a\cdot T^{-1/2},\gamma
	_{0}+g\cdot r_{T}^{-1}\right) \right] \\
	&=o_P(1) \;\; \text{uniformly in $(a,g) \in \mathcal{AG}$}.
	\end{split}
	\end{align}%
	Recall \eqref{expansion-in-R} and write that
	\begin{align}\label{tmp-2}
	\begin{split}
	&T \cdot R\left( \alpha _{0}+a\cdot T^{-1/2} ,\gamma _{0}+g\cdot r_{T}^{-1} \right) \\
	&=a' \mathbb E\left[ Z_{t}\left( \gamma _{0}+g\cdot r_{T}^{-1} \right) Z_{t}\left( \gamma _{0}+g\cdot r_{T}^{-1} \right)^{\prime } \right] a \\
	&+T\cdot \mathbb   E\left( x_{t}^{\prime }\delta _{0}\right) ^{2}\left\vert 1\left\{
	f_{t}^{\prime }\left( \gamma _{0}+g\cdot r_{T}^{-1}\right) >0\right\}
	-1\left\{ f_t'\gamma_0\right\} \right\vert \\
	&+2T^{1/2} \cdot \mathbb E\left( x_{t}^{\prime }\delta _{0}\left( 1_{t}\left( \gamma _{0}+g\cdot r_{T}^{-1} \right)
	-1_{t}\left( \gamma _{0}\right) \right) \right) Z_{t}\left( \gamma _{0}+g\cdot r_{T}^{-1} \right)
	^{\prime }a.
	\end{split}
	\end{align}%
	Then, due to Assumption \ref{A:diminishing dt},
	\begin{align}\label{tmp-3}
	\begin{split}
	a'\left\{ \mathbb E\left[ Z_{t}\left( \gamma _{0}+g\cdot r_{T}^{-1} \right) Z_{t}\left( \gamma _{0}+g\cdot r_{T}^{-1} \right)^{\prime } \right]
	- \mathbb E\left[ Z_{t}\left( \gamma_0 \right) Z_{t}\left( \gamma_0) \right)^{\prime } \right]  \right\} a &= o_P(1), \\
	T^{1/2} \cdot \mathbb E\left[ \left( x_{t}^{\prime }\delta _{0}\left( 1_{t}\left( \gamma _{0}+g\cdot r_{T}^{-1} \right)
	-1_{t}\left( \gamma _{0}\right) \right) \right) Z_{t}\left( \gamma _{0}+g\cdot r_{T}^{-1} \right)
	^{\prime } \right] a &= o_P(1)
	\end{split}
	\end{align}
	uniformly in $(a,g) \in \mathcal{AG}$.
	Then combining \eqref{tmp-1}-\eqref{tmp-3} yields that
	\begin{align}\label{tmp-4}
	\begin{split}
	&T\cdot  \mathbb{R}_{T}\left( \alpha _{0}+a\cdot T^{-1/2},\gamma
	_{0}+g\cdot r_{T}^{-1}\right) \\
	&=
	a' \mathbb E\left[ Z_{t}\left( \gamma _{0} \right) Z_{t}\left( \gamma _{0}\right)^{\prime } \right] a
	+T\cdot \mathbb{E}\left( x_{t}^{\prime }\delta _{0}\right) ^{2}\left\vert 1\left\{
	f_{t}^{\prime }\left( \gamma _{0}+g\cdot r_{T}^{-1}\right) >0\right\}
	-1\left\{ f_t'\gamma_0\right\} \right\vert  \\
	&+ o_P(1) \;\; \text{uniformly in $(a,g) \in \mathcal{AG}$}.
	\end{split}
	\end{align}%

	We now consider the term $T\left[ \mathbb{G}_{T}\left( \alpha _{0}+a\cdot T^{-1/2},\gamma
	_{0}+g\cdot r_{T}^{-1}\right) -\mathbb{G}_{T}\left( \alpha _{0},\gamma
	_{0}\right) \right]$. First, note that due to Lemma \ref{Lem:modul1},
	\begin{align}\label{tmp-5}
	\frac{1}{\sqrt{T}}\sum_{t=1}^{T}\varepsilon _{t}
	\left[ Z_{t}\left( \gamma _{0}+g\cdot r_{T}^{-1}\right)
	- Z_{t}\left( \gamma _{0}\right) \right]^{\prime }a
	= o_P\left( 1\right)
	\end{align}
	uniformly in $(a,g) \in \mathcal{AG}$.
	Then, recall (\ref{G-term-LLN}) and write that
	\begin{align}
	&T\left[ \mathbb{G}_{T}\left( \alpha _{0}+a\cdot T^{-1/2},\gamma
	_{0}+g\cdot r_{T}^{-1}\right) -\mathbb{G}_{T}\left( \alpha _{0},\gamma
	_{0}\right) \right] \nonumber \\
	&=\frac{2}{\sqrt{T}}\sum_{t=1}^{T}\varepsilon _{t}Z_{t}\left( \gamma _{0}+g\cdot r_{T}^{-1}\right) ^{\prime }a+2\sum_{t=1}^{T}\varepsilon _{t}x_{t}^{\prime }\delta
	_{0}\left( 1_{t}\left( \gamma _{0}+g\cdot r_{T}^{-1}\right) -1_{t}\left(
	\gamma _{0}\right) \right) \nonumber \\
	&=\frac{2}{\sqrt{T}}\sum_{t=1}^{T}\varepsilon _{t}Z_{t}\left( \gamma
	_{0}\right) ^{\prime }a+2\sum_{t=1}^{T}\varepsilon _{t}x_{t}^{\prime }\delta
	_{0}\left( 1_{t}\left( \gamma _{0}+g\cdot r_{T}^{-1}\right) -1_{t}\left(
	\gamma _{0}\right) \right) +o_P\left( 1\right), \label{AD-G}
	\end{align}%
	uniformly in $(a,g) \in \mathcal{AG}$, where the last equality follows from \eqref{tmp-5}.
	Then Step 1 follows immediately recalling the decomposition in (\ref{eq:S-S0})
	and collecting the leading terms in \eqref{tmp-4} and \eqref{AD-G}.
 $\qed$

	In view of Step 1,
	the limiting distribution of $a$ is determined by $\mathbb{K}_{1T}\left( a\right)$. That is,
	\begin{align*}
	a = \left[ \mathbb EZ_{t}\left( \gamma _{0}\right) Z_{t}\left( \gamma _{0}\right) ^{\prime } \right]^{-1}
	\frac{1}{\sqrt{T}}\sum_{t=1}^{T}\varepsilon _{t}Z_{t}\left( \gamma _{0}\right)
	+ o_P(1).
	\end{align*}
	Then the first desired result follows directly from the martingale difference central limit theorem \citep[e.g.][]{Hall-Heyde}.

	\noindent
	\textit{Step 2}.
	\begin{equation*}
	T^{1-2\varphi }\left( \widehat{\gamma}-\gamma _{0}\right) \overset{d}{%
		\longrightarrow }\limfunc{argmin}_{g \in \mathcal{G}}
\mathbb E \left[ \left( x_{t}^{\prime }d_{0}\right) ^{2}\left\vert
	f_{t}^{\prime }g\right\vert p_{u_t|f_{2t}}(0) \right] +2W\left( g\right) ,
	\end{equation*}%
	where $W$ is a Gaussian process whose covariance kernel is given by $H\left(
	\cdot ,\cdot \right)$  in \eqref{H-Gaussian-process-covariance-kernel}
	and
	$\mathcal{G} = \{ g \in \mathbb{R}^d: g_1 = 0 \}$.

	\proof[Proof of Step 2]
	The distribution of $g$ is determined by
	$\mathbb{K}_{2T}\left( g\right) -2 \mathbb{K}_{3T}\left( g\right)$.
	For the weak convergence of $\mathbb{K}_{3T}\left( g\right) $, we need to
	verify the tightness of the process and the finite dimensional convergence.
	The tightness is the consequence of Lemma \ref{Lem:modul1} since for any
	finite $g$  and for any $%
	c>0$,
	\begin{eqnarray*}
		&&\mathbb{P }\left\{ \sup_{\left\vert h-g\right\vert <\epsilon }\left\vert
		\mathbb{K}_{3T}\left( g\right) -\mathbb{K}_{3T}\left( h\right) \right\vert
		>c\right\} \\
		&= &\mathbb{P }\left\{ \sup_{\left\vert \vec{\gamma}-\gamma \right\vert
			<\epsilon /r_{T}} \left\vert \frac{1}{\sqrt{T}}\sum_{t=1}^{T}\varepsilon
		_{t}x_{t}^{\prime }d_{0}\left( 1_{t}\left( \vec{\gamma}\right) -1_{t}\left(
		\gamma \right) \right)\right\vert >\frac{c}{2\sqrt{T}}T^{\varphi }\right\}  \\
		&\leq &C\frac{\epsilon ^{2}}{c^{4}},
	\end{eqnarray*}%
	which can be made arbitrarily small by choosing $\epsilon $ small. For the
	fidi, we apply the martingale difference central limit theorem \citep[e.g.][]{Hall-Heyde}. Specifically, let $%
	w_{t}=\sqrt{r_{T}}\varepsilon _{t}x_{t}^{\prime }d_{0}\left( 1_{t}\left(
	\gamma _{0}+g\cdot r_{T}^{-1}\right) -1_{t}\right) $ and verify that $%
	\max_{t}\left\vert w_{t}\right\vert =o_P\left( \sqrt{T}\right) $ and that $%
	\frac{1}{T}\sum_{t=1}^{T}w_{t}^{2}$ has a proper non-degenerate probability
	limit. However, $T^{-2}\mathbb E\max_{t}w_{t}^{4}\leq T^{-1}\mathbb Ew_{t}^{4}$
	since    $\max_{t}\left\vert a_{t}\right\vert \leq
	\sum_{t=1}^{T}\left\vert a_{t}\right\vert $ and $w_t$ is stationary.
	Now,
	$$%
	T^{-1}\mathbb Ew_{t}^{4}=T^{-1}r_{T}^{2}\mathbb E \left[ \left( \varepsilon _{t}x_{t}^{\prime }d_{0} \right)^4 \left|
	1_{t}\left( \gamma _{0}+g\cdot r_{T}^{-1}\right) -1_{t}\right| \right] \leq
	CT^{-1}r_{T}=o\left( 1\right).
	$$
	Furthermore, $\frac{1}{T}%
	\sum_{t=1}^{T}\left( w_{t}^{2}-\mathbb Ew_{t}^{2}\right) =o_P\left( 1\right) $.
	The limit of
	$\mathbb Ew_{t}^{2}$ will be given later while we characterize the covariance kernel
	of the process $\mathbb{K}_{3T}\left( g\right) $.

	To derive the covariance kernel of $\mathbb{K}_{3T}\left( g\right) $ and the
	limit of $\mathbb{K}_{2T}\left( g\right) $, we need to derive the limit of
	the type%
	\begin{equation*}
	\lim_{m\rightarrow \infty }m\mathbb E\eta _{t}^{2}\left\vert 1\left\{ f_{t}^{\prime
	}\left( \gamma _{0}+s/m\right) >0\right\} -1\left\{ f_{t}^{\prime }\left(
	\gamma _{0}+g/m\right) >0\right\} \right\vert
	\end{equation*}%
	for some random variable $\eta_t$ given $s\neq g$.
	We split the remainder of the proof into two cases.

	\begin{remark} \label{rem-ind}
	In the meantime, we note that this proof also implies that the covariance between
	the second term in $ \mathbb{K}_{1T}(a) $ and
	$ \mathbb{K}_{3T}(g) $ degenerates, which implies the asymptotic independence between two processes.
	\end{remark}

	Recall that $\gamma _{1}=1$.
	With this normalization, we need to fix the first element of $g$ in $\mathbb{%
		K}_{2T}\left( g\right) \ $and $\mathbb{K}_{3T}\left( g\right) $ at zero.
	Thus, we assume $g\in \mathbb{R}^{d-1}$ with a slight abuse of notation and
	introduce $u_{t}=f_t'\gamma_0$ and
	\begin{equation*}
	h\left( (\eta _{t},u_{t},f_{2t}),g/m\right) =\eta _{t}1\left\{
	u_{t}+f_{2t}^{\prime }g/m>0\right\}
	\end{equation*}%
	for $g\in \mathbb{R}^{d-1}$ and some random variable $\eta _{t}$, which will
	be made more explicit later. Then, the asymptotic covariances of the process
	$\mathbb{K}_{3T}\left( g\right) $ and the limit of $\mathbb{K}_{2T}\left(
	g\right) $ are characterized by the limit of the type
	\begin{equation*}
	L\left( s,g\right) =\lim_{m\rightarrow \infty }m\mathbb{E}\left( h\left( \cdot
	,s/m\right) -h\left( \cdot ,g/m\right) \right) ^{2},
	\end{equation*}%
	for $g,s\in \mathbb{R}^{d-1}$. That is, for the asymptotic covariance kernel
	$H\left( s,g\right) $ of $\mathbb{K}_{3T}\left( g\right) $, set $\eta
	_{t}=x_{t}^{\prime }d_{0}\varepsilon _{t},$ which is a martingale difference
	sequence to render $\mathbb Eh\left( \cdot ,g/m\right) =0$, and $m=T^{1-2\varphi }$.
	Then,
	\begin{eqnarray*}
		H\left( s,g\right) &=&\func{cov}\left( \mathbb{K}_{3T}\left( s\right) ,%
		\mathbb{K}_{3T}\left( g\right) \right) \\
		&=&\mathbb E\left( \left( h\left( \cdot ,s/m\right) -\eta _{t}1\left\{
		u_{t}>0\right\} \right) \left( h\left( \cdot ,g/m\right) -\eta _{t}1\left\{
		u_{t}>0\right\} \right) \right) \\
		&=&\frac{1}{2}\left( L\left( s,0\right) +L\left( g,0\right) -L\left(
		s,g\right) \right) ,
	\end{eqnarray*}%
	since $2ab=a^{2}+b^{2}-\left( a-b\right) ^{2}$ and $h\left( \cdot ,0\right)
	=\eta _{t}1\left\{ u_{t}>0\right\} $. On the other hand, the limit of $%
	\mathbb{K}_{2T}\left( g\right) $ will be given by $L\left( g,0\right) $ with
	$\eta _{t}=x_{t}^{\prime }d_{0}$.

	Note that%
	\begin{eqnarray*}
		L\left( s,g\right) &=&\lim_{m\rightarrow \infty }m\mathbb E\eta _{t}^{2}\left\vert
		1\left\{ u_{t}+f_{2t}^{\prime }s/m>0\right\} -1\left\{ u_{t}+f_{2t}^{\prime
		}g/m>0\right\} \right\vert \\
		&=&m\mathbb E\eta _{t}^{2}1\left\{ u_{t}+f_{2t}^{\prime }s/m>0\geq
		u_{t}+f_{2t}^{\prime }g/m\right\} \\
		&&+m\mathbb E\eta _{t}^{2}1\left\{ u_{t}+f_{2t}^{\prime }g/m>0\geq
		u_{t}+f_{2t}^{\prime }s/m\right\} .
	\end{eqnarray*}
	Furthermore, let $p_{u|f_{2}}\left( \cdot \right) $ and $P_{2}$ denote the
	conditional density of $u_{t}$ given $f_{2t}=f_{2}$ and the probability
	measure for $f_{2t}$, respectively, and note that
	\begin{eqnarray*}
		&&m\mathbb E\eta _{t}^{2}1\left\{ u_{t}+f_{2t}^{\prime }s/m>0\geq
		u_{t}+f_{2t}^{\prime }g/m\right\} \\
		&=&\int \int \mathbb E\left[ \eta _{t}^{2}|w/m,f_{2}\right] 1\left\{ -f_{2}^{\prime
		}g\geq w>-f_{2}^{\prime }s\right\} p_{u|f_{2}}\left( w/m\right) dwdP_{2} \\
		&\rightarrow &\int \mathbb E\left[ \eta _{t}^{2}|0,f_{2}\right] \left( -f_{2}^{\prime
		}g+f_{2}^{\prime }s\right) 1\left( f_{2}^{\prime }g < f_{2}^{\prime
	}s\right) p_{u|f_{2}}\left( 0\right) dP_{2},
\end{eqnarray*}%
where the equality is by a change of variables, $w=m\cdot u$ and the
convergence is as $m\rightarrow \infty $ by the dominated convergence theorem (DCT). This implies that
\begin{equation*}
L\left( s,g\right) =\int \mathbb E\left[ \eta _{t}^{2}|0,f_{2}\right] \left\vert
f_{2}^{\prime }g-f_{2}^{\prime }s\right\vert p_{u|f_{2}}\left( 0\right)
dP_{2}.
\end{equation*}%
In the special case where $z_{t}^{\prime }g<0<z_{t}^{\prime }s$ almost
surely, $L\left( s,g\right) =L\left( s,0\right) +L\left( g,0\right) .$ This
happens when $f_{t}=\left( q_{t},-1\right) $ and thus $z_{t}$ is a constant
given $u_{t}$.

Therefore, putting together,
\begin{equation*}
T^{1-2\varphi }\left( \widehat{\gamma}-\gamma _{0}\right) \overset{d}{%
	\longrightarrow }\limfunc{argmin}_{g\in \mathbb{R}^{d}:g_{1}=0}
\mathbb  E\left[ \left( x_{t}^{\prime }d_{0}\right) ^{2}\left\vert
f_{t}^{\prime }g\right\vert p_{u_t|f_{2t}}(0) \right] +2W\left( g\right) ,
\end{equation*}%
where $W$ is a Gaussian process whose covariance kernel is given by
\begin{equation*}
H\left( s,g\right) =\frac{1}{2} \mathbb{E}\left[ \left(
x_{t}^{\prime }d_{0}\right) ^{2}\left( \left\vert f_{t}^{\prime
}g\right\vert +\left\vert f_{t}^{\prime }s\right\vert -\left\vert
f_{t}^{\prime }\left( g-s\right) \right\vert \right) p_{u_t|f_{2t}}(0) \right]  .
\end{equation*}
\end{proof}

\noindent
\textit{Step 3}. Asymptotically, $a$ and $g$ are  independent of each other.

\begin{proof}[Proof of Step 3] This is straightforward due to the separability of $ \mathbb{K} $ into functions of $ a $ and $ g $, and due to Remark \ref{rem-ind} that addresses the independence between the processes of $a $ and $ g $.
\end{proof}

\section{Proof of Asymptotics in Section \ref{model:est:factors}: Estimated Factors}
\label{sec:estimated:proof}

\subsection{A Roadmap of the Proof}\label{sec:roadmap}

Due to the  complexity of the proof, we begin with  a roadmap to help readers follow the steps of the proof.

\begin{itemize}
	\item[] \textbf{Step I}.  We first prove a probability bound for $|\widetilde f_t-\widehat f_t|_2$ in Section \ref{sec:E1},
	where
	\begin{equation*}
		\widehat{f}_{t}=H_T'g_{t}+H_T'\frac{h_{t}}{\sqrt{N}}.
	\end{equation*}%

	\item[] \textbf{Step II}.
	We then replace the PCA estimator  $\widetilde f_t$  in the objective function $\widetilde {\mathbb{S}}_{T}\left( \alpha ,\gamma \right)$ with its first-order approximation $\widehat f_t$, and show that the effect of such a replacement is negligible for the convergence rates of the estimators we obtain in the later steps in Section \ref{sec:E2}.

	\item[] \textbf{Step III}.
	We show the consistency of estimators. To do so and to derive the convergence rates in the later steps, we use the alternative parametrization $ \phi = H_T  \gamma $, which helps us derive various uniform convergence lemmas. Note that the reparametrization is fine for the consistency and convergence rate results of the original parameter estimates since $ H_T $ is nonsingular with probability approaching one.
	\medskip

	\item[] \textbf{Step IV}.  We then decompose the objective function into the following form:
	$$
	\widetilde {\mathbb{S}}_{T}\left( \alpha , H_T^{-1}\phi \right)- \widetilde {\mathbb{S}}_{T}\left( \alpha_0 , H_T^{-1}\phi_0 \right)
	= {\mathds R_T(\alpha, \phi)}+ \mathds G_1(\phi)  -\mathds C(\alpha, \phi),
	$$
	where $\mathds R_T(\cdot,\cdot)$ and $\mathds G_1(\cdot)$ are deterministic functions and  $\mathds C(\cdot, \cdot)$ is a stochastic function. The formal definitions are given before Lemma \ref{l:rt}.
	Then as $\mathbb{\widetilde{S}}_{T}\left( \widehat{\alpha },\widehat{\gamma }%
	\right) -\mathbb{\widetilde{S}}_{T}\left( \alpha _{0},\gamma _{0}\right)\leq 0$, the decomposition yields:  for $\widehat \phi=H_T\widehat \gamma$,
	\begin{eqnarray}
		&&C|\widehat\alpha-\alpha_0|_2^2
		+ \mathds G_1(\widehat\phi)
		\leq   \mathds C(\widehat\alpha, \widehat\phi)
	\end{eqnarray}
	where ${\mathbb R_T(\alpha, \phi)}$ is lower bounded by $C|\alpha-\alpha_0|_2^2$ uniformly. Then, Lemmas \ref{l:rt} and \ref{l.1} establish uniform stochastic upper bounds for $ \mathds C(\widehat\alpha, \widehat\phi) $ through maximal inequalities.

	\item[] \textbf{Step V}.
	Next, we derive  a uniform lower bound for $\mathds G_1(\phi)$   over $\phi$ near  $\phi_0$ and over the ratio $\sqrt{N}/T^{1-2\varphi}$ in Lemma \ref{l.2}. In particular, $ \mathds G_1(\phi)$ has a ``kink" lower bound:
	$$
	\mathds G_1(\phi)
	\geq CT^{-2\varphi} |\phi-\phi_0|_2  -  \frac{C}{\sqrt{N}T^{2\varphi}} .
	$$
	These bounds  lead to the rate of convergence:  $$
	|\widehat\alpha-\alpha_0|_2=O_P(T^{-1/2}+N^{-1/4}T^{-\varphi}),\quad
	|\widehat\phi-\phi_0|_2=O_P(T^{-(1-2\varphi)}+ N^{-1/2}).
	$$ These bounds and the rates are sharp in the case $\sqrt{N}/T^{1-2\varphi} \to \infty$, and  are  identical to the case of the known factor.

	\item[] \textbf{Step VI}.
	It turns out the  lower and upper bounds for $\mathds G_1(\cdot)$ and $ \mathds C( \cdot)
	$  are not sharp when $\sqrt{N}/T^{1-2\varphi} \to \omega < \infty$. We then provide sharper bounds for these terms. In particular, obtaining the sharp lower bound for $\mathds G_1(\cdot)$ is most challenging and involves complicated expansions. We establish in Lemma \ref{l:generic} that it has a quadratic lower bound with an unusual error rate: \begin{eqnarray*}
		\mathds G_1(\phi)
		&\geq& CT^{-2\varphi} \sqrt{N} |\phi -\phi_0|_2^2 -  O(\frac{1}{T^{2\varphi}N^{5/6}}).
	\end{eqnarray*}%
	These lead to a sharp rate for $\widehat\phi, \widehat\gamma$ in Proposition \ref{rate:est} in the case of $\omega<\infty.$

	\item[] \textbf{Step VII}.
	Finally, we derive the limiting distributions for $\widehat\alpha$ and $\widehat\gamma.$ This involves utilizing the convergence rates we obtained through the preceding steps to recenter, rescale and reparametrize the original criterion function, which is parametrized not by $ \phi $ but by $ \gamma $. Then, we establish the stochastic equicontinuity of the empirical process part  of the transformed process (i.e. centered process) in Section \ref{sec:EP} and the careful expansion of the drift (i.e. bias) part of the process as a function of the limit $ \omega = \lim_{N,T}\sqrt{N}T^{-1+2\varphi} $ in Section \ref{sec:Bias}.
	Due to the random rotation matrix $ H_T $ incurred by the factor estimation, we prove an extended continuous mapping theorem in Lemma \ref{CMT-extension}, to derive the weak convergence of the transformed criterion function. The remaining step is the application of the argmax continuous mapping theorem.  The new CMT extends Theorem 1.11.1 of \cite{VW} to allowing stochastic drifting functions $\mathbb G_n$ (while \cite{VW} requires $\mathbb G_n$ be deterministic).

\end{itemize}

\subsection{Discussion on Assumption  \ref{as8} }\label{sec:e:discuss}

We discuss the reasons why Assumption  \ref{as8} presents various conditions on several different conditional distributions and why those conditional distributions are well defined.
  A key technical issue in expanding the least squares loss function, in the unknown factor case, is to
 consider the properties of the conditional density of $g_t'\phi_0$, given $g_t'(\phi-\phi_0)$ and $(x_t, h_t)$. It is needed in bounding terms of the form:
  $$
 \mathbb E \left[ (x_t'\delta_0)^2\Psi(h_t'\phi_0, g_t'\phi_0,  g_t'(\phi-\phi_0) \right]
 $$
 with a suitably defined function $\Psi$.
But we should be cautious  that such a conditional density might be degenerated  because  given $ g_t'(\phi-\phi_0)$, there might be no degree of freedom left for $  g_t'\phi_0$. To address this issue, we observe that by the identification condition, we can write $\gamma=(1, \gamma_2)=H_T^{-1}\phi$, where $1$ is the first element of $\gamma$. Let the corresponding factor   be $ f_t= ( f_{1t}, f_{2t})$.  Then $  g_t'(\phi-\phi_0)= f_t'(\gamma-\gamma_0)=   f_{2t}'(\gamma_2-\gamma_{02})$, so it depends on $  f_t$ only through $ f_{2t}$. As such, we can consider the conditional density of  $  f_t'\gamma_0$ given $(  f_{2t}, x_t, h_t)$. Being given $f_{2t}$ still leaves degrees of freedom for $f_{t}'\gamma_0$, so such  conditional density is well defined.

In   the lower bound for  $		\mathds G_1(\phi)$ in Step VI,  the problem eventually reduces to lower bounding
$$
  \mathbb E\left[(x_t'\delta_0)^2p_{f_t'\gamma_0|f_{2t}, x_t, h_t }(0)     |  g_t'(\phi -\phi_0) |^2    1\{|  g_t|_2<M_0\}\right]
$$
 for a sufficiently large $M_0$.
We can apply the above argument to  achieve a tight quadratic lower bound $C|\phi-\phi_0|_2^2$, so long as the conditional density
$ p_{f_t'\gamma_0 |f_{2t}, x_t, h_t   }(0)$
and the  eigenvalues of $ \mathbb E[(x_t'd_0)^2  |g_t, h_t]$ are bounded away from zero.  In addition, here we also need to
upper bound
$\mathbb P(\frac{h_t'\phi}{\sqrt{N}}<g_t'(\phi-\phi_0)<  \frac{h_t'\phi_0}{\sqrt{N}} |h_t)$
and $\mathbb P(\frac{h_t'\phi}{\sqrt{N}}<g_t'\phi<  \frac{h_t'\phi_0}{\sqrt{N}} |h_t)$. This is ensured by the
condition   $\sup_{|u|<c}p_{g_t'r|x_t, h_t}(u) \leq M$.

When we derive a lower bound  for $		\mathds G_1(\phi)$ in Step V,
we also need such an argument   for the conditional density of $\widehat f_t= H_T'\breve g_t$, where $\breve g_t=g_t+\frac{h_t}{\sqrt{N}}$ is the perturbed factors, estimated by the PCA.
 For instance,     we need a lower bound when    $\Psi= \mathbb P \left( 0<\breve g_t'\phi_0< |  \breve g_t'(\phi-\phi_0)|   \right)$.     To derive this lower bound, write $\widehat f_t= (\widehat f_{1t}, \widehat f_{2t})$. Then $\breve g_t'(\phi-\phi_0)$ depends on $\widehat f_t$ only through $\widehat f_{2t}$. As such, we can consider the conditional density of  $\widehat f_t'\gamma_0$ given $(\widehat f_{2t}, x_t)$, and obtain a lower bound
 $$
 \mathbb E\left[(x_t'd_0)^2  1(0<\breve g_t'\phi_0<  |\breve g_t'(\phi-\phi_0)|  )\right]\geq \inf_{m, x, \widehat f_{2t}}p_{\widehat f_t'\gamma_0| \widehat f_{2t}, x_t } (m) \mathbb E\left[| \breve g_t'(\phi-\phi_0)|\right]\geq C|\phi-\phi_0|_2,
 $$
 where it is assumed that $\inf_{|m|<c}\inf_{ x, \widehat f_{2t}}p_{\widehat f_t'\gamma_0| \widehat f_{2t}, x_t } (m) \geq c_0>0.$
 The need for arguments like this  gives rise to Assumption \ref{as8} (i)-(iv).

\subsection{Consistency}

\subsubsection{A probability bound for $|\widetilde f_t-\widehat f_t|_2$} \label{sec:E1}

The stochastic order of the approximation error of $\widetilde f_t-\widehat f_t$ has been well studied in the literature \citep[see, e.g.][]{bai03}.
However, all the existing results in the literature are on the rates of convergence for  $\widetilde f_t-\widehat f_t$ of a fixed $t$ and for  $\frac{1}{T}\sum_t|\widetilde f_t-\widehat f_t|_2^2$. We strengthen these results below by obtaining the following probability bound.

\begin{pro} \label{prop:f}  Suppose $T=O(N)$. Define
$$
\Delta_f=  \frac{(\log T)^{2/c_1} }{T}
$$
Then for a sufficiently large constant $C>0$, and $\widehat  f_t=H_T'( g_{t}+\frac{h_t}{\sqrt{N}})$,
$$
\mathbb P(|\widetilde f_t-\widehat f_t|_2>C\Delta_f)\leq O(T^{-6}).
$$
\end{pro}

\begin{proof} [Proof of Proposition \ref{prop:f}]
The proof consists of several steps.
Recall that $\widetilde f_{1t}$ denotes the $K\times 1$ vector of PCA estimator of $g_{1t}$.
Write $e_t=(e_{1t},...,e_{Nt})'.$

\medskip
\noindent
\textit{Step 1: Decomposition of $\widetilde f_t-H_T'g_t$}
\medskip

Define     $K\times K$ matrix
$
\widetilde H_T'=V_T^{-1}   \frac{1}{T}\sum_{t=1}^T\widetilde f_{1t} g_{1t}' S_\Lambda,$ and $S_\Lambda= \frac{1}{N}\Lambda'\Lambda.
$
Also let  $V_T$ be the $K\times K$ diagonal matrix whose entries are the first $K$ eigenvalues of
$\mathcal Y \mathcal Y'/NT$ (equivalently, the first $K$ eigenvalues of $\frac{1}{NT}\sum_t\mathcal Y_t \mathcal Y_t '$).
We have
\begin{equation}\label{e9.3}
\widetilde f_{1t}-\widetilde H_T'g_{1t}=
\widetilde H_T'S_{\Lambda}^{-1}\frac{1}{N} \Lambda'e_t+
 \sum_{d=1}^6 A_{t,d},
\end{equation}
where
 \begin{eqnarray*}
  A_{t,1}&=&V_T^{-1}\widetilde H_T'\frac{1}{T}\sum_{s=1}^T g_{1s}\frac{1}{N}  \mathbb Ee_s'e_t,\cr
 A_{t,2}&=&V_T^{-1}  \frac{1}{T}\sum_{s=1}^T (\widetilde  f_{1s}-\widetilde H_T'g_{1s})\frac{1}{N}  \mathbb E  e_s'e_t,   \cr
 A_{t,3}&=&V_T^{-1}\frac{1}{T}\sum_{s=1}^T(\widetilde  f_{1s}-\widetilde H_T'g_{1s})\frac{1}{N}  (e_s'e_t-  \mathbb Ee_s'e_t), \cr
  A_{t,4}&=&V_T^{-1}\widetilde H_T'\frac{1}{TN}\sum_{s=1}^T g _{1s}  (e_s'e_t- \mathbb  Ee_s'e_t),\cr
 A_{t,5}&=&V_T^{-1}\frac{1}{T}\sum_{s=1}^T(\widetilde  f_{1s} -\widetilde H_T'g_{1s})g_{1t}' \frac{1}{N} \sum_{i=1}^N   \lambda_ie_{is},  \cr
  A_{t,6}&=&V_T^{-1}\widetilde H_T'\frac{1}{T}\sum_{s=1}^T   g_{1s} g_{1t}' \frac{1}{N} \sum_{i=1}^N \lambda_i e_{is}.
 \end{eqnarray*}
  Hence for $H_T'=\diag\{\widetilde H_T', 1\}$,  $g_t=(g_{1t}', 1)'$,   $\widetilde f_t=(\widetilde f_{1t}', 1)'$,
$
h_t=(S_{\Lambda}^{-1} \frac{ \Lambda' e_t}{\sqrt{N}}, 0)',$ and $ \widehat  f_t=H_T'( g_{t}+\frac{h_t}{\sqrt{N}}),
$
we have
\begin{equation}\label{e:keyexpress}
\widetilde  f_{t}- \widehat f_t=(  \sum_{d=1}^6 A_{t,d}, 0)'.
\end{equation}

\medskip
\noindent
\textit{Step 2: Bounding $\frac{1}{T}\sum_t|\widetilde  f_{1t}-\widetilde H_T'g_{1t}|_2^2$}
\medskip

Note that
  \begin{eqnarray*}
\frac{1}{T}\sum_{t=1}^T|\widetilde  f_{1t}-\widetilde H_T'g_{1t} |_2^2   &\leq &
4\frac{1}{T}\sum_{t=1}^T| \widetilde H_T'\frac{h_t}{\sqrt{N}}|_2^2
+ 4\frac{1}{T}\sum_{t=1}^T| V_T^{-1}\widetilde H_T'\frac{1}{T}\sum_{s=1}^T g_{1s}\frac{1}{N}  \mathbb Ee_s'e_t |_2^2\cr
&&
+ \frac{1}{T}\sum_{s=1}^T|\widetilde  f_{1s}-\widetilde H_T'g_{1s}|_2^2(a_1+a_2+a_3)\cr
&&
+8\frac{1}{T}\sum_{t=1}^T|   V_T^{-1}\widetilde H_T'\frac{1}{TN}\sum_{s=1}^T g_{1s}  (e_s'e_t-  \mathbb  Ee_s'e_t)  |_2^2\cr
&&
+8\frac{1}{T}\sum_{t=1}^T|  V_T^{-1}\widetilde H_T'\frac{1}{T}\sum_{s=1}^T \frac{1}{N} \sum_{i=1}^N     g_{1s} e_{is}   \lambda_i'g_{1t}   |_2^2, \cr
 \end{eqnarray*}
where $$a_1= |  V_T^{-1}|_2^2\frac{1}{T^2}\sum_{t=1}^T
 \sum_{s=1}^T|\frac{1}{N}  (e_s'e_t-  \mathbb  Ee_s'e_t)|_2^2,\quad a_2= |   V_T^{-1}|_2^2     \frac{1}{T^2}\sum_{t=1}^T \sum_{s=1}^T  |g_{1t}' \frac{1}{N}   \Lambda' e_s |_2^2$$
 and assuming $\frac{1}{NT}\sum_{t,s\leq T}\sum_{i\leq N}|\mathbb Ee_{it}e_{is}|<C$,
 $$
 a_3=  |  V_T^{-1}|_2^2\max_{s,t}|\frac{1}{N}  \mathbb E  e_s'e_t| \frac{1}{T^2}\sum_t\sum_{s=1}^T |\frac{1}{N}  \mathbb E  e_s'e_t|\leq  C|  V_T^{-1}|_2^2\frac{1}{T}.
 $$
 Hence  for $c_{NT}= (1-a_1-a_2-a_3)$,
   \begin{eqnarray}\label{eb.5}
 \frac{1}{T}\sum_{t=1}^T|\widetilde  f_{1t}-\widetilde H_T'g_{1t} |_2^2 c_{NT} &\leq& 8\frac{1}{T}\sum_{t=1}^T|    V_T^{-1}\widetilde H_T'\frac{1}{TN}\sum_{s=1}^T g_{1s}  (e_s'e_t-  \mathbb  Ee_s'e_t)  |_2^2\cr
&&+4\frac{1}{T}\sum_{t=1}^T| \widetilde H_T'\frac{h_t}{\sqrt{N}}|_2^2
+ 4\frac{1}{T}\sum_{t=1}^T| V_T^{-1}\widetilde H_T'\frac{1}{T}\sum_{s=1}^T g_{1s}\frac{1}{N}  \mathbb Ee_s'e_t |_2^2\cr
&&+8\frac{1}{T}\sum_{t=1}^T|   V_T^{-1}\widetilde H_T'\frac{1}{T}\sum_{s=1}^T \frac{1}{N} \sum_{i=1}^N     g_{1s} e_{is}   \lambda_i'g_{1t}   |_2^2.
 \end{eqnarray}
Next we provide probability bounds for each term on the right hand side below.

 \medskip
\noindent
\textit{Step 3: Proving that $T^6  \mathbb P(|V_T^{-1}|_2>C_v)+T^6 \mathbb{P}(|\widetilde  H_T|_2>C_H)=o(1) $ for some  $C_v, C_H>0$}
\medskip

 Let $V$ be the diagonal matrix consisting of the first $K$ eigenvalues of  $\Sigma_{\Lambda}^{1/2} \mathbb{E} [g_{1t}g_{1t}'] \Sigma_\Lambda^{1/2}$.
 On the event $|V_T-V|_2<\lambda_{\min}(V)/2$,
 $$
 |V_T^{-1}|_2=\lambda_{\min}^{-1}(V_T)\leq  2\lambda_{\min}^{-1}(V)\leq 2\lambda_{\min}^{-1}(\frac{1}{N}\Lambda'\Lambda)\lambda_{\min}^{-1}(  \mathbb  Eg_{1t}g_{1t}')<C_v.
 $$

 We now show $T^6  \mathbb P(  |V_T-V|_2>\lambda_{\min}(V)/2)=o(1)$. By Weyl's theorem,
\begin{eqnarray*}
 |V_T-V|_2&\leq &|\frac{1}{NT}\sum_t \mathcal Y_t \mathcal Y_t '-\frac{1}{N}\Lambda   \mathbb  Eg_{1t}g_{1t}'\Lambda'|_2
 \leq |\frac{1}{N}\Lambda (  \mathbb  Eg_{1t}g_{1t}'-\frac{1}{T}\sum_tg_{1t}g_{1t}')\Lambda'|_2\cr
 &&+2|\frac{1}{N}\Lambda  \frac{1}{T}\sum_tg_{1t} e_t'|_2+|\frac{1}{N} ( \frac{1}{T}\sum_te_t e_t'-  \mathbb  Ee_te_t')| _2+\frac{1}{N}|  \mathbb  Ee_te_t'|_2\cr
 &\leq& C |  \mathbb  Eg_{1t}g_{1t}'-\frac{1}{T}\sum_tg_{1t}g_{1t}'|_2+C\frac{1}{\sqrt{N}}|\frac{1}{T}\sum_tg_{1t} e_t'|_2+|\frac{1}{N} ( \frac{1}{T}\sum_te_t e_t'-  \mathbb  Ee_te_t')|_2 +\frac{C}{N} \cr
 &=&b_1+b_2+b_3+\frac{C}{N}.
 \end{eqnarray*}
By the Bernstein inequality,  for some $M,c ,\zeta, r>0$,
 \begin{eqnarray*}
  T^6 \mathbb P( b_1>\lambda_{\min}(V)/9) &=&T^6  \mathbb P( C | Eg_{1t}g_{1t}'-\frac{1}{T}\sum_tg_{1t}g_{1t}'|_2>\lambda_{\min}(V)/9)\cr
  &\leq& T^6  \exp(-MT^c)=o(1), \cr
   T^6 \mathbb P( b_2>\lambda_{\min}(V)/9)  &=&T^6  \mathbb P(C |\frac{1}{T}\sum_tg_{1t} e_t'|_2>\sqrt{N}\lambda_{\min}(V)/9)\cr
   &\leq& CT^{-3} \max_{i\leq N}    \mathbb  E |\frac{1}{\sqrt{T}}\sum_tg_{1t} e_{it}|_2^{r}
   \cr
   &=&CT^{-3}\max_i\int_0^{\infty} \mathbb P( |\frac{1}{\sqrt{T}}\sum_tg_{1t} e_{it}|_2>x^{-r})dx\cr
   &\leq& CT^{-3}\int_0^{\infty}  \exp(-Cx^{-\zeta})dx =O(T^{-3}), \cr
      T^6 \mathbb P( b_3>\lambda_{\min}(V)/9)       &=&T^6  \mathbb P( |   \frac{1}{T}\sum_te_t e_t'-  \mathbb  Ee_te_t'|_2>N\lambda_{\min}(V)/9)\cr
     \cr
     &\leq& CT^{-3}\max_{ij}  \mathbb  E|   \frac{1}{\sqrt{T}}\sum_t(e_{it} e_{jt}-  \mathbb  Ee_{it}e_{jt})|^{r}\cr
     &\leq&CT^{-3}\max_{ij}\int_0^{\infty}\mathbb P(|   \frac{1}{\sqrt{T}}\sum_t(e_{it} e_{jt}-  \mathbb  Ee_{it}e_{jt})|>x^{-r})dx\cr
     &\leq& CT^{-3}\int_0^{\infty} \exp(-Cx^{-\zeta})dx=O(T^{-3}).
   \end{eqnarray*}
  Hence
\begin{eqnarray*}
 T^6  \mathbb P(|V_T^{-1}|>C_v)&\leq &T^6  \mathbb P(|V_T^{-1}|_2>C_v, |V_T-V|_2<\lambda_{\min}(V)/2)\cr
 &&+ T^6  \mathbb P(  |V_T-V|_2>\lambda_{\min}(V)/2)\cr
 &=&T^6  \mathbb P(  |V_T-V|_2>\lambda_{\min}(V)/2)\cr
  &\leq& T^6  \mathbb P(b_1+b_2+b_3>\lambda_{\min}(V)/3)\cr
      &\leq& T^6 \sum_{i=1}^3 \mathbb P(b_i >\lambda_{\min}(V)/9)=o(1).
 \end{eqnarray*}
Now
On the event $|V_T^{-1}|_2\leq C_v$,    for $C_H>C_{\lambda}^2C_v(2M_f)^{1/2}K$ (recall $|S_\Lambda|_2\leq C_\lambda$ and $ E|g_{1t}|_2^2<M_f$),
\begin{eqnarray*} &&T^6 \mathbb{P}(|\widetilde H_T|_2>C_H)\cr
&\leq&T^6 \mathbb{P}(|V_T^{-1}|_2>C_v) +T^6 \mathbb{P}(\frac{1}{T}\sum_t|g_{1t}|_2^2>2M_f)\cr
&\leq &o(1)+T^6 \mathbb{P}(\frac{1}{T}\sum_t(|g_{1t}|_2^2-  \mathbb  E|g_{1t}|_2^2)>M_f)=o(1).
    \end{eqnarray*}

  \medskip
\noindent
\textit{Step 4: Proving $T^6  \mathbb P(a_{1,2}>CN^{-1}\log ^cT)=o(1)$ for some $c,C>0$}
\medskip

In step 2, $a_1= |   V_T^{-1}|^2\frac{1}{T^2}\sum_{t=1}^T
 \sum_{s=1}^T|\frac{1}{N}  (e_s'e_t-  \mathbb  Ee_s'e_t)|^2$.  By steps 3 and 4, with probability at least $1-o(T^{-6})$, $  |V_T^{-1}|_2<C$. Thus  for $c=2c_{1}^{-1}$,
\begin{eqnarray}\label{e:6.7}
T^6\mathbb P(a_1 >  CN^{-1}\log ^cT)&\leq&T^6 \mathbb P(C\frac{1}{T^2}\sum_{t=1}^T
 \sum_{s=1}^T|\frac{1}{\sqrt{N}}  (e_s'e_t-  \mathbb  Ee_s'e_t)|^2>C \log ^cT)+o(1)\cr
 &\leq&T^6 \mathbb P(C \max_{st} |\frac{1}{\sqrt{N}}  (e_s'e_t-  \mathbb  Ee_s'e_t)|^2>C \log ^cT)+o(1)\cr
&\leq&T^8\max_{st}\mathbb P(|\frac{1}{\sqrt{N}}  (e_s'e_t-  \mathbb  Ee_s'e_t)|>C\log ^{c/2}T)\cr
&\leq&C \exp(11\log T-C_{1}C^{c_1} \log  T)=o(1),
 \end{eqnarray}
provided that $C_{1}C^{c_1} >11$.
Similarly,
\begin{equation}\label{e:6.6}
T^6  \mathbb P(a_{2}>CN^{-1}\log ^cT)
\leq o(1)+ T^6\max_s  \mathbb P(  |  \frac{1}{N}   \Lambda' e_s |_2^2>CN^{-1}\log ^cT)=o(1).
\end{equation}

   \medskip
\noindent
\textit{Step 5: Prove $T^6\mathbb P(\frac{1}{T}\sum_{t=1}^T|\widetilde  f_{1t}-\widetilde H_T'g_{1t} |_2^2> C(\log T)^c( \frac{1}{N} +\frac{1}{T^2}))=o(1)$ for $c=2/c_1$}
\medskip

By  (\ref{eb.5}),  and steps 3 and 4,  there is $C>0$, with probability  at least $1-o(T^{-6})$,
$$
\frac{1}{T}\sum_{t=1}^T|\widetilde  f_{1t}-\widetilde H_T'g_{1t} |_2^2\leq C(d_1+...+d_4),
$$
where
\begin{eqnarray*}
d_1&=&\frac{1}{T}\sum_{t=1}^T| \frac{1}{TN}\sum_{s=1}^T g_{1s}  (e_s'e_t-  \mathbb  Ee_s'e_t)  |_2^2,
\cr
d_2&=&   \frac{1}{T}\sum_{t=1}^T| \frac{h_t}{\sqrt{N}}|_2^2,
\cr
d_3&=&   \frac{1}{T}\sum_{t=1}^T| \frac{1}{TN}\sum_{s=1}^T     g_{1s} e_{s}'   \Lambda g_{1t}   |_2^2, \cr
d_4&=&
  \frac{1}{T}\sum_{t=1}^T|  \frac{1}{T}\sum_{s=1}^Tg_{1s}\sigma_{st} |_2^2,\quad \sigma_{st}=\frac{1}{N}  \mathbb Ee_s'e_t.
 \end{eqnarray*}
  The tail probability of $d_2$ has already been bounded in (\ref{e:6.6}):
   $$
   T^6\mathbb P(d_2>N^{-1}C\log^{2/c_1}T)=o(1).
   $$
   For $x=(\log T)^{2/c_1} m$,  $y=(\log T)^{2/c_1} m$, $z=(\log T)^{2/c_1}m$
    and sufficiently large $m$,
       \begin{align} \label{e:6.9}
    T^6\max_t \mathbb P(| \frac{1}{\sqrt{TN}}\sum_{s=1}^T g_{1s}  (e_s'e_t-  \mathbb  Ee_s'e_t)  |_2>  x^{1/2})
 &\leq C  \exp(10\log T-C_1 x^{c_1/2}) =o(1),\cr
T^6\mathbb P(  | \frac{1}{TN}\sum_{s=1}^T     g_{1s} u_{s}'   \Lambda    |_2^2>(NT)^{-1}y)
&\leq C  \exp(10\log T-C_1 y^{c_1/2})=o(1),\cr
T^6 \mathbb P(\max_s|g_{1s}|_2^2>z)
&\leq \exp(6\log T- C_1 z^{c_1/2})=o(1).
 \end{align}
  Note that $  \max_t  \sum_{s=1}^T |\sigma_{st} | \leq C_{\sigma}$ for some $C_{\sigma}>0$. Therefore,
    \begin{eqnarray*}
T^6\mathbb P(d_1>(NT)^{-1}x) &\leq& T^6\mathbb P(\frac{1}{T}\sum_{t=1}^T| \frac{1}{TN}\sum_{s=1}^T g_{1s}  (e_s'e_t-  \mathbb  Ee_s'e_t)  |_2^2>(NT)^{-1}x)\cr
&\leq& T^6\max_t\mathbb P(| \frac{1}{\sqrt{TN}}\sum_{s=1}^T g_{1s}  (e_s'e_t-  \mathbb  Ee_s'e_t)  |_2> x^{1/2})=o(1),\cr
T^6\mathbb P(d_3>(NT)^{-1}y) &\leq& T^6\mathbb P(  | \frac{1}{TN}\sum_{s=1}^T     g_{1s} e_{s}'   \Lambda    |_2^2>(NT)^{-1}y) +o(1)=o(1),\cr
T^6\mathbb P(d_4>T^{-2}C_{\sigma}^2z) &\leq&
T^6\max_t\mathbb P(  |  \frac{1}{T}\sum_{s=1}^Tg_{1s}\sigma_{st} |_2^2>T^{-2}C_{\sigma}^2z)\cr
&\leq&T^6\max_t\mathbb P( \max_{s}|g_{1s}|^2 (  \frac{1}{T}\sum_{s=1}^T|\sigma_{st}| )^2>T^{-2}C_{\sigma}^2z)
\cr
&\leq&T^6\max_t\mathbb P( \max_{s}|g_{1s}|^2  > z)=o(1).
 \end{eqnarray*}
 Together, we have, for $c=\log^{2/c_1} $, with probability at least $1-o(T^{-6})$,
 $$
 \frac{1}{T}\sum_{t=1}^T|\widetilde  f_{1t}-\widetilde H_T'g_{1t} |_2^2\leq Cm^2_{NT},\; \text{where} \; m^2_{NT}:= (\log T)^c( \frac{1}{N} +\frac{1}{T^2}).
 $$

    \medskip
\noindent
\textit{Step 6: finishing the proof}
\medskip

    We now work with (\ref{e:keyexpress})
$\widetilde  f_{t}- \widehat f_t=(  \sum_{d=1}^6 A_{t,d}, 0)'.$
Write $
    Q=\frac{1}{T}\sum_{s=1}^T|\widetilde  f_{1s}-\widetilde H_T'g_{1s}|_2^2.$    Step 5 proved  $Q<Cm^2_{NT}$ with probability at least $1-o(T^{-9})$.  In addition,
        $$
    \mathbb P(|f_t|_2> M(\log T)^{1/c_1})\leq C\exp(-C_fM^{c_1}(\log T))=CT^{-C_fM^{c_1}}<o(T^{-9})
    $$
    for large enough $M$.

    Now take
    $$
    x= C(\log T)^{1/c_1},\quad y=C(\log T)^{1/c_1} ,\quad w=C(\log T)^{1/c_1},
    $$
    $$
    z= (\log T)^{1/c_1}w,\quad \widetilde x=C(\log T)^{1/c_1} ,\quad  \widetilde y= (\log T)^{1/c_1}\widetilde x.
    $$
Then, we have,  for sufficiently large $C>0$,
     \begin{align*}
      T^6\mathbb P(|A_{t,1}|_2>C  T^{-1}(\log T)^{1/c_1})&\leq
       T^6\mathbb P(\max_s|g_{1s}|_2 \sum_{s=1}^T  |\frac{1}{N}  \mathbb Ee_s'e_t|>C  (\log T)^{1/c_1})   +o(1)\cr
       &\leq    T^6\mathbb P(\max_s|g_{1s}|_2  >C  (\log T)^{1/c_1})   +o(1) =o(1),
      \cr
         T^6\mathbb P(|A_{t,2}|_2>m_{NT}T^{-1/2} C)&\leq
               T^6\mathbb P(| \frac{1}{T}\sum_{s=1}^T (\widetilde  f_{1s}-\widetilde H_T'g_{1s})\frac{1}{N}  \mathbb E  e_s'e_t  |_2>m_{NT}T^{-1/2} C)\cr
               &\leq
                 T^6\mathbb P(Q\frac{1}{T}\sum_s| \frac{1}{N}  \mathbb E  e_s'e_t  |^2>m_{NT}^2T^{-1} C^2)\cr
  &\leq
                 T^6\mathbb P(\max_{st}| \frac{1}{N}  \mathbb E  e_s'e_t  | \sum_s| \frac{1}{N}  \mathbb E  e_s'e_t  |>   C^2) +o(1)=o(1),
          \cr
 T^6\mathbb P(|A_{t,3}|_2>m_{NT}N^{-1/2}x)
 &=  T^6\mathbb P(C|\frac{1}{T}\sum_{s=1}^T(\widetilde  f_{1s}-\widetilde H_T'g_{1s})\frac{1}{N}  (e_s'e_t-  \mathbb  Ee_s'e_t)|>m_{NT}N^{-1/2}x)+o(1)\cr
 &\leq^{(a)}  T^6\mathbb P(CQ\frac{1}{T}\sum_{s=1}^T| \frac{1}{N}  (e_s'e_t-  \mathbb  Ee_s'e_t)|^2>m_{NT}^2N^{-1}x^2)+o(1) \cr
 &\leq   T^8\max_{st}\mathbb P( | \frac{1}{\sqrt{N}}  (e_s'e_t-  \mathbb  Ee_s'e_t)|>  x)+o(1) =^{(b)}o(1), \cr
T^6\mathbb P(|A_{t,4}|_2>(NT)^{-1/2}y)&= T^6\mathbb P(C|\frac{1}{\sqrt{TN}}\sum_{s=1}^T g_{1s}  (e_s'e_t-  \mathbb  Ee_s'e_t)|_2>y)=^{(c)}o(1)
 \cr
 T^6\mathbb P(|A_{t,5}|_2>m_{NT}N^{-1/2}z)&=T^6\mathbb P(C|\frac{1}{T}\sum_{s=1}^T(\widetilde  f_{1s} -\widetilde H_T'g_{1s})g_{1t}' \frac{1}{N}\Lambda'e_s|_2> m_{NT}N^{-1/2}z)+o(1), \cr
 &\leq T^6\mathbb P(C|g_{1t}|_2^2 \frac{1}{T}\sum_{s=1}^T | \frac{1}{N}\Lambda'e_s|_2^2>  N^{-1}z^2)+o(1)\cr
  &\leq T^7\max_s\mathbb P(C    | \frac{1}{\sqrt{N}}\Lambda'e_s|_2>   w)+o(1)=^{(d)}o(1), \cr
T^6\mathbb P(|A_{t,6}|_2> (NT)^{-1/2}\widetilde y)&=T^6\mathbb P(C |\frac{1}{NT}\sum_{s=1}^T     g_{1s} g_{1t}' \Lambda'e_s  |_2>(NT)^{-1/2}\widetilde y)+o(1) \cr
&\leq  T^6\mathbb P(C  |\frac{1}{NT}\sum_{s=1}^T    g_{1s}  e_s'\Lambda  |_2>(NT)^{-1/2}\widetilde x)+o(1), 
 \end{align*}
 where in (a) we used Cauchy-Schwarz;  (b) comes from (\ref{e:6.7}); (c) and (e) follow from (\ref{e:6.9}); (d) is from (\ref{e:6.6}). Combined together,  $ |\widetilde  f_{t}- \widehat f_t|<C\Delta_f$ with probability at least $1-o(T^{-9})$,      \begin{eqnarray*}
  \Delta_f &=&\frac{\log^{1/c_1} T}{T} +\frac{\log^{1/c_1} T+\log^{1/c_1} T \log^{1/c_1} T}{\sqrt{NT}}+ m_{NT}(\frac{1}{\sqrt{T}} +\frac{\log^{1/c_1} T}{\sqrt{N}} )\cr
  &\leq& 3\frac{\log^{2/c_1} T}{T}.
\end{eqnarray*}
where that last inequality is due to $T=O(N)$.

 \end{proof}

\subsubsection{Defining  notation}\label{sec:def:not}

In the sequel, we show that  $ \left( \widehat{\alpha },\widehat{\gamma }\right) $
defined in Section \ref{model:est:factors}
is asymptotically equivalent to the minimizer of the criterion function that replaces $ \widetilde{f}_t $ in $ \mathbb{\widetilde{S}}_{T}\left( \alpha ,\gamma \right) $ with $ \widehat{f}_t $ in the sense that they have an identical asymptotic distribution. Below we introduce various terms in the form of $ \widetilde{\cdot} $ and $ \widehat{\cdot} $. They indicate that the corresponding terms contain $ \widetilde{f}_t $ and $ \widehat{f}_t $ in their definitions, respectively.


Let $1_{t}=1\left\{ f_{t}^{\prime } \gamma_0>0\right\} $ and recall that
\begin{eqnarray*}
&&\mathbb{\widetilde{S}}_{T}\left( \alpha,\gamma \right) \\
&=&\mathbb{\widetilde{S}}_{T}\left( \alpha _{0},\gamma _{0}\right) +\frac{1}{T}%
\sum_{t=1}^{T}\left( x_{t}^{\prime }\left( \beta -\beta _{0}\right)
+x_{t}^{\prime }\left( \delta1\{\widetilde f_t'\gamma>0\} -\delta _{0}%
1\{\widetilde f_t'\gamma_0>0\} \right) \right) ^{2} \\
&&-\frac{2}{T}\sum_{t=1}^{T}\left( \varepsilon _{t}-x_{t}^{\prime }\delta
_{0}\left( 1\{\widetilde f_t'\gamma_0>0\} -1_{t}\right) \right) \left(
x_{t}^{\prime }\left( \beta -\beta _{0}\right) +x_{t}^{\prime }\left( \delta
1\{\widetilde f_t'\gamma>0\} -\delta _{0}1\{\widetilde f_t'\gamma_0>0\} \right) \right).
\end{eqnarray*}
And introduce the following decomposition:
\begin{eqnarray*}
	\mathbb{\widetilde{S}}_{T}\left( \widehat{\alpha },\widehat{\gamma }%
	\right) -\mathbb{\widetilde{S}}_{T}\left( \alpha _{0},\gamma _{0}\right)
	&=&
	\underset{\mathbb{\widetilde {R}}_{T}\left( \widehat{\alpha },\widehat{\gamma }\right)}
	{\underbrace{\widetilde R_{1}(\widehat\alpha,\widehat\gamma)+\widetilde R_{2}(\widehat\alpha,\widehat\gamma)+\widetilde R_{3}(\widehat\alpha,\widehat\gamma)}} \\
	&& - \
	\underset{\mathbb{\widetilde {G}}_{T}\left( \widehat{\alpha },\widehat{\gamma }\right) -%
		\mathbb{\widetilde {G}}_{T}\left( \alpha _{0}, \gamma _{0}\right)}
	{\underbrace{\left( \widetilde{\mathbb C}_1(\widehat\alpha,\widehat\gamma)+ \widetilde{\mathbb C}_2(\widehat\alpha,\widehat\gamma)
			- \widetilde{\mathbb C}_3(\widehat\alpha,\widehat\gamma)  -\widetilde{\mathbb C}_4(\widehat\alpha,\widehat\gamma)\right) }},
\end{eqnarray*}
where the additional terms are defined in the sequel. Also, note that we suppress the dependence on $T$ to save notational burden as we introduce the more detailed decomposition. Let
$$
 {\widetilde Z_t( \gamma)=(x_t', x'_t1\{\widetilde f_t' \gamma>0\})'},\quad \widehat Z_t(\gamma)=(x_t', x_t'1\{\widehat f_t\gamma>0\})',
$$
\begin{eqnarray*}
\mathbb{\widetilde {R}}_{T}\left( \alpha ,\gamma \right) &=&\frac{1}{T}%
\sum_{t=1}^{T}\left( \widetilde{Z}_{t}\left( \gamma \right) ^{\prime }\alpha -\widetilde{%
Z}_{t}\left( \gamma _{0}\right) ^{\prime }\alpha _{0}\right) ^{2} \\
&=&\underset{\widetilde{R}_{1}\left( \alpha ,\gamma \right) }{\underbrace{\frac{1%
}{T}\sum_{t=1}^{T}\left( \widetilde{Z}_{t}\left( \gamma \right) ^{\prime }\left(
\alpha -\alpha _{0}\right) \right) ^{2}}}+\underset{\widetilde{R}_{2}\left(
\alpha,\gamma \right) }{\underbrace{\frac{1}{T}\sum_{t=1}^{T}\left( x_{t}^{\prime
}\delta _{0}\right) ^{2}\left\vert 1\left\{ \widetilde{f}_{t}^{\prime }\gamma
>0\right\} -1\left\{ \widetilde f_t'\gamma_0>0\right\} \right\vert }%
} \\
&&+\underset{\widetilde{R}_{3}\left( \alpha ,\gamma \right) }{\underbrace{\frac{2%
}{T}\sum_{t=1}^{T}x_{t}^{\prime }\delta _{0}\left( 1\left\{\widetilde f_t'\gamma>0\right\} -1\left\{ \widetilde{f}_{t}^{\prime }\gamma_0>0\right\} \right) \widetilde{Z}_{t}\left( \gamma \right) ^{\prime }\left(
\alpha -\alpha _{0}\right) }}, \\
\mathbb{\widetilde{G}}_{T}\left( \alpha ,\gamma \right) &=&\frac{2}{T}%
\sum_{t=1}^{T}\left( \varepsilon _{t}-x_{t}^{\prime }\delta _{0}\left( 1\{\widetilde f_t'\gamma_0>0\} -1\{f_t'\gamma_0>0\}\right) \right) \left( \widetilde{Z}%
_{t}\left( \gamma \right) ^{\prime }\alpha -\widetilde{Z}_{t}\left( \gamma_0\right) ^{\prime }\alpha _{0}\right) .
\end{eqnarray*}%
Then we have
\begin{eqnarray*}
\mathbb{\widetilde{G}}_{T}\left( \alpha ,\gamma \right) -\mathbb{\widetilde{G}}%
_{T}\left( \alpha _{0},\gamma _{0}\right) &=&\frac{2}{T}\sum_{t=1}^{T}\left(
\varepsilon _{t}-x_{t}^{\prime }\delta _{0}\left( 1\{\widetilde f_t'\gamma_0>0\} -1_{t}\right) \right) \left( \widetilde{Z}_{t}\left( \gamma \right)
^{\prime }\alpha -\widetilde{Z}_{t}\left( \gamma _{0}\right) ^{\prime }\alpha
_{0}\right) \\
&=&\mathbb{\widetilde{C}}_{1}(\alpha,\gamma)+\mathbb{\widetilde{C}}_{2}(\alpha,\gamma)- \mathbb{\widetilde{C}}_{3}(\alpha,\gamma) - \mathbb{\widetilde{C}}_{4}(\alpha,\gamma),
\end{eqnarray*}%
where
\begin{eqnarray*}
\mathbb{\widetilde{C}}_{1}(\alpha,\gamma)&=&
\frac{2}{T}\sum_{t=1}^{T}\varepsilon _{t}x_{t}^{\prime }\delta \left( 1\{\widetilde f_t'\gamma>0\} -1\{\widetilde f_t'\gamma_0>0\} \right),  \cr
\mathbb{\widetilde{C}}_{2}(\alpha,\gamma)&=&    \frac{2}{T}%
\sum_{t=1}^{T}\varepsilon _{t}\widetilde{Z}_{t}\left( \gamma _{0}\right) ^{\prime
}\left( \alpha -\alpha _{0}\right),  \cr
\mathbb{\widetilde{C}}_{3}(\alpha,\gamma)&=&
\frac{2}{T}\sum_{t=1}^{T}x_{t}^{\prime }\delta _{0}x_{t}^{\prime }\delta
\left( 1\{\widetilde f_t'\gamma_0>0\} -1_{t}\right) \left( 1\{\widetilde f_t'\gamma>0\} -1\{\widetilde f_t'\gamma_0>0\} \right),   \\
\mathbb{\widetilde{C}}_{4}(\alpha,\gamma)&=&  \frac{2 }{T}\sum_{t=1}^{T}x_{t}^{\prime }\delta _{0}\left( 1\{\widetilde f_t'\gamma_0>0\} -1_{t}\right) \widetilde{Z}_{t}\left( \gamma _{0}\right) ^{\prime
}\left( \alpha -\alpha _{0}\right)  .
\end{eqnarray*}%
 In addition, the following quantities will be used in the proofs to follow.
\begin{eqnarray*}
\widehat{R}_{1}\left( \alpha ,\gamma \right) &=&\frac{1}{T}\sum_{t=1}^{T}\left( \widehat{Z}_{t}\left( \gamma \right) ^{\prime }\left(\alpha -\alpha _{0}\right) \right) ^{2}, \cr
\widehat{R}_{2}\left( \alpha, \gamma \right) &=&\frac{1}{T}\sum_{t=1}^{T}\left( x_{t}^{\prime
}\delta _{0}\right) ^{2}      |   1\{ \widehat{f}_{t}^{\prime }\gamma
>0\} -1\{ \widehat f_t'\gamma_0>0\}   |,  \cr
\widehat{R}_{3}\left( \alpha ,\gamma \right) &=&\frac{2}{T}\sum_{t=1}^{T}x_{t}^{\prime }\delta _{0}\left( 1\left\{\widehat f_t'\gamma>0\right\} -1\left\{ \widehat f_t'\gamma_0>0\right\} \right) \widehat{Z}_{t}\left( \gamma \right) ^{\prime }\left(\alpha -\alpha _{0}\right), \cr
 \mathbb{\widehat{C}}_{1}\left( \alpha ,\gamma \right) &=&
\frac{2}{T}\sum_{t=1}^{T}\varepsilon _{t}x_{t}^{\prime }\delta \left( 1\{\widehat f_t'\gamma>0\} -1\{\widehat f_t'\gamma_0>0\} \right), \cr
  \mathbb{\widehat{C}}_{2}\left( \alpha,\gamma \right) &=&\frac{2}{T}\sum_{t=1}^{T}\varepsilon _{t}\widehat{Z}_{t}\left( \gamma _{0}\right) ^{\prime}\left( \alpha -\alpha _{0}\right), \cr
\mathbb{\widehat{C}}_{3}\left( \alpha ,\gamma \right) &=&
\frac{2}{T}\sum_{t=1}^{T}x_{t}^{\prime }\delta _{0}x_{t}^{\prime }\delta\left( 1\{\widehat f_t'\gamma_0>0\} -1_{t}\right) \left( 1\{\widehat f_t'\gamma>0\} -1\{\widehat f_t'\gamma_0>0\} \right), \cr
 \mathbb{\widehat{C}}_{4}\left( \alpha, \gamma \right) &=&\frac{2}{T}\sum_{t=1}^{T}x_{t}^{\prime }\delta _{0}\left( 1\{\widehat f_t'\gamma_0>0\} -1_{t}\right) \widehat{Z}_{t}\left( \gamma _{0}\right) ^{\prime}\left( \alpha -\alpha _{0}\right).
\end{eqnarray*}

\subsubsection{Effect of $\widetilde f_t-\widehat f_t$} \label{sec:E2}

\begin{lem}\label{l:relf}
	Uniformly over $\alpha $ and $\gamma $, for $\Delta_f $ defined in Proposition \ref{prop:f}, \newline
	(i) For $j=1,...,4,$ $\left\vert \widetilde{\mathbb{C}}_{j}(\delta ,\gamma )-%
	\widehat{\mathbb{C}}_{j}(\delta ,\gamma )\right\vert \leq (T^{-\varphi }+|\alpha -\alpha
	_{0}|_{2})O_{P}(\Delta_f+T^{-6}).$ \newline
	(ii) $|\widetilde{\mathbb{C}}_{2}(\alpha )|\leq O_{P}(T^{-1/2}+\Delta_f)|\alpha
	-\alpha _{0}|_{2}$.\newline
	(iii) $|\widetilde{\mathbb{C}}_{4}(\alpha )|\leq O_{P}(\Delta_f+N^{-1/2})T^{-\varphi
	}|\alpha -\alpha _{0}|_{2}$.\newline
	(iv)  For $j=1,2,3,$ $|\widetilde{R}_{jT}\left( \alpha ,\gamma \right) -%
	\widehat{R}_{jT}\left( \alpha ,\gamma \right) |\leq [|\alpha -\alpha
	_{0}|_{2}^2+T^{-2\varphi }]O_{P}(\Delta _{f}+T^{-6}).$
\end{lem}

A consequence of this lemma is that the first-order asymptotic distribution
of $\widehat{\alpha}$ and $\widehat{\gamma}$ can be characterized by the minimizer
of $\mathbb{\widehat{S}}_{T}\left( \alpha ,\gamma \right) ,$ which replaces $%
\widetilde{f}_{t}$ in the construction of $\mathbb{\widetilde{S}}_{T}\left( \alpha
,\gamma \right) $ with $\widehat{f}_{t}$, since the difference between the two
is $T^{-\varphi }O_{P}(\Delta_{f}+T^{-6}) $, by Proposition \ref{prop:f}. If in addition $T=O(N)$ then it is $T^{-\varphi }O_{P}(\Delta _{f}+T^{-6})=o_P\left( T^{-1}\right) $.

\begin{proof}
	(i) We prove this for $j=1.$ The others are similarly shown. Note that
		\begin{eqnarray*}
		&&\sup_{\gamma }|\frac{1}{T}\sum_{t=1}^{T}\varepsilon _{t}x_{t}^{\prime }[1\{%
		\widetilde{f}_{t}^{\prime }\gamma >0\}-1\{\widehat{f}_{t}^{\prime }\gamma
		>0\}]|_{2} \\
		&\leq &\sup_{\gamma }\frac{1}{T}\sum_{t=1}^{T}|\varepsilon _{t}x_{t}^{\prime
		}|_{2}1\{\widehat{f}_{t}^{\prime }\gamma <0<\widetilde{f}_{t}^{\prime
	}\gamma \}+\sup_{\gamma }\frac{1}{T}\sum_{t=1}^{T}|\varepsilon
	_{t}x_{t}^{\prime }|_{2}1\{\widetilde{f}_{t}^{\prime }\gamma <0<\widehat{f}%
	_{t}^{\prime }\gamma \}
\end{eqnarray*}%
We bound the first term on the right side of the inequality above. The second term follows similarly.
As $\sup_{\gamma }|\gamma |_{2}\leq C$,
\begin{eqnarray}
	&&\sup_{\gamma }\frac{1}{T}\sum_{t=1}^{T}|\varepsilon _{t}x_{t}^{\prime
	}|_{2}1\{\widehat{f}_{t}^{\prime }\gamma <0<\widetilde{f}_{t}^{\prime
}\gamma \}  \label{e7.9} \\
&\leq& \sup_{\gamma }\frac{1}{T}\sum_{t=1}^{T}|\varepsilon
_{t}x_{t}^{\prime }|_{2}1\{-|\widehat{f}_{t}-\widetilde{f}_{t}|_{2}C<%
\widehat{f}_{t}^{\prime }\gamma <0\}  \notag \\
&\leq &\sup_{\gamma }\frac{1}{T}\sum_{t=1}^{T}|\varepsilon _{t}x_{t}^{\prime
}|_{2}1\{\left\vert \widehat{f}_{t}^{\prime }\gamma \right\vert <C\Delta_f\}+\frac{1}{T}\sum_{t=1}^{T}|\varepsilon _{t}x_{t}^{\prime }|_{2}1\{|%
\widehat{f}_{t}-\widetilde{f}_{t}|_{2}\geq \Delta_f\}  \notag \\
&\leq &\frac{1}{T}\sum_{t=1}^{T}|\varepsilon _{t}x_{t}^{\prime
}|_{2}1\{\inf_{\gamma }\left\vert \widehat{f}_{t}^{\prime }\gamma
\right\vert <C\Delta_f\}+O_P\left( 1\right) C\mathbb{P}\{|\widehat{f}%
_{t}-\widetilde{f}_{t}|_{2}\geq \Delta_f\}  \notag \\
&\leq &O_P\left( 1\right) C\mathbb{P}\left( \inf_{\gamma }|\widehat{f}%
_{t}^{\prime }\gamma |<C\Delta_f\right) +O_P\left( T^{-6}\right)  \notag
\\
&\leq &O_{P}(\Delta_f+T^{-6}),  \notag
\end{eqnarray}%
where the first inequality is by the fact that $1\left\{ A\right\} 1\left\{
B\right\} \leq 1\left\{ A\right\} $ for any events $A$ and $B,$ and the
remaining inequalities are by the law of iterated expectations, the rank
condition and the moment bound that $\mathbb{E}\left( \left\vert \varepsilon
_{t}x_{t}\right\vert _{2}|g_{t},h_{t}\right) \leq C$ a.s. in Assumption \ref%
{as9}, and Proposition \ref{prop:f}.
\medskip

\noindent
(ii) The same proof as in part (i) leads to
$\left\vert \widetilde{\mathbb{C}}_{2}(\delta ,\gamma )-%
	\widehat{\mathbb{C}}_{2}(\delta ,\gamma )\right\vert \leq |\alpha -\alpha
	_{0}|_{2}O_{P}(\Delta_f+T^{-6}).$

 It suffices to show $|\frac{1}{T}\sum_{t=1}^{T}\varepsilon _{t}\widehat{%
	Z}_{t}\left( \gamma _{0}\right) |_{2}\leq O_{P}(\frac{1}{\sqrt{T}})$ due to
(i).  Then
\begin{eqnarray*}
	&&|\frac{1}{T}\sum_{t=1}^{T}\varepsilon _{t}\widehat{Z}_{t}\left( \gamma
	_{0}\right) |_{2}\leq O_{P}(\frac{1}{\sqrt{T}})+|\frac{1}{T}%
	\sum_{t=1}^{T}\varepsilon _{t}x_{t}1\{\widehat{f}_{t}^{\prime }\gamma
	_{0}>0\}|_{2} \\
	&\leq &|\frac{1}{T}\sum_{t=1}^{T}\varepsilon _{t}x_{t}1\{(g_{t}+\frac{h_{t}}{%
		\sqrt{N}})^{\prime }\phi _{0}>0\}|_{2}+O_{P}(\frac{1}{\sqrt{T}})=O_{P}(\frac{%
		1}{\sqrt{T}}).
\end{eqnarray*}
\medskip

\noindent
(iii) The same proof as in part (i) leads to
$\left\vert \widetilde{\mathbb{C}}_{4}(\delta ,\gamma )-%
	\widehat{\mathbb{C}}_{4}(\delta ,\gamma )\right\vert \leq |\alpha -\alpha
	_{0}|_{2}O_{P}(n _{NT}+T^{-6})T^{-\varphi}.$
	Hence it is sufficient to show that
\begin{eqnarray*}
	&&\frac{1}{T}\sum_{t}|x_{t}|_{2}^{2}1\{\widehat{f}_{t}^{\prime }\gamma
	_{0}<0<f_{t}^{\prime }\gamma _{0}\} \\
	&\leq &\frac{1}{T}\sum_{t}|x_{t}|_{2}^{2}1\{0<f_{t}^{\prime }\gamma
	_{0}<|(f_{t}-\widehat{f}_{t})^{\prime }\gamma _{0}|\} \leq \frac{1}{T}\sum_{t}|x_{t}|_{2}^{2}1\{0<f_{t}^{\prime }\gamma_{0}<C
	\frac{|h_{t}|_{2}}{\sqrt{N}}\} \\
	&\leq &O_{P}(1)\frac{1}{T}\sum_{t}\mathbb{E}|x_{t}|_{2}^{2}1\{0<f_{t}^{\prime }\gamma _{0}<C
	\frac{|h_{t}|_{2}}{\sqrt{N}}\} \\
	&\leq &O_{P}(1)\mathbb{E}|x_{t}|_{2}^{2}\mathbb{P}\left( 0<f_{t}^{\prime
	}\gamma _{0}<C\frac{|h_{t}|_{2}}{\sqrt{N}}\bigg{|}x_{t},h_{t}\right) \\
	&\leq &O_{P}(1)\mathbb{E}|x_{t}|_{2}^{2}|h_{t}|_{2}\frac{1}{\sqrt{N}}%
	=O_{P}\left( N^{-1/2}\right) .
\end{eqnarray*}

\noindent
(iv) Similarly as in (i),
\begin{eqnarray*}
	&&\sup_{\gamma }|\frac{1}{T}\sum_{t=1}^{T}x_{t}\left( 1\left\{ \widetilde{f}%
	_{t}^{\prime }\gamma >0\right\} -1\left\{ \widehat{f}_{t}^{\prime }\gamma
	>0\right\} \right) \widetilde{Z}_{t}\left( \gamma \right) ^{\prime }| \\
	&\leq &\sup_{\gamma }\frac{1}{T}\sum_{t=1}^{T}|x_{t}|_{2}^{2}[1\{\widetilde{f%
	}_{t}^{\prime }\gamma <0<\widehat{f}_{t}^{\prime }\gamma \}+1\{\widehat{f}%
	_{t}^{\prime }\gamma <0<\widetilde{f}_{t}^{\prime }\gamma \}] \\
	&\leq &\sup_{\gamma }\frac{2}{T}\sum_{t=1}^{T}|x_{t}|_{2}^{2}1\{|\widehat{f}%
	_{t}^{\prime }\gamma |<C\Delta_f\}+O_{P}(T^{-6})\leq \frac{1}{T}%
	\sum_{t=1}^{T}|x_{t}|_{2}^{2}1\{\inf_{\gamma }|(g_{t}+\frac{h_{t}}{\sqrt{N}}%
	)^{\prime }\gamma |<C\Delta_f\} \\
	&\leq &O_{P}(1)\mathbb{E}|x_{t}|_{2}^{2}\mathbb{P}\left( \inf_{\gamma
	}|(g_{t}+\frac{h_{t}}{\sqrt{N}})^{\prime }\gamma |<C\Delta_f\bigg{|}%
	x_{t}\right) \leq O_{P}(\Delta_f+T^{-6}).
\end{eqnarray*}%
Hence uniformly in $(\alpha ,\gamma )$,
\begin{equation*}
	|\widetilde{R}_{3}\left( \alpha ,\gamma \right) -\widehat{R}_{3}\left(
	\alpha ,\gamma \right) |\leq |\alpha -\alpha _{0}|_{2}T^{-\varphi}O_{P}(\Delta_f+T^{-6})
\end{equation*}%
and the cases for $j=1$ and 2 are similar, so
$|\widetilde{R}_{1}\left( \alpha ,\gamma \right) -\widehat{R}_{1}\left(
	\alpha ,\gamma \right) |\leq |\alpha -\alpha _{0}|_{2}^2O_{P}(\Delta_f+T^{-6})$ and
	$|\widetilde{R}_{2}\left( \alpha ,\gamma \right) -\widehat{R}_{2}\left(
	\alpha ,\gamma \right) |\leq T^{-2\varphi}O_{P}(\Delta_f+T^{-6})$. Together, we have
	$$
 (\Delta_f+T^{-6})[T^{-2\varphi}+|\alpha-\alpha_0|_2^2+|\alpha-\alpha_0|_2T^{-\varphi}]
 \leq 2(\Delta_f+T^{-6})[T^{-2\varphi}+|\alpha-\alpha_0|_2^2].
	$$
\end{proof}

\subsubsection{Consistency} \label{sec:E3}

The introduced notation $\widehat{R}_{i}(\alpha ,\gamma )$ and $\widehat{%
	\mathbb{C}}_{i}(\delta ,\gamma )$ depend on the random rotation matrix $%
H_{T} $, which is inconvenient to carry throughout the study of consistency
and rates of convergence. On the other hand, with $\breve{g}_{t}:=g_{t}+%
\frac{1}{\sqrt{N}}h_{t}$, note that for any $\gamma $ and $\phi =H_{T}\gamma
$, we have $\widehat{f}_{t}^{\prime }\gamma =\breve{g}_{t}^{\prime }\phi $,
which is in fact independent of $H_{T}$. It is therefore more convenient to
work with functions with respect to $\phi $. Hence we introduce the
following functions of reparametrization:
\begin{eqnarray*}
	\breve{\mathds Z}_t(\phi)&=&(x_t^{\prime }, x_t^{\prime }1\{\breve
	g_t^{\prime }\phi>0\})^{\prime },\cr
	\mathds Z_{t}(\phi)&=&(x_t^{\prime },
	x_t^{\prime }1\{ g_t^{\prime }\phi>0\})^{\prime },\cr
	\mathds R(\alpha
	,\phi ) &=&\mathbb{E}[(\alpha -\alpha _{0})^{\prime }\mathds Z_{t}(\phi
	)]^{2}, \cr
	\mathds{R}_{2}\left( \phi \right) &=&\widehat{R}_{2}\left(
	\alpha , H_T^{-1}\phi\right)=\frac{1}{T}\sum_{t=1}^{T}\left( x_{t}^{\prime
	}\delta _{0}\right) ^{2} | 1\{ \breve g_{t}^{\prime }\phi >0\} -1\{\breve
	g_t^{\prime }\phi_0>0\} |, \cr
	\mathds{R}_{3}\left( \alpha ,\phi \right) &=&%
	\widehat{R}_{3}\left( \alpha , H_T^{-1}\phi \right)=\frac{2}{T}%
	\sum_{t=1}^{T}x_{t}^{\prime }\delta _{0}\left( 1\left\{\breve g_t^{\prime
	}\phi>0\right\} -1\left\{\breve g_t^{\prime }\phi_0>0\right\} \right) \breve{%
	\mathds Z}_t\left( \phi \right) ^{\prime }\left(\alpha -\alpha _{0}\right), %
\cr \mathds{C}_{1}\left( \delta ,\phi \right) &=& \mathbb{\widehat{C}}%
_{1}\left( \delta , H_T^{-1}\phi \right)= \frac{2}{T}\sum_{t=1}^{T}%
\varepsilon _{t}x_{t}^{\prime }\delta \left( 1\{\breve g_t^{\prime }\phi>0\}
-1\{\breve g_t^{\prime }\phi_0>0\} \right), \cr
\mathds{C}_{3}\left( \delta
,\phi \right) &=& \mathbb{\widehat{C}}_{3}\left( \delta , H_T^{-1}\phi
\right)= \frac{2}{T}\sum_{t=1}^{T}x_{t}^{\prime }\delta _{0}x_{t}^{\prime
}\delta\left( 1\{\breve g_t^{\prime }\phi_0>0\} -1_{t}\right) \left(
1\{\breve g_t^{\prime }\phi>0\} -1\{\breve g_t^{\prime }\phi_0>0\} \right).
\end{eqnarray*}

\begin{lem}
	\label{l:rhat} Uniformly in $(\alpha ,\phi )$, for an arbitrarily small $%
	\eta >0$,\newline
	(i) $\sup_{\phi }|\widehat{R}_{1}(\alpha ,H_{T}^{-1}\phi )-\mathds %
	R(\alpha ,\phi )|=o_{P}(1)|\alpha -\alpha _{0}|_{2}^{2}$,\newline
	(ii) $|\mathds R_{3}(\alpha ,\phi )|\leq \left( O_{P}\left( T^{-1}\right)
	+CT^{-\varphi }\left\vert \phi -\phi _{0}\right\vert _{2}\right) \left\vert
	\alpha -\alpha _{0}\right\vert _{2}.$\newline
	(iii) $|\mathds C_{1}(\delta ,\phi )-\mathds C_{1}(\delta _{0},\phi )|\leq
	\left( O_{P}\left( T^{-1}\right) +\eta T^{-2\varphi }\left\vert \phi -\phi
	_{0}\right\vert \right) T^{\varphi }\left\vert \delta -\delta
	_{0}\right\vert _{2}$\newline
	(iv) $|\mathds C_{3}(\delta ,\phi )-\mathds C_{3}(\delta _{0},\phi )|\leq
	T^{-\varphi }\left\vert \delta -\delta _{0}\right\vert _{2}O_{P}\left(
	N^{-1/2}\right) .$
\end{lem}

      \begin{proof}
(i) First, note that by uniform law of large numbers, for a sufficiently
large $C>0$,
\begin{equation*}
	\sup_{\phi }\left\vert \frac{1}{T}\sum_{t=1}^{T}\breve{\mathds Z}_{t}(\phi )%
	\breve{\mathds Z}_{t}(\phi )^{\prime }-\mathbb{E}\breve{\mathds Z}_{t}(\phi )%
	\breve{\mathds Z}_{t}(\phi )^{\prime }\right\vert =o_{P}(1).
\end{equation*}%
In addition, $\left\vert \mathbb{E}\mathds Z_{t}(\phi )\mathds Z_{t}(\phi
)^{\prime }-\mathbb{E}\breve{\mathds Z}_{t}(\phi )\breve{\mathds Z}_{t}(\phi
)^{\prime }\right\vert =o_{P}(1)$. Also, $\frac{1}{T}\sum_{t=1}^{T}\widehat{Z%
}_{t}(H_{T}^{-1}\phi )\widehat{Z}_{t}(H_{T}^{-1}\phi )^{\prime }=\frac{1}{T}%
\sum_{t=1}^{T}\breve{\mathds Z}_{t}(\phi )\breve{\mathds Z}_{t}(\phi
)^{\prime }$.
Hence
\begin{eqnarray*}
	&&\sup_{\phi }|\widehat{R}_{1}(\alpha ,H_{T}^{-1}\phi )-\mathds R(\alpha
	,\phi )| \\
	&\leq & |\alpha -\alpha _{0}|_{2}^{2}\sup_{\phi }\left\vert \frac{1}{T}%
	\sum_{t=1}^{T}\breve{\mathds Z}_{t}(\phi )\breve{\mathds Z}_{t}(\phi
	)^{\prime }-\mathbb{E}\breve{\mathds Z}_{t}(\phi )\breve{\mathds Z}_{t}(\phi
	)^{\prime }\right\vert \\
	&&+|\alpha -\alpha _{0}|_{2}^{2}\sup_{\phi }\left\vert \mathbb{E}\mathds %
	Z_{t}(\phi )\mathds Z_{t}(\phi )^{\prime }-\mathbb{E}\breve{\mathds Z}%
	_{t}(\phi )\breve{\mathds Z}_{t}(\phi )^{\prime }\right\vert \\
	&=&o_{P}(1)|\alpha -\alpha _{0}|_{2}^{2}.
\end{eqnarray*}

\noindent
(ii)
By Lemma  \ref{Lem:rate_gm}, uniformly in $\phi$
  \begin{eqnarray*}
  |\mathds R_{3}(\alpha,\phi)|&=&
| \frac{2%
}{T}\sum_{t=1}^{T}x_{t}^{\prime }\delta _{0}\left( 1\left\{\breve g_t'\phi>0\right\} -1\left\{\breve g_t'\phi_0>0\right\} \right) \breve{\mathds  Z}_t\left( \phi \right) ^{\prime }\left(
\alpha -\alpha _{0}\right)|\cr
&\leq&  C |\alpha-\alpha_0|_2 \frac{1
}{T^{1+\varphi}}\sum_{t=1}^{T}|x_{t}|_2^2  \left| 1\left\{\breve g '\phi>0\right\} -1\left\{ \breve g '\phi_0>0\right\} \right |\cr
&\leq&C|\alpha-\alpha_0|_2\big{[}
 O_P(T^{-1})+T^{-2\varphi}|\phi-\phi_0|
\big{] } \cr
&&+ C|\alpha-\alpha_0|_2T^{-\varphi}\mathbb E|x_{t}|_2^2  \left| 1\left\{\breve g '\phi>0\right\} -1\left\{ \breve g '\phi_0>0\right\} \right |\cr
&\leq&C|\alpha-\alpha_0|_2\big{[}
 O_P(T^{-1})+T^{-2\varphi}|\phi-\phi_0|
\big{] }.
\end{eqnarray*}

\noindent
(iii) Due to Lemma \ref{Lem:rate_gm} and H\"{o}lder inequality, for an arbitrarily small $\eta>0$,
  \begin{eqnarray*}
&&  |\mathds C_{1}( \delta,  \phi)-\mathds C_{1}(\delta_0, \phi)|\leq\left|
\frac{2}{T}\sum_{t=1}^{T}\varepsilon _{t}x_{t} \left( 1\{\breve g_t'\phi>0\} -1\{\breve g_t'\phi_0>0\} \right)
\right| |\delta-\delta_0|_2\cr
&=&\left\vert \frac{2}{T^{1+\varphi}}%
\sum_{t=1}^{T}\varepsilon _{t}x_{t}\left( 1\left\{ \breve g'\phi _{0}\leq 0<\breve g'\phi \right\} -1\left\{\breve g'\phi\leq 0<\breve g'\phi_{0}\right\} \right)  \right\vert\\
&&  \notag T^{\varphi} |\delta-\delta_0|_2 \\
&\leq &\left( O_P\left( T^{-1}\right) +\eta T^{-2\varphi }\left\vert \phi -\phi_{0}\right\vert \right) T^{\varphi }\left\vert \delta -\delta
_{0}\right\vert _{2}  .
\end{eqnarray*}

\noindent
(iv)
Uniformly in $\phi$,
\begin{eqnarray*}
 \left\vert \mathds C_{3}\left( \delta _{0},\phi \right) -\mathds C_{3}\left( \delta ,\phi \right) \right\vert
&\leq&
 \frac{2}{T}\sum_{t=1}^{T} |x_{t}|_2^2\left|
1\{\breve g_t'\phi_0>0\} -1\{g_t'\phi_0>0\}\right|
|\delta -\delta
_{0}|_2T^{-\varphi}\cr
&\leq &T^{-\varphi }\left\vert \delta -\delta _{0}\right\vert _{2}O_P\left(
N^{-1/2}\right),
\end{eqnarray*}%
since the modulus of the difference between two indicators is less than equal to $ 1 $.
\end{proof}

 \begin{pro}\label{Pp:consistency}
$$|\widehat\alpha-\alpha_0|_2=o_P(1),\quad  |\widehat\phi-\phi_0|_2=o_P(1).
$$
\end{pro}
Since $H_T^{-1}=O_P(1)$, this proposition implies that $\widehat\gamma- \gamma_0  =H_T^{-1}(\widehat\phi-\phi_0) + o_P(1) = o_P(1)$ as well.

\begin{proof}
	We begin with showing the consistency of $ \widehat{\gamma} $.
	Let $\widetilde P(\gamma)$ and $\widehat P(\gamma)$ respectively be the orthogonal projection matrices on $\widetilde Z_t(\gamma)$ and $\widehat Z_t(\gamma)$.
	Then
	\begin{eqnarray*}
		\widetilde{\mathbb{S}}_{T}\left( \gamma \right) &=&\widetilde{\mathbb{S}}_{T}\left( \widehat{\alpha}%
		\left( \gamma \right) ,\gamma \right) =\frac{1}{T}Y^{\prime }\left(
		I-\widetilde P\left( \gamma \right) \right) Y \\
		&=&\frac{1}{T}\left( e^{\prime }\left( I-\widetilde  P\left( \gamma \right) \right)
		e+2\delta _{0}^{\prime }X_{0}\left( I-\widetilde P\left( \gamma \right) \right)
		e+\delta _{0}^{\prime }X_{0}^{\prime }\left( I-\widetilde P\left( \gamma \right)
		\right) X_{0}\delta _{0}\right) ,
	\end{eqnarray*}%
	where $e,Y,$ and $X_{0}$ are the matrices stacking $\varepsilon _{t}$'s, $%
	y_{t}^{{}}$'s and $x_{t}^{\prime }1_{t}$'s, respectively.

	Let $\widetilde{\gamma}$ be an estimator such that
	\begin{equation}
	\widetilde{\mathbb{S}}_{T}\left( \widetilde{\gamma}\right) \leq \widetilde{\mathbb{S}}_{T}\left( \gamma
	_{0}\right) +o_P\left( T^{-2\varphi }\right) .  \label{eq:approx gm:2}
	\end{equation}%
	Then, $ \widehat{\gamma} $ satisfies this as it is a minimizer. Furthermore, \begin{eqnarray}
	0 &\geq &T^{2\varphi }\left( \widetilde{\mathbb{S}}_{T}\left( \widetilde{\gamma}\right) -%
	\widetilde{\mathbb{S}}_{T}\left( \gamma _{0}\right) \right) -o_P\left( 1\right) \nonumber \\
	&=&\frac{T^{2\varphi }}{T}\left( e^{\prime }\left( \widetilde  P\left( \gamma
	_{0}\right) -\widetilde  P\left( \widetilde{\gamma}\right) \right) e
	+2\delta _{0}^{\prime
	}X_{0}\left( \widetilde  P\left( \gamma _{0}\right) -\widetilde  P\left( \widetilde{\gamma}\right)
	\right) e+\delta _{0}^{\prime }X_{0}^{\prime }\left( \widetilde  P\left( \gamma
	_{0}\right) -\widetilde  P\left( \widetilde{\gamma}\right) \right) X_{0}\delta _{0}\right).\cr \label{eq:F-decomp}
	\end{eqnarray}%
	For the first term in \eqref{eq:F-decomp},  recall $\breve g_t= g_t+h_tN^{-1/2}$
	and note that by Lemma \ref{l:relf}, Lemma \ref{l:rhat} and ULLN lead to uniformly in $\gamma$,  and $\phi=H_T\gamma$, (recall $\mathds Z_{t}(\phi)=Z_t(\gamma)$)
	\begin{eqnarray*}
		&& \frac{1}{T} \widetilde Z(\gamma)'\widetilde Z(\gamma)
		=\frac{1}{T} \widehat Z(\gamma)'\widehat Z(\gamma) +o_P(1)
		=T^{-1}%
		\sum_{t=1}^{T}Z_{t}\left( \gamma \right) Z_{t}\left( \gamma \right) ^{\prime
		}+o_P(1)\cr
		&=&T^{-1} \sum_{t=1}^{T}\mathds Z_{t}(\phi)   \mathds Z_{t}(\phi)  ' +o_P(1)
		=\mathbb E\mathds Z_{t}(\phi)   \mathds Z_{t}(\phi)  ' +o_P(1).
	\end{eqnarray*}
	Then the rank condition for $\mathbb E\mathds Z_{t}(\phi)  \mathds Z_{t}(\phi)  '$ in Assumption \ref{as9} implies that $ \sup_\gamma  [\frac{1}{T} \widetilde Z(\gamma)'\widetilde Z(\gamma)]^{-1}=O_P(1)$. Also,
	$$
	\sup_\gamma | \frac{1}{T} \widetilde Z(\gamma)'e|_2\leq  \sup_\gamma | \frac{1}{T} \widehat Z(\gamma)'e|_2
	+O_P(\Delta_f+ T^{-6})= O_P(\frac{1}{\sqrt{T}}),$$
	by Lemma \ref{l:relf} and an FCLT for VC classes in \cite{arcones1994central}. So
	\begin{eqnarray*}
		&&|\frac{1}{T} e^{\prime }\left( \widetilde  P\left( \gamma
		_{0}\right) -\widetilde  P\left( \widetilde{\gamma}\right) \right) e|
		\leq 2\sup_\gamma\frac{1}{T}  e' \widetilde  P\left( \gamma \right) e
		\leq 2\frac{1}{T}\sup_\gamma|  [ \widetilde Z(\gamma)'\widetilde Z(\gamma)]^{-1}|_2^2| \widetilde Z(\gamma)'e|_2^2\cr
		&\leq&2 \sup_\gamma  [\frac{1}{T} \widetilde Z(\gamma)'\widetilde Z(\gamma)]^{-1}
		\sup_\gamma | \frac{1}{T} \widetilde Z(\gamma)'e|_2^2=O_P( T^{-1}).
	\end{eqnarray*}
	So $\frac{T^{2\varphi }}{T} e^{\prime }\left( \widetilde  P\left( \gamma
	_{0}\right) -\widetilde  P\left( \widetilde{\gamma}\right) \right) e=o_P(1)$.
	For the second term in \eqref{eq:F-decomp},
	\begin{eqnarray*}
		&&\frac{T^{2\varphi }}{T}\delta _{0}^{\prime}X_{0}\left( \widetilde  P\left( \gamma _{0}\right) -\widetilde  P\left( \widetilde{\gamma}\right)
		\right) e\leq  O_P(T^\varphi)\sup_\gamma |\frac{1}{T}X_{0}  \widetilde  P\left( {\gamma}
		\right) e|_2\cr
		&\leq&  O_P(T^\varphi)\sup_\gamma |\frac{1}{T} \sum_t X_t\varepsilon_t1\{\widetilde f_t'\gamma>0\}|\cr
		&=&
		O_P(T^\varphi)\sup_\gamma |\frac{1}{T} \sum_t X_t\varepsilon_t1\{\widehat f_t'\gamma>0\}|
		+O_P(T^\varphi) (\Delta_f+T^{-6})\cr
		&=&o_P(1),
	\end{eqnarray*}
	due to Lemma \ref{l:relf} and FCLT. Applying the same reasoning for the third term in \eqref{eq:F-decomp} and recalling that $ P(\gamma_{0})X_0 = X_0 $,
	\begin{eqnarray*}
		\frac{T^{2\varphi }}{T}\delta _{0}^{\prime }X_{0}^{\prime }\left( \widetilde  P\left( \gamma
		_{0}\right) -\widetilde  P\left( \widetilde{\gamma}\right) \right) X_{0}\delta _{0}
		= o_P(1) + \mathbb E(d_{0}^{\prime}x_{t})^2 1_{t} -A(\widetilde\phi),
\end{eqnarray*}
		where	$ A(\widetilde\phi)=  \mathbb Ed_{0}^{\prime
		}x_{t}1_{t} \mathds Z_{t}\left( \widetilde{\phi}\right) ^{\prime } \left(
		\mathbb E\mathds Z_{t}\left( \widetilde{\phi}\right)  \mathds Z_{t}\left( \widetilde{\phi}\right)
		^{\prime }\right) ^{-1}\mathbb E \mathds Z_{t}\left( \widetilde{\phi}\right) 1_{t}x_{t}^{\prime
		}d_{0} $.
The remaining proof for $\widetilde\phi\overset{P}{\rightarrow }\phi_0$ is the same as the  known factor case.

	Turning to $ \widehat{\alpha} $, recall$$\mathbb{\widetilde {R}}_{T}\left( \alpha ,H_T^{-1}\phi \right) =\frac{1}{T}%
\sum_{t=1}^{T}\left( \widetilde{Z}_{t}\left(H_T^{-1}\phi \right) ^{\prime }\alpha -\widetilde{%
Z}_{t}\left( H_T^{-1}\phi_{0}\right) ^{\prime }\alpha _{0}\right) ^{2}. $$

Write 
\begin{eqnarray*}
\mathds R(\alpha,\phi)&: =& \mathbb E
\left( \breve{\mathds  Z}_t(\phi)'\alpha - \breve{\mathds  Z}_t(\phi_0)'\alpha _{0}\right) ^{2}\cr
\mathds R^0(\alpha,\phi)&: =& \mathbb E
\left( \mathds Z_{t}(\phi)'\alpha - \mathds Z_{t}(\phi_0)'\alpha _{0}\right) ^{2}.
\end{eqnarray*}

We have
\begin{eqnarray*}
&&\sup_{\alpha ,\phi }|\frac{1}{T}%
\sum_{t=1}^{T}\left( \widetilde{Z}_{t}\left( H_T^{-1}\phi \right) ^{\prime }\alpha -\widetilde{%
Z}_{t}\left(H_T^{-1}\phi _{0}\right) ^{\prime }\alpha _{0}\right) ^{2}   - \left( \widehat{Z}_{t}\left( H_T^{-1}\phi\right) ^{\prime }\alpha -\widehat{%
Z}_{t}\left( H_T^{-1}\phi _{0}\right) ^{\prime }\alpha _{0}\right) ^{2}   |\cr
&\leq&\sup_{ \phi}\left(\frac{1}{T}
\sum_{t=1}^{T}     |x_t|_2^21\{|\breve  g_t'\phi|<|\widehat f_t-\widetilde f_t|_2C\}   \right)^{1/2}\cr
&\leq&\left(\frac{1}{T}
\sum_{t=1}^{T}     |x_t|_2^21\{\inf_{ \phi}|\breve g_t'\phi|<|\widehat f_t-\widetilde f_t|_2C\}   \right)^{1/2}\cr
&\leq& \left(\frac{1}{T}
\sum_{t=1}^{T}     |x_t|_2^21\{ \inf_{ \phi}|\breve g_t'\phi|<\Delta_fC\}    +
\frac{1}{T}
\sum_{t=1}^{T}     |x_t|_2^21\{|\widehat f_t-\widetilde f_t|>\Delta_f,\text{ or }|H_T|>C\}
 \right)^{1/2}\cr
 &=&o_P(1).
\end{eqnarray*}
Furthermore,
\begin{eqnarray*}
&& \sup_{\alpha ,\phi }|\frac{1}{T}%
\sum_{t=1}^{T} \left( \widehat{Z}_{t}\left( H_T^{-1}\phi \right) ^{\prime }\alpha -\widehat{%
Z}_{t}\left(H_T^{-1}\phi _{0}\right) ^{\prime }\alpha _{0}\right) ^{2}
- \mathds R(\alpha, \phi)
  |\cr
  &=& \sup_{\alpha ,\phi }|\frac{1}{T}%
\sum_{t=1}^{T} \left( \mathds{Z}_{t}\left( \phi \right) ^{\prime }\alpha -\breve{\mathds  Z}_t\left( \phi _{0}\right) ^{\prime }\alpha _{0}\right) ^{2}
- \mathds R(\alpha,\phi)
  |
  =o_P(1),
\end{eqnarray*}
by uniform law of large numbers.
Also,
\begin{eqnarray*}
  && \sup_{\alpha ,\phi}|    \mathds R(\alpha,\phi)
-    \mathds R^0(\alpha,\phi) |\leq    \left(  \mathbb   E|x_t|_2^21\{\inf_\phi|g_t'\phi|<C|h_t|_2N^{-1/2}\}		\right)^{1/2}=o(1).
\end{eqnarray*}
Hence $\sup_{\alpha ,\phi }\left\vert \mathbb{\widetilde{R}}_{T}\left(
\alpha ,H_T^{-1}\phi\right) -  \mathds R^0(\alpha,\phi)\right\vert\leq
o_P(1).$

Next, we turn to the $ \widehat{\phi} $.
Recall that $\widehat{\alpha}$ and $\widehat{\gamma}$ are minimizers of $\mathbb{\widetilde{S}%
}_{T}$ and thus
\begin{eqnarray*}
0 &\geq &\mathbb{\widetilde{S}}_{T}\left( \widehat{\alpha },\widehat{\gamma }%
\right) -\mathbb{\widetilde{S}}_{T}\left( \alpha _{0},\gamma _{0}\right)
=\mathbb{\widetilde {R}}_{T}\left( \widehat{\alpha },\widehat{\gamma }\right) -%
\mathbb{\widetilde {G}}_{T}\left( \widehat{\alpha },\widehat{\gamma }\right) +%
\mathbb{\widetilde {G}}_{T}\left( \alpha _{0}, \gamma _{0}\right) .
\end{eqnarray*}
Since $\widehat \phi:=H_T\widehat\gamma$,  Lemma \ref{l:relf}, \ref{l:rhat}, and the fact that $\mathds{C}_{i}\left( \delta ,\widehat\phi \right) = \mathbb{\widehat{C}}_{i}\left( \delta , \widehat\gamma  \right)$, $i=1,3$ imply that
      \begin{eqnarray*}
|\mathds R^0(\widehat \alpha, \widehat \phi)|&\leq &
\mathbb{\widetilde{R}}_{T}\left(
\widehat \alpha , \widehat \gamma\right)+  \sup_{\alpha ,\phi }\left\vert \mathbb{\widetilde{R}}_{T}\left(
\alpha ,H_T^{-1}\phi\right) -  \mathds R^0(\alpha,\phi)\right\vert\cr
&\leq &
 o_P(1)+
\mathbb{\widetilde {G}}_{T}\left( \widehat{\alpha },\widehat{\gamma }\right) +%
\mathbb{\widetilde {G}}_{T}\left( \alpha _{0}, \gamma _{0}\right) \cr
&\leq& o_P(1) +|\widetilde{\mathbb C}_1(\widehat\delta,\widehat\gamma)|
+|\widetilde{\mathbb C}_2(\widehat\alpha)|
+|\widetilde{\mathbb C}_3(\widehat\delta,\widehat\gamma)|
+|\widetilde{\mathbb C}_4(\widehat\alpha)|\cr
&\leq& o_P(1) + |\widehat{\mathbb C}_1( \delta_0,\widehat\gamma)|
+|\widehat{\mathbb C}_3(\delta_0,\widehat\gamma)| =o_P(1).
\end{eqnarray*}
By the identification theorem, $\mathds R^0(  \alpha, \phi)$ has a unique minimum  at  $ \left( \alpha _{0},\phi _{0}\right) $.
Then the  continuity of $\mathds R^0$
   implies  $\widehat{\alpha }\overset{P}{\longrightarrow }\alpha_0$ and $\widehat{\phi }\overset{P}{\longrightarrow }\phi_0$ by the argmax continuous mapping theorem
   \citep[see e.g.][p.286]{VW}.
\end{proof}

\subsection{Rate of convergence for $\widehat\phi$ (Proof of Theorem \ref{thm:rate})}

Here, we prove Theorem \ref{thm:rate}.
Let
  \begin{eqnarray}\label{e.11}
\mathds G_1(\phi)&:=& \mathbb E\mathds R_{2}(  \phi)+\mathbb E\mathds C_{3}( \delta_0,  \phi)
\cr
\mathds G_2(\phi)&:=&|  \mathds R_{2}(  \phi)+\mathds C_{3}( \delta_0,  \phi)  -( \mathbb E\mathds R_{2}(  \phi)+\mathbb E\mathds C_{3}( \delta_0,  \phi)  )|.
\end{eqnarray}

Recall that $\mathds R(\alpha,\phi)= \mathbb E[(\alpha-\alpha_0)'\mathds  Z_{t}(\phi)]^2$.

\begin{lem}\label{l:rt} Uniformly in $\alpha,\phi$,  for any $\epsilon>0$,  there is $C>0$ that is independent of $\epsilon$, and $C_\epsilon$ that depends on $\epsilon$, so that
$|\mathds R(\alpha, \phi)
-\mathds R(\alpha,\phi_0)|\leq C|\alpha-\alpha_0|_2^2[   C_\epsilon|\phi-\phi_0|_2+\epsilon
]^{1/2}.$  Hence
$|\mathds R(\alpha, \widehat\phi)
-\mathds R(\alpha,\phi_0)|=o_P(1)|\alpha-\alpha_0|_2^2$.
\end{lem}

      \begin{proof}
      For any $\epsilon>0$, there is $C_1$, so that $\mathbb P(|g_t|_2>C_1)<\epsilon.$
Note that  for any deterministic $\phi$,
  \begin{eqnarray*}
&&|\mathds R(\alpha,  \phi)
-\mathds R(\alpha, \phi_0)|\leq |\alpha-\alpha_0|_2^2 \mathbb E|x_t|_2^21\{|g_t'\phi_0|< |g_t|_2 |\phi-\phi_0|_2\}\cr
&\leq&|\alpha-\alpha_0|_2^2 \mathbb P^{1/2}(|g_t'\phi_0|< |g_t|_2 |\phi-\phi_0|_2)
 (\mathbb E|x_t|_2^4)^{1/2}
\cr
&\leq& C|\alpha-\alpha_0|_2^2[ \mathbb P (|g_t'\phi_0|< C_\epsilon|\phi-\phi_0|_2)
+ \mathbb P (|g_t|_2>C_1 )
]^{1/2}\cr
&\leq& C|\alpha-\alpha_0|_2^2[ C_\epsilon  |\phi-\phi_0|_2+\epsilon
]^{1/2}. 
\end{eqnarray*}
Now let $\phi=\widehat\phi$,  and the consistency implies $|\widehat\phi-\phi_0|_2=o_P(1)$. Thus
 \begin{eqnarray*}
&&|\mathds R(\alpha,  \phi)
-\mathds R(\alpha, \phi_0)|\leq
 C|\alpha-\alpha_0|_2^2[ C_\epsilon  o_P(1)+\epsilon
]^{1/2}. 
\end{eqnarray*}
Since $\epsilon>0$ is arbitrary, we have the desired result.
\end{proof}

\begin{lem}\label{l.1} For an arbitrarily small $\eta>0,$ uniformly in $\phi$,
$$|  \mathds  G_2(  \phi)  )|
\leq  b_{NT} T^{-\varphi}, \quad |\mathds C_{1}\left( \delta _{0},\phi
\right) |\leq  b_{NT}.$$
If in addition,  $\sqrt{N}=O(T^{1-2\varphi})$, then
$$|    \mathds  G_2(   \phi  )|
\leq  a_{NT} T^{-\varphi}, \quad |\mathds C_{1}\left( \delta _{0},\phi
\right) |\leq  a_{NT}.$$
where
\begin{eqnarray*}
a_{NT}&=&T^{-2\varphi}
O_P\left( \frac{\sqrt{N}}{\left( NT^{1-2\varphi }\right) ^{2/3}}\right)
+T^{-2\varphi} \eta \left\vert \phi -\phi
_{0}\right\vert _{2}^{2}\sqrt{N}
\\
b_{NT}&=&O_P(\frac{1}{T}) +\eta T^{-2\varphi }\left\vert \phi -\phi
_{0}\right\vert _{2}.
\end{eqnarray*}
\end{lem}

\begin{proof} Let $z_t=T^{2\varphi}2(x_{t}^{\prime }\delta _{0} )^2
\left( 1\{\breve g_t'\phi_0>0\} -1(g_t'\phi_0>0)\right)$. By  Lemma \ref{Lem:rate_gm}, we have the following  bound:
  \begin{eqnarray*}
| \mathds C_{3}( \delta_0,  \phi)  -  \mathbb E \mathds C_{3}( \delta_0,  \phi)  |&= & T^{-\varphi} | \frac{1}{T^{1+\varphi}}\sum_{t=1}^{T} [z_t \left( 1\{ \breve g_t'\phi>0\} -1\{ \breve g_t'\phi_0>0\}\right)\cr
&&- \mathbb E  z_t \left( 1\{ \breve g_t'\phi>0\} -1\{ \breve g_t'\phi_0>0\} \right)] |\cr
&\leq& O_P(\frac{1}{T^{1+\varphi}}) +\eta T^{-3\varphi }\left\vert \phi-\phi_0\right\vert _{2}.
\end{eqnarray*}

In addition, by Lemma \ref{Lem:rate_gm_1},   when $\sqrt{N}=O(T^{1-2\varphi})$ we have the other upper bound:
 \begin{eqnarray*}
| \mathds C_{3}( \delta_0,  \phi)  -  \mathbb E \mathds C_{3}( \delta_0,  \phi)  |&=& T^{-3\varphi} | \frac{1}{T^{1-\varphi}}\sum_{t=1}^{T} [z_t \left( 1\{ \breve g_t'\phi>0\} -1\{ \breve g_t'\phi_0>0\}\right)\cr
&&- \mathbb E  z_t \left( 1\{ \breve g_t'\phi>0\} -1\{ \breve g_t'\phi_0>0\} \right)] |\cr
&\leq& T^{-3\varphi}
O_P\left( \frac{\sqrt{N}}{\left( NT^{1-2\varphi }\right) ^{2/3}}\right)
+T^{-3\varphi} \eta \left\vert \phi -\phi
_{0}\right\vert _{2}^{2}\sqrt{N}
\end{eqnarray*}

Similarly, the same upper bound applies to
$|  \mathds R_{2}(  \phi) -  \mathbb E\mathds R_{2}(  \phi)|$.

 Furthermore, note that for any $\eta >0$%
\begin{eqnarray*}
\mathds C_{1}\left( \delta _{0},\phi \right)  &\leq&\big\vert
\frac{2}{T}\sum_{t=1}^{T}\varepsilon _{t}x_{t}^{\prime }\delta _{0}\big(
1\left\{ \breve g_t'\phi _{0}\leq 0< \breve g_t'\phi \right\}
-1\left\{ \breve g_t'\phi \leq 0<\breve g_t'\phi _{0}\right\}
\big) \big\vert \\
&\leq &O_P\left( T^{-1}\right) +\eta T^{-2\varphi }\left\vert\phi
-\phi _{0}\right\vert_2
\end{eqnarray*}%
due to Lemma \ref{Lem:rate_gm} and that when  $\sqrt{N}=O(T^{1-2\varphi})$
\begin{equation}
\mathds C_{1}\left( \delta _{0},\phi
\right) \leq T^{-2\varphi }O_P\left( \frac{\sqrt{N}}{\left( NT^{1-2\varphi }\right) ^{2/3}}%
\right) +\eta T^{-2\varphi }\sqrt{N} \left\vert \phi -\phi _{0}\right\vert ^{2},
\label{eq:C1-2}
\end{equation}%
due to Lemma \ref{Lem:rate_gm_1}.
\end{proof}

Lemma \ref{l.2} below  holds regardless of whether $N^{1/2}<T^{1-2\varphi}$ or not, but is crude when $N^{1/2}=o(T^{1-2\varphi})$.  When $N^{1/2}=o(T^{1-2\varphi})$, a sharper bound is given in Lemma \ref{l:generic}.
\begin{lem}  \label{l.2}
Suppose the conditional density of $f_t'\gamma_0$ given $(x_t, h_t)$ is bounded away from above almost surely. Then
  there is a constant $C,c>0$ that do not depend on $\phi$,   \begin{eqnarray*}
 \mathds G_1(\phi)
&\geq&cT^{-2\varphi} |\phi-\phi_0|_2  -  \frac{C}{\sqrt{N}T^{2\varphi}} .
\end{eqnarray*}%
\end{lem}

\begin{proof}    First,
\begin{eqnarray*}%
 | \mathbb E\mathds C_3(\delta_0,\phi)|
&\leq&  \mathbb E (x_{t}^{\prime }\delta _{0})^2
\left| 1\{ \breve g_t'\phi_0>0\} -1\{g_t'\phi_0>0\}\right|   \leq CT^{-2\varphi} \frac{1}{\sqrt{N}}.
\end{eqnarray*}%
Next, we lower bound $ \mathbb E\mathds{R}_{2}\left(
\phi \right)= \mathbb E \left( x_{t}^{\prime
}\delta _{0}\right) ^{2}\left\vert 1\left\{ \breve g_t'\phi
>0\right\} -1\left\{\breve g_t'\phi_0>0\right\} \right\vert$. The proof is similar to Step 1 of Proof of Lemma \ref{rates-thm}]. We  show that there exists a constant $c>0$ and a neighborhood of $%
\phi _{0}$ such that for all $\phi $ in the neighborhood
\[
G\left( \gamma \right) =\mathbb{E}\left\vert 1 \{\breve g_t'\phi>0\}
-1\{\breve g_t'\phi_0>0\} \right\vert \geq c\left\vert \phi -\phi
_{0}\right\vert _{2}.
\]%
Note that   the first element of $%
\left( \gamma -\gamma _{0}\right) $ is zero due to the normalization. Then,
\[
G\left( \gamma \right) =\mathbb{P}\left\{ -\widehat f_{2t}^{\prime }\left( \gamma
_{2}-\gamma _{20}\right) \leq\breve g_t'\phi_0<0\right\} +\mathbb{P}\left\{ 0<\breve g_t'\phi_0\leq
-\widehat f_{2t}^{\prime }\left( \gamma _{2}-\gamma _{20}\right) \right\} .
\]%
Since the conditional density of $\breve g_t'\phi_0$ given $\widehat f_{2t}$   is bounded away from zero and
continuous in a sufficiently small open
neighborhood $\epsilon $ of zero, we can find   $c_{1}>0$ so that
\[
\mathbb{P}\left\{ -\widehat f_{2t}^{\prime }\left( \gamma _{2}-\gamma _{20}\right)
\leq \breve g_t'\phi_0<0\right\} \geq c_{1}\mathbb{E}\left(\widehat  f_{2t}^{\prime }\left(
\gamma _{2}-\gamma _{20}\right) 1\left\{\widehat f_{2t}^{\prime }\left( \gamma
_{2}-\gamma _{20}\right) >0\right\} 1\left\{ \left\vert \widehat f_{2t}^{\prime
}\right\vert \leq M\right\} \right) ,
\]%
where $M$ satisfies that $  \left\vert \gamma -\gamma _{0}\right\vert
_{2}M<\epsilon $. This is always feasible because we can make $%
 \left\vert \gamma -\gamma _{0}\right\vert _{2}$ as small as necessary
due to the consistency of $\widehat{\gamma}$. Similarly,
\[
\mathbb{P}\left\{ 0<\breve g_t'\phi_0\leq - \widehat   f_{2t}^{\prime }\left( \gamma _{2}-\gamma
_{20}\right) \right\} \geq c_{1}\mathbb{E}\left( -  \widehat   f_{2t}^{\prime }\left(
\gamma _{2}-\gamma _{20}\right) 1\left\{  \widehat  f_{2t}^{\prime }\left( \gamma
_{2}-\gamma _{20}\right) <0\right\} 1\left\{ \left\vert  \widehat  f_{2t}^{\prime
}\right\vert \leq M\right\} \right) .
\]%
Thus,
\[
G\left( \gamma \right) \geq c_{1}\mathbb{E}\left( \left\vert \widehat f_{2t}^{\prime
}\left( \gamma _{2}-\gamma _{20}\right) \right\vert 1\left\{ \left\vert
 \widehat  f_{2t} \right\vert \leq M\right\} \right) \geq c_{2}\left\vert
\gamma -\gamma _{0}\right\vert _{2}
\]%
for some $c_{2}>0$ because
\[
\inf_{\left\vert r\right\vert =1}\mathbb{E}\left( \left\vert \widehat   f_{2t}^{\prime
}r\right\vert 1\left\{ \left\vert  \widehat f_{2t} \right\vert \leq M\right\}
\right) >0
\]%
for some $M<\infty $. The last inequality $\inf_{\left\vert r\right\vert =1}\mathbb{E}\left( \left\vert \widehat   f_{2t}^{\prime
}r\right\vert 1\left\{ \left\vert  \widehat f_{2t} \right\vert \leq M\right\}
\right) >0$ follows since
 \begin{align*}
 &\inf_{\left\vert r\right\vert =1}\mathbb{E}\left( \left\vert \widehat   f_{2t}^{\prime
}r\right\vert 1\left\{ \left\vert  \widehat f_{2t} \right\vert \leq M\right\}
\right)\cr
& \geq \inf_{\left\vert r\right\vert =1}\mathbb{E}\left( \left\vert    f_{2t}^{\prime
}r\right\vert 1\left\{ \left\vert   f_{2t} \right\vert \leq M\right\}
\right) - \mathbb E|\widehat f_t-f_t|_2-\mathbb E|f_t|_21\{M-\frac{|h_t|_2}{\sqrt{N}}<|f_t|_2< M+\frac{|h_t|_2}{\sqrt{N}}\} \cr
&\geq c -O(N^{-1/8}) -\mathbb E|f_t|_21\{M-\frac{|h_t|_2}{\sqrt{N}}<|f_t|_2< M+\frac{|h_t|_2}{\sqrt{N}}\} 1\{|h_t|_2<MN^{1/4}\}\cr
&\geq c/2 - c\left[\sup_{|f|<2M, h_t}p_{f_{2t}|h_t}(f) \mathbb E\mu\left( f\in\mathbb R^{\dim(f_{2t})}: M-\frac{|h_t|_2}{\sqrt{N}}<|f|_2< M+\frac{|h_t|_2}{\sqrt{N}}   \right)  1\{ |h_t|_2 <MN^{1/4}\} \right]^{1/2}\cr
&\geq c/2 -c\left[ \mathbb E \left( (M+\frac{|h_t|_2}{\sqrt{N}})^{\dim(f_{2t})}-(M-\frac{|h_t|_2}{\sqrt{N}})^{\dim(f_{2t})}  \right)  1\{ |h_t|_2 <MN^{1/4}\}\right]^{1/2}\geq c/4.
 \end{align*}
where $\mu( A)$ denotes the Lebesgue measure of the set $A$; here $A$ is the difference of two balls in $\mathbb R^{\dim(f_{2t})}$. Here the second inequality follows from: $ \mathbb E|\widehat f_t-f_t|_2=O(N^{-1/2})$, and write
$a_t:=|f_t|_21\{M-\frac{|h_t|_2}{\sqrt{N}}<|f_t|_2< M+\frac{|h_t|_2}{\sqrt{N}}\} .$
\begin{align*}
 \mathbb Ea_t&\leq  \mathbb Ea_t1\{|h_t|_2<MN^{1/4}\}+ (\mathbb Ea_t^2)^{1/2}\mathbb P(|h_t|_2>MN^{1/4})^{1/2}\cr
 &\leq \mathbb Ea_t1\{|h_t|_2<MN^{1/4}\}+(\mathbb E|f_t|_2^2)^{1/2}(\frac{\mathbb E|h_t|_2}{MN^{1/4}})^{1/2}\cr
 &\leq \mathbb Ea_t1\{|h_t|_2<MN^{1/4}\}+O(N^{-1/8}).
    \end{align*}

\end{proof}

\begin{pro}[Preliminary Rate of convergence]\label{p:pre} Suppose  $T^{ 2\varphi}\log^\kappa T=O(N)$ for any $\kappa>0.$ For $\widehat\phi=H_T\widehat\gamma$,
 $$
 |\widehat\alpha-\alpha_0|_2=O_P(T^{-1/2}+N^{-1/4}T^{-\varphi}),\quad
 |\widehat\phi-\phi_0|_2=O_P(T^{-(1-2\varphi)}+ N^{-1/2}).
 $$
\end{pro}
\noindent {\bf Remark} When  $T^{1-2\varphi}=O(\sqrt{N})$,  this rate becomes
 $$
 |\widehat\alpha-\alpha_0|_2=O_P(T^{-1/2}),\quad
 |\widehat\phi-\phi_0|_2=O_P(T^{-(1-2\varphi)}),
 $$
which  is tight and  identical to the case of the known factor, but not so when $\sqrt{N}=o(T^{1-2\varphi})$.

\begin{proof}
As $\widehat{\alpha}$ and $\widehat{\gamma}$ are minimizers of $\mathbb{\widetilde{S}%
}_{T}$,%
\begin{eqnarray}
0 &\geq &\mathbb{\widetilde{S}}_{T}\left( \widehat{\alpha },\widehat{\gamma }%
\right) -\mathbb{\widetilde{S}}_{T}\left( \alpha _{0},\gamma _{0}\right)
=\mathbb{\widetilde {R}}_{T}\left( \widehat{\alpha },\widehat{\gamma }\right) -%
\mathbb{\widetilde {G}}_{T}\left( \widehat{\alpha },\widehat{\gamma }\right) +%
\mathbb{\widetilde {G}}_{T}\left( \alpha _{0}, \gamma _{0}\right) ,  \notag
\end{eqnarray}%

So
$\widetilde R_{1}(\widehat\alpha,\widehat\gamma)
+\widetilde R_{2}( \widehat\gamma)+\mathbb{\widetilde{C}}_{3}(\widehat\delta, \widehat\gamma)
+\widetilde R_{3}(\widehat\alpha,\widehat\gamma)
\leq
\mathbb{\widetilde{C}}_{1}(\widehat\delta, \widehat\gamma)
+\mathbb{\widetilde{C}}_{2}(\widehat\alpha)
-\mathbb{\widetilde{C}}_{4}(\widehat\alpha) .
$
By Lemma \ref{l:relf},
\begin{eqnarray*}
&& \mathds  R(\alpha,  \widehat\phi)
+\widehat R_{2}( \widehat\gamma)+\mathbb{\widehat{C}}_{3}(\widehat\delta, \widehat\gamma)
+\widehat R_{3}(\widehat\alpha,\widehat\gamma)
\leq o_P(1)|\widehat\alpha-\alpha_0|_2^2+O_P(\Delta_f+T^{-6})T^{-\varphi}
\cr
&&+ O_P(\Delta_f+T^{-1/2}) |\widehat\alpha-\alpha_0|_2
+\mathbb{\widehat{C}}_{1}(\widehat\delta, \widehat\gamma).
\end{eqnarray*}
Note that  $\mathds{R}_{3}\left( \alpha ,\phi \right) =\widehat{R}_{3}\left( \alpha , H_T^{-1}\phi \right), $ $\mathds{R}_{2}\left( \phi \right) =\widehat{R}_{2}\left(  H_T^{-1}\phi \right), $ $\mathds{C}_{i}\left( \delta ,\phi \right) =  \mathbb{\widehat{C}}_{i}\left( \delta , H_T^{-1}\phi  \right)$, $i=1,3$.
In addition, since    $\varphi<1/2$, by Lemma \ref{l:rhat},  it follows that  there is $C_1>0$,
\begin{eqnarray*}
&& \mathds  R(\alpha,  \widehat\phi)
+\mathds R_{2}( \widehat\phi)+\mathds C_{3}(\delta_0, \widehat\phi)
\leq o_P(1)|\widehat\alpha-\alpha_0|_2^2+O_P(\Delta_f+T^{-6})T^{-\varphi}
\cr
&&+\mathds C_{1}( \delta_0, \widehat\phi)
+ O_P(\Delta_f+T^{-1/2}+ T^{-\varphi}N^{-1/2}) |\widehat\alpha-\alpha_0|_2
+
 C_1T^{-\varphi }\left|\widehat\phi -\phi_{0}\right\vert _{2}
\left\vert \widehat \alpha -\alpha _{0}\right\vert _{2}.
\end{eqnarray*}
We  now provide a lower bound on the left hand side.
By Lemma \ref{l:rt},  $|\mathds R_T(\widehat\alpha, \widehat\phi)
-\mathds R_T(\widehat\alpha,\phi_0)|=o_P(1)|\widehat\alpha-\alpha_0|_2^2.$ Also, uniformly in $\alpha,$
$$\mathds R(\alpha,\phi)= \mathbb E[(\alpha-\alpha_0)'\mathds  Z_{t}(\phi)]^2\geq C|\alpha-\alpha_0|_2^2.$$ In addition,
 $\mathds R_{2}( \widehat\phi)+\mathds C_{3}(\delta_0, \widehat\phi)\geq\mathds G_1(\widehat\phi)-  \mathds G_2(\widehat\phi)$.
 This implies
\begin{eqnarray}\label{e:prer}
&&(C_0-o_P(1))|\widehat\alpha-\alpha_0|_2^2
+ \mathds G_1(\widehat\phi)
\leq \mathds G_2(\widehat\phi)+\mathds C_{1}( \delta_0, \widehat\phi)
+ O_P(\Delta_f+T^{-6})T^{-\varphi}
\cr
&&+ O_P(\Delta_f+T^{-1/2}  + T^{-\varphi}N^{-1/2} ) |\widehat\alpha-\alpha_0|_2
+
 C_1T^{-\varphi }\left|\widehat\phi -\phi_{0}\right\vert _{2}
\left\vert \widehat \alpha -\alpha _{0}\right\vert _{2}.
\end{eqnarray}
Let  $C_3$ be chosen to be smaller than $C_0/2$ and $C_2$ be  chosen to be smaller than $C_4/4$ below.  Due to the consistency of $\widehat\phi$,  with probability approaching one, $|\widehat\phi -\phi_{0}|_2\leq (C_2C_3)/(8C_1^2)$. Hence   with probability approaching one, for $d=\frac{C_3}{4C_1^2}$, one term on the right hand side:
\begin{eqnarray*}
	&&C_1T^{-\varphi }\left|\widehat\phi -\phi_{0}\right\vert _{2}|\widehat\alpha-\alpha_0|_2 \leq C_1^2
	d|\widehat\alpha-\alpha_0|_2^2 +T^{-2\varphi}\left|\widehat\phi -\phi_{0}\right\vert _{2}^2 d^{-1}\cr
	&\leq& C_3|\widehat\alpha-\alpha_0|_2^2/4
	+C_2T^{-2\varphi}\left|\widehat\phi -\phi_{0}\right\vert _{2}  /2.
\end{eqnarray*}

Given this, the goal becomes lower bounding $\mathds G_1(\widehat\phi)$ and upper bounding $\mathds G_2(\widehat\phi)+\mathds C_{1}( \delta_0, \widehat\phi)$.
Apply  Lemma  \ref{l.1} using the upper bound $b_{NT}$, and reach,
$$
\mathds G_2(\widehat\phi)+\mathds C_{1}( \delta_0, \widehat\phi)
\leq O_P(1)b_{NT}\leq O_P(T^{-1}) +\eta T^{-2\varphi }\left\vert \widehat\phi -\phi
_{0}\right\vert _{2}.
$$
with an arbitrarily small $\eta>0$.
 Lemma \ref{l.2} implies $\mathds G_1(\widehat\phi)\geq C_4T^{-2\varphi} |\widehat\phi-\phi_0|_2  -  \frac{C}{\sqrt{N}T^{2\varphi}} $ almost surely.
 Since $\eta>0$ is arbitrarily small, (\ref{e:prer}) implies,
  \begin{eqnarray}  \label{e:iter2}
&&  C_0|\widehat\alpha-\alpha_0|_2^2/4
+ C_4T^{-2\varphi} |\widehat\phi-\phi_0|_2 /2
\cr
&\leq& O_P(T^{-1}+\frac{C}{\sqrt{N}T^{2\varphi}} )  +
O_P(\Delta_f+T^{-1/2}  + T^{-\varphi}N^{-1/2} )|\widehat\alpha-\alpha_0|_2
+O_P(\Delta_f+T^{-6})T^{-\varphi} \cr
\end{eqnarray}
 which leads to the preliminary rate: when $T^{ 2\varphi}\log^\kappa T=O(N)$ for any $\kappa>0,$
 \begin{eqnarray*}
 |\widehat\alpha-\alpha_0|_2&=&O_P(T^{-1/2}+N^{-1/4}T^{-\varphi}+\Delta_f^{1/2}T^{-\varphi/2}+\Delta_f) = O_P(T^{-1/2}+N^{-1/4}T^{-\varphi})    ,\cr
 |\widehat\phi-\phi_0|_2&=&O_P(T^{-(1-2\varphi)}+ N^{-1/2} +\Delta_fT^{\varphi}+ (\Delta_fT^{\varphi})^2)=O_P(T^{-(1-2\varphi)}+ N^{-1/2} ),
\end{eqnarray*}
where we used    $\Delta_f\leq O(\log ^c T)(\frac{1}{N}+\frac{1}{T})$ proved in Proposition \ref{prop:f}.
\end{proof}

To improve the convergence rate when $N=o(T^{2-4\varphi})$, we need to obtain a sharper lower bound for $\mathds G_1(\phi)$ than that of Lemma \ref{l.2} .
 To present the   lemma below,  we first introduce some notation.   Let $p_{X_t|Y_t}$ denote the conditional density of $X_t$  given $Y_t$, for the random vectors $X_t  $ and $Y_t$ specified in the lemma below, assumed to exist.

\begin{lem} \label{l:generic}   Let $u_t=g_t'\phi_0$ and Assumption  \ref{as8} hold. Suppose  $N=o(T^{2-4\varphi})$. Consider a generic deterministic vector $\phi$  that   is linearly independent of $\phi_0$ and $\sqrt{N}|\phi -\phi_0|\leq L$ for some $L>0$.
Then    uniformly in $\phi$,
\begin{eqnarray*}
|\mathds G_1(\phi)
&\geq& CT^{-2\varphi} \sqrt{N} |\phi -\phi_0|_2^2 -  O(\frac{1}{T^{2\varphi}N^{5/6}}).
\end{eqnarray*}%

\end{lem}

\begin{proof}

Write $1_t=1\{g_t'\phi_0>0\}$. First, we note that a careful calculation yields:
\begin{eqnarray*}
&&2\left( 1\{ \breve g_t'\phi_0>0\} -1_{t}\right) \left( 1\{\breve g_t'\phi>0\}) -1\{ \breve g_t'\phi_0>0\} \right)
 +\left\vert 1\{\breve g_t'\phi>0\} -1\{ \breve g_t'\phi_0>0\} \right\vert \\
&:=& A_{1t}(\phi) +A_{2t}(\phi)-A_{3t}(\phi)-A_{4t}(\phi) \end{eqnarray*}%
where
\begin{eqnarray*}
A_{1t}(\phi)&=&1\left\{ \breve g_t'\phi\leq 0<\breve g_t'\phi_0\right\}
1\left\{ g_{t}^{\prime }\phi _{0}>0\right\} \cr
A_{2t}(\phi)&=& 1\left\{ \breve g_t'\phi_0\leq 0<%
 \breve g_t'\phi \right\} 1\left\{ g_{t}^{\prime }\phi _{0}\leq
0\right\} \cr
A_{3t}(\phi)&=&1\left\{ \breve g_t'\phi\leq 0<\breve g_t'\phi_0\right\}
1\left\{ g_{t}^{\prime }\phi _{0}\leq 0 \right\}\cr
A_{4t}(\phi)&=&1\left\{
\breve g_t'\phi_0\leq 0< \breve g_t'\phi \right\}
1\left\{ g_{t}^{\prime }\phi_{0} > 0\right\}
\end{eqnarray*}

 Therefore,     \begin{eqnarray*}
 \mathds G_1(\phi)&=& \mathbb E\left( x_{t}^{\prime }\delta _{0}\right) ^{2}\left( A_{1t}\left( \phi
\right) +A_{2t}\left( \phi \right) -A_{3t}\left( \phi \right)
-A_{4t}\left( \phi \right) \right).
\end{eqnarray*}%
The goal is to provide a sharp lower bound of the right hand side.  Note that  $\phi-\phi_0$ is linearly independent of $\phi_0$ due to the normalization. And as elsewhere $ C $ is a generic positive constant.


\medskip
\noindent
\textbf{Calculating $A_1$}
\medskip

Take the first term $A_{1t}\left( \phi \right) $ and note that (cf.
notation $u_{t}=g_{t}^{\prime }\phi _{0}$ )
\begin{eqnarray*}
A_1
&=&1\left\{ 0\vee -\frac{h_{t}^{\prime }\phi _{0}}{\sqrt{N}}<u_{t}\leq
-\left( g_{t}+\frac{h_{t}}{\sqrt{N}}\right) ^{\prime }\left( \phi-\phi_0\right) -\frac{h_{t}^{\prime }\phi _{0}}{\sqrt{N}}\right\}  \notag \\
&=&1\left\{ -h_t'\phi_0<\sqrt{N}u_{t}\leq -\sqrt{N} g_t ^{\prime }\left( \phi-\phi_0\right) - h_t'\phi_0\right\} 1\left\{
h_t'\phi_0\leq 0\right\}    \cr
&&+1\left\{ 0<\sqrt{N}u_{t}\leq -\sqrt{N}  g_t^{\prime
}\left( \phi-\phi_0\right) -  h_t'\phi_0\right\} 1\left\{ h_t'\phi_0>0\right\}\cr
&&+\left[1\left\{  \sqrt{N}u_{t}\leq -\sqrt{N} \breve g_t ^{\prime }\left( \phi-\phi_0\right) - h_t'\phi_0\right\} -1\left\{  \sqrt{N}u_{t}\leq -\sqrt{N} g_t ^{\prime }\left( \phi-\phi_0\right) - h_t'\phi_0\right\} \right] \cr
&&\quad \times [1\left\{
h_t'\phi_0\leq 0\right\}   1\{-h_t'\phi_0<\sqrt{N}u_{t}\}  + 1\left\{ h_t'\phi_0>0\right\} 1\{u_t>0\} ] .
\end{eqnarray*}
Now suppose that for any $L>0$, the conditional density of $g_t'\phi$ given $(h_t, x_t)$ is bounded uniformly for $\phi\in\{|\phi-\phi_0|_2<LN^{-1/2}\}$: that is $\sup_{|\phi-\phi_0|_2<LN^{-1/2}}p_{g_t'\phi|h_t, x_t}(\cdot)<C.$
Hence
\begin{eqnarray*}
\mathbb  E\left( x_{t}^{\prime }\delta _{0}\right) ^{2}A_1&=&\mathbb  E\left( x_{t}^{\prime }\delta _{0}\right) ^{2}1\left\{ -h_t'\phi_0<\sqrt{N}u_{t}\leq -\sqrt{N} g_t ^{\prime }\left( \phi-\phi_0\right) - h_t'\phi_0\right\} 1\left\{
h_t'\phi_0\leq 0\right\}    \cr
&&+\mathbb  E\left( x_{t}^{\prime }\delta _{0}\right) ^{2}1\left\{ 0<\sqrt{N}u_{t}\leq -\sqrt{N}  g_t^{\prime
}\left( \phi-\phi_0\right) -  h_t'\phi_0\right\} 1\left\{ h_t'\phi_0>0\right\}+ A_{11},
\end{eqnarray*}
where
\begin{eqnarray*}
A_{11}&:=&\mathbb  E\left( x_{t}^{\prime }\delta _{0}\right) ^{2} [1\left\{
h_t'\phi_0\leq 0\right\}   1\{-h_t'\phi_0<\sqrt{N}u_{t}\}  + 1\left\{ h_t'\phi_0>0\right\} 1\{u_t>0\} ]
\cr
&&\quad \times
\left[1\left\{  \sqrt{N}u_{t}\leq -\sqrt{N} \breve g_t ^{\prime }\left( \phi-\phi_0\right) - h_t'\phi_0\right\} -1\left\{  \sqrt{N}u_{t}\leq -\sqrt{N} g_t ^{\prime }\left( \phi-\phi_0\right) - h_t'\phi_0\right\} \right] \cr
&\leq& C T^{-2\varphi}\mathbb  E\mathbb P\left\{ -    h_t^{\prime }\left( \phi-\phi_0\right) - h_t'\phi_0< \sqrt{N}g_{t}'\phi\leq   - h_t'\phi_0\bigg{|} h_t\right\} \cr
&&+ T^{-2\varphi}\mathbb  E\mathbb P\left\{    - h_t'\phi_0< \sqrt{N}g_{t}'\phi\leq  -  h_t ^{\prime }\left( \phi-\phi_0\right) - h_t'\phi_0    \bigg{|} h_t\right\} \cr
&\leq&  2C\sup_{\|\phi-\phi_0\|<LN^{-1/2}}p_{g_t'\phi|h_t}(\cdot)T^{-2\varphi}\mathbb  E  \frac{|h_t'(\phi-\phi_0)|}{\sqrt{N}}\cr
&\leq& \frac{C}{\sqrt{N}T^{2\varphi}}|\phi-\phi_0|_2 \leq \frac{CL}{NT^{2\varphi}},\quad \text{ given that  $|\phi-\phi_0|_2<LN^{-1/2}$,}
\end{eqnarray*}
due to Assumption \ref{as8} \ref{as8:itm7} for the first inequality.
 On the other hand,
 note that the normalization condition requires the first element of $\gamma-\gamma_0=0$, so $  g_t^{\prime
}\left( \phi-\phi_0\right)= f_t'(\gamma-\gamma_0)=f_{2t}'(\gamma-\gamma_0)_2$. Thus $g_t'(\phi-\phi_0)$ depends on $g_t$ only through $f_{2t}= (H_T'f_t)_2$, where $f_{2t}$ and  $(H_T'f_t)_2$ denote the subvectors of   $f_t$ and $H_T'f_t$, excluding their first elements, corresponding to the 1-element of $\phi.$

 Let $p_{u_t|\bigstar}(\cdot):=p_{f_t'\gamma_0|h_t'\phi_0,   f_{2t},x_{t} }(\cdot)$ denote the conditional density of $u_t = f_t'\gamma_0 = g_t'\phi_0$, given $(h_{t}^{\prime }\phi _{0},   f_{2t},x_{t})$.
Change variable  $a=\sqrt{N}u$, we have,
\begin{align}
& \mathbb E\left( x_{t}^{\prime }\delta _{0}\right) ^{2}A_1- A_{11}\cr
&=\frac{1}{\sqrt{N}} \mathbb E\left( x_{t}^{\prime }\delta _{0}\right) ^{2}\int 1\left\{ -h_t'\phi_0<a\leq -\sqrt{N}   g_t ^{\prime }\left( \phi-\phi_0\right) - h_t'\phi_0\right\} 1\left\{
h_t'\phi_0\leq 0\right\}  p_{u_t|\bigstar}(\frac{a}{\sqrt{N}})da  \cr
&+\frac{1}{\sqrt{N}} \mathbb E\left( x_{t}^{\prime }\delta _{0}\right) ^{2}\int 1\left\{ 0<a\leq -\sqrt{N}   g_t^{\prime
}\left( \phi-\phi_0\right) -  h_t'\phi_0\right\} 1\left\{ h_t'\phi_0>0\right\}p_{u_t|\bigstar}(\frac{a}{\sqrt{N}})da\cr
&=
-  \mathbb E\left( x_{t}^{\prime }\delta _{0}\right) ^{2}p_{u_t|\bigstar}(0)    g_t ^{\prime }\left( \phi-\phi_0\right)1\{   g_t ^{\prime }\left( \phi-\phi_0\right)\leq 0\} 1\left\{
h_t'\phi_0\leq 0\right\}     \cr
&-  \mathbb E\left( x_{t}^{\prime }\delta _{0}\right) ^{2}p_{u_t|\bigstar}(0)  \left(   g_t^{\prime
}\left( \phi-\phi_0\right) + \frac{h_t'\phi_0}{\sqrt{N}}\right)1\left\{   g_t^{\prime
}\left( \phi-\phi_0\right) + \frac{h_t'\phi_0}{\sqrt{N}}<0\right\}1\left\{ h_t'\phi_0>0\right\}
\cr
&+ B_1,
\end{align}
where
\begin{align*}
B_1&=
\frac{ \mathbb E\left( x_{t}^{\prime }\delta _{0}\right) ^{2}}{\sqrt{N}}\int 1\left\{ -h_t'\phi_0<a\leq -\sqrt{N}   g_t ^{\prime }\left( \phi-\phi_0\right) - h_t'\phi_0\right\} 1\left\{
h_t'\phi_0\leq 0\right\} \left( p_{u_t|\bigstar}(\frac{a}{\sqrt{N}})-p_{u_t|\bigstar}(0)\right)da  \cr
&+\frac{1}{\sqrt{N}} \mathbb E\left( x_{t}^{\prime }\delta _{0}\right) ^{2}\int 1\left\{ 0<a\leq -\sqrt{N}   g_t^{\prime
}\left( \phi-\phi_0\right) -  h_t'\phi_0\right\} 1\left\{ h_t'\phi_0>0\right\}\left(p_{u_t|\bigstar}(\frac{a}{\sqrt{N}})-p_{u_t|\bigstar}(0)\right)da. 
\end{align*}
We now show that for some $C$ independent of $\gamma$,   $
|B_1|\leq   \frac{C}{NT^{2\varphi}}.
$
Because   $p_{u_t|\bigstar}(.)$ is Lipschitz,
\begin{eqnarray*}
|B_1|
&\leq&
\frac{C}{N} \mathbb E\left( x_{t}^{\prime }\delta _{0}\right) ^{2}\int 1\left\{ -h_t'\phi_0<a\leq -\sqrt{N}   g_t ^{\prime }\left( \phi-\phi_0\right) - h_t'\phi_0\right\} 1\left\{
h_t'\phi_0\leq 0\right\} |a| da\cr
&&+\frac{C}{N} \mathbb E\left( x_{t}^{\prime }\delta _{0}\right) ^{2}\int 1\left\{ 0<a\leq -\sqrt{N}   g_t^{\prime
}\left( \phi-\phi_0\right) -  h_t'\phi_0\right\} 1\left\{ h_t'\phi_0>0\right\} |a|da\cr
&\leq&
\frac{C'T^{-2\varphi}}{N} \mathbb E (
|\sqrt{N}   g_t ^{\prime }\left( \phi-\phi_0\right) + h_t'\phi_0| +|h_t'\phi_0|) ^2\leq \frac{C'}{N}T^{-2\varphi},
\end{eqnarray*}
due to Assumption \ref{as8} \ref{as8:itm7}.

\medskip
\noindent
\textbf{Calculating $A_2$}
\medskip

The calculation of $A_2$ is very similar to that of $A_1.$ Write
\begin{eqnarray*}
A_2
 &=&1\left\{ - \sqrt{N}   g_t'\left( \phi -\phi_0\right) - h_t'\phi_0<\sqrt{N}u_{t}\leq - h_t'\phi_0\right\} 1\left\{
h_t'\phi_0>0\right\} \\
&&+1\left\{ - \sqrt{N}  g_t'\left( \phi -\phi_0\right) - h_t'\phi_0<\sqrt{N}u_{t}\leq 0\right\} 1\left\{h_t'\phi_0 \leq 0\right\} \cr
&&+[  1\left\{ - \sqrt{N} \breve g_t'\left( \phi -\phi_0\right) - h_t'\phi_0<\sqrt{N}u_{t}\right\}
-1\left\{ - \sqrt{N}   g_t'\left( \phi -\phi_0\right) - h_t'\phi_0<\sqrt{N}u_{t}\right\} ]\cr
&&\quad \times
[1\left\{
h_t'\phi_0>0\right\} 1\left\{\sqrt{N}u_{t}\leq - h_t'\phi_0\right\}+1\left\{h_t'\phi_0 \leq 0\right\}
1\left\{u_{t}\leq 0\right\}   ]  .
\end{eqnarray*}
Hence,
\begin{align*}
 \mathbb  E\left( x_{t}^{\prime }\delta _{0}\right) ^{2}A_2&=
\mathbb E\left( x_{t}^{\prime }\delta _{0}\right) ^{2} 1\left\{ - \sqrt{N}   g_t'\left( \phi -\phi_0\right) - h_t'\phi_0<\sqrt{N}u_{t}\leq - h_t'\phi_0\right\} 1\left\{
h_t'\phi_0>0\right\} \\
&+\mathbb E\left( x_{t}^{\prime }\delta _{0}\right) ^{2}1\left\{ - \sqrt{N}  g_t'\left( \phi -\phi_0\right) - h_t'\phi_0<\sqrt{N}u_{t}\leq 0\right\} 1\left\{h_t'\phi_0 \leq 0\right\} + A_{21}\cr
A_{21}&:=
\mathbb E\left( x_{t}^{\prime }\delta _{0}\right) ^{2}
\left[ 1\left\{
h_t'\phi_0>0\right\} 1\left\{\sqrt{N}u_{t}\leq - h_t'\phi_0\right\}+1\left\{h_t'\phi_0 \leq 0\right\}
1\left\{u_{t}\leq 0\right\}   \right] \cr
&\quad \times
[  1\left\{ - \sqrt{N} \breve g_t'\left( \phi -\phi_0\right) - h_t'\phi_0<\sqrt{N}u_{t}\right\}
-1\left\{ - \sqrt{N}   g_t'\left( \phi -\phi_0\right) - h_t'\phi_0<\sqrt{N}u_{t}\right\} ]\cr
&\leq  \frac{CL}{NT^{2\varphi}},\quad \text{ similar to the bound of } A_{11}.
\end{align*}
So very similar to the bound of $\mathbb E\left( x_{t}^{\prime }\delta _{0}\right) ^{2}A_1 -A_{11}$, we have
\begin{eqnarray*}
&&\mathbb E\left( x_{t}^{\prime }\delta _{0}\right) ^{2}A_2 -A_{21}
\cr
&=& B_2+
\mathbb   E\left( x_{t}^{\prime }\delta _{0}\right) ^{2}     p_{u_t|\bigstar}(0)
   g_t'\left(\phi -\phi_0\right)1\left\{    g_t'\left(\phi-\gamma _{0}\right) >0\right\}  1\{h_t'\phi_0>0\}
\cr
&&+ \mathbb E\left( x_{t}^{\prime }\delta _{0}\right) ^{2}  p_{u_t|\bigstar}(0) \left(    g_t'\left(
\phi -\phi_0\right) +\frac{ h_t'\phi_0}{\sqrt{N}}\right)1\left\{    g_t'\left(
\phi -\phi_0\right) + \frac{h_t'\phi_0}{\sqrt{N}}> 0\right\} 1\left\{h_t'\phi_0
\leq 0\right\}
\end{eqnarray*}
with
$
|B_2|\leq   \frac{C}{NT^{2\varphi}}$.

\medskip
\noindent
\textbf{Calculating $A_3$}
\medskip

First we define events
\begin{eqnarray*}
	E_1&:=& \{\sqrt{N}g_{t}^{\prime }\phi _{0}\leq -\sqrt{N} \breve g_t'\left(\phi  -\phi _{0}\right) - h_t'\phi_0\}\cr
	E_2&:=&\{\sqrt{N}   g_t'\left(\phi  -\phi _{0}\right) +h_t'\phi_0>0\}\cr
	E_3&:=& \{\sqrt{N} \breve g_t'\left(\phi  -\phi _{0}\right) +h_t'\phi_0>0\}\cr
	E_4&:=&\{ \sqrt{N}g_{t}^{\prime }\phi _{0}   \leq -\sqrt{N}   g_t'\left(\phi  -\phi _{0}\right) - h_t'\phi_0\}\cr
	E_5&:=&\{0<\sqrt{N} g_t'\left(\phi  -\phi _{0}\right) +h_t'\phi_0<-h_t'(\phi-\phi_0)\}\cr
	E_6&:=&\{-h_t'(\phi-\phi_0)<\sqrt{N} g_t'\left(\phi  -\phi _{0}\right) +h_t'\phi_0<0\}
\end{eqnarray*}
Careful calcuations yield:
\begin{eqnarray*}
A_3 &=&1\left\{ \breve g_t'\phi \leq 0<\breve g_t'\phi_0\right\}
1\left\{ g_{t}^{\prime }\phi _{0}\leq 0<\breve{g}_{t}^{\prime }\phi_{0}\right\}\cr
      &=&
  1\left\{ -h_t'\phi_0<\sqrt{N}g_{t}^{\prime }\phi _{0}\leq  0\right\} 1\{\sqrt{N}  g_t'\left(\phi  -\phi _{0}\right) + h_t'\phi_0<0\}\cr
   &&+
    1\left\{ -h_t'\phi_0<\sqrt{N}g_{t}^{\prime }\phi _{0}\leq -\sqrt{N}   g_t'\left(\phi  -\phi _{0}\right) - h_t'\phi_0\right\}
    1\{E_2\}+ A_{31}\cr
A_{31}       &:=&
   [  1\left\{E_1\right\}
   +1\left\{  \sqrt{N}g_{t}^{\prime }\phi _{0}\leq  0\right\} ] 1\left\{ -h_t'\phi_0<\sqrt{N}g_{t}^{\prime }\phi _{0}\right\}
    [1\{E_3\}  -1\{E_2\}  ] \cr
    &&+
    1\{E_2\}    1\left\{ -h_t'\phi_0<\sqrt{N}g_{t}^{\prime }\phi _{0}\right\}  \left[  1\left\{ E_1\right\}
  - 1\left\{ E_4\right\}\right ].
\end{eqnarray*}
So
\begin{eqnarray*}
&&\mathbb E(x_t'\delta_0)^2A_3\cr
&=& \mathbb E(x_t'\delta_0)^2
    1\left\{ -h_t'\phi_0<\sqrt{N}g_{t}^{\prime }\phi _{0}\leq -\sqrt{N}   g_t'\left(\phi  -\phi _{0}\right) - h_t'\phi_0\right\}
    1\{E_2\}\cr
    &&+\mathbb E(x_t'\delta_0)^21\left\{ -h_t'\phi_0<\sqrt{N}g_{t}^{\prime }\phi _{0}\leq  0\right\} 1\{\sqrt{N}  g_t'\left(\phi  -\phi _{0}\right) + h_t'\phi_0<0\} +\mathbb E(x_t'\delta_0)^2A_{31}.
\end{eqnarray*}
Note that $\sqrt{N}\breve{g}_t= \sqrt{N} g_t +h_t$, so $|1\{E_3\}-1\{E_2\}|\leq 1\{E_5\}+1\{E_6\}.$ This gives, by Assumption \ref{as8} \ref{as8:itm7} and letting $ M_0=1 $ to simplify the notation,
\begin{eqnarray*}
    \mathbb E(x_t'\delta_0)^2A_{31}
  &\leq& T^{-2\varphi} \mathbb E [1\{E_5\}+1\{E_6\}] [1\left\{E_1\right\}
   +1\left\{  \sqrt{N}g_{t}^{\prime }\phi _{0}\leq  0\right\} ] 1\left\{ -h_t'\phi_0<\sqrt{N}g_{t}^{\prime }\phi _{0}\right\}
\cr  && + T^{-2\varphi} \mathbb E 1\left\{ -h_t'(\phi-\phi_0) - h_t'\phi_0< \sqrt{N}g_{t}^{\prime }\phi    \leq  - h_t'\phi_0\right\} \cr  && + T^{-2\varphi} \mathbb E 1\left\{  - h_t'\phi_0< \sqrt{N}g_{t}^{\prime }\phi <   -h_t'(\phi-\phi_0) - h_t'\phi_0   \right\}  \cr
&\leq&  T^{-2\varphi} \mathbb E 1\{E_5\}  \mathbb P \left\{  -h_t'\phi_0<\sqrt{N}g_{t}^{\prime }\phi _{0}\leq -\sqrt{N} \breve g_t'\left(\phi  -\phi _{0}\right) - h_t'\phi_0\bigg{|}  h_t, g_t'r \right\}
   \cr
   &&+T^{-2\varphi} \mathbb E 1\{E_5\}   \mathbb P\left\{ -h_t'\phi_0<\sqrt{N}g_{t}^{\prime }\phi _{0}<0\bigg{|}  h_t, g_t'r \right\}\cr
&&+   T^{-2\varphi} \mathbb E 1\{E_6\}  \mathbb P \left\{  -h_t'\phi_0<\sqrt{N}g_{t}^{\prime }\phi _{0}\leq -\sqrt{N} \breve g_t'\left(\phi  -\phi _{0}\right) - h_t'\phi_0\bigg{|}  h_t, g_t'r \right\}
   \cr
   &&+T^{-2\varphi} \mathbb E 1\{E_6\}   \mathbb P\left\{ -h_t'\phi_0<\sqrt{N}g_{t}^{\prime }\phi _{0}<0\bigg{|}  h_t,  g_t'r\right\}\cr
   \cr  && + T^{-2\varphi} \mathbb E 1\left\{ -h_t'(\phi-\phi_0) - h_t'\phi_0< \sqrt{N}g_{t}^{\prime }\phi    \leq  - h_t'\phi_0 \bigg{|}  h_t\right\} \cr  && + T^{-2\varphi} \mathbb E 1\left\{  - h_t'\phi_0< \sqrt{N}g_{t}^{\prime }\phi <   -h_t'(\phi-\phi_0) - h_t'\phi_0 \bigg{|}  h_t  \right\}
\cr
&\leq^{(1)}&  T^{-2\varphi} \mathbb E1\{E_5\}   C|\breve g_t'(\phi-\phi_0)|
 + T^{-2\varphi}\mathbb E \mathbb P\{E_5|h_t, x_t\}C|\frac{h_t'\phi_0}{\sqrt{N}}|\cr
&&+  T^{-2\varphi}  \mathbb E 1\{E_6\}      C|\breve g_t'(\phi-\phi_0)|  + T^{-2\varphi}  \mathbb E \mathbb P\{E_6|h_t, x_t\}      C|\frac{h_t'\phi_0}{\sqrt{N}}|\cr
   &&+ \mathbb E  C |\frac{h_t'(\phi-\phi_0)}{\sqrt{N}}| \cr
   &\leq^{(2)}&T^{-2\varphi} |\phi-\phi_0|_2C  (\mathbb  E[ |\breve g_t |] ^q)^{1/q} ( \mathbb E\mathbb P \{E_5|h_t \}  )^{1/p}
   \cr
&&+   T^{-2\varphi} |\phi-\phi_0|_2C ( \mathbb E[ |\breve g_t |] ^q)^{1/q}(  \mathbb E\mathbb P\{E_6|h_t\}     )^{1/p} \cr
&&+ T^{-2\varphi} C  \mathbb E|\frac{h_t'r}{\sqrt{N}}   |   |\frac{h_t'\phi_0}{\sqrt{N}}|+T^{-2\varphi} \mathbb E  C |\frac{h_t'(\phi-\phi_0)}{\sqrt{N}}| \cr
   &\leq^{(3)}& |\phi-\phi_0|_2C (\mathbb E |\frac{h_t'r}{\sqrt{N}}|)^{1/p}T^{-2\varphi}
+ T^{-2\varphi} C  \mathbb E|\frac{h_t'rh_t'\phi_0}{N}   |  + \mathbb E  C |\frac{h_t'(\phi-\phi_0)}{\sqrt{N}}| \cr
&\leq^{(4)}& O(\frac{1}{T^{2\varphi}N^{0.5+1/(2p)}})
\end{eqnarray*}
where inequality (1) follows from the assumption that the conditional density $p_{u_t|\bigstar}$ and the conditional density of $g_t'\phi$ given $( h_t)$ are bounded in a neighborhood of zero, with $r=|\phi-\phi_0|_2^{-1}(\phi-\phi_0)$; (2) (3) follow from
the Holder's inequality for some $p>1$ and $q>0$ and $p^{-1}+q^{-1}=1$, and
that the conditional density of $g_t'r$ given $( h_t)$ is bounded.  (We  take $p=1.5$.); (4) follows from $|\phi-\phi_0|_2<LN^{-1/2}$.

Also,
\begin{eqnarray}
&& \mathbb E(x_t'\delta_0)^2A_3 -\mathbb E(x_t'\delta_0)^2A_{31}
\cr
&=&
 \mathbb E(x_t'\delta_0)^2  \int     1\left\{ -h_t'\phi_0<a\leq  0\right\} 1\{\sqrt{N}     g_t'\left(\phi  -\phi _{0}\right) + h_t'\phi_0<0\}   p_{u_t|\bigstar}(\frac{a}{\sqrt{N}})d\frac{a}{\sqrt{N}}  \cr
&&+
 \mathbb E(x_t'\delta_0)^2  \int     1\left\{ -h_t'\phi_0<a\leq -\sqrt{N}     g_t'\left(\phi  -\phi _{0}\right) - h_t'\phi_0\right\} 1\{     g_t'\left(\phi  -\phi _{0}\right) +\frac{h_t'\phi_0}{\sqrt{N}}>0\}   \cr
 &&p_{u_t|\bigstar}(\frac{a}{\sqrt{N}})d\frac{a}{\sqrt{N}} \cr
&=&
 \mathbb E(x_t'\delta_0)^2  \int     1\left\{ -h_t'\phi_0<a\leq -\sqrt{N}     g_t'\left(\phi  -\phi _{0}\right) - h_t'\phi_0\right\} 1\{     g_t'\left(\phi  -\phi _{0}\right) +\frac{h_t'\phi_0}{\sqrt{N}}>0\}   p_{u_t|\bigstar}(0)d\frac{a}{\sqrt{N}} \cr
&&+
 \mathbb E(x_t'\delta_0)^2  \int     1\left\{ -h_t'\phi_0<a\leq  0\right\} 1\{\sqrt{N}     g_t'\left(\phi  -\phi _{0}\right) + h_t'\phi_0<0\}   p_{u_t|\bigstar}(0)d\frac{a}{\sqrt{N}}  -B_3\cr
&=&
- \mathbb Ep_{u_t|\bigstar}(0)(x_t'\delta_0)^2          g_t'\left(\phi  -\phi _{0}\right) 1\{\frac{h_t'\phi_0}{\sqrt{N}}>     g_t'\left(\phi  -\phi _{0}\right) +\frac{h_t'\phi_0}{\sqrt{N}}>0\}    \cr
&&+
 \mathbb Ep_{u_t|\bigstar}(0)(x_t'\delta_0)^2 \frac{h_t'\phi_0 }{\sqrt{N}}1\{ \sqrt{N}     g_t'\left(\phi  -\phi _{0}\right) + h_t'\phi_0<0\} 1\{h_t'\phi_0>0\}    -B_3,\cr
 \end{eqnarray}
 where
\begin{eqnarray*}
&&|B_3|\leq  \mathbb E(x_t'\delta_0)^2  \int     1\left\{ -h_t'\phi_0<a\leq  0\right\} 1\{\sqrt{N}     g_t'\left(\phi  -\phi _{0}\right) + h_t'\phi_0<0\}   [p_{u_t|\bigstar}(\frac{a}{\sqrt{N}})-p_{u_t|\bigstar}(0)]d\frac{a}{\sqrt{N}}  \cr
&&+C|
 \mathbb E(x_t'\delta_0)^2\frac{1}{N}  \int     1\left\{ -h_t'\phi_0<a\leq [-\sqrt{N}     g_t'\left(\phi  -\phi _{0}\right) - h_t'\phi_0]\right\} 1\{     g_t'\left(\phi  -\phi _{0}\right) >-\frac{h_t'\phi_0}{\sqrt{N}}\}    |a|da\cr
&\leq&  \frac{C}{N}
 \mathbb E(x_t'\delta_0)^2(|h_t'\phi_0| + |\sqrt{N}     g_t'\left(\phi  -\phi _{0}\right)|  )^2\leq  \frac{C}{NT^{2\varphi}}.
 \end{eqnarray*}

\medskip
\noindent
\textbf{Calculating $A_4$}
\medskip

Write
\begin{eqnarray*}
A_4&=&1\left\{
\breve g_t'\phi_0\leq 0< \breve g_t'\phi \right\}
1\left\{ \breve g_t'\phi_0\leq 0<g_t'\phi_0\right\}\cr
&=&
1\left\{0<g_t'\phi_0\leq  -\frac{h_{t}^{\prime }\phi _{0}}{\sqrt{N}}\right\}
 1\left\{    - \breve g_t'\left(\phi  -\phi _{0}\right) -\frac{h_{t}^{\prime }\phi _{0}}{\sqrt{N}}<g_{t}^{\prime }\phi _{0}\leq   -\frac{h_{t}^{\prime }\phi _{0}}{\sqrt{N}}
 \right\} \cr
  &=&
  1\left\{0<g_t'\phi_0\leq  -\frac{h_{t}^{\prime }\phi _{0}}{\sqrt{N}}\right\}
 1\{  \breve g_t'\left(\phi  -\phi _{0}\right) +\frac{h_{t}^{\prime }\phi _{0}}{\sqrt{N}}>0\}
 \cr
 &&
 1\left\{    - \breve g_t'\left(\phi  -\phi _{0}\right) -\frac{h_{t}^{\prime }\phi _{0}}{\sqrt{N}}<g_{t}^{\prime }\phi _{0}\leq   -\frac{h_{t}^{\prime }\phi _{0}}{\sqrt{N}}
 \right\}
 1\{  \breve g_t'\left(\phi  -\phi _{0}\right) +\frac{h_{t}^{\prime }\phi _{0}}{\sqrt{N}}<0\}
\end{eqnarray*}
 The same proof as that of $A_3$ shows
 \begin{eqnarray*}
&&  \mathbb E(x_t'\delta_0)^2A_4\cr
 &=&    \mathbb E(x_t'\delta_0)^2 (-h_t'\phi_0)1\{h_t'\phi_0<0\}
 1\{    g_t'\left(\phi  -\phi _{0}\right) +\frac{h_{t}^{\prime }\phi _{0}}{\sqrt{N}}>0\}   p_{u_t | \bigstar(0)}\frac{1}{\sqrt{N}}
\cr
&&+    \mathbb E(x_t'\delta_0)^2    g_t'\left(\phi  -\phi _{0}\right) 1\{   g_t'\left(\phi  -\phi _{0}\right)>0\} 1\{    g_t'\left(\phi  -\phi _{0}\right) +\frac{h_{t}^{\prime }\phi _{0}}{\sqrt{N}}<0\}   p_{u_t | \bigstar(0) }  \cr
&&+O(\frac{1}{T^{2\varphi}N^{0.5+1/(2p)}}).
\end{eqnarray*}

Combining the above results, we reach,
 \begin{eqnarray}
&& \mathbb E(x_t'\delta_0)^2(A_1-A_3+ A_2-A_4)=\sum_{d=1}^{8} \mathbb  E[(x_t'\delta_0)^2p_{u_t|\bigstar}(0) a_d] + O(\frac{1}{T^{2\varphi}N^{0.5+1/(2p)}}) \cr \label{e0.5}\end{eqnarray}
where
 \begin{eqnarray}
a_1&=&
-     \left(   g_t^{\prime
}\left( \phi-\phi_0\right) + \frac{h_t'\phi_0}{\sqrt{N}}\right)1\left\{   g_t^{\prime
}\left( \phi-\phi_0\right) + \frac{h_t'\phi_0}{\sqrt{N}}<0\right\}1\left\{ h_t'\phi_0>0\right\}\cr
a_2&=&-        g_t ^{\prime }\left( \phi-\phi_0\right)1\{   g_t ^{\prime }\left( \phi-\phi_0\right)\leq 0\} 1\left\{
h_t'\phi_0\leq 0\right\}    \cr
a_3&=&          g_t'\left(\phi  -\phi _{0}\right) 1\{\frac{h_t'\phi_0}{\sqrt{N}}>   g_t'\left(\phi  -\phi _{0}\right) +\frac{h_t'\phi_0}{\sqrt{N}}>0\}    \cr
a_4&=&-    \frac{h_t'\phi_0 }{\sqrt{N}}1\{ \sqrt{N}   g_t'\left(\phi  -\phi _{0}\right) + h_t'\phi_0<0\}  1\{h_t'\phi_0>0\}  \cr
a_5&=&
   g_t'\left( \phi-\phi_0\right)1\left\{    g_t'\left(\phi -\phi_0\right) >0\right\}  1\{h_t'\phi_0>0\}
\cr
a_6&=&      \left(    g_t'\left(
\phi  -\phi _{0}\right) +\frac{ h_t'\phi_0}{\sqrt{N}}\right)1\left\{    g_t'\left(
\phi  -\phi _{0}\right) + \frac{h_t'\phi_0}{\sqrt{N}}> 0\right\} 1\left\{h_t'\phi_0
\leq 0\right\}
\cr
a_7&=&
      \frac{h_t'\phi_0}{\sqrt{N}}     1\{h_t'\phi_0<0\}
 1\{    g_t'\left(\phi  -\phi _{0}\right) +\frac{h_{t}^{\prime }\phi _{0}}{\sqrt{N}}>0\}
\cr
a_8&=&-        g_t'\left(\phi  -\phi _{0}\right) 1\{   g_t'\left(\phi  -\phi _{0}\right)>0\} 1\{    g_t'\left(\phi  -\phi _{0}\right) +\frac{h_{t}^{\prime }\phi _{0}}{\sqrt{N}}<0\}  .\label{e0.5add}
  \end{eqnarray}

We now further simplify the above terms by paying  special attentions to terms involving $a_2$ and $a_5$:
 \begin{eqnarray}
&&-  { \mathbb E(x_t'\delta_0)^2p_{u_t | \bigstar}(0)    g_t ^{\prime }\left( \phi-\phi_0\right)1\{   g_t ^{\prime }\left( \phi-\phi_0\right)\leq 0\} 1\left\{
h_t'\phi_0\leq 0\right\}   } \label{e0.3}\\
&& {  \mathbb E(x_t'\delta_0)^2     p_{u_t | \bigstar}(0)
   g_t'\left( \phi-\phi_0\right)1\left\{    g_t'\left(\phi -\phi_0\right) >0\right\}  1\{h_t'\phi_0>0\} }
\label{e0.4}.
  \end{eqnarray}
The key idea is that $ 1\left\{
h_t'\phi_0\leq 0\right\} $ and $1\{h_t'\phi_0>0\}$   can be exchanged up to an error $O(\frac{T^{-2\varphi}}{N})$. Roughly speaking, this is due to the fact that  given $(x_t, g_t)$, the  conditional distribution of  $h_t'\phi_0$ is approximately normal, and symmetric around zero. The conditional normality of  $h_t'\phi_0$ follows from:
  for $\sigma^2_{h,x_t,g_t}:= \lim_{N\to\infty} \mathbb E(( h_t'\phi_0)^2|x_t, g_t) $, 
$$
 h_t'\phi_0=  \frac{1}{\sqrt{N}}\sum_{i=1}^Ne_{it} \lambda_i'\phi_0  (\frac{1}{N}\Lambda'\Lambda)^{-1}   |(x_t, g_t)\overset{d}{\longrightarrow } \mathcal Z_t
 $$
 where $\mathcal Z_t $ is a Gaussian variable, whose conditional distribution given $(x_t, g_t)$ is
$  \mathcal N(0, \sigma^2_{h,x_t,g_t})$.  For a formal treatment,   we show that $h_t'\phi_0$ in  (\ref{e0.3}) and (\ref{e0.4})  can be replaced with $\mathcal Z_t$.
Under the assumption of the lemma, we have
$$\sup_{x_t, g_t}| \mathbb P(h_t'\phi_0\leq 0|x_t, g_t) -1/2|=O(\frac{1}{\sqrt{N}}).$$

Then   for (\ref{e0.3}),  we have by Assumption \ref{as5} and \ref{as8}
\begin{eqnarray*}
&&\mathbb E(x_t'\delta_0)^2 p_{u_t | \bigstar}(0)  g_t ^{\prime }\left( \phi-\phi_0\right)1\{   g_t ^{\prime }\left( \phi-\phi_0\right)\leq 0\}   [1\left\{
h_t'\phi_0\leq 0\right\} -1\left\{
h_t'\phi_0> 0\right\} ]\cr
&=&\mathbb Ep_{u_t|\bigstar}(0)(x_t'\delta_0)^2   g_t ^{\prime }\left( \phi-\phi_0\right)1\{   g_t ^{\prime }\left( \phi-\phi_0\right)\leq 0\}   [ 1\left\{
h_t'\phi_0\leq 0\right\} - 1/2 ]\cr
&& + \mathbb Ep_{u_t|\bigstar}(0)(x_t'\delta_0)^2   g_t ^{\prime }\left( \phi-\phi_0\right)1\{   g_t ^{\prime }\left( \phi-\phi_0\right)\leq 0\}   [ 1\left\{
h_t'\phi_0> 0\right\} - 1/2 ]\cr
&\leq&O_P\left( \frac{1}{\sqrt{N}}\right) \mathbb E\left( p_{u_t=0|\bigstar}(0)( x_t'\delta_0)^2 |  g_t ^{\prime }\left( \phi-\phi_0\right) |\right)   \cr
&=&O(\frac{T^{-2\varphi}}{N}),\quad \text{ since }  |\phi-\phi_0|_2<LN^{-1/2}.
  \end{eqnarray*}
  Hence (\ref{e0.3}) can be replaced with $\mathbb E(x_t'\delta_0)^2p_{u_t | \bigstar}(0)a_2' + O(\frac{T^{-2\varphi}}{N})$, where
 $$
a_2'=  g_t ^{\prime }\left( \phi-\phi_0\right)1\{   g_t ^{\prime }\left( \phi-\phi_0\right)\leq 0\} 1\left\{
h_t'\phi_0> 0\right\}  .
 $$
     Similarly, (\ref{e0.4}) can be replaced with $\mathbb E(x_t'\delta_0)^2p_{u_t | \bigstar}(0)a_5' + O(\frac{T^{-2\varphi}}{N})$, where
      $$
a_5'=   g_t ^{\prime }\left( \phi-\phi_0\right)1\{   g_t ^{\prime }\left( \phi-\phi_0\right)> 0\} 1\left\{
h_t'\phi_0<0\right\}   .
 $$
Hence with a careful calculation, up to $O(\frac{1}{T^{2\varphi}N^{0.5+1/(2p)}})$ (which is uniform over $\phi$),  it can be shown that
 \begin{eqnarray} \label{e0.10add}
&& \mathbb E(x_t'\delta_0)^2(A_1-A_3+ A_2-A_4)   \cr
&=& \mathbb E(x_t'\delta_0)^2p_{u_t | \bigstar}(0) (a_1+a_2'+a_3+a_4+a_5'+a_6+a_7+a_8).
\cr
&=&
-2  \mathbb E(x_t'\delta_0)^2p_{u_t | \bigstar}(0)  \left(   g_t^{\prime
}\left( \phi-\phi_0\right) + \frac{h_t'\phi_0}{\sqrt{N}}\right)1\left\{   g_t^{\prime
}\left( \phi-\phi_0\right) + \frac{h_t'\phi_0}{\sqrt{N}}<0\right\}1\left\{ h_t'\phi_0>0\right\}\cr
&&+2 \mathbb E(x_t'\delta_0)^2  p_{u_t | \bigstar}(0) \left(    g_t'\left(
\phi  -\phi _{0}\right) +\frac{ h_t'\phi_0}{\sqrt{N}}\right)1\left\{    g_t'\left(
\phi  -\phi _{0}\right) + \frac{h_t'\phi_0}{\sqrt{N}}> 0\right\} 1\left\{h_t'\phi_0
\leq 0\right\}.
\cr
  \end{eqnarray}

Let
$$
R=-\frac{h_t'\phi_0}{\sqrt{N}   g_t^{\prime
}\left( \phi  -\phi _{0}\right)}.
$$

Recall that $\sqrt{N}|\phi -\phi_0|\leq L$. Fix any $M_0>0$, we choose $\epsilon>0$ so that when $|  g_t|_2<M_0$, then
 $
 |(1-\epsilon)\sqrt{N}  g_t'(\phi -\phi_0)|\leq (1-\epsilon)LM_0 $, so that $(1-\epsilon)\sqrt{N}  g_t'(\phi -\phi_0)$  is inside the neighborhood of zero on which the conditional density of $h_t'\phi_0$ given  $(  g_t, x_t)$  is bounded away from zero.
     Thus
  almost surely,
 $$
 \mathbb P\left\{   0<h_t'\phi_0< -(1-\epsilon)\sqrt{N} g_t'(\phi -\phi_0)| x_t, g_t\right\}  \geq
 c |\sqrt{N}  g_t'(\phi -\phi_0) |.
 $$
 So up to $O(\frac{1}{T^{2\varphi}N^{0.5+1/(2p)}})$, by Assumption \ref{as8},
  \begin{eqnarray*}
&& \mathbb E(x_t'\delta_0)^2(A_1-A_3+ A_2-A_4)  \cr
&= &
-2  \mathbb E(x_t'\delta_0)^2p_{u_t | \bigstar}(0)   g_t^{\prime
}\left( \phi-\phi_0\right) (1-R) 1\left\{   0<R<1\right\}  1\left\{   h_t'\phi_0 >0\right\}  \cr
&&+2 \mathbb E(x_t'\delta_0)^2p_{u_t | \bigstar}(0)   g_t'\left(
\phi  -\phi _{0}\right)   \left( 1-R\right)1\left\{ 0<R<1\right\} 1\left\{h_t'\phi_0
\leq 0\right\}
\cr
&\geq &
-2\epsilon T^{-2\varphi} \mathbb E   g_t^{\prime
}\left( \phi-\phi_0\right)    1\left\{   h_t'\phi_0 >0\right\}   1\left\{   0<R<1-\epsilon\right\}1\{|  g_t|_2<M_0\}  \cr
&&+2 \epsilon T^{-2\varphi}\mathbb E   g_t'\left(
\phi  -\phi _{0}\right)   1\left\{h_t'\phi_0
\leq 0\right\}  1\left\{   0<R<1-\epsilon\right\} 1\{|  g_t|_2<M_0\}
\cr
&\geq &
 2\epsilon T^{-2\varphi} \mathbb E      1\left\{   h_t'\phi_0 >0\right\}   1\{|  g_t|_2<M_0\}   c \sqrt{N} | g_t'(\phi -\phi_0) |^2
\cr
&&+2 \epsilon T^{-2\varphi} \mathbb E      1\left\{h_t'\phi_0
\leq 0\right\}    1\{|  g_t|_2<M_0\} c\sqrt{N} |  g_t'(\phi -\phi_0) |^2
\cr
&=&2c \epsilon T^{-2\varphi}\sqrt{N} \mathbb E     |  g_t'(\phi -\phi_0) |^2    1\{|  g_t|_2<M_0\} \cr
&\geq & CT^{-2\varphi} \sqrt{N} |\phi -\phi_0|_2^2,
  \end{eqnarray*}
  where the last ineqaulity follows since    the minimum eigenvalue of
$  \mathbb E  \left( x_{t}^{\prime }d _{0}\right) ^{2} g_tg_t' 1\{|g_t|_2<M_0\} $
is bounded away from zero. It then implies
$$
 \mathbb E(x_t'\delta_0)^2(A_1-A_3+ A_2-A_4)\geq C\sqrt{N}T^{-2\varphi} |\phi -\phi_0|_2^2 -O(\frac{1}{T^{2\varphi}N^{0.5+1/(2p)}}),\quad p=1.5.
$$
\end{proof}

\begin{pro}\label{rate:est}
Suppose $T=O(N)$, the first components of $\gamma_0,\widehat\gamma$ are one.
$$
 |\widehat\phi-\phi_0|_2 \leq O_P\left(\frac{1}{T^{1-2\varphi}}+\frac{1}{\left( NT^{1-2\varphi }\right) ^{1/3}}\right).
$$

\end{pro}

\begin{proof}

Proposition \ref{p:pre} shows $ |\widehat\phi-\phi_0|_2=O_P(T^{-(1-2\varphi)}+ N^{-1/2}).$
 When $T^{1-2\varphi}=O(\sqrt{N})$, the above upper bound leads to
 \begin{equation}\label{e6.6}
  |\widehat\phi-\phi_0|_2 \leq O_P(\frac{1}{T^{1-2\varphi}})   .
 \end{equation}
When $\sqrt{N}=O(T^{1-2\varphi} )$, the above upper bound leads to
$ |\phi-\phi_0|_2 \leq O_P( \frac{1}{\sqrt{N}})   .
$
We now improve this bound in the case $\sqrt{N}=O(T^{1-2\varphi} )$. In this case,
For an arbitrarily small $\epsilon>0$, there is $C_e>0$, with probability at least $1-\epsilon$, $ |\phi-\phi_0|_2 \leq   \frac{C_e}{\sqrt{N}}   .
$ We now proceed the argument conditioning on this event. We use the lower bound in Lemma \ref{l:generic} for $ \mathds G_1(\phi)= \mathbb E\left( x_{t}^{\prime }\delta _{0}\right) ^{2}\left( A_{1t}\left( \phi
\right) +A_{2t}\left( \phi \right) -A_{3t}\left( \phi \right)
-A_{4t}\left( \phi \right) \right).$

If $\widehat \phi-\phi_0$ is linearly dependent of $\phi_0$, there is a scalar $c_T$ so that $\widehat\phi-\phi_0=c_T\phi_0$, implying $\widehat\phi=(1+c_T)\phi_0$. Let $(v)_1$ denote the first component of a vector $v$. Then   $1=(H_T^{-1}\widehat\phi)_1=(H_T^{-1}\phi_0)_1(1+c_T)=1+c_T$, implying $c_T=0$. Hence $\widehat\phi=\phi_0$. Hence  we only need to focus on  the case that $\widehat\phi$ is linearly independent of $\phi_0$.
Then Lemma \ref{l:generic}
yields,  for $p=1.5$
$$
  \mathds G_1(\widehat\phi)
\geq CT^{-2\varphi} \sqrt{N} |\widehat\phi-\phi_0|_2^2 - O(\frac{1}{T^{2\varphi}N^{5/6}}).
$$Write
$$
m_{NT}:=T^{-2\varphi}
  \frac{\sqrt{N}}{\left( NT^{1-2\varphi }\right) ^{2/3}}.
$$

Substitute to (\ref{e:prer}), there are $C_1, C_2, C_3>0$,
\begin{eqnarray*}
&&C|\widehat\alpha-\alpha_0|_2^2
+ CT^{-2\varphi} \sqrt{N} |\widehat\phi-\phi_0|_2^2\cr
&\leq& \mathds G_2(\widehat\phi)+\mathds C_{1}( \delta_0, \widehat\phi)
+ O_P(\Delta_f+T^{-6})T^{-\varphi}
 + O_P(\Delta_f+T^{-1/2}+ T^{-\varphi}N^{-1/2}) |\widehat\alpha-\alpha_0|_2
\cr
&&+
 C_1T^{-\varphi }\left|\widehat\phi -\phi_{0}\right\vert _{2}
\left\vert \widehat \alpha -\alpha _{0}\right\vert _{2} +O(\frac{1}{T^{2\varphi}N^{5/6}}).
\end{eqnarray*}
Next,  replaced $ \mathds G_2$ and $\mathds C_{1}$ with their upper bound based on $a_{NT}$ given in Lemma \ref{l.1}.  In addition,
$C_1T^{-\varphi }\left|\widehat\phi -\phi_{0}\right\vert _{2}
\left\vert \widehat \alpha -\alpha _{0}\right\vert _{2} \leq C_1^2T^{-2\varphi}|\widehat\phi-\phi_0|_2^2N^{1/4} + |\widehat\alpha-\alpha_0|_2^2N^{-1/4}.
$
Also note that  $\frac{1}{T^{2\varphi}N^{5/6}}=O(m_{NT})$ as $T=O(N)$,  and $ T^{-1}=O\left(m_{NT}\right) $ when $\sqrt{N}=O(T^{1-2\varphi} )$. 
\begin{eqnarray*}
&&   C|\widehat\alpha-\alpha_0|_2^2/2
+ CT^{-2\varphi} \sqrt{N} |\widehat\phi-\phi_0|_2^2 /2
\cr
&\leq&  O_P(T^{-1/2}+ \Delta_f+ T^{-\varphi} N^{-1/2} )|\widehat\alpha-\alpha_0|_2
+O_P(\Delta_f+T^{-6})T^{-\varphi}+
O_P\left(   m_{NT} \right)\cr
&\leq&  O_P(T^{-1/2}+\Delta_f)|\widehat\alpha-\alpha_0|_2
+  O_P(m_{NT}  +\Delta_f T^{-\varphi}).
\end{eqnarray*}
This implies
$|\widehat\alpha-\alpha_0|_2^2
\leq O_P(m_{NT} +\Delta_fT^{-\varphi} )  $ with $T^{\varphi}\log^\kappa T=O(N)$ for any $\kappa>0$.
  Hence
  \begin{eqnarray*}
T^{-2\varphi} \sqrt{N} |\widehat\phi-\phi_0|_2^2 &\leq&
   O_P(m_{NT} +   T^{-1/2}\Delta_f^{1/2}T^{-\varphi/2}+\Delta_f\sqrt{m_{NT}} + \Delta_f^{3/2}T^{-\varphi/2}+\Delta_fT^{-\varphi}) \cr
   &\leq&O_P(m_{NT})
\end{eqnarray*}
where in the second inequality we assumed $T=O(N)$. 

Hence
$$
|\widehat\phi-\phi_0|_2^2=O_P(T^{2\varphi}N^{-1/2} m_{NT})=O_P\left( \frac{1}{\left( NT^{1-2\varphi }\right) ^{1/3}}\right) ^2.
$$

 Combining with (\ref{e6.6}), we reach
$$
  |\widehat\phi-\phi_0|_2 \leq O_P\left(\frac{1}{T^{1-2\varphi}}+\frac{1}{\left( NT^{1-2\varphi }\right) ^{1/3}}\right).
$$

\end{proof}

\subsection{Consistency of Regime Classification (Proof of Theorem \ref{thm:classification})}

\begin{proof}[Proof of Theorem \ref{thm:classification}]
To begin with, we consider the case of observed factors, $\widehat{f}_{t}=g_{t}$%
, for which we have $\phi _{0}=\gamma _{0}$ and $\widehat{\gamma}-\gamma
_{0}=O_P\left( T^{-1+2\varphi }\right) .$
Then, it suffices to show that
\[
\sup_{\left\vert \gamma -\gamma _{0}\right\vert \leq CT^{-1+2\varphi
}} \frac{1}{T}\sum_{t=1}^{T}
\left\vert 1\left\{
g_{t}^{\prime }\gamma >0\right\} -1\left\{ g_{t}^{\prime }\gamma
_{0}>0\right\} \right\vert
 =O_P\left( T^{-1+2\varphi }\right), \]
for any $C<\infty $.
It follows by noting that for any $\gamma $ satisfying the
normalization of $\gamma _{1}=1$ and for some finite $c$,
\begin{eqnarray*}
&& \mathbb{E}\left\vert 1\left\{
g_{t}^{\prime }\gamma >0\right\} -1\left\{ g_{t}^{\prime }\gamma
_{0}>0\right\} \right\vert \\
  &=&\mathbb{E} \mathbb{P} \left[ \left( g_{2t}^{\prime }\gamma
_{20}<-g_{1t}\leq g_{2t}^{\prime }\gamma _{2}\right) |g_{1t}\right] +\mathbb{E} \mathbb{P}\left[
\left( g_{2t}^{\prime }\gamma _{20}\geq -g_{1t}>g_{2t}^{\prime }\gamma
_{2}\right) |g_{1t}\right]  \\
&\leq &c\mathbb{E}\left\vert g_{2t}^{\prime }\left( \gamma _{2}-\gamma _{20}\right)
\right\vert  \\
&=&O\left( \left\vert \gamma -\gamma _{0}\right\vert _{2}\right),
\end{eqnarray*}%
and
\begin{eqnarray*}
&&\sup_{\left\vert \gamma -\gamma _{0}\right\vert _{2}\leq CT^{-1+2\varphi
}}\left\vert \frac{1}{T}\sum_{t=1}^{T}\left( \left\vert 1\left\{
g_{t}^{\prime }\gamma >0\right\} -1\left\{ g_{t}^{\prime }\gamma
_{0}>0\right\} \right\vert -\mathbb{E}\left\vert 1\left\{ g_{t}^{\prime }\gamma
>0\right\} -1\left\{ g_{t}^{\prime }\gamma _{0}>0\right\} \right\vert
\right) \right\vert  \\
&=&O_P\left( T^{-1+\varphi }\right)
\end{eqnarray*}%
by the maximal inequality in Lemma \ref{Lem:modul1} and the subsequent
remark.

Next, we move to the case of estimated factors. Recall that $\widehat{f}_{t}=H_{T}^{\prime }g_{t}+H_{T}h_{t}/\sqrt{N}$. By the triangle inequality, for any $%
\gamma $
\begin{eqnarray}
\frac{1}{T}\sum_{t=1}^{T}\left\vert 1\left\{ \widetilde{f}_{t}^{\prime }\gamma
>0\right\} -1\left\{ g_{t}^{\prime }\phi _{0}>0\right\} \right\vert  &\leq &%
\frac{1}{T}\sum_{t=1}^{T}\left\vert 1\left\{ \widehat{f}_{t}^{\prime }\gamma
>0\right\} -1\left\{ \widetilde{f}_{t}^{\prime }\gamma >0\right\} \right\vert
\label{eq:classificaiton01} \\
&&+\frac{1}{T}\sum_{t=1}^{T}\left\vert 1\left\{ \widehat{f}_{t}^{\prime }\gamma
_{0}>0\right\} -1\left\{ \widehat{f}_{t}^{\prime }\gamma >0\right\} \right\vert
\nonumber \\
&&+\frac{1}{T}\sum_{t=1}^{T}\left\vert 1\left\{ \widehat{f}_{t}^{\prime }\gamma
_{0}>0\right\} -1\left\{ g_{t}^{\prime }\phi _{0}>0\right\} \right\vert .
\nonumber
\end{eqnarray}%
Proceeding similarly as the case of the observed factors, we get%
\[
\frac{1}{T}\sum_{t=1}^{T}\left\vert 1\left\{ \widehat{f}_{t}^{\prime }\gamma
_{0}>0\right\} -1\left\{ \widehat{f}_{t}^{\prime }\widehat{\gamma}>0\right\}
\right\vert =O_P\left( \frac{\sqrt{\left\vert \widehat{\gamma}-\gamma
_{0}\right\vert _{2}}}{\sqrt{T}}+\left\vert \widehat{\gamma}-\gamma
_{0}\right\vert _{2}\right)
\]%
and%
\begin{eqnarray*}
\frac{1}{T}\sum_{t=1}^{T}\left\vert 1\left\{ \widehat{f}_{t}^{\prime }\gamma
_{0}>0\right\} -1\left\{ g_{t}^{\prime }\phi _{0}>0\right\} \right\vert  &=&%
\frac{1}{T}\sum_{t=1}^{T}\left\vert 1\left\{ g_{t}^{\prime }\phi
_{0}>-h_{t}^{\prime }\phi _{0}/\sqrt{N}\right\} -1\left\{ g_{t}^{\prime
}\phi _{0}>0\right\} \right\vert  \\
&=&O_P\left( \frac{1}{\sqrt{N}}\right) .
\end{eqnarray*}%
For the remaining term in (\ref{eq:classificaiton01}), note that
\begin{eqnarray*}
&&\sup_{\gamma }|\frac{1}{T}\sum_{t=1}^{T}\left\vert 1\{\widetilde{f}%
_{t}^{\prime }\gamma >0\}-1\{\widehat{f}_{t}^{\prime }\gamma >0\}\right\vert
\\
&\leq &\sup_{\gamma }\frac{1}{T}\sum_{t=1}^{T}1\{\widehat{f}_{t}^{\prime
}\gamma <0<\widetilde{f}_{t}^{\prime }\gamma \}+\sup_{\gamma }\frac{1}{T}%
\sum_{t=1}^{T}1\{\widetilde{f}_{t}^{\prime }\gamma <0<\widehat{f}%
_{t}^{\prime }\gamma \}
\end{eqnarray*}%
and that
\begin{eqnarray}
&&\sup_{\left\vert \gamma \right\vert _{2}\leq C}\frac{1}{T}\sum_{t=1}^{T}1\{%
\widehat{f}_{t}^{\prime }\gamma <0<\widetilde{f}_{t}^{\prime }\gamma \}
\label{e7.9-additional} \\
&=&\sup_{\left\vert \gamma \right\vert _{2}\leq C}\frac{1}{T}%
\sum_{t=1}^{T}1\{-|\widehat{f}_{t}-\widetilde{f}_{t}|_{2}C<\widehat{f}%
_{t}^{\prime }\gamma <0\}  \nonumber \\
&\leq &\sup_{\left\vert \gamma \right\vert _{2}\leq C}\frac{1}{T}%
\sum_{t=1}^{T}1\left\{ \left\vert \widehat{f}_{t}^{\prime }\gamma
\right\vert <C\Delta _{f}\right\} +\frac{1}{T}\sum_{t=1}^{T}1\{|\widehat{f}%
_{t}-\widetilde{f}_{t}|_{2}\geq \Delta _{f}\}  \nonumber \\
&\leq &\frac{1}{T}\sum_{t=1}^{T}1\left\{ \inf_{\left\vert \gamma \right\vert
_{2}\leq C}\left\vert \widehat{f}_{t}^{\prime }\gamma \right\vert <C\Delta
_{f}\right\} +O_P\left( 1\right) \mathbb{P}\{|\widehat{f}_{t}-\widetilde{f%
}_{t}|_{2}\geq \Delta _{f}\}  \nonumber \\
&\leq &O_P\left( 1\right) \mathbb{P}\left( \inf_{\left\vert \gamma
\right\vert _{2}\leq C}|\widehat{f}_{t}^{\prime }\gamma |<C\Delta
_{f}\right) +O_P\left( T^{-6}\right)   \nonumber \\
&\leq &O_{P}(\Delta _{f}+T^{-6}),  \nonumber
\end{eqnarray}%
where the first inequality is by the fact that $1\left\{ A\right\} 1\left\{
B\right\} \leq 1\left\{ A\right\} $ for any events $A$ and $B,$ and the
remaining inequalities are by the law of iterated expectations, the rank
condition in Assumption \ref{as9}, and Proposition \ref{prop:f}. Recall in Proposition \ref{prop:f}  that
notation $\Delta _{f}$ is introduced and  $\Delta _{f}=O\left( T^{-1+2\varphi }\right) $ for any $\varphi >0.$

Putting together, and recalling that $\widehat{\gamma}-\gamma _{0}=O_P\left(
\left( NT^{1-2\varphi }\right) ^{-1/3}+T^{-1+2\varphi }\right) ,$ we
conclude that
\[
\sup_{\left\vert \gamma -\gamma _{0}\right\vert \leq CT^{-1+2\varphi
}} \frac{1}{T}\sum_{t=1}^{T}
\left\vert 1\left\{
\widehat{f}_{t}^{\prime }\gamma >0\right\} -1\left\{ f_{t}^{\prime }\gamma
_{0}>0\right\} \right\vert
=O_P\left( T^{-1+2\varphi }\right). \]

\end{proof}

Proof of Theorem \ref{thm:AD with Estimated f} is divided into two subsections, one for the derivation of the asymptotic distribution of $ \widehat{\alpha} $ and  the other for the derivation of the asymptotic distribution of $ \widehat{\gamma} $. The latter will contain the asymptotic independence proof as well.

\subsection{Limiting distribution of $\widehat{\protect\alpha}$ (Proof of Theorem \ref{thm:AD with Estimated f}: Part I)}\label{sec:alpha-hat}

Recall the notation that
$\widehat{Z}_{t}(\gamma )=(x_{t}^{\prime },x_{t}^{\prime }1\{\widehat{f}_{t}^{\prime }\gamma >0\})^{\prime }$,
$\widetilde{Z}_{t}(\gamma)=(x_{t}^{\prime },x_{t}^{\prime }1\{\widetilde{f}_{t}^{\prime }\gamma >0\})^{\prime }$
and
$Z_{t}(\gamma )=(x_{t}^{\prime },x_{t}^{\prime}1\{f_{t}^{\prime }\gamma >0\})^{\prime }$.
In this subsection, define
$A=(\frac{1}{T}\sum_{t}\widetilde{Z}_{t}(\widehat{\gamma })\widetilde{Z}%
_{t}(\widehat{\gamma })^{\prime} )^{-1}.
$
Then write
\begin{align*}
\widehat{\alpha }&= \left[ \frac{1}{T}\sum_{t}\widetilde{Z}_{t}(\widehat{\gamma })%
\widetilde{Z}_{t}(\widehat{\gamma })^{\prime} \right]^{-1}\frac{1}{T}\sum_{t}%
\widetilde{Z}_{t}(\widehat{\gamma })y_{t}\cr
&=\alpha _{0}+(\frac{1}{T}%
\sum_{t}Z_{t}(\gamma _{0})Z_{t}(\gamma _{0})^{\prime} )^{-1}\frac{1}{T}%
\sum_{t}Z_{t}(\gamma _{0})\varepsilon _{t}+\sum_{l=1}^{5}a_{l},
\end{align*}
where
\begin{align*}
 a_{1}&=A\frac{1}{T}\sum_{t}\widetilde{Z%
}_{t}(\widehat{\gamma })[Z_{t}(\gamma _{0})-\widetilde{Z}_{t}(\gamma
_{0})]^{\prime }\alpha _{0},\cr a_{2}&=A\frac{1}{T}\sum_{t}\widetilde{Z}_{t}(%
\widehat{\gamma })[\widetilde{Z}_{t}(\gamma _{0})-\widetilde{Z}_{t}(\widehat{%
	\gamma })]^{\prime }\alpha _{0},\cr a_{3}&=A\frac{1}{T}\sum_{t}[\widetilde{Z}%
_{t}(\widehat{\gamma })-\widetilde{Z}_{t}(\gamma _{0})]\varepsilon _{t},\cr %
a_{4}&=A\frac{1}{T}\sum_{t}[\widetilde{Z}_{t}(\gamma _{0})-Z_{t}(\gamma
_{0})]\varepsilon _{t},\cr a_{5}&= \left[ A- \left(\frac{1}{T}\sum_{t}Z_{t}(\gamma
_{0})Z_{t}(\gamma _{0})^{\prime} \right)^{-1} \right]\frac{1}{T}\sum_{t}Z_{t}(\gamma
_{0})\varepsilon _{t}.
\end{align*}
In view of Lemma \ref{l:relf}, the fact that $\mathbb{P}(|\widetilde
f_t-\widehat f_t|_2>C\Delta_f)\leq O(T^{-6}) $ implies $A- (\frac{1}{T}%
\sum_t Z_t( \gamma_0) Z_t( \gamma_0) ^{\prime} )^{-1}=o_P(1)$, since $%
\widehat\gamma-\gamma_0=o_P(1)$ and a ULLN applies. Hence $A=O_P(1)$ and $%
a_5=o_P(T^{-1/2})$ by the MDS CLT. Furthermore, Lemma \ref{l7.7} below
implies $\sqrt{T}\sum_{l=1}^4a_l=o_P(1)$. Hence
\begin{equation*}
\sqrt{T}(\widehat\alpha- \alpha_0)= (\frac{1}{T}\sum_t Z_t( \gamma_0) Z_t(
\gamma_0) ^{\prime} )^{-1} \frac{1}{\sqrt{T}}\sum_t Z_t(\gamma_0) \varepsilon_t
+o_P(1).
\end{equation*}
This leads to the desired strong oracle limiting distribution.

Define
\begin{equation}
r_{NT}:=\left( NT^{1-2\varphi }\right) ^{1/3}\wedge T^{1-2\varphi }. \label{eq:r_NT}
\end{equation}%


\begin{lem}
	\label{l7.7} Suppose that $T=O(N)$, the conditional density of $%
	f_t^{\prime }\gamma_0$ given $h_t,x_t$ is bounded a.s. and the density of $%
	\inf_{\gamma \in \Gamma_T} |(g_t+h_tN^{-1/2})^{\prime }\gamma|$ is bounded,
	where $\Gamma_T $ is a $r_{NT}^{-1} $-neighborhood of $\gamma_{0} $. Then,
	\newline
	(i) $\frac{1}{T}\sum_t\widetilde Z_t
	(\widehat\gamma)[Z_t(\gamma_0)-\widetilde Z_t( \gamma_0)]^{\prime }\alpha_0
	= o_P(T^{-1/2})$,\newline
	(ii) $\frac{1}{T}\sum_t\widetilde Z_t (\widehat\gamma)[\widetilde
	Z_t(\gamma_0)-\widetilde Z_t(\widehat \gamma)]^{\prime }\alpha_0 =
	o_P(T^{-1/2})$,\newline
	(iii) $\frac{1}{T}\sum_t[ \widetilde Z_t( \widehat\gamma) - \widetilde
	Z_t(\gamma_0)] \varepsilon_t=o_P(T^{-1/2}),$\newline
	(iv) $\frac{1}{T}\sum_t[ \widetilde Z_t(\gamma_0) - Z_t(\gamma_0) ]
	\varepsilon_t =o_P(T^{-1/2})$.
\end{lem}

\begin{proof}[Proof of Lemma \ref{l7.7}]
	(i) For each $j $,
	\begin{align*}
		&\left| \frac{1}{T}\sum_t\widetilde Z_{jt}
		(\widehat\gamma)[Z_t(\gamma_0)-\widetilde Z_t( \gamma_0)]^{\prime }\alpha_0
		\right| \\
		&=\left| \frac{1}{T}\sum_t\widetilde Z_{jt} (\widehat\gamma)
		x_t^{\prime }\delta_0(1\{f_t^{\prime }\gamma_0>0\}-1\{\widetilde f_t^{\prime
		}\gamma_0>0\})\right| \cr &\leq \frac{|\delta_0|_2 }{T}\sum_t|x_t|_2^2
		|1\{f_t^{\prime }\gamma_0>0\}-1\{\widetilde f_t^{\prime }\gamma_0>0\}| \cr %
		&\leq \frac{|\delta_0|_2 }{T}\sum_t|x_t|_2^2 1\{-|f_t-\widetilde f_t|
		_2|\gamma_0|_2<f_t^{\prime }\gamma_0<0 \} + \frac{|\delta_0|_2 }{T}%
		\sum_t|x_t|_2^2 1\{0<f_t^{\prime }\gamma_0<|f_t-\widetilde f_t|
		_2|\gamma_0|_2 \}.
	\end{align*}
	We bound the first term on the right hand side, and the second term follows
	from a similar argument. In view of Lemma \ref{l:relf} and the boundedness
	of the conditional density of $f_t^{\prime }\gamma_{0} $,
	\begin{align*}
		& \frac{|\delta_0|_2 }{\sqrt{T}}\sum_t|x_t|_2^2 1\{-|f_t-\widetilde f_t|
		_2|\gamma_0|_2<f_t^{\prime }\gamma_0<0 \} \cr &\leq \frac{C }{%
			T^{1/2+\varphi}}\sum_t|x_t|_2^2 1\{-C(\Delta_f+|\frac{h_t}{\sqrt{N}}%
		|_2)<f_t^{\prime }\gamma_0<0 \} +\frac{C }{T^{1/2+\varphi}}%
		\sum_t|x_t|_2^21\{|\widetilde f_t-\widehat f_t|>C\Delta_f\}\cr &\leq
		O_P(T^{1/2-\varphi} ) \mathbb E\left( |x_t|_2^2 \mathbb{P}\{-C(\Delta_f+|\frac{h_t}{%
			\sqrt{N}}|_2)<f_t^{\prime }\gamma_0<0 \left|h_t, x_t \right. \} \right)
		+o_P(1)\cr &\leq O_P(T^{1/2-\varphi} ) \left( \Delta_f  \mathbb E\left( |x_t|_2^2
		\right) + \mathbb   E|x_t|_2^2 |h_t|_2 \frac{1}{\sqrt{N}}\right) +o_P(1) \\
		&=o_P(1),
	\end{align*}
	provided that $T=O(N)$. Hence $\frac{1}{T}\sum_t\widetilde Z_t
	(\widehat\gamma)[Z_t(\gamma_0)-\widetilde Z_t( \gamma_0)]
	\alpha_0=o_P(T^{-1/2}).$

	(ii) For each $j$,
	\begin{align*}
	&\left\vert \frac{1}{T}\sum_{t}\widetilde{Z}_{jt}(\widehat{\gamma })[%
	\widetilde{Z}_{t}(\gamma _{0})-\widetilde{Z}_{t}(\widehat{\gamma })]\alpha
	_{0}\right\vert  \\
&\leq \frac{|\delta
		_{0}|_{2}}{T}\sum_{t}|x_{t}|^{2}1\{0<\widehat{f}_{t}^{\prime }\gamma _{0}<|%
	\widehat{f}_{t}|_{2}|\gamma _{0}-\widehat{\gamma }|_{2}\}+\sup_{\gamma \in
		\Gamma _{T}}\frac{2\sqrt{2}|\delta _{0}|_{2}}{T}\sum_{t}|x_{t}|^{2}1\{0<%
	\widehat{f}_{t}^{\prime }\gamma <|\widehat{f}_{t}-\widetilde{f}%
	_{t}|_{2}|\gamma |_{2}\}\cr
	&+\frac{|\delta _{0}|_{2}}{T}\sum_{t}|x_{t}|^{2}1%
	\{-|\widehat{f}_{t}|_{2}|\gamma _{0}-\widehat{\gamma }|_{2}<\widehat{f}%
	_{t}^{\prime }\gamma _{0}<0\}\cr
	&+\sup_{\gamma \in \Gamma _{T}}\frac{2|\delta
		_{0}|_{2}}{T}\sum_{t}|x_{t}|^{2}1\{-|\widehat{f}_{t}-\widetilde{f}%
	_{t}|_{2}|\gamma |_{2}<\widehat{f}_{t}^{\prime }\gamma <0\}.
	\end{align*}%
	We bound the first two terms on the right hand side; the other two terms can be bounded similarly and thus details are omitted.
	Note that with probability at least $1-o(T^{-1})$, there
	is $c>0$, uniformly in $t$,
	\begin{equation}\label{e9.21}
	|\widehat{f}_{t}|_{2}\leq |H_{T}g_{t}|_{2}+|H_{T}h_{t}|_{2}N^{-1/2}<c(\log
	T)^{c}.
	\end{equation}%
	Moreover, for any $\epsilon >0$, $\mathbb{P}\left\{ |\widehat{\gamma }-\gamma
	_{0}|_{2}>\epsilon r_{NT}^{-1}\log T\right\} \rightarrow 0.$ Thus
	\begin{eqnarray*}
		&&\sqrt{T}\frac{|\delta _{0}|_{2}}{T}\sum_{t}|x_{t}|^{2}1\{0<\widehat{f}%
		_{t}^{\prime }\gamma _{0}<|\widehat{f}_{t}|_{2}|\gamma _{0}-\widehat{\gamma }%
		|_{2}\} \\
		&=&\frac{|\delta _{0}|_{2}}{\sqrt{T}}\sum_{t}|x_{t}|^{2}1\{0<\widehat{f}%
		_{t}^{\prime }\gamma _{0}<c(\log T)^{c}|\gamma _{0}-\widehat{\gamma }%
		|_{2}\}+o_P\left( 1\right) \\
		&=&\frac{|\delta _{0}|_{2}}{\sqrt{T}}\sum_{t}|x_{t}|^{2}1\{0<\widehat{f}%
		_{t}^{\prime }\gamma _{0}<c(\log T)^{c+1}\epsilon r_{NT}^{-1}\}+o_P\left(
		1\right) .
	\end{eqnarray*}
However, due to the boundedness of the conditional density of $\widehat{f}%
	_{t}^{\prime }\gamma _{0},$
	\begin{eqnarray*}
		&&\mathbb{E}\frac{|\delta _{0}|_{2}}{\sqrt{T}}\sum_{t}|x_{t}|^{2}1\left\{ 0<%
		\widehat{f}_{t}^{\prime }\gamma _{0}<c^{\prime }(\log T)^{c+1}r_{NT}^{-1}\right\}
		\\
		&\leq &T^{1/2-\varphi }\mathbb{E}\left[ \mathbb{P}\left\{ (0<\widehat{f}%
		_{t}^{\prime }\gamma _{0}<c(\log T)^{c+1}\epsilon r_{NT}^{-1})|x_{t}\right\}
		|x_{t}|^{2}\right] \\
		&\leq &C\epsilon T^{1/2-\varphi }(\log T)^{c+1}r_{NT}^{-1}\mathbb{E}%
		|x_{t}|^{2}\rightarrow 0\text{ so long as }T^{1-2\varphi }(\log
		T)^{6c+1}=o(N^{2}).
	\end{eqnarray*}
It remains to show $\sqrt{T}\sup_{\gamma \in \Gamma _{T}}\frac{2\sqrt{2}%
		|\delta _{0}|_{2}}{T}\sum_{t}|x_{t}|^{2}1\{0<\widehat{f}_{t}^{\prime }\gamma
	<|\widehat{f}_{t}-\widetilde{f}_{t}|_{2}|\gamma |_{2}\}=o_{P}(1)$, which is
	similar to the proof of (i) due to the boundedness of $\gamma $ and thus
	details are omitted.

Note that
	\begin{eqnarray*}
	&&\sqrt{T}\sup_{\gamma \in \Gamma _{T}}\frac{2\sqrt{2}|\delta _{0}|_{2}}{T}%
	\sum_{t}|x_{t}|^{2}1\{0<\widehat{f}_{t}^{\prime }\gamma <|\widehat{f}_{t}-%
	\widetilde{f}_{t}|_{2}|\gamma |_{2}\}\cr
	&\leq& \sqrt{T}\sup_{\gamma \in \Gamma
		_{T}}\frac{2\sqrt{2}|\delta _{0}|_{2}}{T}\sum_{t}|x_{t}|^{2}1\{0<\widehat{f}%
	_{t}^{\prime }\gamma <C\Delta _{f}\}+\sqrt{T}\sup_{\gamma \in \Gamma _{T}}%
	\frac{2\sqrt{2}|\delta _{0}|_{2}}{T}\sum_{t}|x_{t}|^{2}1\{|\widehat{f}_{t}-%
	\widetilde{f}_{t}|_{2}>C\Delta _{f}\}\cr
	&\leq& \sqrt{T}\frac{2\sqrt{2}|\delta
		_{0}|_{2}}{T}\sum_{t}|x_{t}|^{2}1\{\inf_{\gamma }|\widehat{f}_{t}^{\prime
	}\gamma |<C\Delta _{f}\}\leq O_{P}(T^{1/2-\varphi })\mathbb{P}(\inf_{\gamma
	}|\widehat{f}_{t}^{\prime }\gamma |<C\Delta _{f})\cr
	&=&O_{P}(T^{1/2-\varphi
	}\Delta _{f})=o_{P}(1).
	\end{eqnarray*}

	(iii) For each $j,$
	\begin{eqnarray*}
	&&\left\vert \frac{1}{T}\sum_{t}[\widetilde{Z}_{jt}(\widehat{\gamma })-%
	\widetilde{Z}_{jt}(\gamma _{0})]\varepsilon _{t}\right\vert \\
	&&\leq \left\vert \frac{1}{T}\sum_{t}x_{jt}%
	\varepsilon _{t}[1\{\widehat{f}_{t}^{\prime }\widehat{\gamma }>0\}-1\{%
	\widehat{f}_{t}^{\prime }\gamma _{0}>0\}]\right\vert +2\sup_{\gamma \in
		\Gamma _{T}}\left\vert \frac{1}{T}\sum_{t}x_{jt}\varepsilon _{t}[1\{\widehat{%
		f}_{t}^{\prime }\gamma >0\}-1\{\widetilde{f}_{t}^{\prime }\gamma
	>0\}]\right\vert .
	\end{eqnarray*}%
	Note that $\widehat{f}_{t}^{\prime }\gamma =\breve{g}_{t}^{\prime }\phi $ for $\breve{g%
	}_{t}=g_{t}+h_{t}N^{-1/2}$ and $\phi =H^{-1}\gamma $, and $\breve{g}_{t}$ is
	$\rho $-mixing. Since $\widehat{\phi }$ is consistent, by Lemma \ref%
	{Lem:modul1}, the first term on the right hand side is bounded by: for any $%
	\epsilon _{1},\epsilon _{2}>0$,
	\begin{align*}
	&\mathbb{P}\left( |\frac{1}{T}\sum_{t}x_{jt}\varepsilon _{t}[1\{\widehat{f}%
	_{t}^{\prime }\widehat{\gamma }>0\}-1\{\widehat{f}_{t}^{\prime }\gamma
	_{0}>0\}]|_{2}>T^{-1/2}\epsilon _{1}\right) \cr
	&\leq o(1)+\mathbb{P}\left(
	\sup_{|\phi -\phi _{0}|<\epsilon _{1}^{2}\sqrt{\epsilon _{2}}}|\frac{1}{T}%
	\sum_{t}x_{jt}\varepsilon _{t}[1\{\breve{g}_{t}^{\prime }\phi >0\}-1\{\breve{%
		g}_{t}^{\prime }\phi _{0}>0\}]|_{2}>T^{-1/2}\epsilon _{1}\right) \cr
		&\leq
	o(1)+\frac{C\epsilon _{1}^{4}\epsilon _{2}}{\epsilon _{1}^{4}}\leq
	o(1)+C\epsilon _{2}.
	\end{align*}%
	Because $\epsilon _{1},\epsilon _{2}>0$ are arbitrary, the first term is $%
	o(T^{-1/2})$.

	As for the second term, by (\ref{e7.9}),
	\begin{eqnarray*}
&&	\sup_{\gamma \in \Gamma _{T}}|\frac{1}{T}\sum_{t}x_{jt}\varepsilon _{t}[1\{%
	\widehat{f}_{t}^{\prime }\gamma >0\}-1\{\widetilde{f}_{t}^{\prime }\gamma
	>0\}]|\cr
	&\leq& \sup_{\gamma \in \Gamma _{T}}\frac{1}{T}\sum_{t}|x_{jt}%
	\varepsilon _{t}|1\{\widetilde{f}_{t}^{\prime }\gamma <0<\widehat{f}%
	_{t}^{\prime }\gamma \}+\sup_{\gamma \in \Gamma _{T}}\frac{1}{T}%
	\sum_{t}|x_{jt}\varepsilon _{t}|1\{\widehat{f}_{t}^{\prime }\gamma <0<%
	\widetilde{f}_{t}^{\prime }\gamma \}\cr
	&\leq& O_{P}(\Delta
	_{f}+T^{-6})=o_{P}(T^{-1/2}).
	\end{eqnarray*}

	(iv) By (\ref{e7.9}), for each $j,$%
	\begin{eqnarray*}
	\left\vert \frac{1}{T}\sum_{t}[\widetilde{Z}_{jt}(\gamma _{0})-\widehat{Z}%
	_{jt}(\gamma _{0})]\varepsilon _{t}\right\vert \leq \left\vert \frac{1}{T}%
	\sum_{t=1}^{T}\varepsilon _{t}x_{jt}1\{\widehat{f}_{t}^{\prime }\gamma
	_{0}<0<\widetilde{f}_{t}^{\prime }\gamma _{0}\}\right\vert \cr+\left\vert
	\frac{1}{T}\sum_{t=1}^{T}\varepsilon _{t}x_{j}1\{\widetilde{f}_{t}^{\prime
	}\gamma _{0}<0<\widehat{f}_{t}^{\prime }\gamma _{0}\}\right\vert \leq
	O_{P}(\Delta _{f}+T^{-6})=o_{P}(T^{-1/2}),
	\end{eqnarray*}%
	and%
	\begin{eqnarray*}
	\frac{1}{T}\sum_{t}[\widehat{Z}_{jt}(\gamma _{0})-Z_{jt}(\gamma
	_{0})]\varepsilon _{t}=\frac{1}{T}\sum_{t=1}^{T}\varepsilon _{t}x_{jt}1\{%
	\widehat{f}_{t}^{\prime }\gamma _{0}<0<f_{t}^{\prime }\gamma _{0}\}\cr+\frac{%
		1}{T}\sum_{t=1}^{T}\varepsilon _{t}x_{jt}1\{f_{t}^{\prime }\gamma _{0}<0<%
	\widehat{f}_{t}^{\prime }\gamma _{0}\},
	\end{eqnarray*}%
	unless it is zero. Then, $\mathbb{E}\varepsilon _{t}x_{jt}1\{\widehat{f}%
	_{t}^{\prime }\gamma _{0}<0<f_{t}^{\prime }\gamma _{0}\}=0$ as $\varepsilon
	_{t}$ is an MDS, while
	\begin{equation*}
	\mathrm{var} \left[\frac{1}{T}\sum_{t=1}^{T}\varepsilon _{t}x_{jt}1\{\widehat{f}%
	_{t}^{\prime }\gamma _{0}<0<f_{t}^{\prime }\gamma _{0}\} \right]=\frac{1}{T^{2}}%
	\sum_{t=1}^{T}\mathbb{E}x_{jt}^{2}1\{\widehat{f}_{t}^{\prime }\gamma
	_{0}<0<f_{t}^{\prime }\gamma _{0}\}\mathbb{E}[\varepsilon
	_{t}^{2}|x_{t},g_{t},h_{t}]=o(T^{-1}).
	\end{equation*}%
	Thus $\frac{1}{T}\sum_{t}[\widehat{Z}_{t}(\gamma _{0})-Z_{t}(\gamma
	_{0})]\varepsilon _{t}=o(T^{-1/2}).$
\end{proof}

\subsection{Limiting distribution of $\widehat{\protect\gamma} $ (Proof of Theorem \ref{thm:AD with Estimated f}: Part II)}

Recall the defintion of $ r_{NT} $ in \eqref{eq:r_NT},
which represents the convergence rate as a function of both $N$ and $T,$ and define
\begin{equation*}
l_{NT}=\sqrt{r_{NT}T^{1+2\varphi }}\ \ \ \text{and \ \ }g=r_{NT}\left(
\gamma -\gamma _{0}\right) ,
\end{equation*}%
which are introduced so as to define a reparametrized process that reflects
the convergence rate $r_{NT}$. Then, the following lemma shows that the
estimator $\widehat{\gamma }$ can be represented by the following minimizer
of the reparametrized version of the process:
\begin{equation*}
\limfunc{argmin}_{g:g_{1}=0}l_{NT}\left[ \mathbb{\widetilde{S}}_{T}\left(
\alpha _{0},\gamma _{0}+\frac{g}{r_{NT}}\right) -\mathbb{\widetilde{S}}%
_{T}\left( \alpha _{0},\gamma _{0}\right) \right] .
\end{equation*}%
Note that we fix the first element of $g$ at $0$ to impose the normalization
restriction of $\gamma _{1}=0$.

The following lemma now presents the separability of the centered and scaled
criterion function.

\begin{lem}
	\label{Lem-separation} Let $\alpha =\alpha _{0}+bT^{-1/2}$, and $\gamma
	=\gamma _{0}+gr_{NT}^{-1}$. Then, uniformly in $b,g$ on any compact set,
	\begin{eqnarray*}
		&&l_{NT}\left[ \mathbb{\widetilde{S}}_{T}\left( \alpha ,\gamma \right) -%
		\mathbb{\widetilde{S}}_{T}\left( \alpha _{0},\gamma _{0}\right) \right] \cr
		&=&-l_{NT}\mathbb{\widehat{C}}_{1}\left( \delta _{0},\gamma _{0}+\frac{g}{%
			r_{NT}}\right) +l_{NT}\mathbb{E}\left( \widehat{R}_{2}\left( \gamma _{0}+%
		\frac{g}{r_{NT}}\right) +\widehat{\mathbb{C}}_{3}\left( \gamma _{0}+\frac{g}{%
			r_{NT}}\right) \right)  \\
		&&+l_{NT}T^{-1}\mathbb{E}[b^{\prime }Z_{t}(\gamma _{0})]^{2}+l_{NT}\left[
		\mathbb{\widetilde{C}}_{2}(\alpha _{0}+bT^{-1/2})+\mathbb{\widetilde{C}}%
		_{4}(\alpha _{0}+bT^{-1/2})\right]  \\
		&&+o_{P}(1).
	\end{eqnarray*}%
	Furthermore, the two processes $l_{NT}\mathbb{\widehat{C}}_{1}\left( \delta
	_{0},\gamma _{0}+\frac{g}{r_{NT}}\right) $ and $l_{NT}\left[ \mathbb{%
		\widetilde{C}}_{2}(\alpha _{0}+bT^{-1/2})+\mathbb{\widetilde{C}}_{4}(\alpha
	_{0}+bT^{-1/2})\right] $ are asymptotically independent.
\end{lem}

\begin{proof}
	Uniformly in $\gamma$, and $\phi=H_T{\gamma}$, by Lemmas \ref{l:relf} and %
	\ref{l:rhat}
	\begin{eqnarray*}
		&& |\widetilde{\mathbb{C}}_1(\delta,\gamma) - \widehat{\mathbb{C}}%
		_1(\delta_0,\gamma)|\leq |\widetilde{\mathbb{C}}_1(\delta,\gamma) - \widehat{%
			\mathbb{C}}_1(\delta,\gamma)| +|\mathbb{\widehat{C}}_{1}( \delta, \gamma)-%
		\mathbb{\widehat{C}}_{1}(\delta_0, \gamma)| \\
		&\leq&(T^{-\varphi}+|\alpha-\alpha_0|_2) O_P(\Delta_f+T^{-6}) +\left(
		O_P\left( T^{-1}\right) + \eta T^{-2\varphi }\left\vert \phi -\phi
		_{0}\right\vert \right) T^{\varphi }\left\vert \delta -\delta
		_{0}\right\vert _{2}
	\end{eqnarray*}

	Note that $|\widehat{\gamma }-\gamma _{0}|_{2}=O_{P}(r_{NT}^{-1})$. Hence
	Lemma \ref{l:relf} implies
	\begin{eqnarray*}
		l_{NT} |\widetilde{R}_{2}\left( \gamma \right)-\mathds{R}_{2}\left( \phi
		\right)| &\leq& O_P(\Delta_f+T^{-6})T^{-2\varphi}l_{NT} =o_P(1) \\
		l_{NT} |\widetilde{R}_{3}|&\leq& O_{P}(T^{-1/2}T^{-\varphi }r_{NT}^{-1})
		l_{NT} =o_P(1) \\
		l_{NT} |\widetilde{\mathbb{C}}_1(\delta,\gamma) - \widehat{\mathbb{C}}%
		_1(\delta_0,\gamma)|&\leq& O_P(T^{-1/2})\Delta_fl_{NT}=o_P(1) \\
   l_{NT} \left\vert \mathbb{\widehat{C}}_{3}\left( \delta _{0},\gamma \right) 	-	\mathbb{\widetilde{C}}_{3}\left( \delta ,\gamma \right) \right\vert& \leq&
      l_{NT} \left\vert \mathbb{\widehat{C}}_{3}\left( \delta,\gamma \right) 	-	\mathbb{\widetilde{C}}_{3}\left( \delta ,\gamma \right) \right\vert
      +   l_{NT} \left\vert \mathbb{\widehat{C}}_{3}\left( \delta,\gamma \right) 	-	\mathbb{\widehat{C}}_{3}\left( \delta_0 ,\gamma \right) \right\vert  \cr
      &\leq&  l_{NT}T^{-\varphi} 	O_P(  \Delta_f )   (   T^{-\varphi}			+  |\alpha-\alpha_0|_2) +   l_{NT}	  T^{-\varphi} O_P(N^{-1/2})|\alpha-\alpha_0|_2\cr
      &\leq& o_P(1).
         	\end{eqnarray*}

	In addition, recall $\mathds G_2:=|\widehat{R}_{2}(\gamma )+\mathbb{\widehat{C}}%
	_{3}(\delta _{0},\gamma )-(\mathbb{E}\widehat{R}_{2}(\gamma )+\mathbb{%
		\widehat{C}}_{3}(\delta _{0},\gamma ))|.$ By Lemma \ref{l.1}, when $%
	T^{1-2\varphi }=O(\sqrt{N})$,
$
	l_{NT}\mathds G_2\leq (O_{P}(\frac{1}{T})+\eta T^{-2\varphi }\left\vert \gamma
	-\gamma _{0}\right\vert _{2})T^{-\varphi }l_{NT}=o_{P}(1).
$
	When $\sqrt{N}=o(T^{1-2\varphi })$,
$
	l_{NT}\mathds G_2\leq \left[ T^{-2\varphi }O_P\left( \frac{\sqrt{N}}{\left(
		NT^{1-2\varphi }\right) ^{2/3}}\right) +T^{-2\varphi }\eta r_{NT}^{2}\sqrt{N}%
	\right] T^{-\varphi }l_{NT}=o_{P}(1).
$

	Note that, $\mathds R( \alpha, \phi_0)=\mathbb{E}[b^{\prime }Z_{t}(\gamma
	_{0})]^{2} .$ In addition, Lemma \ref{l:relf} and Lemma \ref{l:rt} show
	uniformly in $\alpha,\gamma$, for any $\epsilon>0$, there is $C>0$ that does
	not depend on $\epsilon$,
	\begin{eqnarray*}
		&& l_{NT}|\widetilde R_{1}(\alpha,\gamma) -\mathds R( \alpha, \phi_0)|
		\leq l_{NT}|\widetilde R_{1}(\alpha,\gamma) - \mathds R(\alpha,
		H_T^{-1}\gamma) |\cr && + l_{NT}| \mathds R(\alpha, H_T^{-1}\gamma_0) - %
		\mathds R(\alpha, H_T^{-1}\gamma) |\cr &\leq& o_P(
		l_{NT})|\alpha-\alpha_0|_2^2+ l_{NT} C|\alpha-\alpha_0|_2^2[ o_P(1)+\epsilon
		]^{1/2}=o_P( l_{NT})|\alpha-\alpha_0|_2^2\cr &=&o_P( l_{NT})T^{-1} =o_P(1)
		\sqrt{r_{NT}T^{-1+2\varphi }} =o_P(1).
	\end{eqnarray*}

	All the above $O_{P},o_{P}$ are uniform in $\alpha ,g$. Then uniformly in $%
	\alpha ,g$, for $\gamma =\gamma _{0}+gr_{NT}^{-1}$, %
	\begin{eqnarray*}
		&&l_{NT}[\mathbb{\widetilde{S}}_{T}\left( {\alpha },{\gamma }\right) -%
		\mathbb{\widetilde{S}}_{T}\left( \alpha _{0},\gamma _{0}\right) ] \\
		&=&l_{NT}[\widetilde{R}_{1}(\alpha ,\gamma )+\widetilde{R}_{2}(\gamma )+%
		\widetilde{R}_{3}(\alpha ,\gamma )-\widetilde{\mathbb{C}}_{1}(\delta
		,\gamma )-\widetilde{\mathbb{C}}_{2}(\alpha )+\widetilde{\mathbb{C}}%
		_{3}(\delta ,\gamma )+\widetilde{\mathbb{C}}_{4}(\alpha )] \\
		&=&o_{P}(1)+l_{NT}[\mathbb{E}\widehat{R}_{2}\left( \gamma \right) +\mathbb{E%
		}\mathbb{\widehat{C}}_{3}\left( \delta _{0},\gamma \right) -\widehat{\mathbb{%
				C}}_{1}(\delta _{0},\gamma )]+l_{NT}[\mathds R\left( \alpha ,\phi
		_{0}\right) -\widetilde{\mathbb{C}}_{2}(\alpha )+\widetilde{\mathbb{C}}%
		_{4}(\alpha )]
	\end{eqnarray*}

	Turning to the last claim, first note that when $l_{NT}=o\left( T\right) ,$ $%
	l_{NT}T^{-1}\mathbb{E}[b^{\prime }Z_{t}(\gamma _{0})]^{2}=o_P\left(
	1\right) \ $and $l_{NT}\left[ \mathbb{\widetilde{C}}_{2}(\alpha
	_{0}+bT^{-1/2})+\mathbb{\widetilde{C}}_{4}(\alpha _{0}+bT^{-1/2})\right]
	=o_P\left( 1\right) $ due to the proof in Section \ref{sec:alpha-hat}.
	When $l_{NT}=T,$ we need to show that $l_{NT}\left[ \mathbb{\widetilde{C}}%
	_{2}(\alpha _{0}+bT^{-1/2})+\mathbb{\widetilde{C}}_{4}(\alpha _{0}+bT^{-1/2})%
	\right] $ is asymptotically uncorrelated to $l_{NT}\mathbb{\widehat{C}}%
	_{1}\left( \delta _{0},\gamma _{0}+\frac{g}{r_{NT}}\right) $. This follows
	from Lemma \ref{Lem-DCTinP} in the ensueing section.
\end{proof}

\subsubsection{Empirical Process Part}\label{sec:EP}

We concern the weak convergence of the empirical process given by%
\begin{eqnarray*}
	l_{NT}\mathbb{\widehat{C}}_{1}\left( \delta _{0},\gamma _{0}+\frac{g}{r_{NT}}%
	\right) &=&l_{NT}\frac{2}{T}\sum_{t=1}^{T}\varepsilon _{t}x_{t}^{\prime
	}\delta _{0}\left( \widehat{1}_{t}\left( \gamma _{0}+\frac{g}{r_{NT}}\right)
	-\widehat{1}_{t}\left( \gamma _{0}\right) \right) \\
	&=&2\mathbb{\breve{C}}_{11}\left( H_{T}g\right) -2\mathbb{\breve{C}}%
	_{12}\left( H_{T}g\right) ,
\end{eqnarray*}%
where $\breve{u}_{t}=\breve{g}_{t}^{\prime }\phi _{0}$ and
\begin{eqnarray*}
	\mathbb{\breve{C}}_{11}\left( \mathfrak{g}\right) &=&\frac{\sqrt{r_{NT}}}{%
		\sqrt{T}}\sum_{t=1}^{T}\varepsilon _{t}x_{t}^{\prime }d_{0}1\left\{ -\breve{g%
	}_{t}^{\prime }\frac{\mathfrak{g}}{r_{NT}}<\breve{u}_{t}\leq 0\right\} , \\
	\mathbb{\breve{C}}_{12}\left( \mathfrak{g}\right) &=&\frac{\sqrt{r_{NT}}}{%
		\sqrt{T}}\sum_{t=1}^{T}\varepsilon _{t}x_{t}^{\prime }d_{0}1\left\{ 0<\breve{%
		u}_{t}\leq -\breve{g}_{t}^{\prime }\frac{\mathfrak{g}}{r_{NT}}\right\} ,
\end{eqnarray*}%
where $\mathfrak{g}\ $belongs to a compact set $\mathfrak{G}$. This is
because $l_{NT}T^{-1-\varphi }=\sqrt{r_{NT}/T}$, $\breve{g}_{t}=g_{t}+h_{t}/%
\sqrt{N}=H_{T}^{-1\prime }\widehat{f}_{t}$, and $\widehat{f}_{t}^{\prime }%
\mathfrak{g}=\breve{g}_{t}^{\prime }H_{T}\mathfrak{g}$.

We introduce this transformation to remove the randomness in $H_{T}$ from
the definition of the processes $\mathbb{\breve{C}}_{11}\left( \mathfrak{g}%
\right) $ and $\mathbb{\breve{C}}_{12}\left( \mathfrak{g}\right) $ and make
use of the stationarity of $\breve{g}_{t}$. Furthermore, in view of the
extended CMT in Lemma \ref{CMT-extension} $\mathbb{\breve{C}}_{11}\left(
H_{T}g\right) $ and $\mathbb{\breve{C}}_{11}\left( Hg\right) $ have the same
weak limit if $H_{T}\overset{p}{\longrightarrow }H\ $and $H$ is a finite
constant. Thus, it is sufficient to derive the weak convergence of $\left(
\mathbb{\breve{C}}_{11}\left( \mathfrak{g}\right) ,\mathbb{\breve{C}}%
_{12}\left( \mathfrak{g}\right) \right) $ to some process, say, $\left(
\mathbb{\mathbb{C}}_{11}\left( \mathfrak{g}\right) ,\mathbb{\mathbb{C}}%
_{12}\left( \mathfrak{g}\right) \right) .$ Since $\mathbb{\breve{C}}%
_{11}\left( \mathfrak{g}\right) $ is of the same type as $\mathbb{\breve{C}}%
_{12}\left( \mathfrak{g}\right) $ and there is no correlation between the
two as $\varepsilon _{t}$ is an mds and the two indicators are orthogonal to
each other, we focus on the stochastic equicontinuity and fidi of $\mathbb{%
	\breve{C}}_{11}\left( \mathfrak{g}\right) .$

The stochastic equicontinuity of $\mathbb{\breve{C}}_{11}\left( \mathfrak{g}%
\right) ,$ however, is a direct consequence of Lemma \ref{Lem:modul1} since $%
\breve{u}_{t}\ $and $\breve{g}_{t}$ are stationary triangular arrays and
thus for any finite $\mathfrak{g}$ and $\gamma =\frac{\mathfrak{g}}{r_{NT}}$
and for any $c,\epsilon >0$
\begin{eqnarray*}
	&&\mathbb{P} \left\{ \sup_{\left\vert \mathfrak{h}-\mathfrak{g}\right\vert
		<\epsilon }\left\vert \mathbb{\breve{C}}_{11}\left( \mathfrak{h}\right) -%
	\mathbb{\breve{C}}_{11}\left( \mathfrak{g}\right) \right\vert >c\right\}  \\
	&=&\mathbb{P} \left\{ \sup_{\left\vert \vec{\gamma}-\gamma \right\vert <\epsilon
		/r_{NT}}\frac{1}{\sqrt{T}}\sum_{t=1}^{T}\varepsilon _{t}x_{t}^{\prime
	}d_{0}\left( 1\left\{ -\breve{g}_{t}^{\prime }\gamma <\breve{u}_{t}\leq
	0\right\} -1\left\{ -\breve{g}_{t}^{\prime }\vec{\gamma}<\breve{u}_{t}\leq
	0\right\} \right) >\frac{c}{\sqrt{r_{NT}}}\right\}  \\
	&\leq &C\frac{\epsilon ^{2}}{c^{4}},
\end{eqnarray*}%
which can be made arbitrarily small by choosing $\epsilon $ small.

Turning to the fidi of $\mathbb{\breve{C}}_{11}\left( \mathfrak{g}\right) $,
we first check $\mathbb{\breve{C}}_{11}\left( \mathfrak{g}\right) $
satisfies the conditions to apply the mds CLT (e.g. Hall and Heyde 1980).
Specifically, let $v_{t}=\sqrt{r_{NT}}\varepsilon _{t}x_{t}^{\prime
}d_{0}1\left\{ -\breve{g}_{t}^{\prime }\frac{\mathfrak{g}}{r_{NT}}<\breve{u}%
_{t}\leq 0\right\} ,$ which is an mds as $\varepsilon _{t}$ is an mds, and
verify that $\max_{t}\left\vert v_{t}\right\vert =o_P\left( \sqrt{T}%
\right) $ and that $\frac{1}{T}\sum_{t=1}^{T}v_{t}^{2}$ has a proper
non-degenerate probability limit. However, $T^{-2}\mathbb E\max_{t}v_{t}^{4}\leq
T^{-1}\mathbb Ev_{t}^{4}$ by the stationarity and by $\max_{t}\left\vert
a_{t}\right\vert \leq \sum_{t=1}^{T}\left\vert a_{t}\right\vert $ and \\$%
T^{-1}\mathbb Ev_{t}^{4}=T^{-1}r_{NT}^{2}\mathbb E\left( \varepsilon _{t}x_{t}^{\prime
}d_{0}\right) ^{4}1\left\{ -\breve{g}_{t}^{\prime }\frac{\mathfrak{g}}{r_{NT}%
}<\breve{u}_{t}\leq 0\right\} \leq CT^{-1}r_{NT}=o\left( 1\right) $.
Furthermore, $\frac{1}{T}\sum_{t=1}^{T}\left( v_{t}^{2}-Ev_{t}^{2}\right)
=o_P\left( 1\right) $ due to Lemma \ref{Lem:modul1}. Thus, it remains to
show that the limit of $\mathbb Ev_{t}^{2}$ does not degenerate, which is shown in
the following.

To that end, we first derive the following limit%
\begin{eqnarray*}
	L\left( \mathfrak{s},\mathfrak{g}\right) &=&\lim_{N,T\rightarrow \infty
	}\mathbb E\left( \mathbb{\breve{C}}_{11}\left( \mathfrak{s}\right) -\mathbb{\breve{C}%
	}_{12}\left( \mathfrak{s}\right) -\mathbb{\breve{C}}_{11}\left( \mathfrak{g}%
	\right) +\mathbb{\breve{C}}_{12}\left( \mathfrak{g}\right) \right) ^{2} \\
	&=&\lim_{N,T\rightarrow \infty }r_{NT}\mathbb E\eta _{t}^{2}\left\vert 1\left\{
	\breve{g}_{t}^{\prime }\left( \phi _{0}+\frac{\mathfrak{s}}{r_{NT}}\right)
	>0\right\} -1\left\{ \breve{g}_{t}^{\prime }\left( \phi _{0}+\frac{\mathfrak{%
			g}}{r_{NT}}\right) >0\right\} \right\vert
\end{eqnarray*}%
for $\mathfrak{s}\neq \mathfrak{g}$ and $\eta _{t}=\varepsilon
_{t}x_{t}^{\prime }d_{0}$.

Note that each element $\mathfrak{g}\in \mathfrak{G}$ is linearly
independent of $\phi _{0}=H\gamma _{0}$, since $g_{1}=0$ while $\gamma
_{01}=1.$ Otherwise, there is $c\neq 0$ such that $\mathfrak{g=c\phi }_{0}.$
Then, $\mathfrak{g}=Hg=cH\gamma _{0},\ $which in turn implies that$\
g=c\gamma _{0}.$ This is a contradiction as $g_{1}=0$ while $\gamma _{01}=1$%
. This allows us to apply Lemma \ref{Lem-DCTinP} below to conclude that
\begin{eqnarray*}
	&&r_{NT}\mathbb E\eta _{t}^{2}1\left\{ \breve{u}_{t}+\breve{g}_{t}^{\prime }\frac{%
		\mathfrak{s}}{r_{NT}}>0\geq \breve{u}_{t}+\breve{g}_{t}^{\prime }\frac{%
		\mathfrak{g}}{r_{NT}}\right\} \\
	&\rightarrow &\mathbb E\left[ \eta _{t}^{2}\left( -g_{t}^{\prime }\mathfrak{g}%
	+g_{t}^{\prime }\mathfrak{s}\right) 1\left( g_{t}^{\prime }\mathfrak{g}%
	<g_{t}^{\prime }\mathfrak{s}\right) |u_{t}=0\right] p_{u}\left( 0\right) ,
\end{eqnarray*}%
and that%
\begin{eqnarray*}
	&&r_{NT}\mathbb E\eta _{t}^{2}1\left\{ \breve{u}_{t}+\breve{g}_{t}^{\prime }\frac{%
		\mathfrak{s}}{r_{NT}}\leq 0<\breve{u}_{t}+\breve{g}_{t}^{\prime }\frac{%
		\mathfrak{g}}{r_{NT}}\right\} \\
	&\rightarrow &\mathbb E\left[ \eta _{t}^{2}\left( g_{t}^{\prime }\mathfrak{g}%
	-g_{t}^{\prime }\mathfrak{s}\right) 1\left( g_{t}^{\prime }\mathfrak{g}%
	>g_{t}^{\prime }\mathfrak{s}\right) |u_{t}=0\right] p_{u}\left( 0\right) .
\end{eqnarray*}%
Thus, we conclude that%
\begin{equation*}
L\left( \mathfrak{s},\mathfrak{g}\right) =\mathbb E_{0}\left[ \left. \eta
_{t}^{2}\left\vert g_{t}^{\prime }\left( \mathfrak{g}-\mathfrak{s}\right)
\right\vert \right\vert u_{t}=0\right] p_{u}\left( 0\right) .
\end{equation*}

Putting these together, we conclude%
\begin{equation*}
l_{NT}\mathbb{\widehat{C}}_{1}\left( \delta _{0},\gamma _{0}+\frac{g}{r_{NT}}%
\right) \Rightarrow 2W\left( g\right) ,
\end{equation*}%
where $W\left( g\right) $ is a centered Gaussian process with the covariance
kernel%
\begin{equation*}
\mathbb EW\left( g\right) W\left( s\right) =\frac{1}{2}\left( L\left( Hs,0\right)
+L\left( Hg,0\right) -L\left( Hs,Hg\right) \right) ,
\end{equation*}%
recalling that $\mathbb EXY=\frac{1}{2}\left( \mathbb EX^{2}+\mathbb EY^{2}-\mathbb E\left( X-Y\right)
^{2}\right) $ and $\mathbb{\breve{C}}_{11}\left( 0\right) =0$.

\begin{lem}
	\label{Lem-DCTinP}Assume Assumption \ref{as8}. Then,
	\begin{equation*}
	r_{NT}\mathbb E\eta _{t}^{2}1\left\{ \breve{u}_{t}+\breve{g}_{t}^{\prime }\frac{s}{%
		r_{NT}}>0\geq \breve{u}_{t}+\breve{g}_{t}^{\prime }\frac{g}{r_{NT}}\right\}
	\rightarrow \mathbb{E}\left[ \eta _{t}^{2}\left( g_{t}^{\prime }s-g_{t}^{\prime
	}g\right) 1\left( g_{t}^{\prime }g<g_{t}^{\prime }s\right) |u_{t}=0\right]
	p_{u_{t}}\left( 0\right) ,
	\end{equation*}%
	as $N,T\rightarrow \infty .$
\end{lem}


\begin{proof}[Proof of Lemma \ref{Lem-DCTinP}] First, we write a conditional density of $ \breve{u}_t $ given a random variable $ Y $ by $ p(u|Y) $ for more clarity. Note that%
\begin{eqnarray*}
	&&r_{NT}\mathbb E\eta _{t}^{2}1\left\{ \breve{u}_{t}+\breve{g}_{t}^{\prime }\frac{s}{%
		r_{NT}}>0\geq \breve{u}_{t}+\breve{g}_{t}^{\prime }\frac{w}{r_{NT}}\right\}
	\cr
	&=&r_{NT}\mathbb E\eta _{t}^{2}1\left\{ -\frac{\breve{g}_{t}^{\prime }s}{r_{NT}}<%
	\breve{u}_{t}\leq -\frac{\breve{g}_{t}^{\prime }w}{r_{NT}}\right\}  \cr
	&=&\mathbb E\left[ \int_{-\breve{g}_{t}^{\prime }s}^{-\breve{g}_{t}^{\prime
		}w}\mathbb E\left( \eta _{t}^{2}|\frac{z}{r_{NT}},\breve{g}_{t}^{\prime }s,\breve{g}%
	_{t}^{\prime }w\right) p\left( \frac{z}{r_{NT}}|\breve{g}_{t}^{\prime }s,%
	\breve{g}_{t}^{\prime }w\right) dz1\left\{ \breve{g}_{t}^{\prime }s>\breve{g}%
	_{t}^{\prime }w\right\} \right]  \cr
	&=&\mathbb E\left[ \int_{-\breve{g}_{t}^{\prime }s}^{-\breve{g}_{t}^{\prime
		}w}\mathbb E\left( \eta _{t}^{2}|0,\breve{g}_{t}^{\prime }s,\breve{g}_{t}^{\prime
	}w\right) p\left( 0|\breve{g}_{t}^{\prime }s,\breve{g}_{t}^{\prime }w\right)
	dz1\left\{ \breve{g}_{t}^{\prime }s>\breve{g}_{t}^{\prime }w\right\} \right]
	\cr
	&&+\mathbb{E}\left[ \int_{-\breve{g}_{t}^{\prime }s}^{-\breve{g}_{t}^{\prime
		}w}\left( \mathbb E\left( \eta _{t}^{2}|\frac{z}{r_{NT}},\breve{g}_{t}^{\prime }s,%
	\breve{g}_{t}^{\prime }w\right) -\mathbb E\left( \eta _{t}^{2}|0,\breve{g}%
	_{t}^{\prime }s,\breve{g}_{t}^{\prime }w\right) \right) p\left( 0|\breve{g}%
	_{t}^{\prime }s,\breve{g}_{t}^{\prime }w\right) dz1\left\{ \breve{g}%
	_{t}^{\prime }s>\breve{g}_{t}^{\prime }w\right\} \right]  \cr
	&&+\mathbb E\left[ \int_{-\breve{g}_{t}^{\prime }s}^{-\breve{g}_{t}^{\prime
		}w}\left( p\left( \frac{z}{r_{NT}}|\breve{g}_{t}^{\prime }s,\breve{g}%
	_{t}^{\prime }w\right) -p\left( 0|\breve{g}_{t}^{\prime }s,\breve{g}%
	_{t}^{\prime }w\right) \right)\mathbb  E\left( \eta _{t}^{2}|0,\breve{g}_{t}^{\prime
	}s,\breve{g}_{t}^{\prime }w\right) dz1\left\{ \breve{g}_{t}^{\prime }s>%
	\breve{g}_{t}^{\prime }w\right\} \right]  \cr
	&&+\mathbb E \int_{-\breve{g}_{t}^{\prime }s}^{-\breve{g}_{t}^{\prime
		}w}\left( \mathbb E\left( \eta _{t}^{2}|\frac{z}{r_{NT}},\breve{g}_{t}^{\prime }s,%
	\breve{g}_{t}^{\prime }w\right) -\mathbb E\left( \eta _{t}^{2}|0,\breve{g}%
	_{t}^{\prime }s,\breve{g}_{t}^{\prime }w\right) \right) \left( p\left( \frac{%
		z}{r_{NT}}|\breve{g}_{t}^{\prime }s,\breve{g}_{t}^{\prime }w\right) -p\left(
	0|\breve{g}_{t}^{\prime }s,\breve{g}_{t}^{\prime }w\right) \right)\cr
&&	dz1\left\{ \breve{g}_{t}^{\prime }s>\breve{g}_{t}^{\prime }w\right\}
\end{eqnarray*}%
by a change-of-variables formula $z=r_{NT}u.$ First,
\begin{eqnarray*}
	&&\mathbb E\left[ \int_{-\breve{g}_{t}^{\prime }s}^{-\breve{g}_{t}^{\prime
		}w}\mathbb E\left( \eta _{t}^{2}|0,\breve{g}_{t}^{\prime }s,\breve{g}_{t}^{\prime
	}w\right) p\left( 0|\breve{g}_{t}^{\prime }s,\breve{g}_{t}^{\prime }w\right)
	dz1\left\{ \breve{g}_{t}^{\prime }s>\breve{g}_{t}^{\prime }w\right\} \right]
	\cr
	&=&\mathbb E\left( 1\left\{ \breve{g}_{t}^{\prime }s>\breve{g}_{t}^{\prime
	}w\right\} \left( \breve{g}_{t}^{\prime }s-\breve{g}_{t}^{\prime }w\right)
	\mathbb E\left( \eta _{t}^{2}|0,\breve{g}_{t}^{\prime }s,\breve{g}_{t}^{\prime
	}w\right) p\left( 0|\breve{g}_{t}^{\prime }s,\breve{g}_{t}^{\prime }w\right)
	\right)  \cr
	&=&\mathbb E\left( \eta _{t}^{2}1\left\{ \breve{g}_{t}^{\prime }s>\breve{g}%
	_{t}^{\prime }w\right\} \left( \breve{g}_{t}^{\prime }s-\breve{g}%
	_{t}^{\prime }w\right) |\breve{u}_{t}=0\right) p_{\breve{u}_{t}}\left(
	0\right)  \cr
	&\rightarrow &\mathbb E\left( \eta _{t}^{2}1\left\{ g_{t}^{\prime }s>g_{t}^{\prime
	}w\right\} \left( g_{t}^{\prime }s-g_{t}^{\prime }w\right) |u_{t}=0\right)
	p_{u_{t}}\left( 0\right) ,
\end{eqnarray*}%
where the convergence holds by the following reasons. Since $\left( \eta
_{t},\breve{g}_{t}^{\prime }\right) ^{\prime }\overset{p}{\longrightarrow }%
\left( \eta _{t},g_{t}^{\prime }\right) ^{\prime }$ as $N\rightarrow \infty ,
$ we have $\eta _{t}^{2}1\left\{ \breve{g}_{t}^{\prime }s>\breve{g}%
_{t}^{\prime }w\right\} \left( \breve{g}_{t}^{\prime }s-\breve{g}%
_{t}^{\prime }w\right) \overset{p}{\longrightarrow }\eta _{t}^{2}1\left\{
g_{t}^{\prime }s>g_{t}^{\prime }w\right\} \left( g_{t}^{\prime
}s-g_{t}^{\prime }w\right) $ and $\breve{u}_{t}\overset{p}{\longrightarrow }%
u_{t}$ by the continuous mapping theorem, which imples by the Lipschitz
continuity of the densities (Assumption \ref{as8} \ref{as8:itm_AD}) the convergence of $p_{\breve{u}_{t}}\left(
0\right) $ and the conditional densities.
This in turn implies the convergence of $\mathbb E\left( \eta _{t}^{2}1\left\{
\breve{g}_{t}^{\prime }s>\breve{g}_{t}^{\prime }w\right\} \left( \breve{g}%
_{t}^{\prime }s-\breve{g}_{t}^{\prime }w\right) |\breve{u}_{t}=0\right) $
due to the uniform integrability, which is implied by the boundedness of $%
\mathbb E\left( \eta _{t}^{4}\left\vert \breve{g}_{t}\right\vert _{2}^{2}|\breve{u}%
_{t}\right) $.

Then, we show the other terms are negligible. We elaborate the first of
these since the reasonings are similar.%
\begin{eqnarray*}
	&&\mathbb E\left[ \int_{-\breve{g}_{t}^{\prime }s}^{-\breve{g}_{t}^{\prime }w}\left(
	\mathbb E\left( \eta _{t}^{2}|\frac{z}{r_{NT}},\breve{g}_{t}^{\prime }s,\breve{g}%
	_{t}^{\prime }w\right) -\mathbb E\left( \eta _{t}^{2}|0,\breve{g}_{t}^{\prime }s,%
	\breve{g}_{t}^{\prime }w\right) \right) \frac{z}{r_{NT}}p\left( 0|\breve{g}%
	_{t}^{\prime }s,\breve{g}_{t}^{\prime }w\right) dz1\left\{ \breve{g}%
	_{t}^{\prime }s>\breve{g}_{t}^{\prime }w\right\} \right]  \\
	&\leq &C\mathbb E\left[ \int_{-\breve{g}_{t}^{\prime }s}^{-\breve{g}_{t}^{\prime }w}%
	\frac{z}{r_{NT}}dzp\left( 0|\breve{g}_{t}^{\prime }s,\breve{g}_{t}^{\prime
	}w\right) 1\left\{ \breve{g}_{t}^{\prime }s>\breve{g}_{t}^{\prime }w\right\} %
	\right]  \\
	&=&C^{\prime }\mathbb E\left( \left( \breve{g}_{t}^{\prime }w\right) ^{2}-\left(
	\breve{g}_{t}^{\prime }s\right) ^{2}\right) \frac{1}{2r_{NT}}=o\left(
	1\right) .
\text{\hfill}
\end{eqnarray*}
\end{proof}

\subsubsection{Bias} \label{sec:Bias}

We show that, as $N,T\rightarrow \infty ,$
\begin{equation*}
l_{NT}(\mathbb{E}\widehat{R}_{2}(g)+\widehat{\mathbb{C}}_{3}(g))\rightarrow
A\left( \omega, g\right),
\end{equation*}%
where
$$
A(\omega, g):=  M_\omega\mathbb{E}\left( (x_{t}^{\prime
}d_{0})^{2}[\left\vert g_{t}^{\prime }Hg+\zeta _{\omega}^{-1}\mathcal{Z}%
_{t}\right\vert -\left\vert \zeta _{\omega}^{-1}\mathcal{Z}_{t}\right\vert ]\bigg{|}u_{t}=0\right) p_{u_{t}}(0).
$$
and that $A\left( \omega, g\right)
\rightarrow +\infty $ as $\left\vert g\right\vert \rightarrow +\infty $ for
any $\omega$.

\proof

For $\gamma =H^{-1}\phi $, and $g=r_{NT}\left[ \gamma -\gamma _{0}\right] $,
we have $\phi -\phi _{0}=H(\gamma -\gamma _{0})=r_{NT}^{-1}Hg,$ with $%
g_{1}=0 $ due to the normalization. Suppose $g\neq 0$. Let
\begin{equation*}
r_g=|\phi-\phi_0|_2^{-1}(\phi-\phi_0)=|Hg|_2^{-1} Hg.
\end{equation*}
We only need to focus on the case that $r_g$ is linearly independent of $%
\phi_0.$ Let
\begin{equation*}
\zeta _{NT}=\sqrt{N}r_{NT}^{-1} .
\end{equation*}%
By the proof of Lemma \ref{l:generic},
\begin{eqnarray*}
	&&l_{NT}\mathbb{E}\left( \mathbb{\widehat{C}}_{3}\left( \delta _{0},\gamma
	\right) +\widehat{R}_{2}\left( \gamma \right) \right) \\
	&=&l_{NT}\mathbb{E}\left( x_{t}^{\prime }\delta _{0}\right) ^{2}\left(
	A_{1t}\left( \phi \right) +A_{2t}\left( \phi \right) -A_{3t}\left( \phi
	\right) -A_{4t}\left( \phi \right) \right)
\end{eqnarray*}

\medskip
\noindent
\textbf{Step I}: obtaining the results for the case of $\omega\in(0,\infty]$.
\medskip

In this case, $\zeta_{NT}\to\zeta_\omega\in (0,\infty]$. We now work with (\ref{e0.5add}). Note
that for $p=1.5$,
\begin{equation*}
\frac{l_{NT}}{T^{2\varphi}N^{0.5+1/(2p)}}=o(1),
\end{equation*}
and
\begin{equation*}
M_{NT}:=\frac{1}{\sqrt{N}}l_{NT}T^{-2\varphi}\zeta_{NT}\to M_\omega:=\max\{1,
\omega^{-1/3}\}\in(0,\infty).
\end{equation*}

We shall use the following equality, which can be verified:
\begin{eqnarray}  \label{e0.7add}
|a+b| -|b|&=& \Xi(a,b), \quad \text{ where }\cr \Xi(a,b) &:=& -a1\{a\leq 0\}
1\left\{b\leq 0\right\} - \left( a + b\right)1\left\{ a +
b<0\right\}1\left\{ b>0\right\}\cr &&- b1\{ a + b<0\} 1\{b>0\} + a 1\{ a
+b>0\} 1\{ a <0\} \cr &&+ a1\left\{ a >0\right\} 1\{b>0\} + \left( a +
b\right)1\left\{ a + b> 0\right\} 1\left\{b\leq 0\right\} \cr && + b
1\{b<0\} 1\{ a +b>0\}- a 1\{ a>0\} 1\{ a +b<0\} .
\end{eqnarray}
Let $g_t^{\prime }(\phi-\phi_0)=a$, $\frac{h_t^{\prime }\phi_0}{\sqrt{N}}=b$%
, Note that (\ref{e0.5add}) can be written exactly as the right hand side of
the above equality, up to $\mathbb{E}_{|u_t=0}(x_t^{\prime
}\delta_0)^2p_{u_t}(0)$. Hence (\ref{e0.5add}) and the above equality imply,
for $\phi-\phi_0=r_{NT}^{-1}H_Tg$,
\begin{eqnarray*}
	&&l_{NT} \mathbb{E}(x_t^{\prime }\delta_0)^2(A_1-A_3+ A_2-A_4) \cr &=^{(1)}&
	l_{NT}\mathbb{E}_{|u_t=0}(x_t^{\prime }\delta_0)^2p_{u_t}(0) \Xi(a,b) +o(1) %
	\cr &=^{(2)}& l_{NT}\mathbb{E}_{|u_t=0}(x_t^{\prime }\delta_0)^2p_{u_t}(0) %
	\left[\left| g_t^{\prime }\left(\phi -\phi _{0}\right) + \frac{h_t^{\prime
		}\phi_0}{\sqrt{N}} \right | - \left |\frac{h_t^{\prime }\phi_0 }{\sqrt{N}}
	\right | \right] +o(1)\cr
 &=& \mathbb{\breve{C}}_{NT}(H_Tg)+o(1),%
	\quad \text{ where} \cr \mathbb{\breve{C}}_{NT}(\mathfrak{g})&:=&M_{NT}%
	\mathbb{E}_{|u_t=0}(x_t^{\prime }d_0)^2p_{u_t}(0) \left(\left| g_t^{\prime }%
	\mathfrak{g} + \zeta_{\omega} ^{-1} h_t^{\prime }\phi_0\right | - \left |
	\zeta_{\omega}^{-1} h_t^{\prime }\phi_0 \right | \right)
\end{eqnarray*}
In the above, (1) is rewriting (\ref{e0.5add}) using the notation of $%
\Xi(a,b)$ for $g_t^{\prime }(\phi-\phi_0)=a$ and $\frac{h_t^{\prime }\phi_0}{%
	\sqrt{N}}=b$; (2) uses the equality $|a+b| -|b|= \Xi(a,b)$.

\medskip
\noindent
\textbf{Step I.1}: pointwise convergence of $\mathbb{\breve{C}}_{NT}(\mathfrak{g})$
\medskip

We now derive the pointwise limit of $\mathbb{\breve{C}}_{NT}(\mathfrak{g})$%
. Define
\begin{equation*}
\widetilde F_{g_t}(z)= \left| g_t^{\prime }\mathfrak{g} + \zeta_\omega ^{-1}
z\right | - \left | \zeta_\omega^{-1} z\right | .
\end{equation*}
Then $\mathbb{\breve{C}}_{NT}(\mathfrak{g})= M_{NT}\mathbb{E}%
_{|u_t=0}(x_t^{\prime }d_0)^2p_{u_t}(0) \mathbb{E}[ \widetilde
F_{g_t}(h_t^{\prime }\phi_0)| x_t, g_t]$. Now we use the following
portmanteau lemma: $X_n\overset{d}{\longrightarrow } X$ if and only if $\mathbb{E}\widetilde F(X_n)\to
\mathbb{E}\widetilde F(X)$ for all bounded continuous functions $\widetilde F$. Note
that $h_t^{\prime }\phi_0|x_t, g_t\overset{d}{\longrightarrow } \mathcal{Z}_t$. Now for each fixed $%
(x_t, g_t)$,
\begin{equation*}
|\widetilde F_{g_t}(z)|\leq |g_t^{\prime } \mathfrak{g}|;
\end{equation*}
the right hand side is independent of $z$, and $\widetilde F_{g_t}(z)$ is
continuous in $z$. So we can apply the portmanteau lemma to conclude that $%
\mathbb{E}[ \widetilde F_{g_t}(h_t^{\prime }\phi_0)| x_t, g_t]\to \mathbb{E}[
\widetilde F_{g_t}(\mathcal{Z}_t)| x_t, g_t]$ for each fixed $x_t, g_t$. This
further implies, $P_N(x_t, g_t) \to P(x_t, g_t)$ for each fixed $(x_t, g_t)$%
, with
\begin{eqnarray*}
	P_N(x_t, g_t) &:=& (x_t^{\prime }d_0)^2p_{u_t}(0) \mathbb{E}[ \widetilde
	F_{g_t}(h_t^{\prime }\phi_0)| x_t, g_t] ,\cr P(x_t, g_t)&:=& (x_t^{\prime
	}d_0)^2p_{u_t}(0) \mathbb{E}[ \widetilde F_{g_t}(\mathcal{Z}_t)| x_t, g_t].
\end{eqnarray*}

In addition, note that for each fixed $x_{t},g_{t}$, $|\mathbb{E}[\widetilde{F}%
_{g_{t}}(h_{t}^{\prime }\phi _{0})|x_{t},g_{t}]|\leq |g_{t}^{\prime } \mathfrak{g}|$.
For all $N$, \newline
$|P_{N}(x_{t},g_{t})|\leq (x_{t}^{\prime
}d_{0})^{2}p_{u_{t}}(0)|g_{t}^{\prime } \mathfrak{g}|$; the right hand side does not
depend on $N$, and has a bounded expectation:
$\mathbb E(x_{t}^{\prime
}d_{0})^{2}p_{u_{t}}(0)|g_{t}^{\prime } \mathfrak{g}|<\infty$.
 Hence by the dominated
convergence theorem, the
pointwise convergence of $P_{N}(x_{t},g_{t})\rightarrow P(x_{t},g_{t})$
implies $\mathbb{E}_{|u_{t}=0}P_{N}(x_{t},g_{t})\rightarrow \mathbb{E}%
_{|u_{t}=0}P(x_{t},g_{t})$, which means
\begin{equation*}
\mathbb{E}_{|u_{t}=0}(x_{t}^{\prime }d_{0})^{2}p_{u_{t}}(0)\mathbb{E}[\widetilde{%
	F}_{g_{t}}(h_{t}^{\prime }\phi _{0})|x_{t},g_{t}]\rightarrow \mathbb{E}%
_{|u_{t}=0}(x_{t}^{\prime }d_{0})^{2}p_{u_{t}}(0)\mathbb{E}[\widetilde{F}%
_{g_{t}}(\mathcal{Z}_{t})|x_{t},g_{t}].
\end{equation*}%
Also, $M_{NT}\rightarrow M_\omega\in (0,\infty )$. Thus
\begin{eqnarray*}
\mathbb{\breve{C}}_{NT}(\mathfrak{g})&=&M_{NT}\mathbb{E}_{|u_{t}=0}\{(x_{t}^{	\prime }d_{0})^{2}p_{u_{t}}(0)\mathbb{E}[\widetilde{F}_{g_{t}}(h_{t}^{\prime}\phi _{0})|x_{t},g_{t}]\}\cr&\rightarrow& M_\omega\mathbb{E}_{|u_{t}=0}(x_{t}^{%
	\prime }d_{0})^{2}p_{u_{t}}(0)\mathbb{E}[\widetilde{F}_{g_{t}}(\mathcal{Z}%
_{t})|x_{t},g_{t}]\cr&=&M_\omega\mathbb{E}\left( (x_{t}^{\prime
}d_{0})^{2}[\left\vert g_{t}^{\prime }\mathfrak{g}+\zeta _{\omega}^{-1}\mathcal{Z}%
_{t}\right\vert -\left\vert \zeta _{\omega}^{-1}\mathcal{Z}_{t}\right\vert ]%
\bigg{|}u_{t}=0\right) p_{u_{t}}(0)\cr
&:=&\breve{A}(\mathfrak{g}).
\end{eqnarray*}%

Hence we have proved  for some $C>0$ and any $|\mathfrak{g}|_2<C$,
\begin{eqnarray*}
l_{NT}\mathbb{E}(x_{t}^{\prime }\delta _{0})^{2}(A_{1}-A_{3}+A_{2}-A_{4})&=&
\mathbb{\breve{C}}_{NT}(H_{T}g)+o(1),
\cr
\mathbb{\breve{C}}_{NT}(\mathfrak{g})
&\to& \breve A(\mathfrak{g}).
\end{eqnarray*}%

\medskip
\noindent
\textbf{Step I.2}: $\mathbb{\breve{C}}_{NT}(H_{T}g)\overset{P}{\longrightarrow }A(\omega, g) $
\medskip

  We apply the extended continuous mapping theorem (CMT) for drifting functions (cf.
Lemma \ref{CMT-extension}).  To do so, first note that $H_T\to ^P H$ for some $K\times K$ invertible nonrandom matrix $H$ (e.g., \cite{bai03}). To applied the extended CMT, we need to show, for any converging sequence $\mathfrak{g}_T\to \mathfrak{g}$ in a compact space, we have
\begin{equation}\label{e9.23}
\mathbb{\breve{C}}_{NT}(\mathfrak{g}_T)\to \breve A(\mathfrak{g}).
\end{equation}
Once this is achieved, then because $H_Tg\overset{P}{\longrightarrow } Hg$, by Theorem 1.11.1 of \cite{VW}, we have $\mathbb{\breve{C}}_{NT}(H_Tg)\overset{P}{\longrightarrow } \breve A(Hg)=A(\omega, g)$.

To prove (\ref{e9.23}), note that $
|\mathbb{\breve{C}}_{NT}(\mathfrak{g}_T)- \breve A(\mathfrak{g})|
\leq |\mathbb{\breve{C}}_{NT}(\mathfrak{g}_T)- \mathbb{\breve{C}}_{NT}(\mathfrak{g})|+
|\mathbb{\breve{C}}_{NT}(\mathfrak{g})- \breve A(\mathfrak{g})|.
$
The second term on the right hand side is $o(1)$ due to the pointwise convergence. It remains to prove the first term on the right is also $o(1)$. By definition,
\begin{eqnarray*}
&& |\mathbb{\breve{C}}_{NT}(\mathfrak{g}_T)- \mathbb{\breve{C}}_{NT}(\mathfrak{g})|
\leq M_{NT}\mathbb{E}_{|u_{t}=0}(x_{t}^{	\prime }d_{0})^{2}p_{u_{t}}(0)\mathbb{E}[    \left| g_t^{\prime }(\mathfrak{g}_T  -  \mathfrak{g} )\right |  |x_{t},g_{t}]\cr
&\leq& O(1)\mathbb{E}_{|u_{t}=0}(x_{t}^{	\prime }d_{0})^{2}|g_t|_2    \left| \mathfrak{g}_T  -  \mathfrak{g} \right |  \leq O(1)
  \left| \mathfrak{g}_T  -  \mathfrak{g} \right |=o(1).
\end{eqnarray*}
Hence by the triangular inequality,  (\ref{e9.23}) holds.
It then immediately follows that
$l_{NT}\mathbb{E}(x_{t}^{\prime }\delta _{0})^{2}(A_{1}-A_{3}+A_{2}-A_{4})\overset{P}{\longrightarrow } A(\omega, g)$.
In particular, when $\omega=\infty $, $\zeta _{\omega}^{-1}=0$ and $M_\omega=1$, so $%
A\left( \omega, g\right)=A(\infty ,g)$.


\medskip
\noindent
\textbf{Step II}: obtaining the results for the case of $\omega=0$
\medskip

In this case, we have that $\zeta_{NT}\to 0$, and
\begin{equation*}
\widetilde M_{NT}:=\frac{l_{NT}\zeta_{NT}^2}{\sqrt{N}}T^{-2\varphi} \to 1.
\end{equation*}
We now work with the last equality of (\ref{e0.10add}), up to $\frac{l_{NT}}{%
	T^{2\varphi}N^{0.5+1/(2p)}}=o(1)$,
\begin{eqnarray*}
	&& l_{NT}\mathbb{E}_{|u_t=0}(x_t^{\prime }\delta_0)^2(A_1-A_3+ A_2-A_4)
:= \mathbb{\breve{C}}	_{NT,2}(H_Tg)+o(1)
\end{eqnarray*}
where
\begin{eqnarray*}
	\mathbb{\breve{C}}_{NT,2}(\mathfrak{g})&:=& -\frac{l_{NT}T^{-2\varphi}%
		\zeta_{NT}}{\sqrt{N}}2 \mathbb{E}_{|u_t=0}(x_t^{\prime }d_0)^2p_{u_t}(0)
	\left( g_t^{\prime } \mathfrak{g} + \zeta_{NT}^{-1} h_t^{\prime
	}\phi_0\right)1\left\{ g_t^{\prime } \mathfrak{g} + \zeta_{NT}^{-1}
	h_t^{\prime }\phi_0<0\right\}1\left\{ h_t^{\prime }\phi_0>0\right\}\cr &&+%
	\frac{l_{NT}T^{-2\varphi}\zeta_{NT}}{\sqrt{N}}2 \mathbb{E}%
	_{|u_t=0}(x_t^{\prime }d_0)^2 p_{u_t}(0)\left( g_t^{\prime }\mathfrak{g}
	+\zeta_{NT}^{-1}h_t^{\prime }\phi_0\right)1\left\{ g_t^{\prime }\mathfrak{g}
	+ \zeta_{NT}^{-1}h_t^{\prime }\phi_0> 0\right\} 1\left\{h_t^{\prime }\phi_0
	\leq 0\right\} .
\end{eqnarray*}

\medskip
\noindent
\textbf{Step II.1}: pointwise convergence of $\mathbb{\breve{C}}_{NT,2}(\mathfrak{g})$
\medskip

We now derive the limit of $\mathbb{\breve{C}}_{NT,2}(\mathfrak{g})$. Change
variable $y=h_t^{\prime }\phi_0\zeta_{NT}^{-1}$, $\mathbb{\breve{C}}_{NT,2}(%
\mathfrak{g})$ equals
\begin{eqnarray*}
&&- \widetilde M_{NT}2
	p_{u_t}(0)\mathbb{E}[(x_t^{\prime }d_0)^2 F_{NT,1}(g_t, x_t) |u_t=0]+
	\widetilde M_{NT}2 p_{u_t}(0) \mathbb{E}[(x_t^{\prime }d_0)^2 F_{NT,2} (g_t,
	x_t) |u_t=0], \cr \cr &\text{where} \cr & & F_{NT,1} (g_t, x_t) :=\int
	\left( g_t^{\prime }\mathfrak{g} + y\right)1\left\{ g_t^{\prime } \mathfrak{g%
	} + y<0\right\}1\left\{ y>0\right\} p_{h_t^{\prime }\phi_0|
		g_t,x_t,  u_t=0}(\zeta_{NT}y)d y \cr && F_{NT,2} (g_t, x_t) := \int \left(
	g_t^{\prime }\mathfrak{g} +y\right)1\left\{ g_t^{\prime }\mathfrak{g} + y>
	0\right\} 1\left\{y \leq 0\right\} p_{h_t^{\prime }\phi_0|
		g_t,x_t,  u_t=0}(\zeta_{NT}y)d y.
\end{eqnarray*}

For each fixed $y$, $x_{t},g_{t}$, as $\zeta _{NT}\rightarrow 0$,
for any $C>0$, for all large $N,T$, $|\zeta_{NT}y|<C$.
Recall  $p_{\mathcal Z_t}(\cdot)$ is  the pdf of $\mathcal N(0,\sigma^2_{h, x_t, g_t})$
with
$\sigma^2_{h, x_t, g_t}:=\plim_{N\to\infty} \mathbb E[(h_t'\phi_0)^2|x_t, g_t, g_t'\phi_0=0]$. By Assumption \ref{as5},
\begin{equation*}
|p_{h_{t}^{\prime }\phi _{0}|g_{t},x_{t}, u_t=0}(\zeta _{NT}y)-p_{\mathcal{Z}%
	_{t}}(0)|\leq \sup_{|z|<C}|p_{h_{t}^{\prime }\phi
	_{0}|g_{t},x_{t}, u_t=0}(z)-p_{\mathcal{Z}_{t}}(z)|+|p_{\mathcal{Z}%
	_{t}}(\zeta _{NT}y)-p_{\mathcal{Z}_{t}}(0)|=o(1).
\end{equation*}%
 and
$\sup_{x_t, g_t}p_{h_t^{\prime }\phi_0|
		g_t,x_t,  u_t=0}(\cdot)<C_0$ for some $C_0>0$ for all $N,T$.   For each fixed $g_{t}$ and all $N,T$, the integrand of $F_{NT,1} (g_t, x_t)$ is bounded by
\begin{equation*}
|\left( g_{t}^{\prime }\mathfrak{g}+y\right) 1\left\{ g_{t}^{\prime }%
\mathfrak{g}+y<0\right\} 1\left\{ y>0\right\} p_{h_{t}^{\prime }\phi
	_{0}|g_{t},x_{t}, u_t=0}(\zeta _{NT}y)|\leq C_0|\left( g_{t}^{\prime }\mathfrak{g}%
+y\right) 1\left\{ g_{t}^{\prime }\mathfrak{g}+y<0\right\} 1\left\{
y>0\right\} |
\end{equation*}%
with the right hand side being  free of $N,T$ and    integrable with respect to $y$:
$$
\int \left|\left( g_{t}^{\prime }\mathfrak{g}%
+y\right) 1\left\{ g_{t}^{\prime }\mathfrak{g}+y<0\right\} 1\left\{
y>0\right\}\right |dy =  \frac{( g_{t}^{\prime }\mathfrak{g})^2}{2}1\{ g_{t}^{\prime }\mathfrak{g}<0\}.
$$
 Hence by the dominated convergence theorem, for each fixed $g_{t},x_{t}$,
\begin{eqnarray*}
F_{NT,1}(g_{t},x_{t})\rightarrow F_{1}(g_{t},x_{t}):=\int \left( g_{t}^{\prime }\mathfrak{g}+y\right)
1\left\{ g_{t}^{\prime }\mathfrak{g}+y<0\right\} 1\left\{ y>0\right\} p_{%
	\mathcal{Z}_{t}}(0)dy=-\frac{1}{2}p_{\mathcal{Z}_{t}}(0)(g_{t}^{\prime }%
\mathfrak{g})^{2}1\{   g_{t}^{\prime }\mathfrak{g}<0 \}.
\end{eqnarray*}
Nnote that $p_{\mathcal{Z}_{t}}(0)$ does not depend on $N,T$, and is a function of $x_t, g_t$ through $\sigma^2_{h, x_t, g_t}$.
 In addition, let $\mathcal R(x_t, g_t)=C_0(x_t'd_0)^2 \frac{( g_{t}^{\prime }\mathfrak{g})^2}{2}1\{ g_{t}^{\prime }\mathfrak{g}<0\}$.
Then for all $N,T$,
\begin{eqnarray*}
 |(x_{t}^{\prime
}d_{0})^{2}F_{NT,1}(g_{t},x_{t})|
&\leq& (x_{t}^{\prime
}d_{0})^{2} |  \int
	\left( g_t^{\prime }\mathfrak{g} + y\right)1\left\{ g_t^{\prime } \mathfrak{g%
	} + y<0\right\}1\left\{ y>0\right\} p_{h_t^{\prime }\phi_0|
		g_t,x_t,  u_t=0}(\zeta_{NT}y)d y | \cr
		&\leq&
		C_0(x_{t}^{\prime
}d_{0})^{2}   \int |\left( g_t^{\prime }\mathfrak{g} + y\right)1\left\{ g_t^{\prime } \mathfrak{g%
	} + y<0\right\}1\left\{ y>0\right\} |dy
	\cr
	&=&C_0(x_t'd_0)^2 \frac{( g_{t}^{\prime }\mathfrak{g})^2}{2}1\{ g_{t}^{\prime }\mathfrak{g}<0\}
	= \mathcal R(x_t, g_t)
 \end{eqnarray*}
 Here $\mathcal R(x_t, g_t)$ is free of $N,T$, and
$\mathbb E(|\mathcal R(x_t, g_t)|| u_t=0)<\infty$.
 Therefore, still by the dominated convergence theorem,
 $\mathbb{E}[(x_{t}^{\prime
}d_{0})^{2}F_{NT,1}(g_{t},x_{t})|u_{t}=0]\rightarrow \mathbb{E}%
[(x_{t}^{\prime }d_{0})^{2}F_{1}(g_{t},x_{t})|u_{t}=0]$.
Using the similar argument, we also reach:
 $\mathbb{E}[(x_{t}^{\prime
}d_{0})^{2}F_{NT,2}(g_{t},x_{t})|u_{t}=0]\rightarrow \mathbb{E}%
[(x_{t}^{\prime }d_{0})^{2}F_{2}(g_{t},x_{t})|u_{t}=0]$,
where
$$
 F_{2} (g_t, x_t) := \int \left(
	g_t^{\prime }\mathfrak{g} +y\right)1\left\{ g_t^{\prime }\mathfrak{g} + y>
	0\right\} 1\left\{y \leq 0\right\} p_{\mathcal Z_t}(0)d y
	=\frac{1}{2}p_{\mathcal{Z}	_{t}}(0)(g_{t}^{\prime }\mathfrak{g})^{2} 1\{   g_{t}^{\prime }\mathfrak{g}>0 \}.
$$
 So
\begin{eqnarray*}
\mathbb{\breve{C}}_{NT,2}(\mathfrak{g})
&=&-\widetilde{M}_{NT}2p_{u_{t}}(0)%
\mathbb{E}[(x_{t}^{\prime }d_{0})^{2}F_{NT,1}(g_{t},x_{t})|u_{t}=0]+\widetilde{M}%
_{NT}2p_{u_{t}}(0)\mathbb{E}[(x_{t}^{\prime
}d_{0})^{2}F_{NT,2}(g_{t},x_{t})|u_{t}=0]\cr
&\rightarrow& -2\mathbb{E}%
[(x_{t}^{\prime }d_{0})^{2}p_{u_{t}}(0)F_{1}(g_{t},x_{t})|u_{t}=0]+2\mathbb{E%
}[(x_{t}^{\prime }d_{0})^{2}p_{u_{t}}(0)F_{2}(g_{t},x_{t})|u_{t}=0]\cr
	&=&(\mathbb{E}(x_{t}^{\prime
}d_{0})^{2}(g_{t}^{\prime }\mathfrak{g})^{2}|u_{t}=0,\mathcal{Z}_{t}=0)p_{u_{t},\mathcal{Z}_{t}}(0,0)\cr
&:=& C(\mathfrak{g}).
\end{eqnarray*}

\medskip
\noindent
\textbf{Step II.2}: $\mathbb{\breve{C}}_{NT,2}(H_{T}g)\overset{P}{\longrightarrow }C(g) $
\medskip

Again by the extended CMT (Lemma \ref{CMT-extension}), due to the pointwise convergence of  $\mathbb{\breve{C}}_{NT,2}(\mathfrak{g})$, similar to the proof of step I.2, it suffices to prove, for any converging sequence $\mathfrak{g}_T\to \mathfrak{g}$ on a compact space,
$
|\mathbb{\breve{C}}_{NT,2}(\mathfrak{g}_T)-  \mathbb{\breve{C}}_{NT,2}(\mathfrak{g})|\to 0.
$
By definition, $|\mathbb{\breve{C}}_{NT,2}(\mathfrak{g}_T)- \mathbb{\breve{C}}_{NT,2}(\mathfrak{g})|\leq a_1+a_2$, where
\begin{eqnarray*}
a_1&=& \widetilde{M}_{NT}2p_{u_{t}}(0)%
\mathbb{E}[(x_{t}^{\prime }d_{0})^{2}|z(\mathfrak{g}_T)-z(\mathfrak{g})||u_{t}=0]\cr
z(\mathfrak{g}_T)&:=&\int
	\left( g_t^{\prime }\mathfrak{g}_T + y\right)1\left\{ g_t^{\prime } \mathfrak{g%
	} _T+ y<0\right\}1\left\{ y>0\right\} p_{h_t^{\prime }\phi_0|
		g_t,x_t,  u_t=0}(\zeta_{NT}y)d y\cr
		a_2&=& \widetilde{M}_{NT}2p_{u_{t}}(0)%
\mathbb{E}[(x_{t}^{\prime }d_{0})^{2}|\widetilde z(\mathfrak{g}_T)-\widetilde z(\mathfrak{g})||u_{t}=0]\cr
\widetilde z(\mathfrak{g}_T)&:=&\int
	\left( g_t^{\prime }\mathfrak{g}_T + y\right)1\left\{ g_t^{\prime } \mathfrak{g%
	} _T+ y>0\right\}1\left\{ y\leq 0\right\} p_{h_t^{\prime }\phi_0|
		g_t,x_t,  u_t=0}(\zeta_{NT}y)d y
\end{eqnarray*}
and $a_2$ is defined similarly. Note that
\begin{eqnarray*}
|z(\mathfrak{g}_T)-z(\mathfrak{g})|&\leq&\int
|	\left( g_t^{\prime }\mathfrak{g} _T+ y\right)1\left\{ g_t^{\prime } \mathfrak{g%
	} _T+ y<0\right\} -\left( g_t^{\prime }\mathfrak{g} + y\right)1\left\{ g_t^{\prime } \mathfrak{g%
	} + y<0\right\}    |1\left\{ y>0\right\} \cr
&&\cdot p_{h_t^{\prime }\phi_0|
		g_t,x_t,  u_t=0}(\zeta_{NT}y)d y\cr
&\leq&\int
|	  g_t^{\prime }(\mathfrak{g}_T-\mathfrak{g}) |1\left\{ g_t^{\prime } \mathfrak{g%
	} _T+ y<0\right\}    1\left\{ y>0\right\} p_{h_t^{\prime }\phi_0|
		g_t,x_t,  u_t=0}(\zeta_{NT}y)d y\cr
&&+\int
|	 \left( g_t^{\prime }\mathfrak{g} + y\right)||1\left\{ g_t^{\prime } \mathfrak{g%
	} + y<0\right\}  -1\left\{ g_t^{\prime } \mathfrak{g%
	} _T+ y<0\right\}   |1\left\{ y>0\right\} \cr
	&&\cdot p_{h_t^{\prime }\phi_0|
		g_t,x_t,  u_t=0}(\zeta_{NT}y)d y\cr
&\leq& C|g_t|_2^2|  \mathfrak{g} _T- \mathfrak{g}    |_2.
\end{eqnarray*}
Thus
$
a_1\leq O(1)  \mathbb{E}[(x_{t}^{\prime }d_{0})^{2} |g_t|_2^2|   u_{t}=0] |\mathfrak{g}_T-\mathfrak{g}|_2=o(1).
$ Similarly, $a_2=o(1)$, implying $\mathbb{\breve{C}}_{NT,2}(\mathfrak{g}_T)\to   \mathbb{\breve{C}}_{NT,2}(\mathfrak{g}).$
Hence by the extended CMT,  $\mathbb{\breve{C}}_{NT,2}(H_{T}g)\overset{P}{\longrightarrow }C(g) $.
So
\begin{eqnarray*}
	&& l_{NT}\mathbb{E}_{|u_t=0}(x_t^{\prime }\delta_0)^2(A_1-A_3+ A_2-A_4) = \mathbb{\breve{C}}_{NT,2}(H_Tg)+o(1)\cr &\overset{P}{\longrightarrow }&(\mathbb{E}(x_t^{\prime
	}d_0)^2 ( (g_t'Hg)^2|u_t=0, \mathcal{Z}_t=0)p_{u_t, \mathcal{Z}%
		_t}(0,0):=C(g).
\end{eqnarray*}

\medskip
\noindent
\textbf{Step II.3}: $C(g)=\lim_{\omega\to0}A\left( \omega, g\right)$
\medskip

 As $\omega\to0$, we have that $%
\zeta_\omega=\omega^{1/3} $, $M_\omega=\omega^{-1/3}$. Still use (\ref{e0.7add}) with $%
g_t^{\prime }Hg=a$, $\zeta_\omega^{-1}\mathcal{Z}_t=b$, and the formula $%
|a+b|-|b|=\Xi(a,b)$:
\begin{eqnarray*}
	A\left( \omega, g\right)&:=&M_\omega \mathbb{E }\left[(xd_0)^2\left(\left| g_t^{\prime }Hg +
	\zeta_\omega^{-1} \mathcal{Z}_t \right | - \left | \zeta_\omega^{-1} \mathcal{Z}_t
	\right | \right) \bigg{|} u_t=0\right] p_{u_t}(0) \cr &=& M_\omega\mathbb{E}%
	_{|u_t=0}(x_t^{\prime }d_0)^2p_{u_t}(0) \left(\left| a + b\right | - \left |
	b\right | \right) \cr &=&-M_\omega\mathbb{E}_{|u_t=0}(x_t^{\prime
	}d_0)^2p_{u_t}(0) a1\{a\leq 0\} 1\left\{b\leq 0\right\} \cr &&+M_\omega\mathbb{E%
	}_{|u_t=0}(x_t^{\prime }d_0)^2p_{u_t}(0) a1\left\{ a >0\right\} 1\{b>0\} \cr
	&& + M_\omega\mathbb{E}_{|u_t=0}(x_t^{\prime
	}d_0)^2p_{u_t}(0)\Delta(a,b)
\end{eqnarray*}
where $\Delta(a,b)  $  denotes the sum of the other terms in the expression of $\Xi(a,b)$ given in (\ref{e0.7add}).
We now aim to  obtain alternative expressions for the first two terms on the right hand side.
Note that conditional on $(x_{t},g_{t},u_{t}=0)$, $b=\zeta_\omega^{-1}\mathcal{Z}_t$ is Gaussian with zero
mean, so the first term on the right hand side can be replaced with
\begin{eqnarray*}
&&-M_\omega\mathbb{E}(x_{t}^{\prime }d_{0})^{2}p_{u_{t}}(0)a1\{a\leq 0\}1\left\{
b\leq 0\right\} \cr
&=&-M_\omega\mathbb{E}(x_{t}^{\prime
}d_{0})^{2}p_{u_{t}}(0)a1\{a\leq 0\}1\left\{ b>0\right\} \cr
&=&-M_\omega\mathbb{E}%
(x_{t}^{\prime }d_{0})^{2}p_{u_{t}}(0)a1\{a\leq 0\}1\left\{ b>-a\right\} -M_\omega\mathbb{E}(x_{t}^{\prime }d_{0})^{2}p_{u_{t}}(0)a1\{a\leq 0\}1\left\{
-a>b>0\right\}
\end{eqnarray*}%
Similarly, $1\{b>0\}$ in the second term on the right hand side of $A(\omega, g)$ can be replaced with
\begin{eqnarray*}
&&M_\omega\mathbb{E}(x_{t}^{\prime
}d_{0})^{2}p_{u_{t}}(0)a1\left\{ a>0\right\} 1\{b<-a\}+M_\omega\mathbb{E}%
(x_{t}^{\prime }d_{0})^{2}p_{u_{t}}(0)a1\left\{ a>0\right\} 1\{-a<b<0\}.
\end{eqnarray*}
These alternative expressions can be combined with $\Delta(a,b)$, to reach:
 (note that $M_\omega=\zeta _{k}^{-1}$ and $%
\zeta _{k}\rightarrow 0$ as $k\rightarrow 0$),
\begin{eqnarray*}
A\left( \omega, g\right)&=&-2\zeta _{k}^{-1}\mathbb{E}_{|u_{t}=0}(x_{t}^{\prime
}d_{0})^{2}p_{u_{t}}(0)\left( a+b\right) 1\left\{ a+b<0\right\} 1\left\{
b>0\right\} \cr
&&+2\zeta _{k}^{-1}\mathbb{E}_{|u_{t}=0}(x_{t}^{\prime
}d_{0})^{2}p_{u_{t}}(0)\left( a+b\right) 1\left\{ a+b>0\right\} 1\left\{
b\leq 0\right\} \cr
&=&-2\mathbb{E}_{|u_{t}=0}(x_{t}^{\prime
}d_{0})^{2}p_{u_{t}}(0)\int \left( a+b\right) 1\left\{ a+b<0\right\}
1\left\{ b>0\right\} p_{\mathcal{Z}_{t}}(\zeta _{k}b)db\cr
&&+2\mathbb{E}
_{|u_{t}=0}(x_{t}^{\prime }d_{0})^{2}p_{u_{t}}(0)\int \left( a+b\right)
1\left\{ a+b>0\right\} 1\left\{ b\leq 0\right\} p_{\mathcal{Z}_{t}}(\zeta
_{k}b)db\cr
&\rightarrow ^{(1)}&-2\mathbb{E}_{|u_{t}=0}(x_{t}^{\prime
}d_{0})^{2}p_{u_{t}}(0)\int \left( a+b\right) 1\left\{ a+b<0\right\}
1\left\{ b>0\right\} p_{\mathcal{Z}_{t}}(0)db\cr
&&+2\mathbb{E}%
_{|u_{t}=0}(x_{t}^{\prime }d_{0})^{2}p_{u_{t}}(0)\int \left( a+b\right)
1\left\{ a+b>0\right\} 1\left\{ b\leq 0\right\} p_{\mathcal{Z}_{t}}(0)db\cr
&=&%
\mathbb{E}_{|u_{t}=0}\left( x_{t}^{\prime }d_{0}\right) ^{2}p_{u_{t}}(0)p_{%
	\mathcal{Z}_{t}}(0)a^{2}\cr
	&=&(\mathbb{E}(x_{t}^{\prime
}d_{0})^{2}(g_{t}^{\prime }Hg)^2|u_{t}=0,\mathcal{Z}_{t}=0)p_{u_{t},\mathcal{Z}%
	_{t}}(0,0):=C(g).
\end{eqnarray*}
It remains to argue that (1) in the above limit holds by applying the DCT.
First, for each fixed $b$,  $ p_{\mathcal{Z}_{t}}(\zeta
_{\omega}b)\to  p_{\mathcal{Z}_{t}}(0)$. Secondly,
$\sup_x p_{\mathcal{Z}_{t}}(x)=\sup_x\frac{1}{\sqrt{2\pi\sigma^2_{h, x_t, g_t}}}\exp(-\frac{x^2}{2\sigma^2_{h, x_t, g_t}})
=(2\pi\sigma^2_{h,x_t, g_t})^{-1/2}<C_0
$
for some $C_0>0$, due to $\inf_{x_t, g_t}\sigma^2_{h,x_t, g_t}>c_0$ (by the assumption).
So
 in the integration: ($a=g_t'Hg$)
$$
\mathcal E_{NT}(a):=\int \left( a+b\right) 1\left\{ a+b<0\right\}
1\left\{ b>0\right\} p_{\mathcal{Z}_{t}}(\zeta _{k}b)db,
$$
$|\left( a+b\right) 1\left\{ a+b<0\right\}
1\left\{ b>0\right\} p_{\mathcal{Z}_{t}}(\zeta _{k}b)|< |\left( a+b\right) 1\left\{ a+b<0\right\}
1\left\{ b>0\right\}| C_0$,
where the right hand side is free of $N,T$ and is integrable:
$\int  |\left( a+b\right) 1\left\{ a+b<0\right\}
1\left\{ b>0\right\}|  db <\infty$ for each fixed $a$. Then  DCT implies
$
\mathcal E_{NT}(a)\to \mathcal E (a):= \int \left( a+b\right) 1\left\{ a+b<0\right\}
1\left\{ b>0\right\} p_{\mathcal{Z}_{t}}(0)db
$ for each fixed $a$. Thirdly,
$$
|(x_t^2d_0)^2 \mathcal E_{NT}(a)|
\leq (x_t^2d_0)^2 C_0 \int |\left( a+b\right) 1\left\{ a+b<0\right\}
1\left\{ b>0\right\}  |db \leq  0.5(x_t^2d_0)^2 C_0a^2
$$
with $a=g_t'Hg$, so that $0.5(x_t^2d_0)^2 C_0a^2$ is free of $N, T$ and is integrable: $\mathbb E_{|u_t=0}0.5(x_t^2d_0)^2 C_0a^2<\infty. $ Also, $(x_t^2d_0)^2 \mathcal E_{NT}(a)\to(x_t^2d_0)^2 \mathcal E(a)$ for each fixed $x_t, g_t$.  Thus applying DCT again yields
$$
\mathbb E_{|u_t=0}(x_t^2d_0)^2 \mathcal E_{NT}(a)\to\mathbb E_{|u_t=0} (x_t^2d_0)^2 \mathcal E(a).
$$
The same argument also applies to the second term on the right hand side of (1).

\section{Proofs for Section \ref{sec:inference}}

\subsection{Proof of Theorem \ref{t6.1}: known factor case}

\subsubsection{Proof of the distribution of $LR$}\label{scg.1.1}

	 Below we prove, under $H_0: h(\gamma_0)=0$,
	 $$
	T\cdot LR\to^d \sigma_{\varepsilon}^{-2}\min_{ g_h'\nabla h=0} \mathbb Q(\infty, g_h)   - \sigma_{\varepsilon}^{-2}\min_{ g} \mathbb Q(\infty, g).
	 $$

	 \begin{proof}
	 	Define
	 	$
	 	\widehat\gamma_h=\arg\min_{\alpha, h(\gamma)=0}\mathbb S_T(\alpha,\gamma),
	 	$
	 	$
	 	\widehat\alpha(\gamma)=\arg\min_{\alpha}\mathbb S_T(\alpha,\gamma),$ and $\widehat\alpha_h=\widehat\alpha(\widehat\gamma_h).
	 	$ Then,
	 	 \begin{eqnarray*}
	 	T	\min_{\alpha,\gamma} {\mathbb S}_T(\alpha,\gamma)LR &=& T[\mathbb S_T(\widehat\alpha_h, \widehat\gamma_h)-\mathbb S_T(\widehat\alpha,\widehat\gamma)] \\
	 		&=&A_1+A_2-A_3,\quad \text{where }\cr
	 		A_1&=&T[\mathbb S_T(\widehat\alpha_h, \widehat\gamma_h)-\mathbb S_T(\widehat\alpha_h,\gamma_0)]\cr
	 		A_2   &=&T[\mathbb S_T(\widehat\alpha_h, \gamma_0)-\mathbb S_T(\widehat\alpha,\gamma_0)]\cr
	 		A_3 &=&T[\mathbb S_T(\widehat\alpha, \widehat\gamma)-\mathbb S_T(\widehat{\alpha},\gamma_0)].
	 	\end{eqnarray*}

	 	Let us first prove a useful equality.   Note that
	 	$$
	 	T[\mathbb S_T(\alpha,\gamma)-\mathbb S_T(\alpha,\gamma_0)]
	 	=T[\mathbb R_T(\alpha,\gamma)- \mathbb R_T(\alpha,\gamma_0)]
	 	-T[\mathbb G_T(\alpha,\gamma)- \mathbb G_T(\alpha,\gamma_0)],
	 	$$
	 	where $\mathbb R_T$ and $\mathbb G_T$ are defined in the proof of Lemma \ref{cons-thm}.
	 	Also recall \begin{eqnarray*}
	 		\mathbb{K}_{2T}\left( g\right) &=&T\cdot  \mathbb E\left( x_{t}^{\prime }\delta
	 		_{0}\right) ^{2}\left\vert 1_{t}\left( \gamma _{0}+g\cdot r_{T}^{-1}\right)
	 		-1_{t}(\gamma_0)\right\vert \\
	 		\mathbb{K}_{3T}\left( g\right) &=&-2\sum_{t=1}^{T}\varepsilon
	 		_{t}x_{t}^{\prime }\delta _{0}\left( 1_{t}\left( \gamma _{0}+g\cdot
	 		r_{T}^{-1}\right) -1_{t}(\gamma_0)\right) .
	 	\end{eqnarray*}
	 	Here    $r_T =T^{1-2\varphi}$, $1_t(\gamma)=1\{f_t'\gamma>0\}$.
	 	Uniformly over $|\gamma-\gamma_0|_2<Cr_T^{-1}, $ $|\alpha-\alpha_0|_2<CT^{-1/2}$,  and $g=r_T(\gamma-\gamma_0)$,
	 	we  have
	 	\begin{eqnarray*}
	 		&&T[\mathbb R_T(\alpha, \gamma)-\mathbb R_T({\alpha}, \gamma_0)]\cr
	 		&=&T\delta'\frac{1}{T}\sum_t\left[x_tx_t' |1\{f_t'\gamma>0\}-1\{f_t'\gamma_0>0\}|- \mathbb Ex_tx_t' |1\{f_t'\gamma>0\}-1\{f_t'\gamma_0>0\}|\right]\delta\cr
	 		&&+T\alpha'\frac{2}{T}\sum_t
	 		[Z_t(\gamma)-Z_t(\gamma_0)]Z_t(\gamma_0)'(\alpha-\alpha_0)\cr
	 		&&+T \mathbb E[(x_t'\delta)^2 -(x_t'\delta_0)^2]|1\{f_t'\gamma>0\}-1\{f_t'\gamma_0>0\}|+T\mathbb  E(x_t'\delta_0)^2 |1\{f_t'\gamma>0\}-1\{f_t'\gamma_0>0\}|\cr
	 		&=&T \mathbb E(x_t'\delta_0)^2 |1\{f_t'\gamma>0\}-1\{f_t'\gamma_0>0\}|+o_P(1)\cr
	 		&=&\mathbb K_{2T}(g)+o_P(1)
	 	\end{eqnarray*}
 	and
	 	\begin{eqnarray*}
	 		-T[\mathbb G_T(\alpha, \gamma)-\mathbb G_T({\alpha}, \gamma_0)]&=&
	 		-2
	 		\sum_{t=1}^{T}\varepsilon _{t}x_t'(\delta-\delta_0)(1\{f_t'\gamma>0\}-1\{f_t'\gamma_0>0\})\cr
	 		&&-2
	 		\sum_{t=1}^{T}\varepsilon _{t}x_t'\delta_0(1\{f_t'\gamma>0\}-1\{f_t'\gamma_0>0\})\cr
	 		&=&\mathbb{K}_{3T}\left( g\right) +o_P(1).
	 	\end{eqnarray*}
	 	Hence uniformly over $|\gamma-\gamma_0|_2<Cr_T^{-1}, $ $|\alpha-\alpha_0|_2<CT^{-1/2}$,  and $g=r_T(\gamma-\gamma_0)$,
	 	\begin{equation}\label{ec.11}
	 	T[\mathbb S_T(\alpha,\gamma)-\mathbb S_T(\alpha,\gamma_0)]
	 	=\mathbb K_{2T}(g)+\mathbb{K}_{3T}\left( g\right) +o_P(1).
	 	\end{equation}

	 	We are now ready to analyze $A_1$.
	 	By Lemma \ref{l5.1}, $|\widehat\gamma_h-\gamma_0|_2=O_P(T^{-(1-2\varphi)})$ under $H_0.$ Also, in the proof of Lemma \ref{l5.1} we have shown that $|\widehat\alpha_h-\alpha_0|_2=O_P(T^{-1/2})$.
	 	Hence  apply (\ref{ec.11}) with $\alpha=\widehat\alpha_h$ and $\gamma=\widehat\gamma_h$,
	 	\begin{eqnarray*}
	 		A_1&=&T[\mathbb S_T(\widehat\alpha_h, \widehat\gamma_h)-\mathbb S_T(\widehat\alpha_h,\gamma_0)]=\mathbb K_{2T}(\widehat g_h)+\mathbb{K}_{3T}\left( \widehat g_h\right) +o_P(1).
	 	\end{eqnarray*}
	 	To analyze the right hand side, recall that in Proof of Theorem \ref{asdist-alpha-gamma},
	 	\begin{eqnarray*}
	 		\mathbb{K}_{T}\left( a,g\right) &=&T\left( \mathbb{S}_{T}\left( \alpha
	 		_{0}+a\cdot T^{-1/2},\gamma _{0}+g\cdot r_{T}^{-1}\right) -\mathbb{S}%
	 		_{T}\left( \alpha _{0},\gamma _{0}\right) \right) \cr
	 		&=&\mathbb{K}_{1T}\left( a\right) +\mathbb{K}%
	 		_{2T}\left( g\right) +\mathbb{K}_{3T}\left( g\right) +o_P\left( 1\right)
	 	\end{eqnarray*}%
	 	where $o_P\left( 1\right) $ is uniform over any compact set.
	 	Define
	 	\begin{eqnarray*}
	 		(\widehat a_h, \widehat g_h)&=& \arg\min_{a,  h(\gamma_0+g_hr_T^{-1})=0} \mathbb{K}_{T}\left( a,g_h\right) \cr
	 		\widehat g_h&=&T^{-1+2\varphi }\left(\widehat \gamma_h -\gamma _{0}\right)
	 		\cr
	 		\widetilde g_h&=&\arg\min_{h(\gamma_0+g_hr_T^{-1})=0}\mathbb K_{2T}( g_h)+\mathbb{K}_{3T}\left(  g_h\right).
	 	\end{eqnarray*}%
	 	Then
	 	$
	 	\mathbb{K}_{T}\left( \widehat a_h,\widehat g_h\right)
	 	\leq  \mathbb{K}_{T}\left( \widehat a_h,\widetilde g_h\right),
	 	$
	 	implying
	 	\begin{eqnarray*}
	 		\mathbb{K}_{T}\left( \widehat a_h,\widehat g_h\right)&= & \mathbb K_{2T}(\widehat g_h)+\mathbb{K}_{3T}\left( \widehat  g_h\right) + \mathbb{K}_{1T}\left(\widehat  a_h\right)+o_P(1)\cr
	 		&\leq&\mathbb{K}_{T}\left( \widehat a_h,\widetilde g_h\right)\cr
	 		\mathbb{K}_{T}\left( \widehat a_h,\widetilde g_h\right)&=&\mathbb K_{2T}(\widetilde g_h)+\mathbb{K}_{3T}\left( \widetilde g_h\right) + \mathbb{K}_{1T}\left(\widehat  a_h\right)+o_P(1)\cr
	 		&\leq& \mathbb K_{2T}(\widehat g_h)+\mathbb{K}_{3T}\left( \widehat  g_h\right) + \mathbb{K}_{1T}\left(\widehat  a_h\right)+o_P(1).
	 	\end{eqnarray*}
	 	Thus
	 	\begin{eqnarray*}
	 		\mathbb K_{2T}(\widehat g_h)+\mathbb{K}_{3T}\left( \widehat  g_h\right) &=&\mathbb K_{2T}(\widetilde g_h)+\mathbb{K}_{3T}\left( \widetilde g_h\right) +o_P(1)\cr
	 		&=&\min_{h(\gamma_0+g_hr_T^{-1})=0}\mathbb K_{2T}( g_h)+\mathbb{K}_{3T}\left(  g_h\right)+o_P(1)
	 	\end{eqnarray*}

	 	These imply, with $\mathbb Q_T(g):=\mathbb{K}_{2T}\left(  g\right) +\mathbb{K}_{3T}\left(  g\right) $,
	 	\begin{eqnarray*}
	 		A_1
	 		&=&\mathbb K_{2T}(\widehat g_h)+\mathbb{K}_{3T}\left( \widehat g_h\right) +o_P(1)=\min_{ h(\gamma_0+g_hr_T^{-1})=0}  {\mathbb{K}_{2T}\left(  g_h\right) +\mathbb{K}_{3T}\left(  g_h\right) } +o_P(1)\cr
	 		&=&\min_{ h(\gamma_0+g_hr_T^{-1})=0} \mathbb Q_T(g_h)+o_P(1) \cr
	 		&=& {\min_{ r_T\{h(\gamma_0+g_hr_T^{-1})-h(\gamma_0)\}=0 }  \mathbb Q_T(g_h)  +o_P(1)}, \quad (\text{under }H_0: h(\gamma_0)=0), \cr
	 		&=& \min_{ g_h'\nabla h =0} \mathbb Q_T(g_h)  +o_P(1).
	 	\end{eqnarray*}

	 	As for $A_2$, Lemma \ref{l5.1}  shows that $A_2=o_P(1).$
	 	As for $A_3$, by definition $(\widehat a, \widehat g)= \arg\min_{a,  g} \mathbb{K}_{T}\left( a,g\right) $ and
	 	$ \widehat g=T^{-1+2\varphi }\left(\widehat \gamma -\gamma _{0}\right)
	 	$.
	 	Apply (\ref{ec.11}) with $\alpha=\widehat\alpha$ and $\gamma=\widehat\gamma$,
	 	\begin{eqnarray*}
	 		A_3&=&T[
	 		\mathbb S_T(\widehat\alpha, \widehat\gamma)-\mathbb S_T(\widehat{\alpha},\gamma_0)]=\mathbb{K}_{2T}\left( \widehat g\right) +\mathbb{K}_{3T} (\widehat g)+o_P(1)\cr
	 		&=&\min_{ g} \mathbb Q_T(g)+o_P(1).
	 	\end{eqnarray*}%

	 	Together, we have
	 	\begin{eqnarray*}
	 	T	\min_{\alpha,\gamma} {\mathbb S}_T(\alpha,\gamma)LR=A_1+A_2-A_3= \min_{ g_h'\nabla h =0} \mathbb Q_T(g_h) -\min_{ g} \mathbb Q_T(g)+o_P(1).
	 	\end{eqnarray*}

	 	Note that
	 	$
	 	\mathbb Q_T(\cdot)
	 	\Rightarrow
	 	\mathbb Q(\infty, \cdot)$.
	 	In addition, the operator $
	 	\mathcal P: f\to \min_{ g_h'\nabla h =0} f(g_h)    -\min_{ g} f(g)
	 	$
	 	is continuous in $f$ with respect to the metric (\textit{essential supremum})$|f_1-f_2|_\infty=\inf\{M: |f_1(x)-f_2(x)|<M\textit{ almost surely}\}$.
	 	Hence by the continuous mapping theorem,  and the fact that $\min_{\alpha,\gamma} {\mathbb S}_T(\alpha,\gamma)\to^P\sigma_{\varepsilon}^2$,   $$
	 T\cdot LR\to^d \sigma_{\varepsilon}^{-2}\min_{ g_h'\nabla h =0} \mathbb Q(\infty, g_h)   - \sigma_{\varepsilon}^{-2}\min_{ g} \mathbb Q(\infty, g).
	 	$$

	  \end{proof}

	 \subsubsection{Proof of the  distribution of $LR_k^*$}

	 \begin{proof}
	 	 We first prove that under $\mathcal H_0: h(\gamma_0)=0$,
	 	 $$
	 	T \mathbb{ {S}}_{T}^{\ast }(\widehat\alpha^*,\widehat\gamma^*) LR^*_k=  {\min_{ g_h'\nabla h=0} \mathbb Q_T^*(g_h)    -\min_{ g} \mathbb Q_T^*(g)+o_{P^*}(1)}
	 	 $$
	 	 where $\mathbb Q_T^*(g)=  \sum_t\left( x_{t}^{\prime }\widehat\delta
	 	 \right) ^{2}\left\vert 1_{t}\left( \widehat\gamma +g\cdot r_{T}^{-1}\right)
	 	 -1_{t}(\widehat\gamma)\right\vert
	 	 -2\sum_{t=1}^{T}\eta_t\widehat\varepsilon
	 	 _{t}x_{t}^{\prime }\widehat\delta \left( 1_{t}\left(\widehat \gamma +g\cdot
	 	 r_{T}^{-1}\right) -1_{t}(\widehat\gamma)\right).
	 	 $

	 	 To do so, define \begin{eqnarray*}
	 	 	\alpha^*(\gamma)&=&\arg\min_{\alpha}\mathbb S_T^*(\alpha,\gamma),
	 	 	\cr
	 	 	\gamma^*(\alpha)&=&\arg\min_{\gamma}\mathbb S_T^*(\alpha,\gamma)
	 	 	\cr
	 	 	\gamma_h^*(\alpha)&=&\arg\min_{h(\gamma)=h(\widehat\gamma)}\mathbb S_T^*(\alpha,\gamma) .
	 	 \end{eqnarray*}
	 	 We have
	 	 $
	 	T \mathbb{ {S}}_{T}^{\ast }(\widehat\alpha^*,\widehat\gamma^*)LR_k^* =T[\mathbb S_T^*(\widehat\alpha_h^*, \widehat\gamma_h^*)-\mathbb S_T^*(\widehat\alpha^*,\widehat\gamma^*)] =A_1^*+A_2^*-A^*_3,
	 	 $
		  where
	 	 \begin{eqnarray*}
	 	 	A^*_1&=&T[\mathbb S_T^*(\widehat\alpha_h^*, \widehat\gamma_h^*)-\mathbb S_T^*(\widehat\alpha_h^*,\widehat\gamma)],\quad A_2^*    = T[\mathbb S_T^*(\widehat\alpha_h^*, \widehat\gamma)-\mathbb S_T(\widehat\alpha^*,\widehat\gamma)] ,\quad A^*_3 =T[\mathbb S_T^*(\widehat\alpha^*, \widehat\gamma^*)-\mathbb S_T^*(\widehat{\alpha}^*,\widehat\gamma)].
	 	 \end{eqnarray*}

	 	 Define
	 	 \begin{align*}
	 	 \mathbb{R}_{T}^*\left( \alpha ,\gamma \right) &:= \frac{1}{T}%
	 	 \sum_{t=1}^{T}\left( Z_{t}\left( \gamma \right) ^{\prime }\alpha
	 	 -Z_{t}\left( \widehat \gamma\right) ^{\prime }\widehat\alpha\right) ^{2} \\
	 	 \mathbb{G}_{T}^*\left( \alpha ,\gamma \right) &:= \frac{2}{T}%
	 	 \sum_{t=1}^{T}\eta_t\widehat\varepsilon _{t}Z_{t}\left( \gamma \right) ^{\prime }\alpha
	 	 \cr
	 	 \mathbb{K}_{1T}^*\left( a\right) &:=a^{\prime }\frac{1}{T}\sum_tZ_{t}\left(\widehat \gamma \right)
	 	 Z_{t}\left( \widehat\gamma \right) ^{\prime }a-\frac{2}{\sqrt{T}}%
	 	 \sum_{t=1}^{T}\eta_t\widehat\varepsilon _{t}Z_{t}\left( \widehat\gamma \right) ^{\prime }a, \cr
	 	 \mathbb{K}_{2T}^*\left( g\right)&:= T\cdot  \frac{1}{T}\sum_t\left( x_{t}^{\prime }\widehat\delta
	 	 \right) ^{2}\left\vert 1_{t}\left( \widehat\gamma +g\cdot r_{T}^{-1}\right)  -1_{t}(\widehat\gamma)\right\vert, \\
	 	 \mathbb{K}_{3T}^*\left( g\right) &:= -2\sum_{t=1}^{T}\eta_t\widehat\varepsilon _{t}x_{t}^{\prime }\widehat\delta \left( 1_{t}\left(\widehat \gamma +g\cdot
	 	 r_{T}^{-1}\right) -1_{t}(\widehat\gamma)\right)\cr
	 	 \mathbb{K}^*_{T}\left( a,g\right) &:=T\left( \mathbb{S}^*_{T}\left( \widehat \alpha
	 	 +a\cdot T^{-1/2},\widehat\gamma +g\cdot r_{T}^{-1}\right) -\mathbb{S}^*
	 	 _{T}\left(\widehat \alpha ,\widehat\gamma \right) \right) .
	 	 \end{align*}
	 	 We first show two important equalities: \\
	 	 (i) $ T[\mathbb S_T^*( \alpha,  \gamma)-\mathbb S_T^*(\alpha,\widehat\gamma)]=\mathbb K_{2T}^*(g)+\mathbb{K}_{3T}^*\left( g \right) +o_{P^*}(1)$
	 	 \\
	 	 (ii) $ \mathbb{K}_{T} ^*  \left( a,g\right) =
	 	 \mathbb{K}_{1T}^*\left( a\right) +\mathbb{K}^*
	 	 _{2T}\left( g\right) +\mathbb{K}_{3T}^*\left( g\right) +o_{P^*}\left( 1\right)$, \\
	 	 where $o_{p^*}\left( 1\right) $ is uniform over any compact set.

	 	 For (i), note that
	 	 $
	 	 T[\mathbb S_T^*( \alpha,  \gamma)-\mathbb S_T^*(\alpha,\widehat\gamma)]
	 	 =T[\mathbb R_T^*( \alpha,  \gamma)-\mathbb R_T^*({\alpha}, \widehat\gamma)]-T[\mathbb G_T^*( \alpha ,  \gamma )-\mathbb G_T^*({\alpha} , \widehat\gamma)].
	 	 $
	 	 To bound the right hand side,  note that uniformly for $|\alpha-\widehat\alpha|_2=O_{P^*}(T^{-1/2})$,
	 	 $|\gamma-\widehat\gamma|_2=O_{P^*}(r_T^{-1})$ and
	 	 $g=r_T(\gamma-\widehat\gamma)$,
	 	 \begin{eqnarray*}
	 	 	&&T[\mathbb R_T^*  (\alpha, \gamma)-\mathbb R^* _T({\alpha}, \widehat\gamma)]\cr
	 	 	&=& T\frac{1}{T}\sum_t[\delta'x_t]^2 |1\{f_t'\gamma>0\}-1\{f_t'\widehat\gamma>0\}|  +T\frac{2}{T}\sum_t\delta'x_t(1\{f_t'\gamma>0\}-1\{f_t'\widehat\gamma>0\})Z_t(\widehat\gamma)'(\alpha-\widehat\alpha)\cr
	 	 	&=&\mathbb K^* _{2T}(g)+O_P(1) T^{1-\varphi} |\delta-\widehat\delta|_2 \frac{1}{T}\sum_t|x_t|^2_2 |1\{f_t'\gamma>0\}-1\{f_t'\widehat\gamma>0\}| \cr
	 	 	&=&\mathbb K^* _{2T}(g)+O_P(1) T^{1-\varphi} |\delta-\widehat\delta|_2 |\gamma-\widehat\gamma|_2+o_P(1) =\mathbb K_{2T}^* (g)+o_{P^*}(1)\cr
	 	 	&&-T[\mathbb G_T^* (\alpha, \gamma)-\mathbb G_T^* ({\alpha}, \widehat\gamma)]=
	 	 	-2
	 	 	\sum_{t=1}^{T}\eta_t\widehat\varepsilon _{t}x_t'(\delta-\widehat\delta)(1\{f_t'\gamma>0\}-1\{f_t'\widehat\gamma>0\})+\mathbb{K}_{3T}^*\left( g\right) \cr
	 	 	& =&\mathbb{K}_{3T}^*\left( g\right)   +o_{P^*}(1),
	 	 \end{eqnarray*}
	 	 {where we applied  Lemma \ref{A-rates:lemma-used} on the bootstrap sampling space}
	 	 to show \\$\sum_{t=1}^{T}\eta_t\widehat\varepsilon _{t}x_t'(\delta-\widehat\delta)(1\{f_t'\gamma>0\}-1\{f_t'\widehat\gamma>0\})= o_{P^*}(1)$.
	 	 Therefore,  uniformly in $g$, $|\gamma-\widehat\gamma|_2=O_{P^*}(r_T^{-1})$, and $|\alpha-\widehat\alpha|_2=O_P(T^{-1/2})$,

	 	 \begin{equation}\label{e.20}
	 	 T[\mathbb S_T^*( \alpha,  \gamma)-\mathbb S_T^*(\alpha,\widehat\gamma)]=\mathbb K_{2T}^*(g)+\mathbb{K}_{3T}^*\left( g \right) +o_{P^*}(1).
	 	 \end{equation}

	 	 For (ii),  note that uniformly for $\alpha-\widehat\alpha=T^{-1/2}a$ and $\gamma-\widehat\gamma= r_T^{-1}g$, we have
	 	 $$
	 	 \mathbb{K}^*_{T}\left( a,g\right) = \mathbb{K}_{1T}^*\left( a\right) +\mathbb{K}^*
	 	 _{2T}\left( g\right) +\mathbb{K}_{3T}^*\left( g\right)+ \Delta^*_1+\Delta_2^*+\Delta^*_3
	 	 $$
	 	 where
	 	 \begin{eqnarray*}
	 	 	\Delta_1^*(\alpha, \gamma) &=&
	 	 	2\sum_t x_t' (\widehat\delta-\delta) \eta_t\widehat\epsilon_t (1_{t}\left(\gamma \right) -1_{t}(\widehat\gamma)) =o_{P^*}(1)\cr
	 	 	\Delta_2^*(\alpha, \gamma) &=&
	 	 	\frac{2}{\sqrt{T}}\sum_t a'Z_t(\widehat\gamma)   x_t'\delta (1_{t}\left(\gamma \right) -1_{t}(\widehat\gamma))    \cr
	 	 	\Delta_3^*(\alpha, \gamma) &=&o_P(1)
	 	 	\sum_t  [( x_t'\delta)^2-( x_t'\widehat\delta)^2 ] |1_{t}\left(\gamma \right) -1_{t}(\widehat\gamma)|    =o_P(1)
	 	 \end{eqnarray*}%
	 	 {where we applied  Lemma \ref{A-rates:lemma-used} on the bootstrap sampling space to bound the first term}, and applied the same lemma on the original space to bound the other two terms.

	 	 We are now ready to analyze $A^*_1.$
	 	 By Lemma \ref{bootstrap.high_level},
	 	 $	|\widehat\alpha_h^*-\widehat\alpha|_2=O_{P^*}(T^{-1/2})$, and $	|\widehat\gamma_h^*-\widehat\gamma|_2=O_P(T^{-(1-2\varphi)}).$
	 	 Apply  (\ref{e.20}) with $\alpha=\widehat\alpha_h^*$ and $\gamma=  \widehat\gamma_h^*=\widehat\gamma +\widehat g_h^*r_T^{-1}
	 	 $,
	 	 \begin{eqnarray*}
	 	 	A^*_1&=&T[\mathbb S_T^*(\widehat\alpha_h^*, \widehat\gamma_h^*)-\mathbb S_T^*(\widehat\alpha_h^*,\widehat\gamma)]=\mathbb K_{2T}^*(\widehat g_h^*)+\mathbb{K}_{3T}^*\left( \widehat g_h^*\right) +o_{P^*}(1).
	 	 \end{eqnarray*}
	 	 Define
	 	 \begin{eqnarray*}
	 	 	\widehat a^*&=&\sqrt{T}(\widehat\alpha_h^*-\widehat\alpha) ,\quad  \widehat g_h^*=r_T(\widehat\gamma_h^*-\widehat\gamma)
	 	 	\cr
	 	 	\widetilde g_h^*&:=&\arg\min_{h(\widehat\gamma+gr_T^{-1})=h(\widehat\gamma)}\mathbb K_{2T}^*( g)+\mathbb{K}_{3T}^*\left(  g\right),\quad h(\widehat\gamma+\widetilde g^*_hr_T^{-1})=h(\widehat\gamma).
	 	 \end{eqnarray*}

	 	 By  Lemma \ref{bootstrap.high_level},
	 	 $
	 	 \mathbb S_T^*(\widehat\alpha_h^*,\widehat\gamma_h^*)
	 	 \leq \min_{\alpha, h(\gamma)=h(\widehat\gamma)}\mathbb S_T^*(\alpha,\gamma)+ o_{P^*}(T^{-1})$, and $h(\widehat\gamma_h^*)=h(\widehat\gamma)
	 	 $
	 	 we have
	 	 \begin{eqnarray*}
	 	 	\mathbb{K}^*_{T}\left( \widehat a^*,\widehat g^*_h\right)
	 	 	&=&T\left( \mathbb{S}^*_{T}\left( \widehat\alpha_h^*,\widehat\gamma_h^*\right) -\mathbb{S}^*
	 	 	_{T}\left(\widehat \alpha ,\widehat\gamma \right) \right)
	 	 	\leq
	 	 	T\left( \min_{\alpha, h(\gamma)=h(\widehat\gamma)}\mathbb S_T^*(\alpha,\gamma) -\mathbb{S}^*
	 	 	_{T}\left(\widehat \alpha ,\widehat\gamma \right) \right)  + o_{P^*}(1)\cr
	 	 	&= &\min_{a, h(\widehat\gamma+r_T^{-1}g)=h(\widehat\gamma)}\mathbb K_T^*(a, g)
	 	 	+ o_{P^*}(1)\leq  \mathbb{K}_{T}\left( \widehat a^*,\widetilde g_h^*\right)+ o_{P^*}(1).
	 	 \end{eqnarray*}%
	 	 So by $ \mathbb{K}_{T} ^*  \left( a,g\right) =
	 	 \mathbb{K}_{1T}^*\left( a\right) +\mathbb{K}^*
	 	 _{2T}\left( g\right) +\mathbb{K}_{3T}^*\left( g\right) +o_{P^*}\left( 1\right)$,
	 	 \begin{eqnarray*}
	 	 	&&\mathbb K_{2T}^*(\widehat g_h^*)+\mathbb{K}^*_{3T}\left( \widehat  g_h^*\right) + \mathbb{K}^*_{1T}\left(\widehat  a_h^*\right)+o_{P^*}(1)
	 	 	=  \mathbb{K}_{T}^*\left( \widehat a^*,\widehat g_h^*\right)
	 	 	\leq \mathbb{K}^*_{T}\left( \widehat a^*,\widetilde g_h^*\right)+ o_{P^*}(1).
	 	 \end{eqnarray*}
	 	 On the other hand,  by the definition of $ \widetilde g_h^*$,
	 	 \begin{eqnarray*}
	 	 	\mathbb{K}_{T}^*\left( \widehat a^*,\widetilde g^*_h\right)&=&\mathbb K^*_{2T}(\widetilde g^*_h)+\mathbb{K}^*_{3T}\left( \widetilde g^*_h\right) + \mathbb{K}^*_{1T}\left(\widehat  a^*_h\right)+o_{P^*}(1)\cr
	 	 	&\leq& \mathbb K^*_{2T}(\widehat g_h^*)+\mathbb{K}^*_{3T}\left( \widehat  g^*_h\right) + \mathbb{K}^*_{1T}\left(\widehat  a^*_h\right)+o_{P^*}(1)
	 	 \end{eqnarray*}
	 	 and note that  $
	 	 \mathbb Q_T^*(g)= \mathbb K^*_{2T}(g)+\mathbb{K}^*_{3T}\left( g\right)
	 	 $.  So
	 	 \begin{eqnarray*}
	 	 	\mathbb K_{2T}^*(\widehat g^*_h)+\mathbb{K}^*_{3T}\left( \widehat  g^*_h\right) &=&\mathbb K^*_{2T}(\widetilde g^*_h)+\mathbb{K}_{3T}^*\left( \widetilde g^*_h\right) +  o_{P^*}(1)\cr
	 	 	&=&\min_{h(\widehat\gamma+gr_T^{-1})=h(\widehat\gamma)}\mathbb K_{2T}^*( g)+\mathbb{K}_{3T}^*\left(  g\right)
	 	 	+o_{P^*}(1)\cr
	 	 	&=&\min_{h(\widehat\gamma+gr_T^{-1})=h(\widehat\gamma)}\mathbb Q_T^*(g)   +o_{P^*}(1).
	 	 \end{eqnarray*}
	 	 These imply
	 	 \begin{eqnarray*}
	 	 	A^*_1&=&\mathbb K_{2T}^*(\widehat g_h^*)+\mathbb{K}_{3T}^*\left( \widehat g_h^*\right) +o_{P^*}(1)=\min_{h(\widehat\gamma+gr_T^{-1})=h(\widehat\gamma)}\mathbb Q_T^*(g)     +o_{P^*}(1).
	 	 \end{eqnarray*}

	 	 As for $A_2^*$, Lemma \ref{bootstrap.high_level}  shows that
		 $\widehat\alpha_h^*-\widehat\alpha^*=o_{P^*}(T^{-1/2})$.
	 Hence similar proof as in Lemma \ref{l5.1} shows 	 $A_2^*=o_{P^*}(1).$

	 	 As for $A^*_3$, let
	 	 $\widehat g^*=r_T(\widehat\gamma^*-\widehat\gamma).
	 	 $
	 	 Apply  (\ref{e.20}) with $\alpha=\widehat\alpha^*, $ and $\gamma=\widehat\gamma^*$,  then
	 	 $
	 	 A^*_3 =[\mathbb S_T^*(\widehat\alpha^*, \widehat\gamma^*)-\mathbb S_T^*(\widehat\alpha^*,\widehat\gamma)]
	 	 =\mathbb{K}_{2T}\left( \widehat g^*\right) +\mathbb{K}_{3T} (\widehat g^*)+o_{P^*}(1).
	 	 $
	 	 Now let  $
	 	 \widetilde g^*=\arg\min_g \mathbb K_{2T}^*(g)+\mathbb K_{3T}^*(g).
	 	 $ Then by Lemma \ref{bootstrap.high_level} ,
	 	 $ \mathbb S_T^*(\widehat\alpha^*,\widehat\gamma^*)
	 	 \leq  \min_{\alpha, \gamma}\mathbb S_T^*(\alpha,\gamma)+ o_{P^*}(T^{-1}) $.
	 	 Hence,
	 	 \begin{eqnarray*}
	 	 	&&\mathbb{K}_{1T}^*(\widehat a^*)
	 	 	+\mathbb{K}_{2T}^*(\widehat g^*)
	 	 	+\mathbb{K}_{3T}^*(\widehat g^*)
	 	 	+o_{P^*}(1)\cr
	 	 	&=&
	 	 	\mathbb K_T^*\left( \widehat a^*,\widehat g^*\right) \cr
	 	 	&=&T\left( \mathbb{S}^*_{T}\left( \widehat \alpha^*,\widehat\gamma^*\right) -\mathbb{S}^*
	 	 	_{T}\left(\widehat \alpha ,\widehat\gamma \right) \right)
	 	 	\leq
	 	 	T\left( \min_{\alpha, \gamma}\mathbb S_T^*(\alpha,\gamma) -\mathbb{S}^*
	 	 	_{T}\left(\widehat \alpha ,\widehat\gamma \right) \right)  +o_{P^*}(1)\cr
	 	 	&= &\min_{a, g}\mathbb K_T^*(a, g)
	 	 	+o_{P^*}(1) \leq  \mathbb K_T^*(\widehat a^*, \widetilde g^*)
	 	 	+o_{P^*}(1)
	 	 	\cr
	 	 	&=&
	 	 	\mathbb{K}_{1T}^*(\widehat a^*)
	 	 	+\mathbb{K}_{2T}^*(\widetilde g^*)
	 	 	+\mathbb{K}_{3T}^*(\widetilde g^*)
	 	 	+o_{P^*}(1)\cr
	 	 	&\leq&
	 	 	\mathbb{K}_{1T}^*(\widehat a^*)
	 	 	+\mathbb{K}_{2T}^*(\widehat g^*)
	 	 	+\mathbb{K}_{3T}^*(\widehat g^*)
	 	 	+o_{P^*}(1).
	 	 \end{eqnarray*}%

	 	 This implies
	 	 $
	 	 \mathbb{K}_{2T}^*(\widetilde g^*)
	 	 +\mathbb{K}_{3T}^*(\widetilde g^*)
	 	 \leq \mathbb{K}_{2T}^*(\widehat g^*)
	 	 +\mathbb{K}_{3T}^*(\widehat g^*)
	 	 \leq
	 	 \mathbb{K}_{2T}^*(\widetilde g^*)
	 	 +\mathbb{K}_{3T}^*(\widetilde g^*)
	 	 +o_{P^*}(1).
	 	 $
	 	 So
	 	 \begin{eqnarray*}%
	 	 	A^*_3 &=& \mathbb{K}_{2T}^*(\widehat g^*)
	 	 	+\mathbb{K}_{3T}^*(\widehat g^*) +o_{P^*}(1)  =
	 	 	\mathbb{K}_{2T}^*(\widetilde g^*)
	 	 	+\mathbb{K}_{3T}^*(\widetilde g^*)
	 	 	+o_{P^*}(1) \cr
	 	 	&=&\min_g\mathbb Q_T^*(g)+o_{P^*}(1) .
	 	 \end{eqnarray*}%

	 	 Together, we have
	 	 \begin{eqnarray*}
	 	 T	{\mathbb S}_T^*(\widehat\alpha^*,\widehat\gamma^*)LR^*_k
	 	 	&=&A^*_1+A^*_2 -A^*_3=\min_{g_h: h(\widehat\gamma+g_hr_T^{-1})=h(\widehat\gamma)} \mathbb Q_T^*(g_h) -\min_{ g} \mathbb Q_T^*(g)+o_{P^*}(1)\cr
	 	 	& =& {\min_{ r_T\{h(\widehat\gamma+g_hr_T^{-1})-h(\widehat\gamma)\}=0 }  \mathbb Q_T^*(g_h) -\min_{ g} \mathbb Q_T^*(g)+o_{P^*}(1)}\cr
	 	 	&=& {\min_{ g_h'\nabla h=0} \mathbb Q_T^*(g_h)    -\min_{ g} \mathbb Q_T^*(g)+o_{P^*}(1)}.
 	 	\end{eqnarray*}
	 	 Here $\nabla h$ is constant since $h$ is linear.

	 	 Next,  recall
	 	 \begin{align*}
	 	 \mathbb{K}_{2T}^*\left( g\right)&:= T\cdot  \frac{1}{T}\sum_t\left( x_{t}^{\prime }\widehat\delta
	 	 \right) ^{2}\left\vert 1_{t}\left( \widehat\gamma +g\cdot r_{T}^{-1}\right)
	 	 -1_{t}(\widehat\gamma)\right\vert\cr
	 	 &= T\cdot  \frac{1}{T}\sum_t\left( x_{t}^{\prime }\delta_0
	 	 \right) ^{2}\left\vert 1_{t}\left( \widehat\gamma +g\cdot r_{T}^{-1}\right)
	 	 -1_{t}(\widehat\gamma)\right\vert+o_P(1)\cr
	 	 &=M_T(\widehat\gamma, g)+ o_P(1)
	 	 \end{align*}
	 	 where
	 	 $$
	 	 M_T(\gamma, g)=  T\cdot  \mathbb E\left( x_{t}^{\prime }\delta_0
	 	 \right) ^{2}\left\vert 1_{t}\left( \gamma +g\cdot r_{T}^{-1}\right)
	 	 -1_{t}(\gamma)\right\vert.
	 	 $$
	 	  For any  $\gamma_T\to \gamma_0$, and fixed $g$, we have
	 	 $
	 	 M_T(\gamma_T, g) \to \mathbb Q(\infty, g).
	 	 $
	 	 It then follows from the extended continuous mapping theorem that $M_T(\widehat\gamma, g)\to^P \mathbb Q(\infty, g)$ for each $g$.
	 	 So $\mathbb{K}_{2T}^*=\mathbb Q(\infty, g)+o_P(1)$ pointwise for each $g$.

	 	 Next, for
 $  \mathbb{K}_{3T}^*\left( g\right) := -2\sum_{t=1}^{T}\eta_t\widehat\varepsilon
	 	 _{t}x_{t}^{\prime }\widehat\delta \left( 1_{t}\left(\widehat \gamma +g\cdot
	 	 r_{T}^{-1}\right) -1_{t}(\widehat\gamma)\right)$,
	Lemma \ref{bootstrapempirical} shows that in the known factor case,
	 	  $ \mathbb{K}_{3T}^*\left( g\right) \Rightarrow^{*} 2W(g)$.

	 	 So $ \mathbb Q_T^*(.)= \mathbb K^*_{2T}(.)+\mathbb{K}^*_{3T}\left(.\right)\Rightarrow ^{*} \mathbb Q(\infty,.)$.  Here $\Rightarrow^*$ denotes the weak convergence with respect to the bootstrap distribution. It follows  that
	 	 \begin{eqnarray*}
	 	T 	\mathbb{ {S}}_{T}^{\ast }(\widehat\alpha^*,\widehat\gamma^*) LR^*_k&=&  {\min_{ g_h'\nabla h=0} \mathbb Q_T^*(g_h)    -\min_{ g} \mathbb Q_T^*(g)+o_{P^*}(1)}\cr
	 	 	&\to^{d^*}& \min_{ g_h'\nabla h=0} \mathbb Q(\infty, g_h)    -\min_{ g} \mathbb Q(\infty, g).
	 	 \end{eqnarray*}
	 	 In addition,
	 	 \begin{eqnarray*}
	 	 	{\mathbb S}_T^*(\widehat\alpha^*,\widehat\gamma^*)&=&
	 	 	{\mathbb S}_T^*(\widehat \alpha,\widehat \gamma)+o_{P^*}(1)
	 	 	=\frac{1}{T}\sum_t(\eta_t\widehat\varepsilon_t)^2+o_{P^*}(1)\cr
	 	 	& =&\mathbb E^*\frac{1}{T}\sum_t(\eta_t\widehat\varepsilon_t)^2+o_{P^*}(1)=\frac{1}{T}\sum_t\widehat\varepsilon_t^2+o_{P^*}(1)=\sigma_{\varepsilon}^2+o_{P^*}(1).
	 	 \end{eqnarray*}
	 	 Thus
	 	 $
	 	T\cdot  LR^*_k\to^{d^*}   \sigma_{\varepsilon}^{-2} \min_{ g_h'\nabla h=0} \mathbb Q(\infty, g_h)    - \sigma_{\varepsilon}^{-2} \min_{ g} \mathbb Q(\infty, g).
	 	 $

	 \end{proof}

	 \subsubsection{Technical Lemmas}

	 \begin{lem}\label{l5.1} Under $\mathcal H_0$,

	 	(i) $|\widehat\gamma_h-\gamma_0|_2=O_P(T^{-(1-2\varphi)})$.

	 	(ii)
	 	$T[\mathbb S_T(\widehat\alpha_h, \gamma_0)-\mathbb S_T(\widehat\alpha,\gamma_0)]=o_P(1)$
	 \end{lem}

	 \begin{proof}

	 	(i) The proof is ver similar to that of the rate for $\widehat\gamma$, so we only briefly sketch the main steps. First of all,    $h(\gamma_0)=0$ and $h(\widehat\gamma_h)=0$. By definition, we have
	 	\begin{eqnarray*}
	 		0 &\geq &\mathbb{S}_{T}\left( \widehat{\alpha }_h,\widehat{\gamma }_h\right) -%
	 		\mathbb{S}_{T}\left( \alpha _{0},\gamma _{0}\right)  \notag \\
	 		&=&\mathbb{R}_{T}\left( \widehat{\alpha }_h, \widehat{\gamma }_h\right) -\mathbb{G%
	 		}_{T}\left( \widehat{\alpha }_h, \widehat{\gamma }_h\right) +\mathbb{G}_{T}\left(
	 		\alpha _{0},\gamma _{0}\right) ,  \label{eq:S-S0-add}
	 	\end{eqnarray*}

	 	Similarly as before, we  can find some $c,c^{\prime }>0$ such that for
	 	sufficiently small $\left\vert \alpha -\alpha _{0}\right\vert _{2}$%
	 	\begin{eqnarray*}
	 		R\left( \alpha ,\gamma \right) &=&\mathbb E\left( Z_{t}\left( \gamma \right)
	 		^{\prime }\left( \alpha -\alpha _{0}\right) \right) ^{2}+\mathbb E\left(
	 		x_{t}^{\prime }\delta _{0}\left( 1_{t}\left( \gamma \right) -1_{t}\left(
	 		\gamma _{0}\right) \right) \right) ^{2} \\
	 		&&+2\mathbb E\left( x_{t}^{\prime }\delta _{0}\left( 1_{t}\left( \gamma \right)
	 		-1_{t}\left( \gamma _{0}\right) \right) \right) Z_{t}\left( \gamma \right)
	 		^{\prime }\left( \alpha -\alpha _{0}\right) \\
	 		&\geq &c\left\vert \alpha -\alpha _{0}\right\vert _{2}^{2}+cT^{-2\varphi
	 		}\left\vert \gamma -\gamma _{0}\right\vert _{2}-c^{\prime }\left\vert \alpha
	 		-\alpha _{0}\right\vert _{2}\left\vert \gamma -\gamma _{0}\right\vert
	 		_{2}T^{-\varphi },
	 	\end{eqnarray*}%
	 	where the first inequality is from the bounds and the second from the
	 	condition that $\left\vert \alpha -\alpha _{0}\right\vert _{2}$ is small. (this is guaranteed since $\widehat\alpha_h$ is consistent under $H_0$)
	 	Furthermore, we  still have, for  $0<\eta <c,$%
	 	\begin{eqnarray*}
	 		\left\vert \mathbb{G}_{T}\left( \alpha ,\gamma \right) -\mathbb{G}_{T}\left(
	 		\alpha _{0},\gamma _{0}\right) \right\vert &\leq &O_{P}\left( \frac{1}{\sqrt{%
	 				T}}\right) \left\vert \alpha -\alpha _{0}\right\vert _{2}+\eta T^{-2\varphi
	 		}\left\vert \gamma -\gamma _{0}\right\vert _{2}+O_{P}\left( \frac{1}{T}%
	 		\right)  \\
	 		\left\vert \mathbb{R}_{T}\left( \alpha ,\gamma \right) -R\left( \alpha
	 		,\gamma \right) \right\vert &\leq &\eta \left\vert \alpha -\alpha
	 		_{0}\right\vert _{2}^{2}+\eta T^{-2\varphi }\left\vert \gamma -\gamma
	 		_{0}\right\vert _{2}+O_{P}\left( \frac{1}{T}\right),
	 	\end{eqnarray*}%
	 	where the inequality is uniform in $\alpha $ and $\gamma $ in the sense that
	 	the sequences $O_{P}\left( \cdot \right) $ and $o_P\left( \cdot \right) $
	 	do not depend on $\alpha $ and $\gamma $. Since
	 	\begin{equation*}
	 	R\left(  \widehat{\alpha }_h, \widehat{\gamma }_h\right) \leq \left\vert \mathbb{G}%
	 	_{T}\left(  \widehat{\alpha }_h, \widehat{\gamma }_h\right) -\mathbb{G}_{T}\left(
	 	\alpha _{0},\gamma _{0}\right) \right\vert +\left\vert \mathbb{R}_{T}\left(
	 	\widehat{\alpha }_h, \widehat{\gamma }_h\right) -R\left(  \widehat{\alpha }_h,%
	 	\widehat{\gamma }_h\right) \right\vert ,
	 	\end{equation*}%
	 	we conclude that
	 	\begin{equation*}
	 	\left( c-\eta \right) \left( \left\vert  \widehat{\alpha }_h-\alpha
	 	_{0}\right\vert _{2}^{2}+T^{-2\varphi }\left\vert  \widehat{\gamma }_h-\gamma
	 	_{0}\right\vert _{2}\right) \leq O_{P}\left( \frac{1}{\sqrt{T}}\right)
	 	\left\vert  \widehat{\alpha }_h-\alpha _{0}\right\vert _{2}+O_{P}\left( \frac{1%
	 	}{T}\right) .
	 	\end{equation*}%
	 	implying
	 	\begin{equation*}
	 	\left\vert  \widehat{\gamma }_h-\gamma
	 	_{0}\right\vert _{2}=O_{P}\left( \frac{1}{T^{1-2\varphi }}\right) .
	 	\end{equation*}

	 	(ii)
	 	First we show that
	 	$
	 	|\widehat\alpha_h-\widehat\alpha|_2=o_P(T^{-1/2})$ under $H_0.$ Let $\widehat Z_h=Z(\widehat\gamma_h)$. Straightforward calculations yield
	 	\begin{eqnarray*}
	 		\widehat\alpha_h-\widehat\alpha&=&
	 		(\widehat Z_h'\widehat Z_h)^{-1}(\widehat Z_h-\widehat Z)'(Z-\widehat Z_h)\alpha_0+(\widehat Z_h'\widehat Z_h)^{-1}\widehat Z'(\widehat Z-\widehat Z_h)\alpha_0+(\widehat Z_h'\widehat Z_h)^{-1}(\widehat Z_h-\widehat Z)'\epsilon\cr
	 		&&+[(\widehat Z_h'\widehat Z_h)^{-1}-(\widehat Z'\widehat Z)^{-1}][\widehat Z'(Z-\widehat Z)\alpha_0+ Z'\epsilon+(\widehat Z-Z)'\epsilon]
	 	\end{eqnarray*}
	 	which is $o_P(T^{-1/2})$ since $|\widehat\gamma_h-\gamma_0|_2=o(T^{-(0.5-\varphi)})$  under $H_0$. Then
	 	\begin{eqnarray*}
	 		&&T[\mathbb S_T(\widehat\alpha_h,  \gamma_0)-\mathbb S_T(\widehat\alpha,\gamma_0)] \cr
	 		&=&T(\widehat\alpha_h-\widehat\alpha)'\frac{1}{T}\sum_tZ_t(\gamma_0)Z_t(\gamma_0)'(\widehat\alpha_h-\widehat\alpha)
	 		+T(\widehat\alpha-\widehat\alpha_h) '\frac{2}{T}\sum_tZ_t(\gamma_0)\varepsilon_t\cr
	 		&&+T\frac{2}{T}(\alpha_0-\widehat\alpha)\sum_tZ_t(\gamma_0)Z_t(\gamma_0)'(\widehat\alpha-\widehat\alpha_h)\cr
	 		&=&O_P(T)|\widehat\alpha_h-\widehat\alpha|_2^2 +O_P(\sqrt{T})|\widehat\alpha_h-\widehat\alpha|_2=o_P(1).
	 	\end{eqnarray*}

	 \end{proof}

	 \begin{lem} \label{bootstrap.high_level} In the known factor case, the k-step bootstrap estimators $(\widehat\alpha^*, \widehat\gamma^*, \widehat\gamma_h^*)$ satisfy:
	 	\begin{eqnarray*}
	 		\mathbb S_T^*(\widehat \alpha^*,\widehat \gamma^*)& \leq&   \min_{\alpha, \gamma}\mathbb S_T^*(\alpha,\gamma) + o_{P^*}(T^{-1}).\cr
	 		\mathbb S_T^*(\widehat\alpha_h^*,\widehat\gamma_h^*)
	 		&\leq& \min_{\alpha, h(\gamma)=h(\widehat\gamma)}\mathbb S_T^*(\alpha,\gamma)+ o_{P^*}(T^{-1}), \quad h(\widehat\gamma_h^*)=h(\widehat\gamma)
	 		\cr
	 		|\widehat\alpha^*-\widehat\alpha|_2&=&O_{P^*}(T^{-1/2}),\quad |\widehat\alpha_h^*-\widehat\alpha|_2=O_{P^*}(T^{-1/2}),\quad |\widehat\alpha^*-\widehat\alpha^*_h|_2=o_{P^*}(T^{-1/2})\cr
	 		|\widehat\gamma_h^*-\widehat\gamma|_2&=&O_P(T^{-(1-2\varphi)})\cr
	 		|\widehat\gamma^*-\widehat\gamma|_2&=&O_P(T^{-(1-2\varphi)}).
	 	\end{eqnarray*}

	 \end{lem}

	 \begin{proof}
	 	Define
	 	\begin{eqnarray*}
	 		(\alpha^*_g, \gamma_g^*)&=&\arg\min \mathbb S_T^*(\alpha,\gamma).
	 		\cr
	 		(\alpha^*_{g,h}, \gamma_{g,h}^*)&=&\arg\min_{\alpha, h(\gamma)=h(\widehat\gamma)} \mathbb S_T^*(\alpha,\gamma)\cr
	 		\alpha^*(\gamma)&=& \arg\min_\alpha \mathbb S_T^*(\alpha,\gamma),  \cr
	 		\gamma^*(\alpha)&=& \arg\min_\gamma \mathbb S_T^*(\alpha,\gamma),\cr
	 		\gamma_h^*(\alpha)&=& \arg\min_{\gamma: h(\gamma)=h(\widehat\gamma)} \mathbb S_T^*(\alpha,\gamma).
	 	\end{eqnarray*}

	 	Our proof is divided into the following steps.
	 	\\
	 	\textit{step 0:  }$ |\gamma_{g,h}^*-\widehat\gamma|_2=O_{P^*}(T^{-(1-2\varphi)})$, $|\gamma_g^*-\widehat\gamma|_2=O_{P^*}(T^{-(1-2\varphi)})$ and $|\alpha^*_g-\widehat\alpha|_2=O_{P^*}(T^{-1/2}).$
	 	\\
	 	\textit{step 1: } if $|\gamma-\widehat\gamma|_2=O_{P^*}(T^{-(1-2\varphi)})$, then $	|\alpha^*(\gamma)-\widehat\alpha|_2=O_{P^*}(T^{-1/2}) $.
	 	\\
	 	\textit{step 2:}    in addition,   $	|\alpha^*(\gamma)-\alpha^*_g|_2=o_{P^*}(T^{-1/2}) $, and $	|\alpha^*(\gamma)-\alpha_{g,h}^*|_2=o_{P^*}(T^{-1/2}) $.
	 	\\
	 	\textit{step 3:  } if  $	|\alpha -\alpha^*_g|_2=o_{P^*}(T^{-1/2}) $, and $	|\alpha -\alpha_{g,h}^*|_2=o_{P^*}(T^{-1/2}) $, then\\ $\mathbb S_T^*(  \alpha,  \gamma^*(\alpha) )\leq    \min_{\alpha, \gamma}\mathbb S_T^*(\alpha,\gamma) + o_{P^*}(T^{-1})$, and $\mathbb S_T^*(  \alpha,  \gamma_h^*(\alpha) )\leq     \min_{\alpha, h(\gamma)=h(\widehat\gamma)}\mathbb S_T^*(\alpha,\gamma)+ o_{P^*}(T^{-1})$.
	 	\\
	 	\textit{step 4:} in addition,  $	 |\gamma^*(\alpha)- \widehat\gamma|_2=O_{P^*}(T^{-(1-2\varphi)})$ and $	 |\gamma_h^*(\alpha)- \widehat\gamma|_2=O_{P^*}(T^{-(1-2\varphi)})$.

	 	Once the above steps are successfully achieved, then the proof is completed  by the following argument.
	 	Recall that $\widehat\gamma^{*,0}=\widehat\gamma^{*,0 }_h=\widehat\gamma$. Also, for $l\geq1$,
	 	$\widehat\alpha^{*,l}=\alpha^*(\widehat\gamma^{*,l-1} )$, $\widehat\alpha^{*,l}_h=\alpha^*(\widehat\gamma_h^{*,l-1} )$,
	 	$\widehat\gamma^{*,l} = \gamma^*(\widehat\alpha^{*,l})$, $\widehat\gamma^{*,l} _h= \gamma_h^*(\widehat\alpha_h^{*,l})$,
	 	and $\widehat\alpha^*=\widehat\alpha^{*,k}$,   $\widehat\gamma^*=\widehat\gamma^{*,k}$, and  $\widehat\gamma^*_h=\widehat\gamma^{*,k}_h$.

	 	For $k=1$, $\widehat\gamma^{*,0}=\widehat\gamma^{*,0 }_h=\widehat\gamma$.  Hence by step 1, $|\widehat\alpha^{*,1}-\widehat\alpha|_2=O_{P^*}(T^{-1/2}) =|\widehat\alpha_h^{*,1}-\widehat\alpha|_2$. Conditions of step 3 are satisfied due to step 2, hence  for $\alpha=\alpha^*(\widehat\gamma^{*, 0})$ in step 3,
	 	$$\mathbb S_T^*( \widehat\alpha^{*,1},  \widehat\gamma^{*,1}  )\leq    \min_{\alpha, \gamma}\mathbb S_T^*(\alpha,\gamma) + o_{P^*}(T^{-1})$$ and  for $\alpha=\alpha^*(\widehat\gamma_h^{*, 0})$ in step 3,
	 	$$\mathbb S_T^*( \widehat\alpha_h^{*,1},  \widehat\gamma^{*,1} _h )\leq     \min_{\alpha, h(\gamma)=h(\widehat\gamma)}\mathbb S_T^*(\alpha,\gamma)+ o_{P^*}(T^{-1}).$$
	 	By step 4, $	 | \widehat\gamma^{*,1} - \widehat\gamma|_2=O_{P^*}(T^{-(1-2\varphi)})$ and $	 | \widehat\gamma^{*,1} _h- \widehat\gamma|_2=O_{P^*}(T^{-(1-2\varphi)})$.  Thus results of Lemma \ref{bootstrap.high_level}  are verified if we stop after $k$ step(s) for  $k=1$.

	 	For $k=2$,  the previous step 4  ensures that we can apply step 1 respectively with $\gamma= \widehat\gamma^{*,1}$ and $\gamma= \widehat\gamma_h^{*,1}$.  Thus the same argument yields  Lemma \ref{bootstrap.high_level}  is verified for  $k=2$. We can employ the mathematical induction  to   conclude that Lemma  \ref{bootstrap.high_level}  is verified for  all $k\geq 1$.

	 	\textit{Proof of Step 0.}

	 	In the bootstrap world, $\widehat\gamma$ is the  true value while $\gamma_g^*$ is the least squares estimator.
	 	Also, by the definition of $(\alpha^*_{g,h}, \gamma_{g,h}^*)$, we have
	 	$$
	 	\mathbb S_T^*(\alpha_{g,h}^*,\gamma_{g,h}^*)\leq \mathbb S_T^*(\widehat\alpha, \widehat\gamma),\quad
	 	\mathbb S_T^*(\alpha_{g}^*,\gamma_{g}^*)\leq \mathbb S_T^*(\widehat\alpha, \widehat\gamma).
	 	$$
	 	Hence the proof of this step  is simply the bootstrap version of the proof of the rates of convergence  in the original sampling space. We thus omit its proof to avoid repetitions.

	 	\textit{Proof of Step 1.}

	 	For a generic $\gamma$, let $A(\gamma):= \frac{1}{T} \sum_{t=1}^{T}Z_{t}\left( \gamma \right) Z_{t}\left( \gamma \right)'$.
	 	\begin{eqnarray*}
	 		{\alpha}^*\left( \gamma \right) -\widehat \alpha &=&A(\gamma) ^{-1}\left( \frac{1}{T}\sum_{t=1}^{T}Z_{t}\left( \gamma \right)
	 		\widehat\varepsilon _{t}\eta_t+\frac{1}{T}\sum_{t=1}^{T}Z_{t}\left( \gamma \right)
	 		x_{t}^{\prime }\widehat \delta \left( 1_{t}\left( \gamma \right) -1_{t}(\widehat\gamma)\right)
	 		\right).    \end{eqnarray*}
	 	So  conditional on the event $|\widehat\gamma-\gamma_0|_2\leq CT^{-(1-2\varphi)}$ and uniformly in  $|\gamma-\widehat\gamma|\leq  CT^{-(1-2\varphi)}$,
	 	\begin{eqnarray*}
	 		&&\left| {\alpha}^*\left( \gamma \right) -\widehat \alpha  - A(\widehat\gamma)^{-1}  \frac{1}{T}\sum_{t=1}^{T}Z_{t}\left( \widehat\gamma \right)
	 		\widehat\varepsilon _{t}\eta_t \right|\cr
	 		&\leq &
	 		\left|(A(\gamma)^{-1} -A(\widehat\gamma)^{-1}) \frac{1}{T}\sum_{t=1}^{T}Z_{t}\left( \widehat\gamma \right)
	 		\widehat\varepsilon _{t}\eta_t  \right|+
	 		|A(\gamma)^{-1}| \sup_{|\gamma-\widehat\gamma|_2\leq CT^{-(1-2\varphi)}}\left|\frac{1}{T}\sum_{t=1}^{T}[Z_{t}\left(\gamma \right)- Z_t(\widehat\gamma)]
	 		\widehat\varepsilon _{t}\eta_t \right |\cr
	 		&&+ |A(\gamma)^{-1} O_P(T^{-\varphi})|\sup_{|\gamma-\gamma_0|_2\leq CT^{-(1-2\varphi)}}\left| \frac{1}{T}\sum_{t=1}^{T}|x_{t}|_2^2|1_{t}\left( \gamma \right) -1_{t}(\gamma_0)|
	 		-\mathbb E|x_{t}|_2^2|1_{t}\left( \gamma \right) -1_{t}(\gamma_0)|
	 		\right|
	 		\cr
	 		&&+  |A(\gamma)^{-1} O_P(T^{-\varphi})  \sup_{|\gamma-\gamma_0|_2\leq CT^{-(1-2\varphi)}} \mathbb E|x_{t}|_2^2|1_{t}\left( \gamma \right) -1_{t}(\gamma_0)|\cr
	 		&=&o_{P^*}(T^{-1/2}).
	 	\end{eqnarray*}

	 	Thus we have proved, uniformly over  $|\gamma-\widehat\gamma|\leq  CT^{-(1-2\varphi)}$,
	 	\begin{equation}\label{e.22}
	 	{\alpha}^*\left( \gamma \right) -\widehat \alpha=A(\widehat\gamma)^{-1}) \frac{1}{T}\sum_{t=1}^{T}Z_{t}\left( \widehat\gamma \right)
	 	\widehat\varepsilon _{t}\eta_t+o_{P^*}(T^{-1/2}).
	 	\end{equation}

	 	\textit{Proof of Step 2.}

	 	Note  that
	 	$\alpha^*_g=\alpha^*(\gamma_g^*)$ and
	 	$\alpha_{m, h}^*=\alpha^*(\gamma_{g,h}^*)$.
	 	Respectively letting    $\gamma=\gamma_{g,h}^*$    and $\gamma=\gamma_g^*$ in (\ref{e.22}) yields (by step 2)
	 	\begin{eqnarray*}
	 		{\alpha}^*\left( \gamma \right) -\widehat \alpha&=&A(\widehat\gamma)^{-1}) \frac{1}{T}\sum_{t}Z_{t}\left( \widehat\gamma \right)
	 		\widehat\varepsilon _{t}\eta_t+o_{P^*}(T^{-1/2})
	 		\cr
	 		{\alpha}^* _m -\widehat\alpha&=&A(\widehat\gamma)^{-1}) \frac{1}{T}\sum_{t}Z_{t}\left( \widehat\gamma \right)
	 		\widehat\varepsilon _{t}\eta_t+o_{P^*}(T^{-1/2})
	 		\cr
	 		{\alpha}^*_{g,h} -\widehat\alpha&=&A(\widehat\gamma)^{-1}) \frac{1}{T}\sum_{t}Z_{t}\left( \widehat\gamma \right)
	 		\widehat\varepsilon _{t}\eta_t+o_{P^*}(T^{-1/2}).
	 	\end{eqnarray*}
	 	Thus
	 	$
	 	\alpha^*(\gamma)-\alpha^*_g=o_{P^*}(T^{-1/2}) $ and
	 	$
	 	\alpha^*(\gamma)-\alpha_{g,h}^*=o_{P^*}(T^{-1/2}) .$

	 	\textit{Proof of Step 3.}

	 	By the definition of $\gamma^*(\alpha)$ and $\gamma_h^*(\alpha)$,
	 	\begin{eqnarray*}
	 		\mathbb S_T^*( \alpha,\gamma^*( \alpha))&\leq& \mathbb S_T^*( \alpha,\gamma^*_g )
	 		=\mathbb S_T^*(\alpha^*_g,\gamma^*_g ) +\mathbb S_T^*(\alpha,\gamma^*_g )  -\mathbb S_T^*(\alpha^*_g,\gamma^*_g ) \cr
	 		\mathbb S_T^*(\alpha, \gamma_h^*(\alpha))
	 		&	\leq& \mathbb S_T^*(\alpha, \gamma_{g,h}^*)
	 		= \mathbb S_T^*(\alpha_{g,h}^*,\gamma_{g,h}^*)+ \mathbb S_T^*( \alpha, \gamma_{g,h}^*)-  \mathbb S_T^*(\alpha_{g,h}^*,\gamma_{g,h}^*).
	 	\end{eqnarray*}
	 	By definition,  $ \min_{\alpha, h(\gamma)=h(\widehat\gamma)}\mathbb S_T^*(\alpha,\gamma)=\mathbb S_T^*(\alpha_{g,h}^*,\gamma_{g,h}^*)$ and
	 	$ \min_{\alpha, \gamma}\mathbb S_T^*(\alpha,\gamma)=\mathbb S_T^*(\alpha_{g}^*,\gamma_{g}^*)$. Hence it suffices to show   if $	|\alpha -\alpha^*_g|_2=o_{P^*}(T^{-1/2}) $, and $	|\alpha -\alpha_{g,h}^*|_2=o_{P^*}(T^{-1/2}) $,
	 	\begin{eqnarray*}
	 		&&	\mathbb S_T^*(  \alpha,\gamma^*_g )  -\mathbb S_T^*(\alpha^*_g,\gamma^*_g )\leq o_{P^*}(T^{-1})\cr
	 		&&	\mathbb S_T^*( \alpha, \gamma_{g,h}^*)-  \mathbb S_T^*(\alpha_{g,h}^*,\gamma_{g,h}^*)\leq  o_{P^*}(T^{-1}).
	 	\end{eqnarray*}
	 	By  $|\alpha^*_g-\widehat\alpha|_2=O_{P^*}(T^{-1/2})$
	 	and  the triangular inequality, $|\alpha-\widehat\alpha|_2=O_{P^*}(T^{-1/2})$.
	 	Uniformly in   $\gamma$ so that  $|\gamma-\widehat\gamma|_2\leq CT^{-(1-2\varphi)}$,
	 	\begin{eqnarray}\label{e.23}
	 &&	\mathbb S_T^*( \alpha,\gamma )  -\mathbb S_T^*(\alpha^*_g,\gamma )
	=(\alpha^*_g-\alpha)'\frac{2}{T}\sum_tZ_t(\gamma) \eta_t\widehat\varepsilon_t
	 	+(\alpha^*_g-\alpha) '\frac{1}{T}\sum_tZ_t(\gamma)  [Z_t(\widehat\gamma)-Z_t(\gamma)]'\widehat\alpha
	 	\cr
	 	&& +(\alpha^*_g-\alpha)'\frac{1}{T}\sum_tZ_t(\gamma) Z_t(\gamma)'(\widehat\alpha -\alpha)
	 	+(\alpha^*_g-\alpha)'\frac{1}{T}\sum_tZ_t(\gamma) Z_t(\gamma)'(\widehat\alpha -\alpha^*_g) \cr
	 	&=&o_{P^*}(T^{-1}) + o_{P^*}(T^{-1/2})\frac{1}{T}\sum_tZ_t(\gamma) \eta_t\widehat\varepsilon_t
	 	+ o_{P^*}(T^{-1/2})    \frac{1}{T}\sum_tZ_t(\gamma)  [Z_t(\widehat\gamma)-Z_t(\gamma)]'\widehat\alpha\cr
	 	&=&
	 	o_{P^*}(T^{-1}) + o_{P^*}(T^{-1/2})\frac{1}{T}\sum_tZ_t(\widehat\gamma) \eta_t\widehat\varepsilon_t
	 	+o_{P^*}(T^{-1/2})\frac{1}{T}\sum_t[ Z_t(\gamma)-Z_t(\widehat\gamma)]   \eta_t\widehat\varepsilon_t
	 	\cr
	 	&&
	 	+ o_{P^*}(T^{-1/2})    \frac{1}{T}\sum_tZ_t(\gamma)  [1\{f_t'\widehat\gamma>0\}-1\{f_t'\gamma>0\}]'x_t'\widehat\delta\cr
	 	&\leq & o_{P^*}(T^{-1})
	 	+o_{P^*}(T^{-1/2})\sup_{|\gamma-\widehat\gamma|<CT^{-(1-2\varphi)}}|\frac{1}{T}\sum_t[ Z_t(\gamma)-Z_t(\widehat\gamma)]   \eta_t\widehat\varepsilon_t  |_2\cr
	 	&&  + o_{P^*}(T^{-1/2-\varphi})   \sup_{|\gamma-\gamma_0|<CT^{-(1-2\varphi)}}  \left|   \frac{1}{T}\sum_t |x_t|_2^2  |1(\gamma_0)-1(\gamma)| -\mathbb E  |x_t|_2^2  |1(\gamma_0)-1(\gamma)|    \right|  \cr
	 	&&  + o_{P^*}(T^{-1/2-\varphi})   \sup_{|\gamma-\gamma_0|<CT^{-(1-2\varphi)}}   \mathbb E  |x_t|_2^2  |1(\gamma_0)-1(\gamma)|      \cr
	 	&\leq & o_{P^*}(T^{-1})
	 	+    o_{P^*}(T^{-1/2})  T^{-(1-\varphi)}
	 	+o_{P^*}(T^{-1/2-\varphi}) T^{-(1-2\varphi)}= o_{P^*}(T^{-1}) .
	 	\end{eqnarray}

	 	Applying the above to  $\gamma= \gamma_g^*$,
	 	which satisfies  $|\gamma-\widehat\gamma|_2\leq CT^{-(1-2\varphi)}$ with bootstrap probability measure arbitrarily close to one by step 1 due to step 0,
	 	we have
	 	$$
	 	\mathbb S_T^*( \alpha,\gamma_g^* )  -\mathbb S_T^*(\alpha^*_g,\gamma_g^* )= o_{P^*}(T^{-1}) .
	 	$$

	 	In addition,  	 (\ref{e.23})   also applies when $\alpha^*_g$ is replaced with $\alpha_{g,h}^*$.
	 	That is,\\  $\mathbb S_T^*( \alpha,\gamma  )  -\mathbb S_T^*(\alpha_{g,h}^*,\gamma  )= o_{P^*}(T^{-1}) $ uniformly  in  $|\gamma-\widehat\gamma|_2\leq CT^{-(1-2\varphi)}$.
	 	By step 0,    $ |\gamma_{g,h}^*-\widehat\gamma|_2=O_{P^*}(T^{-(1-2\varphi)})$. Hence let $\gamma=\gamma^*_{m,h}$, we have
	 	$$
	 	\mathbb S_T^*(  \alpha, \gamma^*_{m,h}  )  -\mathbb S_T^*(\alpha_{g,h}^*,\gamma^*_{m,h}  )= o_{P^*}(T^{-1}) .
	 	$$

	 	\textit{Proof of Step 4.}
	 	Note that $|\alpha-\widehat\alpha|=O_{P^*}(T^{-1/2})$.
	 	The proof is then simply the bootstrap version of     step 3 of the iterative estimator in the known factor case.  Thus we just sketch the proof for $	| \gamma_h^*(\alpha)-\widehat\gamma|_2=O_{P^*}(T^{-(1-2\varphi)})$ for brevity.

	 	For   generic $\gamma$,
	 	\begin{eqnarray*}
	 		\mathbb{S}_{T}^*\left(  {\alpha},\gamma \right) -\mathbb{S}^*
	 		_{T}\left(  {\alpha},\widehat\gamma \right)
	 		& =& {\delta}^{\prime }\frac{1}{T}\sum_{t=1}^{T}x_{t}x_{t}^{\prime
	 		}\left\vert 1_{t}\left( \gamma \right) -1_{t} (\widehat\gamma) \right\vert   {\delta}-\frac{2%
	 		}{T}\sum_{t=1}^{T}\eta_t\widehat\varepsilon _{t}x_{t}^{\prime }\left( 1_{t}\left( \gamma
	 		\right) -1_{t} (\widehat\gamma)   \right)   {\delta}  \notag \\
	 		&&+\left(   {\alpha}-\widehat\alpha \right) ^{\prime }\frac{2}{T}%
	 		\sum_{t=1}^{T}Z_{t}\left( \widehat\gamma \right) x_{t}^{\prime }\left(
	 		1_{t}\left( \gamma \right) -1_{t}  (\widehat\gamma) \right)   {\delta} .
	 	\end{eqnarray*}
	  {Apply Lemma \ref{A-rates:lemma-used} with $\gamma_0$ replaced by a generic $\gamma_2$},   uniformly for $\gamma, \gamma_2$, and an arbitrarily small $\eta >0, $
	 	\begin{eqnarray*}
	 		\delta ^{\prime }\frac{1}{T}\sum_{t=1}^{T}x_{t}x_{t}^{\prime }\left\vert
	 		1_{t}\left( \gamma \right) -1_{t}(\gamma_2)\right\vert \delta  &=&O_{p}\left(
	 		\frac{1}{T^{1+\varphi }}\right) +\left\vert \delta \right\vert _{2}\eta
	 		T^{-2\varphi }\left\vert \gamma -\gamma _{2}\right\vert _{2}\cr
	 		&&+T^{-2\varphi
	 		}\mathbb E\left( d_{0}^{\prime }x_{t}\right) ^{2}\left\vert 1_{t}\left( \gamma
	 		\right) -1_{t}(\gamma_2)\right\vert \\
	 		&\geq &O_{p}\left( \frac{1}{T}\right) +cT^{-2\varphi }\left\vert \gamma
	 		-\gamma _{2}\right\vert _{2}\cr
	 		\frac{2}{T}\sum_{t=1}^{T}\eta_t\widehat\varepsilon _{t}x_{t}^{\prime }\left( 1_{t}\left(
	 		\gamma \right) -1_{t}(\gamma_2)\right) \delta  &\leq &O_{P^*}\left( \frac{1}{T}\right)
	 		+\eta T^{-2\varphi }\left\vert \gamma -\gamma _{2}\right\vert _{2}
	 		\\
	 		\left( {\alpha}-\widehat \alpha \right) ^{\prime }\frac{2}{T}%
	 		\sum_{t=1}^{T}Z_{t}\left( \widehat \gamma \right) x_{t}^{\prime }\left(
	 		1_{t}\left( \gamma \right) -1_{t}(\gamma_2)\right) \delta  &=&\left( O_{p}\left( \frac{1}{T}\right) +\eta		T^{-2\varphi }\left\vert \gamma -\gamma _{2}\right\vert _{2}+T^{-\varphi		}\left\vert \gamma -\gamma _{2}\right\vert _{2}\right)\cr		&&\times O_{P^*}\left(		T^{-1/2 }\right)  .
	 	\end{eqnarray*}
	 	Combining these bounds and setting  $\gamma=\gamma_h^*(  \alpha)$, and $\gamma_2=\widehat\gamma$,
	 	\begin{eqnarray*}
	 		0\geq 	\mathbb{S}_{T}^*\left( \alpha ,  \gamma_h^* (\alpha)\right) -\mathbb{S}^*
	 		_{T}\left( \alpha ,\widehat\gamma \right)
	 		& \geq & O_{P^*}\left(  T^{-1}\right) +cT^{-2\varphi }\left\vert  \gamma_h^*(\alpha)
	 		-\widehat\gamma \right\vert .
	 	\end{eqnarray*}
	 	This implies
	 	$\left\vert  \gamma_h^*(\alpha)
	 	-\widehat\gamma \right\vert   \leq O_{P^*}(T^{-(1-2\varphi)}).
	 	$
	 	The same argument yields
	 	$\left\vert \gamma^*(\alpha)
	 	-\widehat\gamma \right\vert   \leq O_{P^*}(T^{-(1-2\varphi)}).
	 	$
	 \end{proof}

			\subsection{Proof of Theorem \ref{t6.1}: estimated factor case}\label{sec:appendix:est:factors}

		\subsubsection{The bootstrap with re-estimated factors using PCA}
	We now present details on our main  bootstrap procedure that re-estimates factors in the bootstrap sample. This is given by \cite{gonccalves2018bootstrapping}. To maintain the cross-sectional dependence among the idiosyncratic components in the bootstrap factor models, we rely on a consistent (in spectral norm) covariance matrix for $e_t$, given by $\widehat \var(e_t)$. Such a covariance estimator can be obtained via thresholding as in \cite{POET}. Let $\mathcal W_t^* $ be a sequence of independent  $N\times 1$   multivariate standard normal vectors. As in \cite{gonccalves2018bootstrapping},   generate bootstrap data
	$$
	\mathcal Y_t^*= \widehat\Lambda \widetilde f_{1t} + \widehat \var(e_t)^{1/2}\mathcal W_t^*.
	$$
	Then  apply PCA to estimate factors, obtaining $\widetilde F_{1t}^*$ as the estimated factors in the bootstrap sample. Let $F_t^*=(F_{1t}^*, -1)'$. It   has the following expansion
	$$
	F_t^*= H_T^{*'}(\widetilde f_t +\frac{1}{\sqrt{N}}h_t^* )+ r_t^*
	$$
	where $r_t^*$ is some remainder term. In addition, let $V^*$ be the diagonal matrix containing the top dim$(\widetilde f_{1t})$ eigenvalues of $\mathcal Y^*\mathcal Y^{*'}/(NT)$.
	Then
		$$
	H_T^{*'}=
	\begin{pmatrix}
	H^{*'} & 0\\
	0& 1
	\end{pmatrix},\quad H^{*'}=V^{*-1}\frac{1}{T}\sum_tF_{1t}^*\widetilde f_{1t}'\widehat S_\Lambda, \quad h_t^*=\begin{pmatrix}
	h_{1t}^*\\
	0
	\end{pmatrix}.
	$$
	where
	$\widehat S_\Lambda=\frac{1}{N}\widehat\Lambda'\widehat\Lambda$, and $$h_{1t}^*=\widehat S_\Lambda^{-1}\frac{1}{N}\widehat\Lambda \widehat \var(e_t)^{1/2}\mathcal W_t^*:=\mathcal Z_t^*.$$
Note that $F_t^*$ estimates $\widetilde f_t$, the ``true factors" in the bootstrap sample, up to a new rotation matrix $H_T^*$. Hence the bootstrap distribution of $F_t^*$ would not be able to exactly mimic the sampling distribution of $\widetilde f_t$. Fortunately, such a rotation discrepancy can be removed in the bootstrap estimation because $H_T^*$ is known. As such, we can define
$$
f_t^*:= H_T^{*'-1} F_t^*
$$
as the final ``estimated factors" in the bootstrap sample, whose asymptotic bootstrap distribution would then exactly mimic that of $\tilde f_t$ without rotations.

 \subsubsection{Gaussian-perturbed factors}\label{sec:Gaussian:perturbed-factors}

We now present an alternative bootstrap procedure using Gaussian-perturbed factors.
Let $\{\mathcal W_t^*:t\leq T\}$ be a sequence of independent $K \times 1$    multivariate standard normal random vectors.
We simply use
$$
f_t^*:=\widetilde f_t +N^{-1/2} \widehat\Sigma_h^{1/2}\mathcal W_t^*,
$$
where $N^{-1/2} \widehat\Sigma_h^{1/2}\mathcal W_t^*$ is a Gaussian perturbation; $\widehat\Sigma_h$ is an estimator for the asymptotic variance of $H'h_{1t}$, as defined in (\ref{e:6.7:h})
$$
\Sigma_h:=\var\left[(\frac{1}{N}\Lambda'\Lambda)^{-1}\frac{1}{\sqrt{N}}\Lambda' e_t\right].
$$
Hence we can use
$$
\widehat\Sigma_h=N ( \widehat \Lambda'\widehat \Lambda)^{-1}\widehat \Lambda'\widehat \var(e_t)\widehat \Lambda (\widehat \Lambda'\widehat \Lambda)^{-1},
$$
where $\widehat \var(e_t)$ is a high-dimensional covariance estimator for $\var(e_t)$.  For instance,
\citet{POET} assumed that $\var(e_t)$ is a sparse covariance matrix, and  constructed $\widehat \var(e_t)$ using thresholding.

		\subsubsection{Proof of the distribution of $LR$}

Let $$l_{NT}=\sqrt{r_{NT}T^{1+2\varphi }}.$$

\begin{proof}
Define
	 	$
	 	\widehat\gamma_h=\arg\min_{\alpha, h(\gamma)=0}\mathbb S_T(\alpha,\gamma),
	 	$
	 	$
	 	\widehat\alpha(\gamma)=\arg\min_{\alpha}\mathbb S_T(\alpha,\gamma),$ and $\widehat\alpha_h=\widehat\alpha(\widehat\gamma_h).
	 	$
	 	Then \begin{eqnarray*}
l_{NT}	\min_{\alpha,\gamma} {\mathbb S}_T(\alpha,\gamma)LR &=& l_{NT}[\mathbb S_T(\widehat\alpha_h, \widehat\gamma_h)-\mathbb S_T(\widehat\alpha,\widehat\gamma)] \\
&=&A_1+A_2-A_3\quad \text{ where, }\cr		A_1&=&l_{NT}[\mathbb S_T(\widehat\alpha_h, \widehat\gamma_h)-\mathbb S_T(\widehat\alpha_h,\gamma_0)]\cr
	 		A_2   &=&l_{NT}[\mathbb S_T(\widehat\alpha_h, \gamma_0)-\mathbb S_T(\widehat\alpha,\gamma_0)]\cr
	 		A_3 &=&l_{NT}[\mathbb S_T(\widehat\alpha, \widehat\gamma)-\mathbb S_T(\widehat{\alpha},\gamma_0)].
	 	\end{eqnarray*}
	We first analyze $\mathbb S_T(\alpha,\gamma)-\mathbb S_T(\alpha,\gamma_0)$.

	 We have
	$
\mathbb S_T(\alpha,\gamma)-\mathbb S_T(\alpha,\gamma_0)
=  \widetilde{\mathbb C}_3(\alpha,\gamma) - \widetilde{\mathbb C}_1(\alpha,\gamma) +\sum_{d=2}^6\widetilde R_d(\alpha,\gamma)
	$
	where  $ \widetilde{\mathbb C}_1(\alpha,\gamma) $, $ \widetilde{\mathbb C}_3(\alpha,\gamma) , \widetilde R_2(\alpha,\gamma),\widetilde R_3(\alpha,\gamma)$ are as defined in
		Section \ref{sec:def:not}, and
\begin{eqnarray*}
\widetilde R_4(\alpha,\gamma)&=&\frac{1}{T}\sum_t2\widetilde Z_t(\gamma_0)'(\alpha-\alpha_0)
x_t'(\delta-\delta_0) (1\{\widetilde f_t'\gamma>0\}-1\{\widetilde f_t'\gamma_0>0\}) \cr
\widetilde R_5(\alpha,\gamma)&=& \frac{1}{T}\sum_t
(x_t'(\delta-\delta_0) )^2|1\{\widetilde f_t'\gamma>0\}-1\{\widetilde f_t'\gamma_0>0\}| \cr
\widetilde R_6(\alpha,\gamma)&=& \frac{1}{T}\sum_t
x_t'(\delta-\delta_0) x_t'\delta |1\{\widetilde f_t'\gamma>0\}-1\{\widetilde f_t'\gamma_0>0\}|.
\end{eqnarray*}
Uniformly  for $|\alpha-\alpha_0|= O_P(T^{-1/2})$ and $|\gamma-\gamma_0|_2=O_P(r_{NT}^{-1})$, we have $l_{NT} |\widetilde R_d(\alpha,\gamma)|=o_P(1)$ for $d=3\sim 6$.			In addition,  $\widetilde R_2(\alpha,\gamma)+ \widetilde{\mathbb C}_3(\alpha,\gamma)
= \mathds G_1(H_T\gamma) +o_P(l_{NT}^{-1})
 $ for $\mathds G_1(\phi)$	  as defined in (\ref{e.11}),  by Lemmas \ref{l:relf}, \ref{l:rhat}, \ref{l.1}. So
 $$
 l_{NT}[\mathbb S_T(\alpha,\gamma)-\mathbb S_T(\alpha,\gamma_0)]
 = l_{NT}[\mathds G_1(H_T\gamma)- \widehat{\mathbb C}_1(\alpha_0,\gamma)]  +o_P(1).
 $$
Next, $\mathbb S_T(\widehat\alpha_h,\widehat\gamma_h)\leq \mathbb S_T(\alpha_0,\gamma_0)$ under  $h(\gamma_0)=h(\widehat\gamma_h)=0$. Thus the same proof as the rate of convergence for $(\widehat\alpha, \widehat\gamma)$ also carries over to prove that $|\widehat\gamma_h-\gamma_0|_2=O_P(r_{NT}^{-1})$ and $|\widehat\alpha_h-\alpha_0|_2=O_P(T^{-1/2}).$
Now let  $\alpha=\alpha
	 		_{0}+a\cdot T^{-1/2}$,
			$\gamma=\gamma _{0}+g\cdot r_{NT}^{-1} ,$
\begin{eqnarray*}
	 		\mathbb{K}_{T}\left( a,g\right) &=&l_{NT}\left( \mathbb{S}_{T}\left( \alpha
	 		_{0}+a\cdot T^{-1/2},\gamma _{0}+g\cdot r_{NT}^{-1}\right) -\mathbb{S}%
	 		_{T}\left( \alpha _{0},\gamma _{0}\right) \right) \cr
&=& l_{NT}\sum_{d=1}^3\widetilde R_d(\alpha,\gamma)
- l_{NT}\sum_{d=1}^2\widetilde{\mathbb C}_d(\alpha,\gamma) +l_{NT}\sum_{d=3}^4\widetilde{\mathbb C}_d(\alpha,\gamma) \cr
&=&l_{NT}[\mathds G_1(H_T\gamma)- \widehat{\mathbb C}_1(\alpha_0,\gamma)] +l_{NT} [\widehat R_1(\alpha,\gamma_0)-\widehat{\mathbb C}_2(\alpha,\gamma)] +o_P(1)\cr
&=&\mathbb K_{4T}(g)+\mathbb K_{1T}(a)+o_P(1),
	 	\end{eqnarray*}%
where $\mathbb K_{1T}(a):=l_{NT} [\widehat R_1(\alpha,\gamma_0)-\widehat{\mathbb C}_2(\alpha,\gamma)] $, which  does not depend on $\gamma$, and
$$
\mathbb K_{4T}(g):=l_{NT}[\mathds G_1(H_T (\gamma _{0}+g\cdot r_{NT}^{-1}))- \widehat{\mathbb C}_1(\alpha_0, \gamma _{0}+g\cdot r_{NT}^{-1})] .
$$
 Define
	 	\begin{eqnarray*}
	 		(\widehat a_h, \widehat g_h)&=& \arg\min_{a,  h(\gamma_0+g_hr_{NT}^{-1})=0} \mathbb{K}_{T}\left( a,g_h\right) ,\quad (\widehat a, \widehat g)= \arg\min_{a,  g } \mathbb{K}_{T}\left( a,g\right) \cr
	 		\widehat g_h&=& r_{NT}^{-1}\left(\widehat \gamma_h -\gamma _{0}\right),\qquad \widehat g= r_{NT}^{-1}\left(\widehat \gamma -\gamma _{0}\right)
	 		\cr
	 		\widetilde g_h&=&\arg\min_{h(\gamma_0+g_hr_T^{-1})=0}  \mathbb K_{4T}(g) ,\qquad
			\widetilde g=\arg\min_{g}  \mathbb K_{4T}(g).
	 	\end{eqnarray*}%
Then the same proof as in Section \ref{scg.1.1} shows that
$$
 \mathbb K_{4T}(\widehat g_h)= \mathbb K_{4T}(\widetilde g_h)+o_P(1),\quad  \mathbb K_{4T}(\widehat g)= \mathbb K_{4T}(\widetilde g)+o_P(1).
$$ Thus
\begin{eqnarray*}
A_1&=&  l_{NT}[\mathds G_1(H_T\widehat \gamma_h)- \widehat{\mathbb C}_1( \alpha_0,\widehat \gamma_h)]  +o_P(1)\cr
&=& \mathbb K_{4T}(\widehat g_h)  +o_P(1)
=\arg\min_{h(\gamma_0+g_hr_T^{-1})=0}  \mathbb K_{4T}(g_h)+o_P(1)\cr
&=&\arg\min_{g_h'\nabla h=0}  \mathbb K_{4T}(g_h)+o_P(1)\cr
A_3&=&  l_{NT}[\mathds G_1(H_T\widehat \gamma)- \widehat{\mathbb C}_1(  \alpha_0,\widehat \gamma)]  +o_P(1)\cr
&=& \arg\min_g\mathbb K_{4T}(g) +o_P(1).
\end{eqnarray*}
Also $A_2=o_P(1)$  following  a  similar proof as in Lemma \ref{l5.1}.
Sections \ref{sec:EP} \ref{sec:Bias} show that  $\mathbb K_{4T}(\cdot)\Rightarrow \mathbb Q(\omega, \cdot)$, where $l_{NT}\mathds G_1(H_T (\gamma _{0}+g\cdot r_{NT}^{-1}))$ is the bias part and $ l_{NT}\widehat{\mathbb C}_1(\alpha_0, \gamma _{0}+g\cdot r_{NT}^{-1})$ is the empirical process part.
 Hence by the continuous mapping theorem,
$$
	 l_{NT} \cdot LR\to^d \sigma_{\varepsilon}^{-2}\min_{ g_h'\nabla h =0} \mathbb Q(\omega, g_h)   - \sigma_{\varepsilon}^{-2}\min_{ g} \mathbb Q(\omega, g).
	 	$$

\end{proof}
			\subsubsection{Proof of the  distribution of $LR_k^*$}

		\begin{proof}
	 	The proof below holds for both  cases of bootstrap estimated factors:

	 	\textbf{re-estimated factors using PCA}
		Recall that $f_t^*=H^{*'-1}_TF_t^*$. Then we have
			\begin{equation}\label{eqg.5}
		f_t^* 
		=\widehat f_t^* +H_T^{*'-1}r_t^*+m_t,\quad
		\widehat f_t^*:=\widehat f_t +\frac{1}{\sqrt{N}}h_t^*
		\end{equation}
		where $m_t=\widetilde f_t-\widehat f_t$, whose probability bound is given in Section \ref{sec:E1}. Note that the bootstrap probability bound for $r_t^*$ is similar to that of $m_t$. Also,
		$$
		h_{t}^*=(h_{1t}^*, 0),\quad h_{1t}^*\sim \mathcal N(0,\widehat\Sigma_h)
		$$

	\textbf{perturbed bootstrap factors}

	 In the case of directly using perturbed factors for bootstrap: $ f_t^*=\widetilde f_t+ N^{-1/2} h_t^*$,
	 we can still write
	 		\begin{equation}\label{eqg.5.perturb}
	 f_t^* 
	 =\widehat f_t^* +H_T^{*'-1}r_t^*+m_t,\quad
	 \widehat f_t^*:=\widehat f_t +\frac{1}{\sqrt{N}}h_t^*
	 \end{equation}
	 where $r_t^*=0$ and
	 	$$
	 h_{t}^*=(h_{1t}^*, 0),\quad h_{1t}^*
	 = \widehat\Sigma_h^{1/2} \mathcal W_t^*
	 \sim \mathcal N(0,\widehat\Sigma_h).
	 $$

  \textbf{Step 1. Expansion of $l_{NT}(\mathbb S_T^*(\alpha,\gamma)-\mathbb S^*_T(\alpha,\widehat\gamma))$.}

				 We have
			$
		l_{NT} \mathbb{ {S}}_{T}^{\ast }(\widehat\alpha^*,\widehat\gamma^*)LR_k^* =	l_{NT} [\mathbb S_T^*(\widehat\alpha_h^*, \widehat\gamma_h^*)-\mathbb S_T^*(\widehat\alpha^*,\widehat\gamma^*)] =A_1^*
		+A_2^*	-A^*_3,
			$
			where
			\begin{eqnarray*}
				A^*_1&=&	l_{NT} [\mathbb S_T^*(\widehat\alpha_h^*, \widehat\gamma_h^*)-\mathbb S_T^*(\widehat\alpha_h^*,\widehat\gamma)],
				A^*_2 =	l_{NT} [\mathbb S_T^*(\widehat\alpha_h^*, \widehat\gamma )-\mathbb S_T^*(\widehat\alpha^*,\widehat\gamma)],  A^*_3 =	l_{NT} [\mathbb S_T^*(\widehat\alpha^*, \widehat\gamma^*)-\mathbb S_T^*(\widehat{\alpha}^*,\widehat\gamma)].
			\end{eqnarray*}

Define
		\begin{eqnarray*}
		\widetilde{R}^*_{1}\left( \alpha ,\gamma \right) &=&\frac{1}{T}\sum_{t=1}^{T}\left({Z}^*_{t}\left( \gamma \right) ^{\prime }\left(\alpha -\widehat\alpha \right) \right) ^{2}\cr
 \widetilde{R}_{2}^*\left(  \gamma \right) &=&\frac{1}{T}\sum_{t=1}^{T}\left( x_{t}^{\prime}\delta_0  \right) ^{2}      |   1\{ {f}_{t}^{*\prime }\gamma >0\} -1\{ f_t^{*'}\widehat\gamma>0\}   |,  \cr
			 \widetilde{R}_{3}^*\left( \alpha ,\gamma \right) &=&\frac{2}{T}\sum_{t=1}^{T}x_{t}^{\prime }\widehat\delta \left( 1\left\{f_t^{*'}\gamma>0\right\} -1\left\{ f_t^{*'}\widehat\gamma>0\right\} \right) {Z}^*_{t}\left( \gamma \right) ^{\prime }\left(\alpha -\widehat\alpha \right), \cr
				\widetilde R_4^*(\alpha,\gamma)&=&\frac{1}{T}\sum_t2  Z_t^*(\widehat\gamma)'(\alpha-\widehat\alpha)
			x_t'(\delta-\widehat\delta) (1\{f_t^{*'}\gamma>0\}-1\{f_t^{*'}\widehat\gamma>0\}) \cr
\widetilde R_5^*(\alpha,\gamma)&=& \frac{1}{T}\sum_t
			(x_t'(\delta-\delta_0) )^2|1\{f_t^{*'}\gamma>0\}-1\{f_t^{*'}\widehat\gamma>0\}| \cr
\widetilde R_6^*(\alpha,\gamma)&=& \frac{2}{T}\sum_t
			x_t'(\delta-\delta_0) x_t'\delta |1\{f_t^{*'}\gamma>0\}-1\{f_t^{*'}\widehat\gamma>0\}| \cr
			\mathbb{\widetilde{C}}^*_{1}\left( \alpha ,\gamma \right) &=&
			\frac{2}{T}\sum_{t=1}^{T}\eta_t\widehat\varepsilon _{t}x_{t}^{\prime }\delta \left( 1\{  f_t^{*'}\gamma>0\} -1\{  f_t^{*'}\widehat\gamma>0\} \right), \cr
			\mathbb{\widetilde{C}}^*_{2}\left( \alpha \right) &=&\frac{2}{T}\sum_{t=1}^{T}\eta_t\widehat\varepsilon _{t} {Z}^*_{t}\left( \widehat \gamma \right) ^{\prime}\left( \alpha -\widehat \alpha \right), \cr
			\mathbb{\widetilde{C}}^*_{3}\left( \alpha ,\gamma \right) &=&
			\frac{2}{T}\sum_{t=1}^{T}x_{t}^{\prime }\widehat\delta x_{t}^{\prime }\delta\left( 1\{f_t^{*'}\widehat\gamma>0\} -1\{\widetilde f_t'\widehat\gamma>0\}\right) \left( 1\{f_t^{*'}\gamma>0\} -1\{f_t^{*'}\widehat\gamma>0\} \right), \cr
			\mathbb{\widetilde{C}}^*_{4}\left( \alpha \right) &=&\frac{2}{T}\sum_{t=1}^{T}x_{t}^{\prime }\widehat\delta  \left( 1\{f_t^{*'}\widehat\gamma>0\} -1\{\widetilde f_t'\widehat\gamma>0\}\right) {Z}^*_{t}\left(\widehat \gamma  \right) ^{\prime}\left( \alpha -\widehat\alpha \right),\cr
			\mathbb{\widehat{C}}^*_{1}\left(   \gamma \right) &=&
			\frac{2}{T}\sum_{t=1}^{T}\eta_t\widehat\varepsilon _{t}x_{t}^{\prime }\widehat\delta \left( 1\{ \widehat f_t^{*'}\gamma>0\} -1\{  \widehat f_t^{*'}\widehat\gamma>0\} \right), \cr
			\widehat{R}_{2}^*\left(  \gamma ,\widehat\gamma\right) &=&\frac{1}{T}\sum_{t=1}^{T}\left( x_{t}^{\prime}\delta_0  \right) ^{2}      |   1\{ \widehat{f}_{t}^{*\prime }\gamma >0\} -1\{ \widehat f_t^{*'}\widehat\gamma>0\}   |,  \cr
			\mathbb{\widehat{C}}^*_{3}\left(  \gamma,\widehat\gamma \right) &=&
			\frac{2}{T}\sum_{t=1}^{T}(x_{t}^{\prime }\delta_0   )^2\left( 1\{ \widehat  f_t^{*'}\widehat\gamma>0\} -1\{\widehat f_t'\widehat\gamma>0\}\right) \left( 1\{\widehat f_t^{*'}\gamma>0\} -1\{\widehat f_t^{*'}\widehat\gamma>0\} \right).
		\end{eqnarray*}

Uniformly in $|\alpha-\widehat\alpha|_2=O_P(T^{-1/2})$ and $|\gamma-\widehat\gamma|_2=O_P(r_{NT}^{-1})$,  we   have $l_{NT} |\widetilde R^*_d(\alpha,\gamma)|=o_{P^*}(1)$ for $d=3\sim 6$,
$l_{NT}|\mathbb{\widehat{C}}^*_{3}\left(  \gamma ,\widehat\gamma\right) -\mathbb{\widetilde{C}}^*_{3}\left( \alpha ,\gamma \right) | + l_{NT}|\widehat{R}_{2}^*\left(  \gamma,\widehat\gamma \right) -\widetilde{R}_{2}^*\left(  \gamma \right) |=o_{P^*}(1)$,  and
   $l_{NT}|\widehat{\mathbb C}_1^*(  \gamma)-  \mathbb{\widetilde{C}}^*_{1}\left( \alpha ,\gamma \right) |=o_{P^*}(1)$.   These convergences are straightforward to verify as in the original sampling space.  We verify $l_{NT}|\widehat{\mathbb C}_1^*(  \gamma)-  \mathbb{\widetilde{C}}^*_{1}\left( \alpha ,\gamma \right) |=o_{P^*}(1)$ at the end of the proof (step 5) for illustration.

   Next, write
   $$
  \breve  g_{t,H,\Sigma}:= \breve g_t +N^{-1/2} H^{'-1}\Sigma^{1/2}\mathcal W_t^*,
   $$
   where $\mathcal W_t^*$ is standard normal.  Also write
    \begin{eqnarray*}
\mathds G_{H,\Sigma}(\phi_2, \phi_1)&:=&	 2\mathbb E(x_{t}^{\prime }\delta_0   )^2\left( 1\{ \breve g_{t, H,\Sigma}'\phi_1>0\} -1\{\breve g_t'\phi_1>0\}\right) \left( 1\{\breve g_{t, H,\Sigma}'\phi_2>0\} -1\{\breve g_{t, H,\Sigma}'\phi_1>0\} \right)\cr
&&+ \mathbb E\left( x_{t}^{\prime}\delta_0  \right) ^{2}      |   1\{\breve g_{t, H,\Sigma}'\phi_2>0\} -1\{\breve g_{t, H,\Sigma}'\phi_1>0\} |.
    \end{eqnarray*}
  where $\mathbb E$ is with respect to the joint distribution of the sampling distribution and  $\mathcal W_t^*$.
   Then $\widehat f_t^*= H_T'\breve g_{t,H_T,\widehat\Sigma}$, and
   $
   \mathds G_{H_T,\widehat\Sigma}(H_T\gamma, H_T\widehat\gamma)=\mathbb E(\mathbb{\widehat{C}}^*_{3}\left(  \gamma_1  , \gamma_2\right)
   +\widehat{R}_{2}^*\left(  \gamma_1 ,\gamma_2\right)).
   $
Then for $\phi=H_{T}\gamma$ and $\widehat\phi=H_T\widehat\gamma$,
	\begin{eqnarray}\label{ef.5}
l_{NT}(\mathbb S_T^*(\alpha,\gamma)-\mathbb S^*_T(\alpha,\widehat\gamma))
&=& l_{NT}(  \widetilde{\mathbb C}^*_3(\alpha,\gamma) - \widetilde{\mathbb C}^*_1(\alpha,\gamma) +\sum_{d=3}^6\widetilde R_d^*(\alpha,\gamma)+\widetilde R_2^*(\gamma))\cr
		&=&l_{NT}(\mathbb{\widehat{C}}^*_{3}\left(  \gamma ,\widehat\gamma\right)
+\widehat{R}_{2}^*\left(  \gamma ,\widehat\gamma\right) )- l_{NT}\mathbb{\widehat{C}}^*_{1}\left( \gamma \right) +o_{P^*}(1)\cr
&=&l_{NT}\mathds G_{H_T,\widehat\Sigma}(\phi,\widehat\phi) - l_{NT}\mathbb{\widehat{C}}^*_{1}\left(  \gamma \right) +o_{P^*}(1).
\end{eqnarray}

 \textbf{Step 2.  Probability limit of $l_{NT}\mathds G_{H_T,\widehat\Sigma}(\phi,\widehat\phi) $.}

Fix  $\phi_1$ in a neighborhood of $\phi_0$. We first obtain a similar expansion as in (\ref{e0.5add}).
 \begin{eqnarray*}
A_{1t}^*(\phi_2, \phi_1)&=&1\left\{ \breve   g_{t,H,\Sigma}'\phi_2\leq 0< \breve g_{t,H,\Sigma}'\phi_1\right\}
1\left\{ \breve g_{t}^{\prime }\phi _{1}>0\right\} \cr
A_{2t}^*(\phi_2, \phi_1)&=& 1\left\{ \breve  g_{t,H,\Sigma}'\phi_1\leq 0<%
 \breve  g_{t,H,\Sigma}'\phi_2 \right\} 1\left\{\breve  g_{t}^{\prime }\phi _{1}\leq
0\right\} \cr
A_{3t}^*(\phi_2, \phi_1)&=&1\left\{ \breve  g_{t,H,\Sigma}'\phi_2\leq 0< \breve g_{t,H,\Sigma}'\phi_1\right\}
1\left\{ \breve g_{t}^{\prime }\phi _{1}\leq 0 \right\}\cr
A_{4t}^*(\phi_2, \phi_1)&=&1\left\{\breve
g_{t,H,\Sigma}'\phi_1\leq 0<\breve   g_{t,H,\Sigma}'\phi_2 \right\}
1\left\{ \breve g_{t}^{\prime }\phi_{1} > 0\right\}.
\end{eqnarray*}

 Therefore,     $
\mathds G_{H,\Sigma}(\phi_2,\phi_1)= \mathbb E\left( x_{t}^{\prime }\delta _{0}\right) ^{2}\left( A_{1t}^*\left( \phi_2, \phi_1
\right) +A_{2t}^*\left( \phi_2, \phi_1 \right) -A_{3t}^*\left(\phi_2, \phi_1 \right)
-A_{4t}^*\left( \phi_2, \phi_1 \right) \right).
 $
 Let us calculate $A_{1t}^*$ first. For notational simplicity, write
 $$
 h_{t,H,\Sigma}^{*'}=\mathcal W_t^{*'}\Sigma^{1/2}H^{-1},\quad u_{Nt}:=\breve g_t'\phi_1.
 $$
Then
 \begin{eqnarray*}
 	\mathbb  E\left( x_{t}^{\prime }\delta _{0}\right) ^{2}A_{1t}^*&=&\mathbb  E\left( x_{t}^{\prime }\delta _{0}\right) ^{2}1\left\{ -h_{t,H,\Sigma}^{*'}\phi_1<\sqrt{N}u_{Nt}\leq -\sqrt{N} \breve g_t ^{\prime }\left( \phi_2-\phi_1\right) - h_{t,H,\Sigma}^{*'}\phi_1\right\} 1\left\{
 	h_{t,H,\Sigma}^{*'}\phi_1\leq 0\right\}    \cr
 	&&+\mathbb  E\left( x_{t}^{\prime }\delta _{0}\right) ^{2}1\left\{ 0<\sqrt{N}u_{Nt}\leq -\sqrt{N} \breve  g_t^{\prime
 	}\left( \phi_2-\phi_1\right) -  h_{t,H,\Sigma}^{*'}\phi_1\right\} 1\left\{ h_{t,H,\Sigma}^{*'}\phi_1>0\right\}+ A_{11},
 \end{eqnarray*}
 where the same proof as of Lemma \ref{l:generic} implies
 $
 	A_{11}\leq   \frac{CL}{NT^{2\varphi}} ,
$
given the assumption that $p_{\breve g_t'\phi|h_t^*}(\cdot)$ is bounded.
 Let $p_{u_{Nt}|\bigstar}(\cdot):=p_{u_{Nt}|h_{t,H,\Sigma}^{*'}\phi_1,   f_{2t},x_{t} }(\cdot)$ denote the conditional density of $u_{Nt}$.
 Change variable  $a=\sqrt{N}u$, we have,
 \begin{align}
 & \mathbb E\left( x_{t}^{\prime }\delta _{0}\right) ^{2}A_1- A_{11}\cr
 &=\frac{1}{\sqrt{N}} \mathbb E\left( x_{t}^{\prime }\delta _{0}\right) ^{2}\int 1\left\{ -h_{t,H,\Sigma}^{*'}\phi_1<a\leq -\sqrt{N}   \breve g_t ^{\prime }\left( \phi_2-\phi_1\right) - h_{t,H,\Sigma}^{*'}\phi_1\right\} 1\left\{
 h_{t,H,\Sigma}^{*'}\phi_1\leq 0\right\}  p_{u_{Nt}|\bigstar}(\frac{a}{\sqrt{N}})da  \cr
 &+\frac{1}{\sqrt{N}} \mathbb E\left( x_{t}^{\prime }\delta _{0}\right) ^{2}\int 1\left\{ 0<a\leq -\sqrt{N}  \breve g_t^{\prime
 }\left( \phi_2-\phi_1\right) -  h_{t,H,\Sigma}^{*'}\phi_1\right\} 1\left\{ h_{t,H,\Sigma}^{*'}\phi_1>0\right\}p_{u_{Nt}|\bigstar}(\frac{a}{\sqrt{N}})da\cr
 &=
 -  \mathbb E\left( x_{t}^{\prime }\delta _{0}\right) ^{2}p_{u_{Nt}|\bigstar}(0)    \breve g_t ^{\prime }\left( \phi_2-\phi_1\right)1\{   \breve g_t ^{\prime }\left( \phi_2-\phi_1\right)\leq 0\} 1\left\{
 h_{t,H,\Sigma}^{*'}\phi_1\leq 0\right\}     \cr
 &-  \mathbb E\left( x_{t}^{\prime }\delta _{0}\right) ^{2}p_{u_{Nt}|\bigstar}(0)  \left(   \breve g_t^{\prime
 }\left( \phi_2-\phi_1\right) + \frac{h_{t,H,\Sigma}^{*'}\phi_1}{\sqrt{N}}\right)1\left\{ \breve   g_t^{\prime
 }\left( \phi_2-\phi_1\right) + \frac{h_{t,H,\Sigma}^{*'}\phi_1}{\sqrt{N}}<0\right\}1\left\{ h_{t,H,\Sigma}^{*'}\phi_1>0\right\}
 \cr
 &+ B_1,
 \end{align}
 where the same proof as of Lemma \ref{l:generic} implies, due to
     $p_{u_{Nt}|\bigstar}(.)$ is Lipschitz,
 \begin{eqnarray*}
 	|B_1|
  \leq \frac{C'}{N}T^{-2\varphi}.
 \end{eqnarray*}
So the same proof as of Lemma \ref{l:generic} carries over to $A_{1t}^*...A_{4t}^*$, showing that a similar expansion as in (\ref{e0.5add}) holds: for $\Xi(a,b)$ is as defined in (\ref{e0.7add}), for $\breve g_t^{\prime }(\phi_2-\phi_1)=a$, $\frac{h_t^{*\prime }\phi_1}{\sqrt{N}}=b$, and $\phi_2= \phi_1+Hgr_{NT}^{-1}$, 
  \begin{eqnarray}\label{eqf.6}
 && l_{NT}  \mathds G_{H,\Sigma}(\phi_2,\phi_1)=\mathbb  E[(x_t'\delta_0)^2p_{u_{Nt}|\bigstar}(0) \Xi(a,b)] + o(l_{NT}^{-1} )\cr
 &=&  l_{NT}\mathbb{E}_{|u_{Nt}=0}(x_t^{\prime }\delta_0)^2p_{u_{Nt}}(0) %
 \left[\left| \breve g_t^{\prime }\left(\phi _2-\phi _{1}\right) + \frac{h_t^{*\prime
 	}\phi_1}{\sqrt{N}} \right | - \left |\frac{h_t^{*\prime }\phi_1 }{\sqrt{N}}
 \right | \right] +o(1).
    \end{eqnarray}

 When $\omega\in(0,\infty]$, $l_{NT}  \mathds G_{H,\Sigma}(\phi_2,\phi_1) = \mathbb{\breve{C}}_{N,H,\Sigma, \phi_1}(Hg)+o(1) $, where
   \begin{eqnarray}
\mathbb{\breve{C}}_{N,H,\Sigma, \phi_1}(\mathfrak{g}) &:=&M_{\omega}%
 \mathbb{E}_{|u_{Nt}=0}(x_t^{\prime }d_0)^2p_{u_{Nt}}(0) \left(\left| \breve g_t^{\prime }%
 \mathfrak{g} + \zeta_{\omega} ^{-1} h_t^{*\prime }\phi_1\right | - \left |
 \zeta_{\omega}^{-1} h_t^{*\prime }\phi_1 \right | \right)\cr
 \mathbb{\breve{C}}_{H,\Sigma, \phi_1}(\mathfrak{g}) &:=&M_{\omega}%
 \mathbb{E}_{|g_{t}'\phi_1=0}(x_t^{\prime }d_0)^2p_{g_{t}'\phi_1}(0) \left(\left|   g_t^{\prime }%
 \mathfrak{g} + \zeta_{\omega} ^{-1} h_t^{*\prime }\phi_1\right | - \left |
 \zeta_{\omega}^{-1} h_t^{*\prime }\phi_1 \right | \right).
   \end{eqnarray}
Note that $\mathbb{\breve{C}}_{N,H,\Sigma, \phi_1}(\mathfrak{g})=M_{\omega}%
 \mathbb{E}(x_t^{\prime }d_0)^2p_{u_{Nt}|x_t, g_t}(0) \left(\left| \breve g_t^{\prime }%
 \mathfrak{g} + \zeta_{\omega} ^{-1} h_t^{*\prime }\phi_1\right | - \left |
 \zeta_{\omega}^{-1} h_t^{*\prime }\phi_1 \right | \right)$, so  by the assumption
 $
\sup_{x_t, h_t, g_t,\phi_1}| p_{\breve g_t'\phi_1|x_t,  f_{2t}}(0)-p_{g_t'\phi_1|x_t, f_{2t}}(0)|=o(1),
$  uniformly in $\phi_1$, $$l_{NT}  \mathds G_{H,\Sigma}(\phi_2,\phi_1) = \mathbb{\breve{C}}_{H,\Sigma, \phi_1}(Hg)+o(1).  $$

Let   $\phi=  \widehat\phi+H_Tgr_{NT}^{-1}$. Note that $\widehat \Sigma\to^PH'\Sigma H .$
By the assumption that  $\mathbb{\breve{C}}_{H,\Sigma, \phi_1}(\mathfrak{g})$ is continuous  in $(H,\Sigma, \phi_1, \mathfrak{g})$,
\begin{eqnarray*}
&& l_{NT}\mathds G_{H_T,\widehat\Sigma}(\phi,\widehat\phi)
= \mathbb{\breve{C}}_{H_T,\widehat\Sigma , \widehat\phi}(H_Tg)+o_P(1)
= \mathbb{\breve{C}}_{H, H' \Sigma H,  \phi_0}(Hg)+o_P(1)
\cr
&=&M_{\omega}%
 \mathbb{E}_{|u_t=0}(x_t^{\prime }d_0)^2p_{u_t}(0) \left(\left|   g_t^{\prime }%
 H g+ \zeta_{\omega} ^{-1}  \mathcal W_t^{*'} (H'\Sigma H)^{1/2}H^{-1}\phi_0\right | - \left |
 \zeta_{\omega}^{-1}  \mathcal W_t^{*'}(H'\Sigma H)^{1/2}H^{-1}\phi_0 \right | \right)\cr
 &=&M_{\omega}%
 \mathbb{E}_{|u_t=0}(x_t^{\prime }d_0)^2p_{u_t}(0) \left(\left|   g_t^{\prime }%
  H g+ \zeta_{\omega} ^{-1}   \mathcal Z_t \right | - \left |
 \zeta_{\omega}^{-1}  \mathcal Z_t \right | \right)\cr
 &=&A(\omega, g),
\end{eqnarray*}
 where we note that  $ \mathcal W_t^{*'}(H'\Sigma H)^{1/2}H^{-1}\phi_0\sim\mathcal N(0,\sigma_h^2)$, with
 $\sigma_h^2= \phi_0' \Sigma  \phi_0= \lim \var(h_t'\phi_0)=\sigma_{h,x_t,g_t}^2$ in the homoskedastic case.  So $\mathcal W_t^{*'}(H'\Sigma H)^{1/2}H^{-1}\phi_0=^d\mathcal Z_t$.

  When $\omega\in(0,\infty]$,  we now work with  (\ref{eqf.6}), where  we treat   terms
  $a_2,  a_5$ in the definition of    $\Xi(a,b)$   as we did for  (\ref{e0.3})  (\ref{e0.4}) . Here
  \begin{eqnarray}
a_2&=&-       \breve  g_t ^{\prime }\left( \phi_2-\phi_1\right)1\{   \breve    g_t ^{\prime }\left( \phi_2-\phi_1\right)\leq 0\} 1\left\{
h_{t,H,\Sigma}^{*'}\phi_1\leq 0\right\}    \cr
a_5&=&
    \breve   g_t'\left(\phi_2-\phi_1\right)1\left\{      \breve  g_t'\left(\phi_2-\phi_1\right) >0\right\}  1\{h_{t,H,\Sigma}^{*'}\phi_1>0\} \cr
    a_2'&=&-       \breve  g_t ^{\prime }\left( \phi_2-\phi_1\right)1\{   \breve    g_t ^{\prime }\left( \phi_2-\phi_1\right)\leq 0\} 1\left\{
h_{t,H,\Sigma}^{*'}\phi_1> 0\right\}    \cr
a_5'&=&
    \breve   g_t'\left(\phi_2-\phi_1\right)1\left\{      \breve  g_t'\left(\phi_2-\phi_1\right) >0\right\}  1\{h_{t,H,\Sigma}^{*'}\phi_1\leq0\} .
  \end{eqnarray}
   We note that
  $h_{t,H,\Sigma}^{*'}\phi_1$ is symmetric around zero, due to the Gaussianity of $h_t^*$.  Hence
  $$
  \mathbb P(h_{t,H,\Sigma}^{*'}\phi_1\leq0| x_t, \breve g_t)= \mathbb P(h_{t,H,\Sigma}^{*'}\phi_1>0| x_t, \breve g_t)  =1/2.$$
  So
  $\mathbb E(x_t'\delta_0)^2p_{u_{NT} | \bigstar}(0)a_d= \mathbb E(x_t'\delta_0)^2p_{u_{NT} | \bigstar}(0)a_d'$ for $d=2,5$,  and we reach an expansion similar to (\ref{e0.10add}): for $\phi_2= \phi_1+Hgr_{NT}^{-1}$,
  \begin{eqnarray*}
&& l_{NT}  \mathds G_{H,\Sigma}(\phi_2,\phi_1) =o_P(1)  \cr
&&
-2  l_{NT}  \mathbb E(x_t'\delta_0)^2p_{u_{Nt} | \bigstar}(0)  \left(   \breve g_t^{\prime
}\left( \phi_2-\phi_1\right) + \frac{h_{t,H,\Sigma}^{*'}\phi_1}{\sqrt{N}}\right)1\left\{   \breve g_t^{\prime
}\left( \phi_2-\phi_1\right) + \frac{h_{t,H,\Sigma}^{*'}\phi_1}{\sqrt{N}}<0\right\}1\left\{ h_{t,H,\Sigma}^{*'}\phi_1>0\right\}\cr
&&+2 l_{NT}  \mathbb E(x_t'\delta_0)^2  p_{u_{Nt} | \bigstar}(0) \left(    \breve g_t'\left(
\phi_2  -\phi _{1}\right) +\frac{ h_{t,H,\Sigma}^{*'}\phi_1}{\sqrt{N}}\right)1\left\{ \breve    g_t'\left(
\phi_2 -\phi _{1}\right) + \frac{h_{t,H,\Sigma}^{*'}\phi_1}{\sqrt{N}}> 0\right\} 1\left\{h_{t,H,\Sigma}^{*'}\phi_1
\leq 0\right\}\cr
&=&\mathbb{\breve{C}}_{N,H,\Sigma, \phi_1,2}( Hg) +o(1)
  \end{eqnarray*}
where we used a similar  change-variable as in  Step II.1  in Section \ref{sec:Bias}:
\begin{eqnarray*}
	&&\mathbb{\breve{C}}_{N,H,\Sigma, \phi_1,2}(\mathfrak{g})\cr
	&:=&  - \widetilde M_{NT}2
	p_{u_{Nt}}(0)\mathbb{E}[(x_t^{\prime }d_0)^2 F_{1}(\breve g_t, x_t,  \mathfrak{g }) |u_{Nt}=0]+
	\widetilde M_{NT}2 p_{u_{Nt}}(0) \mathbb{E}[(x_t^{\prime }d_0)^2 F_{2} (\breve g_t,
	x_t,  \mathfrak{g }) |u_{Nt}=0]
\cr
 && F_{1} (\breve g_t, x_t,  \mathfrak{g } ) :=\int
	\left( \breve g_t^{\prime }\mathfrak{g} + y\right)1\left\{ \breve g_t^{\prime } \mathfrak{g%
	} + y<0\right\}1\left\{ y>0\right\} p_{h_t^{*\prime }\phi_1 }(\zeta_{NT}y)d y \cr && F_{2} (\breve g_t, x_t,  \mathfrak{g }) := \int \left(
	\breve g_t^{\prime }\mathfrak{g} +y\right)1\left\{ \breve g_t^{\prime }\mathfrak{g} + y>
	0\right\} 1\left\{y \leq 0\right\} p_{h_t^{*\prime }\phi_1 }(\zeta_{NT}y)d y.
\end{eqnarray*}
Note $\zeta_{NT}\to0$,  $\widetilde M_{NT}\to1$,
$ |p_{h_t^{*\prime }\phi_1 }(\zeta_{NT}y)- p_{h_t^{*\prime }\phi_1 }(0)|\leq C\zeta_{NT}y$ (Gaussian densities with bounded variance),  so
\begin{eqnarray*}
&&\mathbb E_{|u_{Nt}=0}(x_t'd_0)^2\int| \left( \breve g_t^{\prime }\mathfrak{g} + y\right)1\left\{ \breve g_t^{\prime } \mathfrak{g%
	} + y<0\right\}1\left\{ y>0\right\} |p_{h_t^{*\prime }\phi_1 }(\zeta_{NT}y)-p_{h_t^{*\prime }\phi_1 }(0)|d y  \cr
&\leq&  C\zeta_{NT}\mathbb E_{|u_{Nt}=0}(x_t'd_0)^2	\int |\left( \breve g_t^{\prime }\mathfrak{g} + y\right)|1\left\{ \breve g_t^{\prime } \mathfrak{g%
	} + y<0\right\}1\left\{ y>0\right\}  yd y \cr
	&\leq&   C\zeta_{NT}\mathbb E_{|u_{Nt}=0}(x_t'd_0)^2(\breve g_t^{\prime }\mathfrak{g})^31\{\breve g_t^{\prime }\mathfrak{g}<0\}=o(1).
\end{eqnarray*}
 In addition, by the assumption that  $|p_{u_{Nt}, h_t^{*\prime }\phi_1|x_t, f_{2t}}(0,0)-p_{g_{t}'\phi_1, h_t^{*\prime }\phi_1|x_t, f_{2t}}(0,0)|=o(1)$,
\begin{eqnarray*}
 \mathbb{\breve{C}}_{N,H,\Sigma, \phi_1,2}(\mathfrak{g})
 &=&(\mathbb{E}(x_{t}^{\prime
}d_{0})^{2}(g_{t}^{\prime }\mathfrak{g})^{2}|u_{Nt}=0, h_t^{*\prime }\phi_1=0)p_{u_{Nt}, h_t^{*\prime }\phi_1}(0,0)+o(1)\cr
&=& \mathbb{\breve{C}}_{H,\Sigma, \phi_1,2}(\mathfrak{g})+o(1)\cr
 \mathbb{\breve{C}}_{H,\Sigma, \phi_1,2}(\mathfrak{g})&:=&(\mathbb{E}(x_{t}^{\prime
}d_{0})^{2}(g_{t}^{\prime }\mathfrak{g})^{2}|g_{t}'\phi_1=0, h_t^{*\prime }\phi_1=0)p_{g_{t}'\phi_1, h_t^{*\prime }\phi_1}(0,0)+o(1).
\end{eqnarray*}
Let   $\phi=  \widehat\phi+H_Tgr_{NT}^{-1}$.
By the assumption that  $  \mathbb{\breve{C}}_{H,\Sigma, \phi_1,2}(\mathfrak{g})$ is continuous  in $(H,\Sigma, \phi_1, \mathfrak{g})$,
 \begin{eqnarray*}
&& l_{NT}\mathds G_{H_T,\widehat\Sigma}(\phi,\widehat\phi)
= \mathbb{\breve{C}}_{H_T,\widehat \Sigma, \widehat \phi,2}(H_Tg)+o_P(1)=  \mathbb{\breve{C}}_{H, H' \Sigma H,   \phi_0,2}(Hg)+o_P(1)\cr
&=&(\mathbb{E}(x_{t}^{\prime
}d_{0})^{2}(g_{t}^{\prime } Hg)^{2}|g_{t}'\phi_0=0,  \mathcal W_t^{*'}(H'\Sigma H)^{1/2}H^{-1}\phi_0=0)p_{g_{t}'\phi_0,  \mathcal W_t^{*'}(H'\Sigma H)^{1/2}H^{-1}\phi_0}(0,0)+o_P(1)\cr
&=&(\mathbb{E}(x_{t}^{\prime
}d_{0})^{2}(g_{t}^{\prime } Hg)^{2}| u_t=0,  \mathcal Z_t=0)p_{ u_t, \mathcal Z_t}(0,0)+o_P(1)\cr
&=&A(0,g)+o_P(1).
\end{eqnarray*}
  Together
  $$
  l_{NT}\mathds G_{H_T,\widehat\Sigma}(\widehat\phi+H_Tgr_{NT}^{-1},\widehat\phi)= A(\omega, g)+o_P(1).
  $$

  \textbf{Step 3.  Empirical process part.}

 Lemma \ref{bootstrapempirical} shows that in the  estimated factor case,
	 	  $   l_{NT}\mathbb{\widehat{C}}^*_{1}\left(  \widehat\gamma
 +r_{NT}^{-1}g\right) \Rightarrow^{*} 2W(g)$, where $\widehat f_t^*=\widehat f_t+ N^{-1/2} \mathcal Z_t^*$, and
$ \mathbb{\widehat{C}}^*_{1}\left(   \gamma \right) =
			\frac{2}{T}\sum_{t=1}^{T}\eta_t\widehat\varepsilon _{t}x_{t}^{\prime }\widehat\delta \left( 1\{ \widehat f_t^{*'}\gamma>0\} -1\{  \widehat f_t^{*'}\widehat\gamma>0\} \right). $

   \textbf{Step 4.   Finish the proof.}

 Together, we have shown that
 \begin{eqnarray}\label{eqf.9}
l_{NT}(\mathbb S_T^*(\widehat\alpha+ aT^{-1/2}, \widehat\gamma+r_{NT}^{-1}g)-\mathbb S^*_T( \widehat\alpha+ aT^{-1/2},\widehat\gamma))
&=& \mathbb K_{4T}^*(g) +o_{P^*}(1),
\end{eqnarray}
 and $\mathbb K_{4T}^*(.)\Rightarrow ^* \mathbb Q(\omega, .)$, where
$$
\mathbb K_{4T}^*(g):=l_{NT}[\mathds G_{H_T,\widehat\Sigma}( H_T(\widehat\gamma+r_{NT}^{-1}g), H_T\widehat\gamma) - \mathbb{\widehat{C}}^*_{1}\left( \delta,  \widehat\gamma+r_{NT}^{-1}g \right)].
$$ In addition, let $\mathbb K_1^*(a):=l_{NT} [\widetilde R_1^*( \widehat \alpha +a\cdot T^{-1/2},\widehat\gamma)-\widetilde{\mathbb C}^*_2( \widehat \alpha +a\cdot T^{-1/2})] $ and
 \begin{eqnarray*}
\mathbb{K}^*_{T}\left( a,g\right) &:=&l_{NT}\left( \mathbb{S}^*_{T}\left(\widehat \alpha +a\cdot T^{-1/2},\widehat \gamma +g\cdot r_{NT}^{-1}\right) -\mathbb{S}^*
	 		_{T}\left( \widehat \alpha,\widehat\gamma  \right) \right) \cr
&=& l_{NT}\sum_{d=1}^3\widetilde R^*_d(\alpha,\gamma)
- l_{NT}\sum_{d=1}^2\widetilde{\mathbb C}^*_d(\alpha,\gamma) +l_{NT}\sum_{d=3}^4\widetilde{\mathbb C}^*_d(\alpha,\gamma) \cr
&=&  \mathbb K_{4T}^*(g)+\mathbb K_1^*(a)+o_P(1).
	 	\end{eqnarray*}%
	 By Lemma \ref{bootstrapd},
	 	 $	|\widehat\alpha^*-\widehat\alpha|_2=O_{P^*}(T^{-1/2})$, and $	|\widehat\gamma_h^*-\widehat\gamma|_2=O_P(r_{NT}^{-1}).$
 Define
	 	 \begin{eqnarray*}
	 	 	\widehat a^*&=&\sqrt{T}(\widehat\alpha_h^*-\widehat\alpha) ,\quad  \widehat g_h^*=r_{NT}(\widehat\gamma_h^*-\widehat\gamma)
	 	 	\cr
	 	 	\widetilde g_h^*&:=&\arg\min_{h(\widehat\gamma+gr_{NT}^{-1})=h(\widehat\gamma)}\mathbb K_{4T}^*( g)  .
	 	 \end{eqnarray*}
Then 	 because $h(\widehat\gamma+ \widetilde g_h^*r_{NT}^{-1})= h(\widehat\gamma)$,
$$
 \mathbb K_{4T}^*(\widehat g_h^*)+\mathbb K_1^*(\widehat a^* )+o_P(1)=\mathbb{K}^*_{T}\left( \widehat a^*,\widehat g_h^*\right)
\leq \mathbb{K}^*_{T}\left( \widehat a^*,\widetilde g_h^*\right) +o_{P^*}(1)
= \mathbb K_{4T}^*(\widetilde g_h^*)+\mathbb K_1^*(\widehat a^*)+o_{P^*}(1).
$$
where the  inequality is due to Lemma \ref{bootstrapd} that  $  \mathbb S_T^*(\widehat\alpha_h^*,\widehat\gamma_h^*)
	 	 \leq \min_{\alpha, h(\gamma)=h(\widehat\gamma)}\mathbb S_T^*(\alpha,\gamma)+ o_{P^*}(l_{NT}^{-1})$. This implies
	$ \mathbb K_{4T}^*(\widehat g_h^*)= \mathbb K_{4T}^*(\widetilde g_h^*)+o_{P^*}(1).$
Therefore, 	 by (\ref{eqf.9}),
	\begin{eqnarray*}
				A^*_1&=&	l_{NT} [\mathbb S_T^*(\widehat\alpha_h^*, \widehat\gamma_h^*)-\mathbb S_T^*(\widehat\alpha_h^*,\widehat\gamma)]
= \mathbb K_{4T}^*(\widehat g_h^*) +o_{P^*}(1),
\cr
&=&\min_{h(\widehat\gamma+gr_{NT}^{-1})=h(\widehat\gamma)}\mathbb K_{4T}^*( g) 	 +o_{P^*}(1) =\min_{g_h'\nabla h=0 }\mathbb K_{4T}^*( g) 	 +o_{P^*}(1)\cr
&\to^{d^*}&	\min_{g_h'\nabla h=0 }\mathbb Q(\omega,  g_h) .
\end{eqnarray*}

	 As for $A_2^*$,  it follows from
$\widehat\alpha_h^*-\widehat\alpha^*=o_{P^*}(T^{-1/2})$ that  $A_2^*=o_{P^*}(1).$

Next, by Lemma \ref{bootstrapd},
	 	 $	|\widehat\gamma^*-\widehat\gamma|_2=O_P(r_{NT}^{-1})$
and   $  \mathbb S_T^*(\widehat\alpha^*,\widehat\gamma^*)
	 	 \leq \min_{\alpha, \gamma}\mathbb S_T^*(\alpha,\gamma)+ o_{P^*}(l_{NT}^{-1})$, we have $A^*_3\to^{d^*} 	\min_{g}\mathbb Q(\omega,  g) .$
		 Hence
 $$
		l_{NT} \mathbb{ {S}}_{T}^{\ast }(\widehat\alpha^*,\widehat\gamma^*)LR^*  \to^{d^*}
		\min_{g_h'\nabla h=0 }\mathbb Q(\omega,  g_h)-
				\min_{g}\mathbb Q(\omega,  g). $$ This finishes the proof since $ \mathbb{ {S}}_{T}^{\ast }(\widehat\alpha^*,\widehat\gamma^*)\to^{P^*} \sigma^2$.

   \textbf{Step 5. verify   $l_{NT}|\widehat{\mathbb C}_1^*(  \gamma)-  \mathbb{\widetilde{C}}^*_{1}\left( \alpha ,\gamma \right) |=o_{P^*}(1)$. }

Note that   we can bound
$
|\widehat\epsilon_t| \leq |\epsilon_t| +C |x_t|_2
$
with high probability. Thus for $w_t:=2|\eta_t \varepsilon_t| |x_{t}|_2+2|\eta_t | |x_t|_2^2 $,
uniformly for $|\gamma-\gamma_0|_2<C r_{NT}^{-1}$ and $|\alpha-\alpha_0|_2< CT^{-1/2}$,
 \begin{eqnarray*}
b&:=& \frac{1}{T}\sum_{t=1}^{T}\eta_t\widehat\varepsilon _{t}x_{t}^{\prime }\delta [1\{  f_t^{*'}\gamma>0\}
-\widehat f_t^{*'}\gamma>0\}
\leq \frac{1}{T}\sum_{t=1}^{T}|\eta_t\widehat\varepsilon _{t}x_{t}^{\prime }\delta| 1\{   0<|\widehat f_t^{*'}\gamma|< C |\widetilde f_t-\widehat f_t|_2\} \cr
&\leq &\frac{1}{T}\sum_{t=1}^{T}|\eta_t\widehat\varepsilon _{t}x_{t}^{\prime }\delta| 1\{   0<\inf_{\gamma}|\widehat f_t^{*'}\gamma|< C\Delta_f \} +   \mathbb P( |\widetilde f_t-\widehat f_t|_2> \Delta_f)^{1/2}\cr
&\leq &o(l_{NT}^{-1})+O_P(T^{-\varphi})   \mathbb P(   0<\inf_{\gamma}|\widehat f_t^{*'}\gamma|< C\Delta_f  ) =o_P(l_{NT}^{-1})
 \end{eqnarray*}
 given that  $\inf_{\gamma}|\widehat f_t^{*'}\gamma|$ has a density bounded and continuous  at zero.
 Next,  write
\begin{eqnarray*}
 a_T:= \frac{2}{T}\sum_{t=1}^{T}\eta_t\widehat\varepsilon _{t}x_{t}^{\prime }  \left( 1\{ \widehat f_t^{*'}\gamma>0\} -1\{  \widehat f_t^{*'}\widehat\gamma>0\} \right).
  \end{eqnarray*}
 Then $\mathbb E^*a_T=0$, where $\mathbb E^*$ is the conditional expectation with respect to the distribution of $(\eta_t, \mathcal W_t^*)$, and note that $\eta_t, \mathcal W_t^*$ are independent. Now we apply Lemma \ref{Lem:rate_gm}  to the bootstrap distribution, to reach
 $
 a_T=O_{P^*}( T^{-\varphi}  ) \left[|\gamma-\widehat \gamma|_2
 +\frac{1}{T^{1-2\varphi}} \right].
 $

 Thus
 $
 |\widehat{\mathbb C}_1^*(  \gamma)-  \mathbb{\widetilde{C}}^*_{1}\left( \alpha ,\gamma \right) |\leq b+
| a_T|_2|\delta-\widehat\delta|_2 =o_{P^*}(l_{NT}^{-1}).
 $
		\end{proof}

Recall 	
	$$
r_{NT}^{-1}=\max\left( \frac{1}{\left( NT^{1-2\varphi }\right) ^{1/3}} , \frac{1}{T^{1-2\varphi}}\right).
$$	
 
	 \begin{lem} \label{bootstrapd} In the estimated  factor case, the k-step bootstrap estimators $(\widehat\alpha^*, \widehat\gamma^*, \widehat\gamma_h^*)$ satisfy:
	 	\begin{eqnarray*} 
	 		\mathbb S_T^*(\widehat \alpha^*,\widehat \gamma^*)& \leq&   \min_{\alpha, \gamma}\mathbb S_T^*(\alpha,\gamma) + o_{P^*}(l_{NT}^{-1}).\cr
	 		\mathbb S_T^*(\widehat\alpha_h^*,\widehat\gamma_h^*)
	 		&\leq& \min_{\alpha, h(\gamma)=h(\widehat\gamma)}\mathbb S_T^*(\alpha,\gamma)+ o_{P^*}(l_{NT}^{-1}), \quad h(\widehat\gamma_h^*)=h(\widehat\gamma)
	 		\cr
	 	 	|\widehat\alpha^*-\widehat\alpha|_2&=&O_{P^*}(T^{-1/2}),\quad |\widehat\alpha_h^*-\widehat\alpha|_2=O_{P^*}(T^{-1/2}),\quad |\widehat\alpha^*-\widehat\alpha^*_h|_2=o_{P^*}(T^{-1/2})\cr
	 		|\widehat\gamma_h^*-\widehat\gamma|_2&=&O_P( r_{NT}^{-1}) ,\quad 
	 		|\widehat\gamma^*-\widehat\gamma|_2=O_P(r_{NT}^{-1}). 
	 	\end{eqnarray*}
	 	
	 \end{lem}
	
	\begin{proof}
	 	Define
\begin{eqnarray*} 
	(\alpha^*_g, \gamma_g^*)&=&\arg\min \mathbb S_T^*(\alpha,\gamma),\quad 
	(\alpha^*_{g,h}, \gamma_{g,h}^*)=\arg\min_{\alpha, h(\gamma)=h(\widehat\gamma)} \mathbb S_T^*(\alpha,\gamma),\cr
	\alpha^*(\gamma)&=& \arg\min_\alpha \mathbb S_T^*(\alpha,\gamma),  \cr
	\gamma^*(\alpha)&=& \arg\min_\gamma \mathbb S_T^*(\alpha,\gamma),\quad 	\gamma_h^*(\alpha)= \arg\min_{\gamma: h(\gamma)=h(\widehat\gamma)} \mathbb S_T^*(\alpha,\gamma).
\end{eqnarray*} 

Also recall the following definitions
\begin{eqnarray*} 
	 \widetilde Z_t(\gamma) &=&(x_t', x_t' 1\{  \widetilde f_t'\gamma>0 \})' \cr
	  \widehat Z_t(\gamma) &=&(x_t', x_t' 1\{  \widehat f_t'\gamma>0 \})' \cr
	 Z_t^*(\gamma) &=&(x_t', x_t' 1\{  f_t^{*'}\gamma>0 \})' \cr
	  \widehat Z_t^*(\gamma) &=&(x_t', x_t' 1\{    \widehat f_t^{*'}\gamma>0 \})',
\end{eqnarray*} 
where 
\begin{eqnarray*} 
\widetilde f_t &=& \text{estimated factors in the original data world}, \cr
  f_t^*&=& \text{estimated factors in the bootstrap  data world}, \cr
   \widehat   f_t &=& H_T f_t+H_T \frac{1}{\sqrt{N}}h_t, \cr
 \widehat   f_t^*&=&\widehat   f_t + \frac{1}{\sqrt{N}}h_t^*.
\end{eqnarray*}

Our proof is divided into the following steps.

\textbf{step 0:  } $ |\gamma_{g,h}^*-\widehat\gamma|_2=O_{P^*}(r_{NT}^{-1})$, $|\gamma_g^*-\widehat\gamma|_2=O_{P^*}( r_{NT}^{-1})$ and $|\alpha^*_g-\widehat\alpha|_2=O_{P^*}(T^{-1/2}).$ 

Step 0 is regarding the statistical convergence of the global minimiums in the bootstrap sample. So the proof is the same as that for $|\widehat\gamma-\gamma_0|_2$ and $|\widehat\alpha-\alpha_0|_2.$

\textbf{step 1: } if $|\gamma-\widehat\gamma|_2=O_{P^*}(r_{NT}^{-1})$, then $	|\alpha^*(\gamma)-\widehat\alpha|_2=O_{P^*}(T^{-1/2}) $.   In addition,   $	|\alpha^*(\gamma)-\alpha^*_g|_2=o_{P^*}(T^{-1/2}) $, and $	|\alpha^*(\gamma)-\alpha_{g,h}^*|_2=o_{P^*}(T^{-1/2}) $.

In the proof of step 1 we shall show that $ {\alpha}^*\left( \gamma
 \right) $ is an oracle estimator in the  sense:
 \begin{eqnarray}\label{e.4a}
 {\alpha}^*\left( \gamma \right) -\widehat \alpha &=& 
 [\frac{1}{T}%
\sum_{t=1}^{T}   \widehat  Z^*_{t}\left( \widehat  \gamma  \right)        \widehat  Z^*_{t}\left( \widehat \gamma  \right) ^{\prime}]^{-1}
 \frac{1}{T}\sum_{t=1}^{T}    \widehat    Z^*_{t}\left( \widehat \gamma  \right)
\eta_t\widehat\varepsilon _{t}+   o_{P^*}(T^{-1/2}). 
\end{eqnarray}

Write $\mathcal A^*(\gamma)=\frac{1}{T}%
\sum_{t=1}^{T}  Z^*_{t}\left( \gamma \right)  Z^*_{t}\left( \gamma \right) ^{\prime
}=\frac{1}{T}%
\sum_{t=1}^{T} \widehat  Z^*_{t}\left( \gamma \right)  \widehat   Z^*_{t}\left( \gamma \right) ^{\prime
}+o_{P^*}(T^{-1/2})$.  Note that $ y_t^* = \widetilde Z_t(\widehat\gamma)'\widehat\alpha +\eta_t\widehat\varepsilon_t$, 
\begin{eqnarray} 
 {\alpha}^*  \left( \gamma \right) -\widehat \alpha  &=&
 \mathcal A^*(\gamma) ^{-1} \frac{1}{T}\sum_{t=1}^{T}  Z^* _{t}(\gamma)\ y_t^*-\widehat\alpha \cr
 &=& \mathcal A^*(\gamma)^{-1}(a_1+a_2),\quad \text{ where }\cr
a_1&=&  \frac{1}{T}\sum_{t=1}^{T}  Z^*_{t}\left( \gamma \right)\eta_t\widehat
\varepsilon _{t}\cr
a_2&=& \frac{1}{T}\sum_{t=1}^{T}  Z^*_{t}\left( \gamma \right)
x_{t}^{\prime }\widehat \delta  \left( 1\left(\widetilde f_t'\widehat \gamma>0 \right) -1 \left(f^{*'}_t\gamma>0\right)\right)  
 \end{eqnarray}
By the same argument in the  proof of lemma \ref{l7.7}, 
$$
a_1=\frac{1}{T}\sum_{t=1}^{T}\widehat Z_{t}^*\left( \gamma  \right)\eta_t\widehat
\varepsilon _{t} +o_{P^*}(T^{-1/2})
=\frac{1}{T}\sum_{t=1}^{T}\widehat Z_{t}^*\left( \widehat \gamma  \right)\eta_t\widehat
\varepsilon _{t}+o_{P^*}(T^{-1/2}).
$$

On the other hand,  by lemma \ref{Lem:rate_gm},  uniformly in $\gamma$, since $T=O(N)$,
\begin{eqnarray}
a_2  &=& \frac{1}{T}\sum_{t=1}^{T} \widehat  Z _{t}\left( \gamma \right)
x_{t}^{\prime }\widehat \delta  \left( 1\left(\widehat f_t'\widehat \gamma>0 \right) -1 \left(\widehat f'_t\gamma>0\right)\right)+o_{P^*}(T^{-1/2}) +o_{P^*}(T^{-\varphi} N^{-1/2})  \cr
&=&
\mathbb E\widehat Z_{t}\left( \gamma \right)
x_{t}^{\prime }\left( 1\left(\widehat f_t'\widehat \gamma>0 \right) -1\left(\widehat f_t'\gamma >0 \right)\right)\widehat \delta 
+  \eta T^{-2\varphi}|\gamma-\widehat \gamma |_2+o_{P^*}(T^{-1/2})\cr
&\leq& O(T^{-\varphi}  )  |\gamma-\widehat \gamma|_2+o_{P^*}(T^{-1/2}) =o_{P^*}(T^{-1/2}). 
 \end{eqnarray}
 This gives rise to (\ref{e.4a}).
 
 Now $\alpha_{g,h}^* =\alpha^*(\gamma_g^*)$ $\alpha_{g}^*=\alpha^*(\gamma_{g,h}^*)$
 while  $ |\gamma_{g,h}^*-\widehat\gamma|_2=O_{P^*}(r_{NT}^{-1})$, $|\gamma_g^*-\widehat\gamma|_2=O_{P^*}( r_{NT}^{-1})$ by step 0. Hence both $\alpha_{g,h}^* $ and $\alpha_{g}^* $ satisfy the expansion (\ref{e.4a}), whose leading term does not depend on the general choice of $\gamma$.

\textbf{step 2:  } if  $	|\alpha -\alpha^*_g|_2=o_{P^*}(T^{-1/2}) $, and $	|\alpha -\alpha_{g,h}^*|_2=o_{P^*}(T^{-1/2}) $, then\\ $\mathbb S_T^*(  \alpha,  \gamma^*(\alpha) )\leq    \min_{\alpha, \gamma}\mathbb S_T^*(\alpha,\gamma) + o_{P^*}(l_{NT}^{-1})$, and $\mathbb S_T^*(  \alpha,  \gamma_h^*(\alpha) )\leq     \min_{\alpha, h(\gamma)=h(\widehat\gamma)}\mathbb S_T^*(\alpha,\gamma)+ o_{P^*}(l_{NT}^{-1})$.

Note that 
$
\mathbb S_T^*(  \alpha,  \gamma^*(\alpha) )\leq    \min_{\alpha, \gamma}\mathbb S_T^*(\alpha,\gamma) +\mathbb S_T^*(  \alpha,  \gamma^*_g ) -\mathbb S_T^*(  \alpha^*_g,  \gamma^*_g ).
$ So we need to bound $\mathbb S_T^*(  \alpha,  \gamma^*_g ) -\mathbb S_T^*(  \alpha^*_g,  \gamma^*_g ).$ From (\ref{ef.5}),
	\begin{eqnarray*} 
&&	\mathbb S_T^*(  \alpha,  \gamma^*_g ) -\mathbb S_T^*(  \alpha^*_g,  \gamma^*_g )\cr
&=&
(\mathbb S_T^*(\alpha,\gamma_g^*)-\mathbb S^*_T(\alpha,\widehat\gamma))
-(\mathbb S_T^*(\alpha_g^*,\gamma_g^*)-\mathbb S^*_T(\alpha_g^*,\widehat\gamma))
+(\mathbb S_T^*(\alpha,\widehat\gamma)-\mathbb S^*_T(\alpha_g^*,\widehat\gamma))
 \cr
&=& \mathbb S_T^*(\alpha,\widehat\gamma)-\mathbb S^*_T(\alpha_g^*,\widehat\gamma)
 +o_{P^*}(l_{NT}^{-1}).
\end{eqnarray*}

Now given that $	|\alpha -\alpha^*_g|_2=o_{P^*}(T^{-1/2}) $,	\begin{eqnarray*} 
&&	\mathbb S_T^*(\alpha,\widehat\gamma)-\mathbb S^*_T(\alpha_g^*,\widehat\gamma)=\frac{1}{T}\sum_{t=1}^{T}\left({Z}^*_{t}\left( \widehat\gamma \right) ^{\prime }\left(\alpha - \alpha_g^* \right) \right) ^{2}+(\widehat\alpha-\alpha_g^*)'\frac{2}{T}\sum_{t=1}^{T}\widetilde Z_{t}(\widehat\gamma)Z_t^*(\widehat\gamma)(\alpha_g^*-\alpha)\cr
&&+\frac{2}{T}\sum_{t=1}^{T}\eta_t\widehat\varepsilon _{t} {Z}^*_{t}\left( \widehat \gamma \right) ^{\prime}\left( \alpha_g^* -  \alpha \right)+\frac{2}{T}\sum_{t=1}^{T}x_{t}^{\prime }\delta_g^*  \left( 1\{f_t^{*'}\widehat\gamma>0\} -1\{\widetilde f_t'\widehat\gamma>0\}\right) {Z}^*_{t}\left(\widehat \gamma  \right) ^{\prime}\left( \alpha -\alpha_g^* \right)\cr
&=&o_{P^*}(l_{NT}^{-1}),
\end{eqnarray*}	
The same result applies when $\alpha_g^*$ is replaced with $\alpha_{g,h}^*.$

\textbf{step 3:} if $\alpha=\widehat \alpha  +O_{P^*}\left( T^{-1/2 }  \right) ,$ then    $	 |\gamma^*(\alpha)- \widehat\gamma|_2=O_{P^*}( r_{NT}^{-1})$ and $	 |\gamma_h^*(\alpha)- \widehat\gamma|_2=O_{P^*}( r_{NT}^{-1})$.

 Note that for any $\gamma, \alpha$, 
   \begin{eqnarray}
 {\mathbb{S}}^*_{T}\left( {\alpha},\gamma \right) -{\mathbb{S}}^*
_{T}\left( {\alpha}, \widehat \gamma \right)  
&=&\widetilde{  R}^*_{1}(\alpha,\gamma)-\widetilde{  R}^*_{1}(\alpha,\widehat \gamma )
+\widetilde{  R}^*_{2}( \gamma)+\widetilde{  R}_{3}^*(\alpha,\gamma)
- \widetilde{\mathbb C}^*_1(\alpha,\gamma) \cr
&&+\widetilde{\mathbb C}^*_3(\alpha,\gamma). \label{e.3}
\end{eqnarray}%

We divide the proof of step 3 into the following sub-steps.

\subsection*{ step 3-i: preliminary rate of convergence }

We now study each term on the right of (\ref{e.3}) using similar arguments of lemma \ref{l:relf} and \ref{l:rhat}.
Uniformly in $\gamma$,

  \begin{eqnarray*}
|\widetilde{  R}^*_{1}(\alpha,\gamma)-   \widetilde{  R}^*_{1}(\alpha,\widehat \gamma ) | 
&=& |\alpha -\widehat \alpha  |_{2}^2 O_{P^*}(\Delta _{f}+T^{-6}  +T^{-1+\varphi}) 
	+ C |\gamma-\widehat\gamma|_2 |\alpha-\widehat\alpha |_2^2 + T^{-2\varphi }O_{P^*}(\Delta _{f}+T^{-6}). 
	\cr
	\widetilde{  R}_{3}(\alpha,\gamma) &\leq &   [|\alpha -\widehat \alpha
	 |_{2}^2+T^{-2\varphi }]O_{P^*}(\Delta _{f}+T^{-6})  + \left( O_{P^*}\left( T^{-1}\right)
	+CT^{-\varphi }\left\vert \gamma-\widehat\gamma\right\vert _{2}\right) \left\vert
	\alpha -\widehat \alpha  \right\vert _{2}.\cr
	 \widetilde{\mathbb C}^*_1(\alpha,\gamma) &=&\widehat{\mathbb C}^*_{1}\left(  \gamma 
\right)+ (T^{-\varphi }+|\alpha -\widehat\alpha|_{2}) O_{P^*}(\Delta _{f}+T^{-6}) \cr
	&&+\left(  O_{P^*}\left( T^{-1}\right) +\eta T^{-2\varphi }\left\vert \gamma-\widehat\gamma\right\vert \right) T^{\varphi }\left\vert \alpha-\widehat\alpha\right\vert _{2}.\cr
	 \widetilde{  R}_{2}^*( \gamma)  +\widetilde{\mathbb C}_3^*(\alpha,\gamma)
	&=&  \mathds G_{H_T,\widehat\Sigma}(\phi,\widehat\phi)  +(T^{-\varphi }+|\alpha -\widehat\alpha|_{2}) O_{P^*}(\Delta _{f}+T^{-6})+T^{-\varphi }\left\vert \alpha-\widehat\alpha\right\vert _{2} O_{P^*}\left( N^{-1/2}\right)
\end{eqnarray*}
where $ \mathds G_{H_T,\widehat\Sigma}(\phi,\widehat\phi) $ is as defined in (\ref{ef.5}).
 
 Putting together, $\alpha=\widehat \alpha  +O_{P^*}\left( T^{-1/2 }  \right) ,$ for $\gamma=\gamma^*(\alpha)$,  $
 {\mathbb{S}}^*_{T}\left( {\alpha},\gamma \right) - {\mathbb{S}}^*
_{T}\left( {\alpha},\widehat\gamma  \right) \leq 0
$ implies 
\begin{equation}\label{e.4}
\mathds G_{H_T,\widehat\Sigma}(\phi,\widehat\phi) \leq \widehat{\mathbb C}^*_{1}\left(  \gamma 
\right)+  o_{P^*}(T^{-2\varphi})N^{-1/2} +o_{P^*}(T^{-1})
    +T^{-\varphi }  O_{P^*}(\Delta _{f}+T^{-6})
    +o_{P^*}(T^{-2\varphi})\left\vert \gamma-\widehat\gamma\right\vert _{2} . 
\end{equation}
By a similar argument as Lemmas \ref{l.1}, \ref{l.2}, 
 $$   \widehat{\mathbb C}^*_{1}\left(  \gamma 
\right)\leq  b_{NT}, \quad  \mathds G_{H_T,\widehat\Sigma}(\phi,\widehat\phi)
 \geq CT^{-2\varphi} |\phi-\widehat\phi|_2  -  \frac{C}{\sqrt{N}T^{2\varphi}}, 
$$
where for an arbitrarily small $\eta>0$, 
$
b_{NT}=O_{P^*}(T^{-1}) +\eta T^{-2\varphi }\left\vert \gamma-\widehat\gamma
\right\vert _{2}.
$
Then 
 \begin{eqnarray*}
&& CT^{-2\varphi} |\phi-\widehat\phi|_2   
\cr
    &\leq & o_{P^*}(T^{-2\varphi})N^{-1/2} +O_{P^*}(T^{-1})
    +T^{-\varphi } O_{P^*}(\Delta _{f}+T^{-6})
    + \eta T^{-2\varphi }\left\vert \phi -\phi
_{0}\right\vert _{2} + \frac{C}{\sqrt{N}T^{2\varphi}}.
	 \end{eqnarray*}%
	 Since $\eta>0$ is arbitrarily small, we have 
 \begin{eqnarray*}
&&  |\gamma^*(\alpha)- \widehat\gamma|_2=O_{P^*}(  |\phi-\widehat\phi|_2   )
\leq O_{P^*}(N^{-1/2} +T^{-(1-2\varphi)}) . 
	 \end{eqnarray*}%

\subsection*{ step 3-ii: sharp rate of convergence }

The preliminary rate is sharp when $\sqrt{N}\gg T^{1-2\varphi}$. 
 Now suppose $\sqrt{N}=O(T^{1-2\varphi})$. 
By proofs similar to lemmas \ref{l.1},   \ref{l:generic},  for $\phi=H_T\widehat{\gamma}(\alpha)$,
	  $$  \widehat{\mathbb C}^*_{1}\left(  \gamma 
\right)\leq  a_{NT},
\quad    \mathds G_{H_T,\widehat\Sigma}(\phi,\widehat\phi)
 \geq CT^{-2\varphi} \sqrt{N} |\phi -\widehat\phi|_2^2 -  O_{P}(\frac{1}{T^{2\varphi}N^{5/6}}),
$$ where for an arbitrarily small $\eta>0$, 
$
a_{NT}=T^{-2\varphi}
O_P\left( \frac{\sqrt{N}}{\left( NT^{1-2\varphi }\right) ^{2/3}}\right)
+T^{-2\varphi} \eta \left\vert \phi -\phi
_{0}\right\vert _{2}^{2}\sqrt{N}.
$
Substituting these bounds to (\ref{e.4}) yields 
$$
  |\gamma^*(\alpha)- \widehat\gamma|_2=O_{P^*}(  |\phi-\widehat\phi|_2   )
\leq 
O_{P^*}\left( \frac{1}{\left( NT^{1-2\varphi }\right) ^{1/3}}\right). 
$$
Combining with the   rates proved in claim 3,  we obtain the desired result.

	\end{proof}

	\begin{lem}\label{bootstrapempirical}

	(i)  In the known factor case, $ \mathbb{K}_{3T}^*\left( g\right) \Rightarrow^{*} 2W(g)$, where
 $$  \mathbb{K}_{3T}^*\left( g\right)
 := -2\sum_{t=1}^{T}\eta_t\widehat\varepsilon
	 	 _{t}x_{t}^{\prime }\widehat\delta \left( 1_{t}\left(\widehat \gamma +g\cdot
	 	 r_{T}^{-1}\right) -1_{t}(\widehat\gamma)\right).$$
(ii) In the  estimated factor case,      $   \sqrt{r_{NT}T^{1+2\varphi}}\mathbb{\widehat{C}}^*_{1}\left(  \widehat\gamma
 +r_{NT}^{-1}g\right) \Rightarrow^{*} 2W(g)$, where
 $
  \mathbb{\widehat{C}}^*_{1}\left(   \gamma \right) =
			\frac{2}{T}\sum_{t=1}^{T}\eta_t\widehat\varepsilon _{t}x_{t}^{\prime }\widehat\delta \left( 1\{ \widehat f_t^{*'}\gamma>0\} -1\{  \widehat f_t^{*'}\widehat\gamma>0\} \right)
		$, $\widehat f_t^*=\widehat f_t+ N^{-1/2} \mathcal Z_t^*$,
		and $\mathcal Z_t^*$ is iid $\mathcal N(0,\widehat\Sigma_h)$.
	\end{lem}

 \begin{proof}
 	(i) We first show the stochastic equicontinuity of $ \mathbb{K}_{3T}^*\left( g\right)  $, for which it is sufficient to show that of $ \sum_{t=1}^{T}\eta_t \varepsilon
 	_{t}x_{t}^{\prime }\delta_0 \left( 1_{t}\left(\gamma_T +g\cdot
 	r_{T}^{-1}\right) -1_{t}(\gamma_T)\right)  $ for any $ \gamma_{T} \to \gamma_{0} $
 	since $ \widehat\delta - \delta_0 = O_P(T^{-1/2}) $, $ \widehat\gamma $ is consistent, and
 	$ \widehat\varepsilon_{t} = \varepsilon_{t} + remainder_t $, where the remainder terms are treated as before.
 	However, we can apply the maximal inequality in Lemma \ref{A:Lem} here since $ \eta_{t} $ is a centered iid sequence independent of the other variables.
 	Next, to derive the finite dimensional convergence we can apply the conditional CLT e.g. Hall and Heyde (1980) for the MDS.
 	The conditions are checked similarly as in Section \ref{sec:asymptotic-distn:known-f}.

 	(ii) The argument for the stochastic equicontinuity is similar to the case (i). Also, the derivation in Section \ref{sec:EP} and the proof of Lemma \ref{Lem-DCTinP} in particular reveals that the finite dimensional limits are not affected by the change of $ \widehat f_t $ by  $\widehat f_t^*=\widehat f_t+ N^{-1/2} \mathcal Z_t^*$.

 \end{proof}

\subsection{Proof of Theorem \ref{Thm:Linearity Test}}

\begin{proof}[Proof of Theorem \ref{Thm:Linearity Test}]%
We begin with the known factor case.
For each $\gamma $, our $Q_{T}\left(
\gamma \right) $ corresponds to a modified version of the  Wald statistic $T_{n}\left( \gamma
\right) $  used in \citet{Hansen:1996}. Specifically, let $\widehat\alpha(\gamma)=\arg\min_\alpha\mathbb S_T(\alpha,\gamma)$ and $R=(0_{d_x}, I_{d_x})$.
Then it can be proved that
$$
\min_{\alpha :\delta =0}\mathbb{S}_{T}\left( \alpha ,\gamma
\right) -\min_{\alpha ,\gamma }\mathbb{S}_{T}\left( \alpha ,\gamma \right)
=\widehat\alpha(\gamma)'R'[R(\sum_tZ_t(\gamma)Z_t(\gamma)')^{-1}R']^{-1}R\widehat\alpha(\gamma).
$$
We then replace the term $\widehat{%
		V}_{n}\left( \gamma \right) $ in \citet{Hansen:1996} with
\begin{equation}\label{redefine-V_n}
\widehat{V}_{n}\left( \gamma \right) =\frac{1}{T}\sum_{t=1}^{T}x_{t}x_{t}^{
\prime }1\left\{ f_{t}^{\prime }\gamma >0\right\} \mathbb{{S}}
_{T}\left( \widehat{\alpha},\widehat{\gamma}\right) .
\end{equation}
We now verify regularity conditions imposed by \citet{Hansen:1996}. His Assumption 1  concerns the mixing and
moment conditions that are satisfied by our Assumption \ref{A-mixing} (with $v=r=2$ in the notation used in \citet{Hansen:1996}).
His Assumption 2 is a sufficient condition to ensure the tightness of the
empirical process $T^{-1/2}\sum_{t=1}^{T}x_{t}1\left\{ f_{t}^{\prime }\gamma
>0\right\} \varepsilon _{t},$ which is guaranteed by our maximal inequality
Lemma \ref{Lem:modul1}. Finally, his Assumption 3 follows from the ULLN.
Then, the theorem is proved with the replaced $\widehat{%
		V}_{n}\left( \gamma \right) $ in \eqref{redefine-V_n}.

Turning to the estimated factor case, we need to establish the asymptotic equivalence between
the known and unknown factors. For this purpose,  it suffices to show that
\begin{align}
\sup_{\gamma }\left\vert \frac{1}{T}\sum_{t=1}^{T}x_{t}x_{t}^{\prime }\left(
1\left\{ \widetilde{f}_{t}^{\prime }\gamma >0\right\} -1\left\{ f_{t}^{\prime
}\gamma >0\right\} \right) \right\vert &= o_P\left( 1\right), \label{approx0001} \\
\sup_{\gamma }\left\vert \frac{1}{T}\sum_{t=1}^{T}x_{t}x_{t}^{\prime }\left(
1\left\{ \widetilde{f}_{t}^{\prime }\gamma >0\right\} -1\left\{ f_{t}^{\prime
}\gamma >0\right\} \right) \varepsilon _{t}^{2}\right\vert &= o_P\left(
1\right), \label{approx0002} \\
\sup_{\gamma }\left\vert \frac{1}{\sqrt{T}}\sum_{t=1}^{T}x_{t}\left(
1\left\{ \widetilde{f}_{t}^{\prime }\gamma >0\right\} -1\left\{ f_{t}^{\prime
}\gamma >0\right\} \right) \varepsilon _{t}\right\vert &= o_P\left(
1\right). \label{approx0003}
\end{align}
Recall that $\widehat{f}_{t}$ is defined as $\widehat{f}_{t}=H_T' (g_{t}+h_{t}/\sqrt{N})$.
The last condition \eqref{approx0003}  follows directly if we show that
\begin{equation}\label{approx111}
\sup_{\gamma }\left\vert \frac{1}{\sqrt{T}}\sum_{t=1}^{T}x_{t}\left(
1\left\{ \widetilde{f}_{t}^{\prime }\gamma >0\right\} -1\left\{ \widehat{f}
_{t}^{\prime }\gamma >0\right\} \right) \varepsilon _{t}\right\vert
=o_P\left( 1\right)
\end{equation}
and
\begin{equation}\label{approx222}
\sup_{\gamma }\left\vert \frac{1}{\sqrt{T}}\sum_{t=1}^{T}x_{t}\left(
1\left\{ \widehat{f}_{t}^{\prime }\gamma >0\right\} -1\left\{ f_{t}^{\prime
}\gamma >0\right\} \right) \varepsilon _{t}\right\vert = o_P\left(
1\right).
\end{equation}
By Lemma \ref{l:relf}, \eqref{approx111} follows. To show \eqref{approx222}, note that
in view of the maximal inequality in Lemma \ref{Lem:modul1}, Theorem
16.1 of \citet{billingsley1968} and the extended CMT in Lemma \ref{CMT-extension},
the empirical process
\begin{equation*}
\frac{1}{\sqrt{T}}\sum_{t=1}^{T}x_{t}1\left\{ \widehat{f}_{t}^{\prime }\gamma
>0\right\} \varepsilon _{t}
\end{equation*}
is stochastically equicontinuous.
This implies \eqref{approx222}.
The other two conditions \eqref{approx0001} and \eqref{approx0002} can be shown similarly and thus omitted.
In case of the estimated factors, $ f_t = H_T 'g_t $ and $ \gamma = H_T^{-1}\phi $ and the supremum is understood as taken over $ \phi $ after cancelling out $ H_T $ in $ f_t ' \gamma $.
Finally, the CMT yields the desired result.
\end{proof}%

\section{Technical Lemmas}\label{sec:tech:lem:appendix}

This section proves technical lemmas, which are repeatedly used to prove main theorems. Their proofs are given in the subsequent subsection. They are proven under the following assumption.
\begin{assum} \label{A:Lem}
Assume that $\left\{ z_{t},q_{t}\right\} _{t=1}^{T}$ be a
sequence of strictly stationary, ergodic, and $\rho $-mixing array with $%
\sum_{m=1}^{\infty }\rho _{m}^{1/2}<\infty $, $ \mathbb E\left\vert z_{t}\right\vert_2
^{4}<\infty $, and, for all $\gamma $ in a neighborhood of $ \gamma_0 $, $ \mathbb E%
\left( \left\vert z_{t}\right\vert ^{4}|q_{t}=\gamma \right) <C<\infty $
and $ q_t'\gamma $ has a density that is continuous and bounded by some $ C<\infty $.
\end{assum}
Similar to the previous notation, we define $1_{t}\left( \gamma \right) \equiv 1\left\{ q_{t}^{\prime }\gamma >0\right\} $
while  $1_{t}\left(
\gamma ,\bar{\gamma}\right) \equiv 1\left\{ q_{t}^{\prime }\gamma \leq
0<q_{t}^{\prime }\bar{\gamma}\right\} $, which should not cause much confusion. Furthermore, we let the last element of $ q_t $ equal to $ -1 $.

\begin{lem}\label{Lem:modul1}
Let Assumption \ref{A:Lem} hold.
	Then, there exists $T_{0}<\infty $ such that for any $\vec{\gamma}\ $in a
	neighbourhood of $\gamma _{0}\ $,  $ K>0 $ and for all $T>T_{0}$ and $\epsilon \geq
	T^{-1}$,
	\begin{equation*}
	\mathbb{P }\left\{ \sup_{\left\vert \gamma -\vec{\gamma}\right\vert_2
		<\epsilon }
	\left\vert \frac{1}{T^{1/2}}\sum_{t=1}^{T}\left( z_{t}1_t \left(\vec{\gamma},\gamma \right)
	-\mathbb E z_{t} 1_t\left(\vec{\gamma},\gamma\right) \right) \right\vert >K\right\} \leq \frac{C}{K^{4}}\epsilon ^{2}.
	\end{equation*}%

\end{lem}

An obvious implication of this lemma is that when $\epsilon =a_{T}^{-1}$ for
some sequence $a_{T}=O\left( T\right) $ the process in the display is $%
O_P\left( a_{T}^{-1/2}\right) $.
It also leads to the following uniform bounds for empirical processes of mixing arrays.

\begin{lem}\label{Lem:rate_gm}
	Let Assumption \ref{A:Lem} hold.
	For any $\eta >0$ and some $C,c>0,$
	\begin{align*}
	\sup_{cT^{-1+2\varphi }\leq \left\vert \gamma -\gamma_{0}\right\vert _{2}<C}
	&\bigg[ \left\vert \frac{1}{T^{1+\varphi }}\sum_{t=1}^{T}\left( z_{t}\left(
	1_{t}\left( \gamma \right) -1_{t}\left( \gamma _{0}\right) \right)
	- \mathbb Ez_{t}\left( 1_{t}\left( \gamma \right) -1_{t}\left( \gamma _{0}\right)
	\right) \right) \right\vert \\
	& -\eta T^{-2\varphi }\left\vert \gamma -\gamma
	_{0}\right\vert _{2}\bigg]
	\leq O_P\left( \frac{1}{T}\right).
	\end{align*}
\end{lem}

\begin{lem}\label{Lem:rate_gm_1}
Let Assumption \ref{A:Lem} hold.
	For any $\eta >0$ and some $C, c>0,$
\begin{eqnarray*}
&&\sup_{cT^{-1+2\varphi }\leq \left\vert \gamma -\gamma _{0}\right\vert _{2}<C}
\left[ \left\vert \frac{1}{\sqrt{N}T^{1-\varphi }}\sum_{t=1}^{T}\left(
z_{t}\left( 1_{t}\left( \gamma \right) -1_{t}\left( \gamma _{0}\right)
\right) -\mathbb{E} z_{t}\left( 1_{t}\left( \gamma \right) -1_{t}\left( \gamma
_{0}\right) \right) \right) \right\vert -\eta \left\vert \gamma -\gamma
_{0}\right\vert _{2}^{2}\right] \\
&\leq &O_P\left( \frac{1}{\left( NT^{1-2\varphi }\right) ^{2/3}}\right) .
\end{eqnarray*}
\end{lem}

We derive an extended continuous mapping theorem (CMT) in Lemma \ref{CMT-extension}, in the sense that we consider a transformation by a continuous stochastic process.    This lemma extends  Theorem 1.11.1 of \cite{VW} to allowing stochastic drifting functions $\mathbb G_n$ (while \cite{VW} requires $\mathbb G_n$ be deterministic).

\begin{lem}\label{CMT-extension}
Suppose that as $n \rightarrow\infty$,
\begin{equation*}
\mathbb{G}_{n}\left( x\right) \Rightarrow \mathbb{G}\left( x\right)
\end{equation*}%
over any compact set in $\mathbb{R}^{m},$ where $\mathbb{G}\left( \cdot
\right) $ is a Gaussian process with continuous sample paths. Let $f_{n}$ be
a sequence of random functions from $\mathbb{R}^{k}$ onto $\mathbb{R}^{m}$
and assume that%
\begin{equation*}
f_{n}\left( z\right) \overset{P}{\longrightarrow }   f\left( z\right) ,
\end{equation*}%
uniformly, where $f$ is a deterministic function, and that for any $\eta >0$
there exists $C_{\eta }<\infty $ such that%
\begin{equation*}
\mathbb P\left\{ \left\vert f_{n}\left( z\right) -f_{n}\left( z^{\prime }\right)
\right\vert_2  >C_{\eta }\left\vert z-z^{\prime }\right\vert_2 \; \text{ for all }%
z,z^{\prime }\right\} <\eta ,
\end{equation*}%
for all $n$. Then,
\begin{equation*}
\mathbb{G}_{n}\left( f_{n}\left( z\right) \right) \Rightarrow \mathbb{G}%
\left( f\left( z\right) \right)
\end{equation*}%
over any compact set.
\end{lem}

\subsection{Proofs of Lemmas}

\begin{proof}[Proof of Lemma \ref{Lem:modul1}]
	In this proof, $c,C$ and so on denote generic constants. Let the dimension
	of $q_{t}$ be denoted by $d_f = d+1 $ and partition $\gamma =\left( \psi ^{\prime},c\right) ^{\prime }$ and $ q_t =\left( q_{1t} ^{\prime },-1\right) ^{\prime }$ Also let%
	\begin{equation*}
	J_{T}\left( \gamma \right) =\frac{1}{T^{1/2}}\sum_{t=1}^{T}\left(
	z_{t}1_{t}\left( \vec{\gamma},\gamma \right) -  \mathbb  Ez_{t}1_{t}\left( \vec{\gamma}%
	,\gamma \right) \right) .
	\end{equation*}%
	First, note that Lemma 3.6 of \citet{Peligrad:1982} implies that
	there is a universal constant $C$, depending only on the
	$\rho_m$'s, such that
	for any $\gamma _{1}$ and $\gamma _{2},$%
	\begin{align}\label{eq:Peligrad3.6}
	\begin{split}
	& \mathbb E\left\vert \frac{1}{T^{1/2}}\sum_{t=1}^{T}\left( z_{t}1_{t}\left( \gamma
	_{1},\gamma _{2}\right) - \mathbb Ez_{t}1_{t}\left( \gamma _{1},\gamma _{2}\right)
	\right) \right\vert ^{4} \\
	&\leq C\left( T^{-1} \mathbb E\left\vert z_{t}\right\vert
	^{4}~1_{t}\left( \gamma _{1},\gamma _{2}\right) +\left(  \mathbb E\left\vert
	z_{t}\right\vert ^{2}~1_{t}\left( \gamma _{1},\gamma _{2}\right) \right)
	^{2}\right) \text{.}
	\end{split}
	\end{align}
	Consider $\gamma _{1}=\left( \psi ^{\prime },c_{1}\right) ^{\prime }$ and $%
	\gamma _{2}=\left( \psi ^{\prime },c_{2}\right) ^{\prime },$ which are
	identical other than the last elements. Then,
	\begin{equation*}
	1_{t}\left( \gamma _{1},\gamma _{2}\right) =1\left\{ c_{2}<q_{1t}^{\prime
	}\psi \leq c_{1}\right\}
	\end{equation*}%
	and thus there is a universal constant $C$ such that
	\begin{eqnarray*}
		 \mathbb E\left\vert z_{t}\right\vert ^{k}~1_{t}\left( \gamma _{1},\gamma _{2}\right)
		&=& \mathbb E\left[  \mathbb E\left( \left\vert z_{t}\right\vert ^{k} \Big| q_{t}\right) 1_{t}\left(
		\gamma _{1},\gamma _{2}\right) \right]  \\
		&\leq &C \mathbb E1_{t}\left( \gamma _{1},\gamma _{2}\right) \leq C^{\prime
		}\left\vert c_{1}-c_{2}\right\vert
	\end{eqnarray*}%
	for $k=2,4$, as the densities of $q_{t}'\gamma $ are bounded uniformly. Thus, for
	any $c_{1},c_{2}$ such that $\left\vert c_{1}-c_{2}\right\vert \geq T^{-1},$%
	\begin{equation}
	\sup_{\psi }\mathbb{E}\left\vert \frac{1}{T^{1/2}}\sum_{t=1}^{T}\left(
	z_{t}1_{t}\left( \gamma _{1},\gamma _{2}\right) - \mathbb Ez_{t}1_{t}\left( \gamma
	_{1},\gamma _{2}\right) \right) \right\vert ^{4}\leq C\left\vert
	c_{1}-c_{2}\right\vert ^{2}.  \label{eq:rho-mom4}
	\end{equation}
	Here, recall that  $\psi$ is the common element between $\gamma_1$ and $\gamma_2$.

	Next, by \citet{bickel1971}, their equation (1), that
	\begin{equation*}
	\sup_{\gamma }\left\vert J_{T}\left( \gamma \right) \right\vert \leq d\cdot
	M^{\prime \prime }+\left\vert J_{T}\left( \widetilde{\gamma}\right) \right\vert ,
	\end{equation*}%
	where $\widetilde{\gamma}$ is the elementwise increament of $\vec{\gamma} $ by $\epsilon $ and the supremum is taken
	over a hyper cube $\left\{ \gamma :0\leq \gamma _{j}-\vec{\gamma}_{j}\leq
	\epsilon ,j=1,...,d\right\} $ and an upper bound for $M^{\prime \prime }$ is
	given by their Theorem 1. The precise definition of $M^{\prime \prime }$ is
	referred to Bickel and Wichura. It is sufficient to show that each of $%
	M^{\prime \prime }$ and $\left\vert J_{T}\left( \widetilde{\gamma}\right)
	\right\vert $ satisfies the conclusion of the lemma since $\left\vert
	a\right\vert +\left\vert b\right\vert >2c$ implies that $\left\vert
	a\right\vert >c$ or $\left\vert b\right\vert >c$.

	To apply their Theorem 1, we need to consider the increment of the process $%
	J_{T}$ around a block\footnote{%
		It is sufficient to consider blocks with side length at least $n^{-1}$ for
		the same reason as the remarks in the last paragraph in p. 1665.} $B=(\gamma
	_{1},\gamma _{2}]=(\gamma _{12,}\gamma _{22}]\times \cdots \times
	(c_{1},c_{2}]$ with each side of length greater than equal to $T^{-1}$, that
	is, consider%
	\begin{eqnarray*}
		J_{T}\left( B\right)  &=&\sum_{k_{1}=0,1}\cdots \sum_{k_{d+1}=0,1}\left(
		-1\right) ^{d-k_{1}-\cdots -k_{d+1}}J_{T}\left( \gamma _{11}+k_{1}\left(
		\gamma _{21}-\gamma _{11}\right) ,...,c_{1}+k_{d+1}\left( c_{2}-c_{1}\right)
		\right)  \\
		&=&\sum_{k_{1}=0,1}\cdots \sum_{k_{d}=0,1}\left( -1\right)
		^{d-k_{1}-\cdots -k_{d}} \\
		&&\times \left( J_{T}\left( \gamma _{11}+k_{1}\left( \gamma _{21}-\gamma
		_{11}\right) ,...,c_{1}\right) -J_{T}\left(\gamma _{11}+k_{1}\left(
		\gamma _{21}-\gamma _{11}\right) ,...,c_{2}\right) \right) .
	\end{eqnarray*}%
	Then, it follows from the $c_{r}$-inequality and (\ref{eq:rho-mom4}) that
	for some $C,C^{\prime },C^{\prime \prime }<\infty $
	\begin{eqnarray*}
		&& \mathbb E\left\vert J_{T}\left( B\right) \right\vert ^{4} \\
		&\leq &C\sum_{k1=0,1}\cdots \sum_{k_{d}=0,1}\mathbb{E}\left\vert J_{T}\left(
		\gamma _{11}+k_{1}\left( \gamma _{21}-\gamma _{11}\right)
		,...,c_{1}\right) -J_{T}\left( \gamma _{11}+k_{1}\left( \gamma
		_{21}-\gamma _{11}\right) ,...,c_{2}\right) \right\vert ^{4} \\
		&\leq &C^{\prime }\sup_{\psi }\mathbb{E}\left\vert \frac{1}{T^{1/2}}%
		\sum_{t=1}^{T}\left( z_{t}1_{t}\left( \gamma _{1},\gamma _{2}\right)
		-\mathbb{E} z_{t}1_{t}\left( \gamma _{1},\gamma _{2}\right) \right) \right\vert ^{4},\
		\text{for }\gamma _{j}=\left( \psi ^{\prime },c_{j}\right) ,j=1,2 \\
		&\leq &C^{\prime \prime }\left\vert c_{1}-c_{2}\right\vert ^{2}.
	\end{eqnarray*}%
	Now, without loss of generality we can assume that $\mu \left( B\right) \geq
	C^{\prime \prime \prime }\left\vert c_{1}-c_{2}\right\vert ^{d},$ where $\mu
	$ denotes the Lebesque measure in $\mathbb{R}^{d}$, since we can derive the
	same bound by choosing the smallest side length of $B$ as $c_{2}-c_{1}$.
	This implies by the Cauchy-Schwarz inequality that their $\mathcal{C}\left(
	\beta ,\gamma \right) $ condition holds with $\beta =4$ and $\gamma =2/d$,
	and thus, by their Theorem 1, we conclude
	\begin{equation*}
	\mathbb P \left\{ M^{\prime \prime }>K\right\} \leq \frac{C}{K^{4}}\mu \left(
	T\right) ^{2/d}\leq \frac{C}{K^{4}}\epsilon ^{2},
	\end{equation*}%
	for some $C<\infty $.

	Furthermore, the Markov inequality, the moment bound in (\ref{eq:Peligrad3.6}), the boundedness of
	the density of $ q_{t}'\gamma $ imply that
	\begin{equation*}
	\mathbb P \left\{ \left\vert J_{T}\left( \widetilde{\gamma}\right) \right\vert
	>K\right\} \leq \frac{C}{K^{4}}\epsilon ^{2},
	\end{equation*}%
	for some $C<\infty $. This completes the proof.
\end{proof}

\begin{proof}[Proof of Lemma \ref{Lem:rate_gm}]
	Define $A_{T,j}=\{\theta :(j-1)T^{-1+2\varphi }\leq \left\vert \gamma
	-\gamma _{0}\right\vert _{2}<jT^{-1+2\varphi }\}$ and
	\begin{equation*}
	R_{T}^{2}=T\sup_{T^{-1+2\varphi} < \left\vert  \gamma -\gamma_{0}\right\vert \leq C}\left[
	\left\vert \mathbb{D}_{T}\left( \gamma \right) \right\vert -\eta
	T^{-2\varphi }\left\vert  \gamma -\gamma_{0}\right\vert _{2}\right] \text{,}
	\end{equation*}%
	where $\mathbb{D}_{T}\left( \gamma \right) =\frac{1}{T^{1+\varphi }}%
	\sum_{t=1}^{T}\left( z_{t}\left( 1_{t}\left( \gamma \right) -1_{t}\left(
	\gamma _{0}\right) \right) - \mathbb Ez_{t}\left( 1_{t}\left( \gamma \right)
	-1_{t}\left( \gamma _{0}\right) \right) \right) $. Then, for any $m>0$,
	\begin{align*}
	& \mathbb P \left\{ R_{T}>m\right\}  \\
	& =\mathbb P \left\{ T\left\vert \mathbb{D}_{T}\left( \gamma \right) \right\vert
	>\eta | \gamma -\gamma_{0}|T^{1-2\varphi }+m^{2}\ {{\text{for some }}}\gamma
	\right\}  \\
	& \leq \sum_{\ell =2}^{\infty }\mathbb P \left\{ T\left\vert \mathbb{D}_{T}\left(
	\gamma \right) \right\vert >\eta (\ell -1)+m^{2}\ {{\text{for some }}}\gamma
	\in A_{T\ell }\right\}  \\
	& \leq C^{\prime }\sum_{\ell =2}^{\infty }\frac{\ell ^{2}}{\left( \eta (\ell
		-1)+m^{2}\right) ^{4}},
	\end{align*}%
	\newline
	where the last equality is due to Lemma \ref{Lem:modul1} with $%
	K=T^{-1/2+\varphi }\left( \eta (\ell -1)+m^{2}\right) $ and $\epsilon
	=\ell T^{-1+2\varphi }$. The last term is finite for any $\eta >0\ $and can be
	made arbitrarily small by choosing sufficiently large $m$, which completes
	the proof.
\end{proof}

\begin{proof}[Proof of Lemma \ref{Lem:rate_gm_1}]
	Define $A_{T,j}=\{\gamma :(j-1)\leq \widetilde{n}^{2/3}\left\vert \gamma -\gamma
	_{0}\right\vert _{2}^{2}<j\}$ with $\widetilde{n}=NT^{1-2\varphi }$ and
	\begin{equation*}
	R_{T}^{2}=\widetilde{n}^{2/3}\sup_{T^{-1+2\varphi} < \left\vert \gamma -\gamma _{0}\right\vert \leq C}
	\left[ \left\vert \mathbb{D}_{T}\left( \gamma \right) \right\vert -\eta
	\left\vert \gamma -\gamma _{0}\right\vert _{2}^{2}\right] \text{,}
	\end{equation*}%
	where $\mathbb{D}_{T}\left( \gamma \right) =\frac{1}{\sqrt{N}T^{1-\varphi }}%
	\sum_{t=1}^{T}\left( z_{t}\left( 1_{t}\left( \gamma \right) -1_{t}\left(
	\gamma _{0}\right) \right) -\mathbb{E}z_{t}\left( 1_{t}\left( \gamma \right)
	-1_{t}\left( \gamma _{0}\right) \right) \right) $. Then, for any $
	\varepsilon >0$, we can find $m$ such that%
	\begin{align*}
	& \mathbb P \left\{ R_{T}>m\right\}  =\mathbb P \left\{ \widetilde{n}^{2/3}\left\vert \mathbb{D}_{T}\left( \gamma \right)
	\right\vert >\eta \widetilde{n}^{2/3}|\gamma -\gamma _{0}|^{2}+m^{2}\ {{\text{for some }}%
	}\gamma \right\} \\
	& \leq \sum_{\ell =2}^{\infty }\mathbb P \left\{ \widetilde{n}^{2/3}\left\vert \mathbb{D}%
	_{T}\left( \gamma \right) \right\vert >\eta (\ell -1)+m^{2}\ {{\text{for
				some }}}\gamma \in A_{T\ell }\right\} \\
	& \leq C^{\prime }\sum_{\ell =2}^{\infty }\frac{%
		\widetilde{n}^{2/3}}{\left( \eta (\ell -1)+m^{2}\right) ^{4}}\frac{\ell }{\widetilde{n}^{2/3}}  \leq \varepsilon
	\end{align*}%
	where the first and second inequalities follow from the union bound and Lemma \ref{Lem:modul1} with $%
	K=\widetilde{n}^{-1/6}\left( \eta (\ell -1)+m^{2}\right) $ and $\epsilon =\sqrt{\frac{%
			\ell }{\widetilde{n}^{2/3}}}$, respectively, and the third by choosing sufficiently large $m$. This
	completes the proof.
\end{proof}

\begin{proof}[Proof of Lemma \ref{CMT-extension}]
	First, we show the stochastic equicontinuity of $\mathbb{G}_{n}\left(
	f_{n}\left( z\right) \right) .$ For any positive $\varepsilon \ $and $\eta ,$
	there exist $\delta >0$ and $N$ such that for all $n>N,$
	\begin{align*}
	&\mathbb{P}\left\{ \sup_{\left\vert z-z^{\prime }\right\vert_2 <\delta }\left\vert
	\mathbb{G}_{n}\left( f_{n}\left( z\right) \right) -\mathbb{G}_{n}\left(
	f_{n}\left( z^{\prime }\right) \right) \right\vert_2 >\eta \right\}  \\
	&\leq \mathbb{P}\bigg\{ \sup_{\left\vert z-z^{\prime }\right\vert_2 <\delta
	}\left\vert \mathbb{G}_{n}\left( f_{n}\left( z\right) \right) -\mathbb{G}%
	_{n}\left( f_{n}\left( z^{\prime }\right) \right) \right\vert_2 >\eta \ \text{%
		and }\left\vert f_{n}\left( z\right) -f_{n}\left( z^{\prime }\right)
	\right\vert_2 \leq C\left\vert z-z^{\prime }\right\vert_2 \ \\
	&\;\;\;\;\;\;\;\;\;\; \text{and }%
	\sup_{z}\left\vert f_{n}\left( z\right) \right\vert_2 \leq C \bigg\}  \\
	&\;\;\;\;+\mathbb{P}\left\{ \left\vert f_{n}\left( z\right) -f_{n}\left( z^{\prime }\right)
	\right\vert_2 >C\left\vert z-z^{\prime }\right\vert_2 \right\} +\mathbb{P}\left\{
	\sup_{z}\left\vert f_{n}\left( z\right) \right\vert_2 >C\right\}  \\
	&\leq \mathbb{P}\left\{ \sup_{\left\vert x-x^{\prime }\right\vert_2 <\delta
		/C}\left\vert \mathbb{G}_{n}\left( x\right) -\mathbb{G}_{n}\left( x^{\prime
	}\right) \right\vert_2 >\eta \right\} +\frac{\varepsilon }{2} \\
	&\leq \varepsilon ,
	\end{align*}%
	where the second inequality is due to the set inclusion and the given
	condition on $f_{n}$ with boundedness of $z$ and the last one follows from
	the stochastic equicontinuity of $\mathbb{G}_{n}$.

	Second, for the fidi note that
	\begin{equation*}
	\mathbb{G}_{n}\left( f_{n}\left( z\right) \right) -\mathbb{G}_{n}\left(
	f\left( z\right) \right) \overset{p}{\longrightarrow }0
	\end{equation*}%
	due to the stochastic equicontinuity of $\mathbb{G}_{n}$ as $f_{n}\left(
	z\right) \overset{p}{\longrightarrow }f\left( z\right) $. Therefore, for any
	finite collection $\left( z_{1},...,z_{p}\right) ^{\prime }$, $\left(
	\mathbb{G}_{n}\left( f_{n}\left( z_{1}\right) \right) ,...,\mathbb{G}%
	_{n}\left( f_{n}\left( z_{p}\right) \right) \right) ^{\prime }=\left(
	\mathbb{G}_{n}\left( f\left( z_{1}\right) \right) ,...,\mathbb{G}_{n}\left(
	f\left( z_{p}\right) \right) \right) ^{\prime }+o_P\left( 1\right) \overset%
	{d}{\longrightarrow }\left( \mathbb{G}\left( f\left( z_{1}\right) \right)
	,...,\mathbb{G}\left( f\left( z_{p}\right) \right) \right) ^{\prime }$ due
	to the weak convergence of $\mathbb{G}_{n}$.
\end{proof}

\section{Additional Tables of Simulation Results}\label{sec:add:sim:appendix}

In this part of the appendix, we collect the additional tables of the simulation results that are omitted from the main text.

\begin{table}[htbp]
\centering
\begin{threeparttable}
\caption{Simulation Results: Baseline Model $(T=N=200)$}\label{tb1:base}
\begin{tabular}{rrrr}
\\
  \hline
 & Mean Bias & RMSE & Coverage \\
  \hline
  \multicolumn{4}{l}{Scenario (i): \underline{Oracle}}\\
  $\bt_1$ & -0.0025 & 0.0427 & 0.948 \\
  $\bt_2$ & 0.0015 & 0.0383 & 0.947 \\
  $\dt_1$ & 0.0012 & 0.0749 & 0.962 \\
  $\dt_2$ & -0.0039 & 0.0678 & 0.959 \\
\\
  \multicolumn{4}{l}{Scenario (ii): \underline{Observed Factors/No Selection on $g_t$}}\\
  $\bt_1$ & -0.0033 & 0.0430 & 0.943 \\
  $\bt_2$ & 0.0013 & 0.0385 & 0.942 \\
  $\dt_1$ & 0.0042 & 0.0759 & 0.956 \\
  $\dt_2$ & -0.0027 & 0.0684 & 0.954 \\
  $\phi_2$ & 0.0002 & 0.0655 &  \\
  $\phi_4$ & -0.0011 & 0.0495 &  \\
  \multicolumn{3}{l}{Ave.\ Cor.\ Regime Prediction:}& 0.9929 (0.0074) \\
\\
  \multicolumn{4}{l}{Scenario (iii): \underline{Observed Factors/Selection on $g_t$}}\\
  $\bt_1$ &  -0.0034 & 0.0431 & 0.943 \\
  $\bt_2$ &  0.0013 & 0.0385 & 0.940 \\
  $\dt_1$ &  0.0045 & 0.0759 & 0.959 \\
  $\dt_2$ &  -0.0027 & 0.0685 & 0.954 \\
  $\phi_2$ & -0.0053 & 0.0646 &  \\
  $\phi_3$ & 0.0010 & 0.0110 &  \\
  $\phi_4$ & -0.0023 & 0.0526 &  \\
  \multicolumn{3}{l}{Ave.\ Cor.\ Regime Prediction:}& 0.9925 (0.0080) \\
  \multicolumn{3}{l}{Correct Factor Selection:}& 0.985  \\
\\
  \multicolumn{4}{l}{Scenario (iv): \underline{Unobserved Factors}}\\
  $\bt_1$ &  -0.0002 & 0.0435 & 0.945 \\
  $\bt_2$ &  0.0032 & 0.0391 & 0.940 \\
  $\dt_1$ &  -0.0062 & 0.0795 & 0.952 \\
  $\dt_2$ &  -0.0085 & 0.0702 & 0.957 \\
  $\gm_2$ &  -0.0003 & 0.5098 &  \\
  $\gm_3$ &  -0.0061 & 0.4977 &  \\
  $\gm_4$ &  -0.0061 & 0.3784 &  \\
  \multicolumn{3}{l}{Ave.\ Cor.\ Regime Prediction:}& 0.9799 (0.0122) \\
   \hline
\end{tabular}
\begin{tablenotes}
\item \footnotesize \emph{Note: }The average correct regime prediction (Ave.\ Cor.\ Regime Prediction) measures the average proportion such that the predicted regime of $1\{g_t'\widehat{\phi}>0\}$ (or $1\{f_t'\widehat{\gm}>0\}$ in (iv)) is equal to the true regime of $1\{g_t'\phi_0 > 0\}$ (or $1\{f_t'\gm_0>0\}$ in (iv)):
$
\widehat{E}\left(\frac{1}{T}\sum_{t=1}^T 1\left\{1\{g_t'\hat{\phi}>0\} = 1\{g_t'\phi_0>0\} \right\}\right),
$
where the expectation $\widehat{E}$ is taken over simulation draws.
\end{tablenotes}
\end{threeparttable}
\end{table}

\begin{table}[htbp]
\centering
\begin{threeparttable}
\caption{Unobserved Factors with Different $N$ Sizes}\label{tb2:diff_N}
\begin{tabular}{lrr}
  \hline
 & Mean Bias & RMSE  \\
  \hline
  \underline{$N=100$} \\
   $\bt_1$ & 0.0097 & 0.0473  \\
   $\bt_2$ & 0.0077 & 0.0407  \\
   $\dt_1$ & -0.0397 & 0.1015 \\
   $\dt_2$ & -0.0376 & 0.0939 \\
   $\gm_2 / \gm_1$ & 0.0016 & 0.0802   \\
  \multicolumn{2}{l}{Ave.\ Cor.\ Regime Prediction:}& 0.9741 (0.0133) \\
  \\
  \underline{$N=200$} \\
  $\bt_1$ &  0.0067 & 0.0462  \\
  $\bt_2$ &  0.0050 & 0.0386  \\
  $\dt_1$ &  -0.0252 & 0.0966 \\
  $\dt_2$ &  -0.0241 & 0.0850 \\
  $\gm_2 / \gm_1$ & -0.0014 & 0.0629   \\
  \multicolumn{2}{l}{Ave.\ Cor.\ Regime Prediction:}& 0.9821 (0.0107) \\
  \\
  \underline{$N=400$} \\
  $\bt_1$ & 0.0038 & 0.0460  \\
  $\bt_2$ & 0.0028 & 0.0379  \\
  $\dt_1$ & -0.0129 & 0.0880 \\
  $\dt_2$ & -0.0142 & 0.0795 \\
  $\gm_2 / \gm_1$ & -0.0010 & 0.0500 \\
  \multicolumn{2}{l}{Ave.\ Cor.\ Regime Prediction:}& 0.9870 (0.0087) \\
\\
  \underline{$N=1600$} \\
  $\bt_1$ &  0.0010 & 0.0443 \\
  $\bt_2$ & 0.0006  & 0.0373  \\
  $\dt_1$ & -0.0029 & 0.0851 \\
  $\dt_2$ & -0.0056 & 0.0759 \\
  $\gm_2 / \gm_1$ & 0.0011 & 0.0392  \\
  \multicolumn{2}{l}{Ave.\ Cor.\ Regime Prediction:}& 0.9934 (0.0062) \\
   \hline
\end{tabular}
\begin{tablenotes}
\item \footnotesize \emph{Note: }See the note under Table \ref{tb1:base} for the definition of {Ave.\ Cor.\ Regime Prediction}.
\end{tablenotes}

\end{threeparttable}
\end{table}

\begin{table}[htbp]
\centering
\begin{threeparttable}
\caption{Computation Time for Different Sample Sizes (unit=second)} \label{tb:comtime.T}
\begin{tabular}{llrrrr}
  \hline
         & Algorithm & T=200 & T=300 & T=400 & T=500 \\
  \hline
  \noalign{\vskip 2mm}
{Min}    & MIQP & 1.87 & 2.85 & 3.97 & 5.23 \\
  \noalign{\vskip 2mm}
{Median} & MIQP & 1.99 & 3.04 & 4.39 & 5.66 \\
  \noalign{\vskip 2mm}
  {Mean} & MIQP & 1.99 & 3.09 & 4.34 & 5.66 \\
  \noalign{\vskip 2mm}
{Max}    & MIQP & 2.21 & 3.69 & 4.73 & 6.07 \\
   \hline
\end{tabular}

\end{threeparttable}
\end{table}

\begin{table}[htbp]
\centering
\begin{threeparttable}
\caption{Computation Time for Different Sizes of $x_t$ (unit=second) }\label{tb:comtime.X}
\begin{tabular}{llrrrr}
  \hline
  & Algorithm      & $d_x=1$ & $d_x=2$ & $d_x=3$ & $d_x=4$ \\
  \hline
  \noalign{\vskip 2mm}
{Min}    & MIQP & 1.87 & 2.16 & 2.39 & 2.46 \\
  \noalign{\vskip 2mm}
{Median} & MIQP & 1.99 & 2.31 & 2.52 & 2.76 \\
  \noalign{\vskip 2mm}
  {Mean} & MIQP & 1.99 & 2.30 & 2.51 & 2.76 \\
  \noalign{\vskip 2mm}
{Max}    & MIQP & 2.21 & 2.54 & 2.85 & 3.06 \\
   \hline
\end{tabular}
\end{threeparttable}
\end{table}

\begin{table}[htp]
\centering
\begin{threeparttable}
\caption{Computation Time for Different Sizes of $g_t$ (unit=second) }\label{tb:comtime.G}
\begin{tabular}{llrrrr}
  \hline
    & Algorithm      & $d_g=2$ & $d_g=3$ & $d_g=4$ & $d_g=5$ \\
  \hline
  \noalign{\vskip 2mm}
{Min}
                        & MIQP & 1.87 & 2.04 & 4.79 & 78.78 \\
  \noalign{\vskip 2mm}
{Median}
                        & MIQP & 1.99 & 2.17 & 6.42 & 410.35 \\
  \noalign{\vskip 2mm}
{Mean}
                        & MIQP & 1.99 & 2.18 & 6.56 & 445.15 \\
  \noalign{\vskip 2mm}
{Max}
                        & MIQP & 2.21 & 2.38 & 9.68 & 1389.86 \\
   \hline
\end{tabular}
\end{threeparttable}
\end{table}

\begin{table}[thp]
\centering
\caption{Computation Time of Large Models: $T=500$}
\begin{tabular}{llrrr}
\hline
$T=500$   & Algorithm & $d_g=6$ & $d_g=8$ & $d_g=10$\\
\hline
{\underline{$d_x=6$}}  &        &         &         &          \\
\multirow{2}{*}{Min Time}    & BCD   & 611.34  & 610.69  & 610.42  \\
            & MIQP  & 1806.03 & 1806.08 & 1806.15 \\
\multirow{2}{*}{Median Time} & BCD   & 612.35  & 611.61  & 611.76  \\
            & MIQP  & 1806.21 & 1806.18 & 1806.22 \\
\multirow{2}{*}{Max Time}    & BCD   & 614.57  & 612.63  & 612.34  \\
            & MIQP  & 1806.77 & 1807.28 & 1806.48 \\
\multirow{2}{*}{Median Obj.} & BCD   & 0.23    & 0.24    & 0.24    \\
            & MIQP  & 0.25    & 0.3     & 0.29    \\
\multirow{2}{*}{Convergence} & BCD   & 1       & 1       & 1       \\
            & MIQP  & 0       & 0       & 0       \\
            &        &         &         &         \\
{\underline{$d_x=8$}}  &        &         &         &          \\
\multirow{2}{*}{Min Time}    & BCD   & 611.70   & 611.75  & 611.97  \\
            & MIQP  & 1806.58 & 1806.59 & 1806.7  \\
\multirow{2}{*}{Median Time} & BCD   & 613.49  & 612.03  & 640.21  \\
            & MIQP  & 1806.74 & 1806.74 & 1806.78 \\
\multirow{2}{*}{Max Time}    & BCD   & 614.59  & 612.69  & 1151.22 \\
            & MIQP  & 1807.32 & 1807.38 & 1807.51 \\
\multirow{2}{*}{Median Obj.} & BCD   & 0.24    & 0.24    & 0.59    \\
            & MIQP  & 0.27    & 0.37    & 1.24    \\
\multirow{2}{*}{Convergence} & BCD   & 1       & 1       & 1       \\
            & MIQP  & 0       & 0       & 0       \\
\\
{\underline{$d_x=10$}}  &        &         &         &          \\
\multirow{2}{*}{Min Time}    & BCD   & 613.51  & 612.74  & 612.84  \\
            & MIQP  & 1807.77 & 1807.79 & 1807.77 \\
\multirow{2}{*}{Median Time} & BCD   & 614.69  & 613.14  & 614.23  \\
            & MIQP  & 1807.84 & 1807.92 & 1807.84 \\
\multirow{2}{*}{Max Time}    & BCD   & 616.05  & 613.66  & 1693.54 \\
            & MIQP  & 1808.37 & 1808.32 & 1808.04 \\
\multirow{2}{*}{Median Obj.} & BCD   & 0.24    & 0.23    & 0.27    \\
            & MIQP  & 0.27    & 0.69    & 1.42    \\
\multirow{2}{*}{Convergence} & BCD   & 1       & 1       & 1       \\
            & MIQP  & 0       & 0       & 0       \\
          \hline
\end{tabular}
\end{table}

\begin{table}[thp]
\centering
\caption{Computation Time of Large Models: $T=1000$} \label{tb:com-large-1000}
\begin{tabular}{llrrr}
\hline
$T=1,000$   & Algorithm & $d_g=6$ & $d_g=8$ & $d_g=10$\\
\hline
{\underline{$d_x=6$}}        &        &         &         &          \\
{Min Time}    & BCD   &  628.65  &  625.61  &  627.10  \\
{Median Time} & BCD   &  635.58  &  645.17  &  634.25  \\
{Max Time}    & BCD   & 1344.56  & 1344.56  & 1163.16  \\
{Median Obj.} & BCD   &    0.25  &    0.26  &    0.24  \\
{Convergence} & BCD   &    1.00  &    1.00  &    1.00  \\
\hline
\end{tabular}
\end{table}

\clearpage

{\singlespacing

\bibliographystyle{economet}
\bibliography{LLSS_bib}
}

\end{document}